\documentclass[11pt]{ociamthesis}  % default square logo 
\usepackage{amsmath, amssymb, amsthm, amsfonts, mathtools, mathrsfs}

\usepackage{enumitem}

\usepackage{dsfont}

\usepackage[margin=3cm]{geometry}

\usepackage{cite}

\usepackage{fancyhdr}

\usepackage{setspace}

\usepackage{tikz}
\usetikzlibrary{calc,intersections}
\usetikzlibrary{shapes}

\usepackage{float}
\restylefloat{table}
\restylefloat{figure}

\usepackage{appendix}

\usepackage[colorinlistoftodos]{todonotes}

\usepackage[symbol]{footmisc}

\numberwithin{equation}{section}

\DeclareMathOperator{\sgn}{sgn}
\DeclareMathOperator{\Tr}{Tr}

\theoremstyle{plain}
\newtheorem{theorem}{Theorem}
\newtheorem{proposition}[theorem]{Proposition}
\newtheorem{lemma}[theorem]{Lemma}
\newtheorem{corollary}[theorem]{Corollary}
\newtheorem{definition}[theorem]{Definition}
\newtheorem{example}[theorem]{Example}
\newtheorem{remark}[theorem]{Remark}

\usepackage[utf8]{inputenc} % allow utf-8 input
\usepackage[T1]{fontenc}    % use 8-bit T1 fonts
\usepackage{hyperref}       % hyperlinks
\usepackage{url}            % simple URL typesetting
\usepackage{booktabs}       % professional-quality tables
\usepackage{amsfonts}       % blackboard math symbols
\usepackage{nicefrac}       % compact symbols for 1/2, etc.
\usepackage{microtype}      % microtypography

\title{Fisher-Hartwig Asymptotics \\[1ex] 
        and Log-Correlated Fields \\[1ex]
        in Random Matrix Theory}   %note \\[1ex] is a line break in the title

\author{Johannes Forkel}             %your name
\college{The Queen's College}        %your college

\degree{Doctor of Philosophy}     %the degree
\degreedate{Trinity 2023}         %the degree date

\begin{document}

%this baselineskip gives sufficient line spacing for an examiner to easily
%markup the thesis with comments
\baselineskip=16pt plus1pt

%set the number of sectioning levels that get number and appear in the contents
\setcounter{secnumdepth}{3}
\setcounter{tocdepth}{3}

\maketitle                  % create a title page from the preamble info
\begin{dedication}

This thesis is dedicated to my beloved parents.

\end{dedication}        % include a dedication.tex file
\begin{acknowledgements}

During the past four years of my DPhil I was supported by many people in various ways.

First of all I want to thank my supervisor Jon Keating, for making it possible for me to do my DPhil in Oxford, for suggesting interesting problems for me to work on, for his encouragement at all stages of my DPhil, for supporting me also in exploring areas that are only remotely related to the topic of my DPhil, like attending a machine learning conference and doing an internship in finance, and for supporting me in many other ways.

Next, I want to thank Tom Claeys, Mo Dick Wong, and Isao Sauzedde, for the very pleasant collaborations, and for greatly contributing to my understanding of different areas of mathematics. In particular I want to thank Tom Claeys for inviting me to visit Louvain-la-Neuve, an enriching experience I thoroughly enjoyed.

I also want to thank my examiners Ben Hambly and Igor Krasovsky, for their careful reading of this thesis and for their valuable comments and suggestions.

Throughout the past four years I have had countless enlightening conversations about mathematics with many people, which I am very grateful for. In particular I want to thank Theo Assiotis, who in the first year of my DPhil helped me to get oriented in the vast field of random matrix theory, and Brian Conrey, with whom I had inspiring discussions at the American Institute of Mathematics.

I also want to thank all of my friends and colleagues who've made my time in Oxford so enjoyable. In particular I want to thank, in no specific order, Alexis Chevalier,  Andreas Heilmann,  Anna Glieden,  Annina Lieberherr,  Asha Vettour,  Axel Schmid,  Camila Ramos,  Constantin Kogler,  David Seiferth,  Elle Yang,  Estelle Paulus,  Francis Aznaran,  Hannah Hughes,  Hans Chan,  Haonan Xu,  Håvard Damm-Johnsen,  Jessica Ceynar,  Kirsten Smith,  Konrad Marx,  Kyung Chan Park,  Lorenz Olbrich,  Maxime Kayser,  Nils Fitzian,  Samuel Fink,  Saoud Al-Khuzaei,  Shayan Hundrieser (also for his helpful comments on a draft of this thesis),  Weilin Wang,  Wissam Ghantous, and all members of the Queen's College Boat Club.

Finally, I want to express my profound gratitude to my family for their unwavering support throughout this journey.
\end{acknowledgements}   % include an acknowledgements.tex file
\begin{abstract}
    
This thesis is concerned with establishing and studying connections between random matrices and log-correlated fields. This is done with the help of formulae, including some newly established ones, for the asymptotics of Toeplitz, and Toeplitz+Hankel determinants with Fisher-Hartwig singularities.
    
In Chapter 1, we give an introduction to the mathematical objects that we are interested in. In particular we explain the relations between the characteristic polynomial of random matrices, log-correlated fields, Gaussian multiplicative chaos, the moments of moments, and Toeplitz and Toeplitz+Hankel determinants with Fisher-Hartwig singularities. Chapter 1 is based on joint work with Jon Keating \cite{Forkel2021}, Tom Claeys and Jon Keating \cite{CFK23}, and Isao Sauzedde \cite{FS22}.

In Chapter 2, based on joint work with Jon Keating \cite{Forkel2021}, we establish formulae for the asymptotics of Toeplitz, and Toeplitz+Hankel determinants with two complex conjugate pairs of merging Fisher-Hartwig singularities. We prove those formulae using Riemann-Hilbert techniques, which are heavily inspired by the ones used by Deift, Its, and Krasovsky \cite{Deift2011}, and Claeys and Krasovsky \cite{Claeys2015}.

In Chapter 3, based on joint work with Jon Keating \cite{Forkel2021}, we complete the connection between the classical compact groups and Gaussian multiplicative chaos, by showing that analogously to the case of the unitary group first established by Webb \cite{Webb2015}, the characteristic polynomial of random orthogonal and symplectic matrices, when properly normalized, converges to a Gaussian multiplicative chaos measure on the unit circle. 

In Chapter 4, based on joint work with Tom Claeys and Jon Keating \cite{CFK23}, we compute the asymptotics of the moments of moments of random orthogonal and symplectic matrices, which can be expressed in terms of integrals over Toeplitz+Hankel determinants. The phase transitions we observe are in stark contrast to the ones proven for the unitary group by Fahs \cite{Fahs2021}. 

In Chapter 5, based on joint work with Isao Sauzedde \cite{FS22}, we establish convergence in Sobolev spaces, of the logarithm of the characteristic polynomial of unitary Brownian motion to the Gaussian free field on the cylinder, thus proving the dynamical analogue of the classical stationary result by Hughes, Keating and O'Connell \cite{Hughes2005}. 

\end{abstract}          % include the abstract
\begin{originalitylong}
    This thesis is a result of my own efforts, except where indicated otherwise.
\end{originalitylong}

\begin{romanpages}          % start roman page numbering
\tableofcontents            % generate and include a table of contents
%\listoffigures              % generate and include a list of figures
\end{romanpages}            % end roman page numbering

\chapter{Introduction} \label{chapter:Intro}

Characteristic polynomials of large random matrices are fundamental objects in random matrix theory. Besides encoding the eigenvalues of the random matrices, their statistics show remarkable similarities with those of the Riemann zeta function and other number-theoretic $L$-functions on the critical line \cite{Conrey2005, Fyodorov2012, FK14, Keating2000b, Keating2000c, Saksman2020, Arguin2019} - see \cite{Bailey2022} for a review. They are also closely connected to the theory of log-correlated fields and to Gaussian multiplicative chaos, a connection that we are concerned with in this thesis. See for example \cite{Spohn1998, Hughes2005, Bourgade2022, Webb2015, Nikula2020, Forkel2021, Chhaibi2018a, PZ18, PZ22, Najnudel2023, Lambert2021, KW22} for results on the circular ensembles, \cite{Joh98, Berestycki2018, Claeys2019, Lam20} for Hermitian ensembles and the complex Ginibre ensemble, among others. For many classical random matrix ensembles the limit, as the matrix size goes to infinity, of the logarithm of the characteristic polynomial, has been identified to be a log-correlated Gaussian field, i.e. a Gaussian random generalized function whose covariance kernel blows up logarithmically on the diagonal. Such fields cannot be defined pointwise as their value at each point would be a Gaussian with infinite variance, thus they need to be understood either as stochastic processes indexed by test functions, or as random elements in a space of generalized functions, for example Sobolev spaces of negativ regularity. However log-correlated fields can still be exponentiated in a certain sense, which gives rise to a family of random measures, called \textit{Gaussian multiplicative chaos} measures. For various random matrix ensembles it has been shown that powers of the the exponential of the real and imaginary part of the logarithm of the characteristic polynomial, after an appropriate normalization procedure, converge to such Gaussian multiplicative chaos measures, as the matrix size goes to infinity. \\

The log-correlated fields and Gaussian multiplicative chaos measures arising from the characteristic polynomials of random matrices are also closely related to the so called \textit{moments of moments} of the characteristic polynomial, which have received a lot of attention recently \cite{Fyodorov2012, FK14, Claeys2015, Fyodorov2018, Bailey2019, Assiotis2019, Assiotis2020, BarhoumiAndreani2020, Fahs2021, Andrade2022, Bailey2022, KW22, CFK23}. The term {\em moments of moments} refers to the fact that one first takes a moment of the characteristic polynomial with respect to the spectral variable, and then a moment with respect to the random matrix distribution. The moments of moments can be viewed as the moments of the partition function of the random energy landscape given by the logarithm of the characteristic polynomial, thus their asymptotics can be interpreted as the moments of the partition function of the limiting log-correlated field. Furthermore they are related to the moments of the total mass of the Gaussian multiplicative chaos measures associated to the characteristic polynomial. \\

In this thesis we prove results regarding log-correlated fields, Gaussian multiplicative chaos, and the moments of moments in random matrix theory. In our proofs we often need to evaluate certain expectations over random matrix ensembles which can be expressed in terms of Toeplitz, Hankel, or Toeplitz+Hankel determinants with Fisher-Hartwig singularities. The asymptotics of such determinants have been studied for a long time and with various techniques (see the references in Section \ref{section:FH} and Chapter \ref{chapter:FH}), and we will prove certain new asymptotic formulae for Toeplitz and Toeplitz+Hankel determinants which we will then use in the proof of our results regarding Gaussian multiplicative chaos.\\

%This thesis is based on the following three papers: "The Classical Compact Groups and Gaussian Multiplicative Chaos" \cite{Forkel2021}, which is joint work with Jon Keating. "Moments of moments of the characteristic polynomials of random orthogonal and symplectic matrices" \cite{Claeys2022}, which is joint work with Tom Claeys and Jon Keating. "Convergence of the logarithm of the characteristic polynomial of unitary Brownian motion in Sobolev space" \cite{FS22}, which is joint work with Isao Sauzedde.\\

In the following sections of this introductory chapter we define all the mathematical objects that we are concerned with, explain how they are related to each other, point out the results that we contributed, and then provide an outline to the rest of the thesis. We start by recalling some basic properties of random matrices from the classical compact groups. 

\section{Random Matrices from the Classical Compact Groups} \label{section:CCG}
In this section we follow Sections 1 and 3 of \cite{Meckes2019} and recall the definitions of the classical compact groups and state basic properties of Haar distributed matrices from those groups. We then follow \cite{FS22} for a defintion of unitary Brownian motion.
 
\begin{definition} \label{def:groups}
\begin{enumerate}
\item An $n \times n$ matrix $U$ over $\mathbb{R}$ is \textbf{orthogonal} if
\begin{equation}
UU^T = U^TU = I,
\end{equation}
where $I$ denotes the identity matrix, and $U^T$ is the transpose of $U$. The group of $n \times n$ orthogonal matrices over $\mathbb{R}$ is denoted $O(n)$. 

\item The set $SO(n) \subset O(n)$ of \textbf{special orthogonal matrices} is defined by 
\begin{equation}
SO(n):= \left\{ U \in O(n): \det (U) = 1 \right\}.
\end{equation}

\item The set $SO^-(n) \subset O(n)$ (the \textbf{negative coset}) is defined by 
\begin{equation}
SO^-(n):= \left\{ U \in O(n): \det (U) = -1 \right\}.
\end{equation}

\item An $n \times n$ matrix $U$ over $\mathbb{C}$ is \textbf{unitary} if
\begin{equation}
UU^* = U^*U = I,
\end{equation}
where $U^*$ is the conjugate transpose of $U$. The group of $n \times n$ unitary matrices over $\mathbb{C}$ is denoted $U(n)$.

\item 
The set $SU(n) \subset U(n)$ of \textbf{special unitary matrices} is defined by 
\begin{equation}
SU(n):= \left\{ U \in U(n): \det (U) = 1 \right\}.
\end{equation}

\item An $2n \times 2n$ matrix $U$ over $\mathbb{C}$ is \textbf{symplectic} if $U \in U(2n)$ and 
\begin{equation}
UJU^T = U^TJU = J,
\end{equation}
where
\begin{equation}
J = \left( \begin{array}{cc} 0 & I \\ -I & 0 \end{array} \right).
\end{equation}
The group of $2n \times 2n$ symplectic matrices over $\mathbb{C}$ is denoted $Sp(2n)$. Note that $Sp(2n)$, defined like this, is isomorphic to the group of $n \times n$ matrices $U$ with entries in the quaternions $\mathbb{H}$, that satisfy $U^*U = I$, where $U^*$ denotes the quaternionic conjugate transpose of $U$ (see \cite{Meckes2019}[p.12] for details). Recall that $\mathbb{H} = \{a + bi + cj + dk: \, a, b, c, d \in \mathbb{R} \}$, with $i,j,k$ satisfying the relation $i^2 = j^2 = k^2 = ijk = -1$, and that quaternionic conjugation is defined by 
\begin{align}
    \overline{a + bi + cj + dk} = a - bi - cj - dk.
\end{align}
Thus $Sp(2n)$ can also be defined as the isometry group of $\mathbb{H}^n$.
\end{enumerate}
\end{definition}

Let $G$ denote any of $O(n)$, $SO(n)$, $U(n)$, $SU(n)$ or $Sp(2n)$. We want to define a natural probability measure on the Borel $\sigma$-algebra of $G$. On the unit circle $U(1)$ the most natural probability measure is the uniform measure, and it can be defined as the unique rotation-invariant probability measure on $U(1)$. Similarly we want a probability measure $\mu$ on $G$ that is \textit{translation invariant}; that is, for any measurable subset $A \subset G$ and fixed $M \in G$ it holds that 
\begin{equation}
\mu(MA) = \mu(AM) = \mu(A),
\end{equation}
where $MA:= \{ MU:U \in A\}$ and $AM:= \{ UM:U \in A\}$.   

The following theorem (which is not stated here in its most general form), due to Alfr\'{e}d Haar \cite{Haa33}, states that such a measure on $G$ exists and that it is unique:

\begin{theorem} [\cite{Haa33}]
Let $G$ be any of $O(n)$, $SO(n)$, $U(n)$, $SU(n)$ or $Sp(2n)$. Then there is a unique translation-invariant probability measure (called \textbf{Haar measure}) on $G$.
\end{theorem} 

By conditioning the Haar measure on $O(n)$ onto $SO^-(n)$ we also get a natural probability measure on $SO^-(n)$ (even though it is not a group). In the following all random matrices from the sets in Definition \ref{def:groups} are distributed according to (conditioned) Haar measure.\\

It is easy to see that the Riemannian volume measure on any of the above groups is translation invariant, thus it is Haar measure. Another way to understand Haar measure, which is useful for simulations, is to sample a matrix with i.i.d. Gaussian entries (for $U(n)$ the real and imaginary part are sampled as i.i.d. real Gaussians and similarly for the real and $i,j,k$ parts in the case $Sp(2n)$), and then perform the Gram-Schmidt algorithm on that matrix. In the quaternionic case one needs to then map to the $2n \times 2n$ complex matrix version of $Sp(2n)$. To sample from $SO(n)$, $SO^-(n)$, or $SU(n)$ one needs in an extra step to multiply the final column by the necessary scalar to ensure that the matrix has the desired determinant. \\

Since matrices from $O(n)$, $U(n)$ and $Sp(2n)$ all act as isometries of $\mathbb{C}^n$, all their eigenvalues lie on the unit circle $S^1 \subset \mathbb{C}$. In the orthogonal and symplectic cases there are some further restrictions, which are stated in the following lemma which is easy to prove:

\begin{lemma} 
Each $U \in SO(2n+1)$ has $+1$ as an eigenvalue, each $U \in SO^-(2n+1)$ has $-1$ as an eigenvalue and each $U \in SO^-(2n+2)$ has $+1$ and $-1$ as eigenvalues. The remaining $2n$ non-fixed eigenvalues of matrices in $SO(2n+1)$, $SO^-(2n+1)$ and $SO^-(2n+2)$, as well as the $2n$ eigenvalues of matrices in $SO(2n)$ and $Sp(2n)$ appear in complex conjugate pairs.
\end{lemma}

Denote the $2n$ non-fixed eigenvalues of matrices in $SO(2n+1)$, $SO^-(2n+1)$, $SO(2n)$, $SO^-(2n+2)$ and $Sp(2n)$, by $e^{i\theta_1},...,e^{i\theta_n}, e^{-i\theta_1},...,e^{-i\theta_n}$, with $\theta_1,...,\theta_n \in [0,\pi]$. The $n$ eigenvalue angles $\theta_1,...,\theta_n$, corresponding to the eigenvalues in the upper half circle are called the \textbf{nontrivial eigenangles}. For matrices in $U(n)$ all the eigenvalue angles are considered nontrivial, as there are no automatic symmetries in this case. 

Explicit formulae for the joint distribution of the non-trivial eigenvalues of Haar-distributed random matrices from any of the above groups have been found by Weyl. They are known as the Weyl integration formulae for those groups and are given in the following two theorems:  
  
\begin{theorem} [Weyl] The unordered eigenvalues of an $n \times n$ random unitary matrix have eigenvalue density
\begin{equation} \label{thm:Weyl U}
\frac{1}{n!} \prod_{1 \leq k < j \leq n} |e^{i\theta_k} - e^{i\theta_j}|^2 \prod_{k = 1}^n \frac{\text{d}\theta_k}{2\pi},
\end{equation} 
with respect to $\text{d}\theta_1,...,\text{d}\theta_n$ on $[0,2\pi)^n$. That is, for any $g:U(n) \rightarrow \mathbb{R}$ with 
\begin{equation}
g(U) = G(VUV^*) \quad \text{for any } U,V \in U(n),
\end{equation}
(i.e. $g$ is a class function and only depends on the unordered eigenvalues), if $U$ is Haar-distributed on $U(n)$, then 
\begin{equation}
\mathbb{E}(g(U)) = \frac{1}{n!} \int_{[0,2\pi)^n} \tilde{g}(\theta_1,...,\theta_n) \prod_{1 \leq k < j \leq n} |e^{i\theta_k} - e^{i\theta_j}|^2 \prod_{k = 1}^n \frac{\text{d}\theta_k}{2\pi},
\end{equation} 
where $\tilde{g}:[0,2\pi)^n \rightarrow \mathbb{R}$ is the (necessarily symmetric) expression of $g(U)$ as a function of the eigenvalues of $U$.
\end{theorem}

\begin{theorem} [Weyl]
Let $U$ be a Haar-distributed random matrix in $G$, where $G$ denotes either $SO(2n+1)$, $SO^-(2n+1)$, $SO(2n)$, $SO^-(2n+2)$ and $Sp(2n)$. Then a function $g$ of $U$ for which $g(VUV^T) = g(U)$ for all $V \in G$, is associated with a function $\tilde{g}:[0,\pi)^n \rightarrow \mathbb{R}$ (of the unordered nontrivial eigenangles) which is invariant under permutation of coordinates, and if $U$ is Haar-distributed on $G$ then 
\begin{equation}
\mathbb{E}\left(g(U)\right) = \int_{[0,\pi)^n} \tilde{g} \mu_{G}^W,
\end{equation}
where the measures $\mu_{G}^W$ on $[0,\pi)^n$ have densities with respect to $\text{d}\theta_1,...,\text{d}\theta_n$ as follows:
\begin{center}
\def\arraystretch{3}
\begin{tabular}{ll}
\hline
$G$ & $\mu_G^W$ \\ 
\hline 
\hline
$SO(2n)$ & $\frac{2}{n!(2\pi)^n} \prod_{1 \leq j < k \leq n} (2\cos(\theta_k) - 2\cos(\theta_j))^2$ \\
\hline 
$SO(2n+1)$ & $\frac{2^n}{n!\pi^n} \prod_{1 \leq j \leq n} \sin^2\left(\frac{\theta_j}{2}\right) \prod_{1 \leq j < k \leq n} \left(2\cos(\theta_k) - 2\cos(\theta_j)\right)^2$ \\
\hline 
$SO^-(2n+1)$ & $\frac{2^n}{n!\pi^n} \prod_{1 \leq j \leq n} \cos^2\left(\frac{\theta_j}{2}\right) \prod_{1 \leq j < k \leq n} \left(2\cos(\theta_k) - 2\cos(\theta_j)\right)^2$ \\
\hline
$Sp(2n)$ & $\frac{2^n}{n!\pi^n} \prod_{1 \leq j \leq n} \sin^2(\theta_j) \prod_{1 \leq j < k \leq n} \left(2\cos(\theta_k) - 2\cos(\theta_j)\right)^2$ \\
\hline
\end{tabular}
\end{center}
\end{theorem}

One can see that the eigenvalues repel each other, i.e. they are very unlikely to be close to each other. The factors $\sin^2\left(\frac{\theta_j}{2} \right)$, $\cos^2\left(\frac{\theta_j}{2}\right)$, and $\sin^2 (\theta_j)$ are due to the deterministic eigenvalues at $\pm 1$. Figure \ref{fig:eigenvalues} shows samples of eigenvalues from the classical compact groups, as well as points sampled independently on the unit circle. Due to the independent points not repelling each other they are much more clustered, and large gaps are much more common. \\

The following famous theorem relates the moments of traces of powers of random unitary matrices to the moments of complex Gaussians, and we will need it at various points in the following sections and chapters:

\begin{theorem}(Diaconis, Shahshahani \cite{Diaconis1994} and Diaconis, Evans \cite{Diaconis2001}) \label{thm:DiaconisShahshahani}
    Let $U$ be a Haar-distributed random matrix in $U(n)$, $\ell\in \mathbb{N}$, and let $Z_1,...,Z_l$ be i.i.d. standard complex Gaussian random variables, i.e. their real and imaginary parts are independent real Gaussians with variance $1/2$. Let $a_k, b_k \in \mathbb{N}$, $k = 1,...,l$, and let $n \in \mathbb{N}$ be such that
    \begin{align}
        \max \left\{ \sum_{j = 1}^k ja_j, \sum_{j = 1}^k jb_j \right\} \leq n.
    \end{align}
    Then
    \begin{align}
        \mathbb{E} \left( \prod_{k = 1}^\ell (\text{Tr}(U^k))^{a_k} \overline{(\text{Tr}(U^k))^{b_k}} \right) = \prod_{k = 1}^\ell \delta_{a_k b_k} k^{a_k} a_k! = \mathbb{E} \left( \prod_{k = 1}^\ell (\sqrt{k} Z_k)^{a_k} \overline{(\sqrt{k} Z_k)^{b_k}} \right).
    \end{align}
\end{theorem}
Note that from Theorem \ref{thm:DiaconisShahshahani} it follows that for any $\ell \in \mathbb{N}$ the random vector $(\Tr (U^k))_{k = 1}^\ell$ converges to $(\sqrt{k}Z_k)_{k = 1}^\ell$ as $n \rightarrow \infty$. Note that no normalization is needed for the traces of powers to converge, which is due the repulsion between the eigenvalues.\\

\begin{figure} \label{fig:eigenvalues}
    \centering
    \includegraphics[scale=0.66]{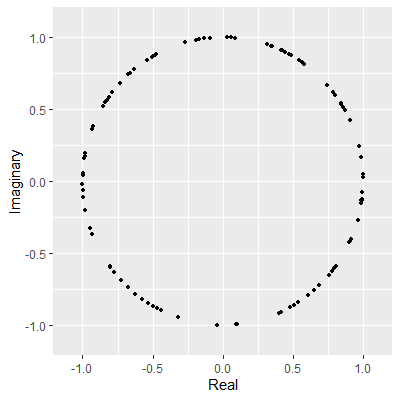}
    \includegraphics[scale=0.66]{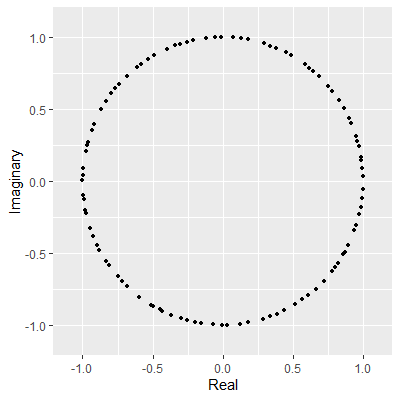}
    \includegraphics[scale=0.66]{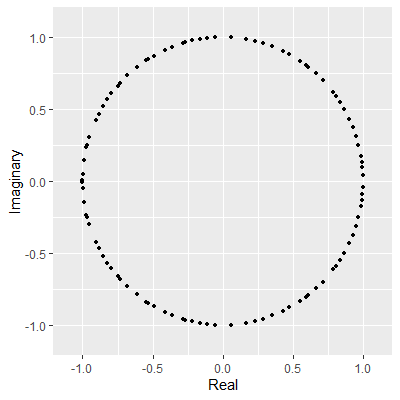}
    \includegraphics[scale=0.66]{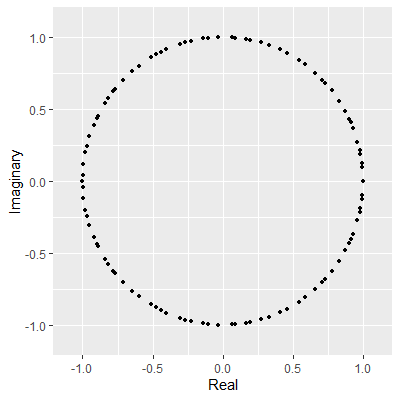}
    \caption{Top left: 100 i.i.d. points drawn uniformly. Top right: Eigenvalues of a random $U \in U(100)$. Bottom left: Eigenvalues of a random $U \in SO(100)$. Bottom right: Eigenvalues of a random $U \in SO^-(100)$.}
\end{figure}

When instead of looking at a random variable taking values in the unitary group one looks at a stochastic process whose state space is the unitary group, the natural dynamic analogue to Haar distributed random unitary matrices is Brownian motion on the unitary group, which is usually referred to as \textit{unitary Brownian motion}. We define unitary Brownian motion $(U_t)_{t \geq 0}$ on $U(n)$ as the diffusion governed by the stochastic differential equation 
\begin{align} \label{eqn:UBM SDE}
    \text{d}U_t = \sqrt{2} U_t \text{d}B_n(t) - U_t \text{d}t,
\end{align}
with $(B_n(t))_{t \geq 0}$ denoting a Brownian motion on the space of $n \times n$ skew-Hermitian matrices. That is
\begin{align}
    B_n(t) = \sum_{k = 1}^{n^2} X_k \tilde{B}^{(k)}(t),
\end{align}
where $\tilde{B}^{(k)}$, $k = 1,...,n^2$, are independent one-dimensional standard Brownian motions, and where the matrices $X_k$, $k = 1,...,n^2$, are an orthonormal basis of the real vector space of skew-Hermitian matrices w.r.t. the scalar product $\langle A, B \rangle := n \Tr (AB^*)$. One such basis is given by the matrices $\frac{1}{\sqrt{2n}}(E_{k,l} - E_{l,k})$, $\frac{i}{\sqrt{2n}}(E_{k,l} + E_{l,k})$, $1 \leq k < \ell \leq n$, and $\frac{i}{\sqrt{n}} E_{k,k}$, $1 \leq k \leq n$, where $(E_{k,l})_{mn} = \delta_{km} \delta_{ln}$.

\begin{remark}
    Unitary Brownian motion is usually defined using a different normalisation, i.e. satisfying the SDE $\text{d} \tilde{U}_t = \tilde{U}_t \text{d} B_n(t) - \frac{1}{2}\tilde{U}_t \text{d}t$. With this normalisation the generator is given by one half times the Laplacian on $U(n)$, which is the usual definition of Brownian motion on a Riemannian manifold. The relation between the two normalisations is $\tilde{U}_{2t} = U_t$. We chose our normalisation to be consistent with the work of Spohn \cite{Spohn1998}, and Bourgade and Falconet \cite{Bourgade2022}, to which our work in Chapter \ref{chapter:UBM} is strongly related.
\end{remark}

The study of unitary Brownian motion was initiated by Dyson \cite{Dyson1962}, who showed that the eigenangles $\theta_1,...,\theta_n \in [0,2\pi)$ of unitary Brownian motion satisfy the stochastic differential equation
\begin{equation} \label{eqn:UBM eigenvalues}
    \text{d}\theta_j(t) = \frac{1}{n} \sum_{k \neq j} \cot \left( (\theta_j(t) - \theta_k(t))/2 \right) \text{d}t + \sqrt{\frac{2}{n}} \text{d} B^{(j)}(t),
\end{equation}
where $B^{(j)}$, $j = 1,...,n$, are one-dimensional standard Brownian motions. In the same work he also showed that analogously the eigenvalues of a Brownian motion on the set of Hermitian matrices satisfy a similar stochastic differential equation. \\

Since the normalized Riemannian volume measure is the stationary distribution of Brownian motion on any compact Riemannian manifold, and since the normalized Riemannian volume measure on $U(n)$ equals Haar measure, we see that Haar measure on $U(n)$ is the stationary distribution of Unitary Brownian motion. Thus unitary Brownian motion started at Haar measure is reversible and can in turn be defined for all $t \in \mathbb{R}$. In this thesis we always consider unitary Brownian motion at its equilibrium, so that $U_t$ is Haar distributed for all $t \in \mathbb{R}$. This is the stationary setting we consider in Chapter \ref{chapter:UBM}.

\section{Log-correlated fields}
%\todo{perhaps explain how log-correlated fields arise from the empirical spectral measure, linear statistics, to suspect an underlying Gaussian field}
We say a field $h$ on some metric space $(D,d)$ is a centered log-correlated Gaussian field if its covariance kernel $K$ has a logarithmic singularity on the diagonal, i.e. 
\begin{align}
    K(x,y) = \mathbb{E}\left( h(x) h(y) \right) = - \log d (x,y) + \mathcal{O}(1), \quad \text{as } x \rightarrow y.
\end{align} 
Note that due to the logarithmic singularity on the diagonal such a field cannot be understood as a function, as its value at any point would be a Gaussian with infinite variance. Thus it has to be interpreted as a centered generalised Gaussian field, which we define as follows: a centered generalized Gaussian field $h$ with kernel $K:D \times D \rightarrow \mathbb{R} \cup \{\pm \infty \}$ on a metric space $(D,d)$ with Radon measure $\sigma$, is a centered Gaussian process $h = \left( h(f) \right)_{f \in \mathcal{F}}$, where the index set $\mathcal{F}$ is some family of test functions (e.g. smooth functions with compact support in the case of $D$ being a smooth manifold), whose covariance function is given by  
\begin{align} \label{eqn:Kernel Gaussian field}
    \mathbb{E}\left( h(f)h(g) \right) = \langle f, g\rangle_K := \int_D \int_D f(x)g(y) K(x,y) \sigma(\text{d}x) \sigma(\text{d}y).
\end{align}
For any kernel $K$ which is positive definite on $\mathcal{F}$ in the sense that the r.h.s. of (\ref{eqn:Kernel Gaussian field}) is non-negative and finite for $f = g$, for all $f \in \mathcal{F}$, such a field exists due to Kolmogorov's extension theorem.\\

In the case that $(\mathcal{F}, \langle \cdot, \cdot \rangle_K)$ is (or can be extended to) a Hilbert space, which is always the case for the fields considered in this thesis, the Gaussian process $(h(f))_{f \in \mathcal{F}}$ becomes a Gaussian Hilbert space, i.e. a closed subspace of $L^2$ consisting of centered Gaussian variables (see \cite{Jan97} for more details on Gaussian Hilbert spaces). One can choose an ordered orthonormal basis $(\phi_j)_{j \in \mathbb{N}}$ of $\mathcal{F}$ (for simplicity we here assume $\mathcal{F}$ to be infinite-dimensional but separable), and interpret $h$ as the formal Fourier series (or "infinite-dimensional standard Gaussian")
\begin{align}
    h := \sum_{j = 1}^\infty \mathcal{N}_j \phi_j,
\end{align}
where $\mathcal{N}_j$, $i \in \mathcal{I}$, is a family of i.i.d. standard Gaussians. Defined like this, $h$ is almost surely not an element of $\mathcal{F}$, as 
\begin{align}
    \langle h,h\rangle_K = \sum_{j = 1}^\infty \mathcal{N}_j^2  = \infty \quad \text{almost surely}.
\end{align}
However when choosing an $f \in \mathcal{F}$ and setting 
\begin{align}
    h(f) := \langle h,f\rangle_K = \sum_{j = 1}^\infty \langle \phi_j, f\rangle_K \mathcal{N}_j,
\end{align}
one can see that this sum converges in $L^2$ and almost surely, and that defined like this the field $h$ has the right covariance structure: 
\begin{align}
    \mathbb{E}\left( h(f)h(g) \right) = \sum_{j = 1}^\infty \langle \phi_j, f \rangle_K \langle \phi_j, g \rangle_K = \langle f, g \rangle_K. 
\end{align}
%One natural choice of orthonormal basis for this formal Fourier series %representation of $h$ is to use an appropriate generalized version of %Mercer's theorem to write $K$ as 
%\begin{align}
%    K(x,y) = \sum_{k = 1}^{\infty} \lambda_k \phi_k(x) \phi_k(y),
%\end{align}
%where $\phi_k$, $k \in \mathbb{N}$, are the continuous eigenfunctions %to the non-zero eigenvalues of the integral operator 

We now explain in which way logarithmically correlated fields appear as the limit of the logarithm of the characteristic polynomial of random matrices. Let $G(n) \in \left\{ U(n), \, O(n), \, SO(n), \, SO^-(n), \, Sp(2n) \right\}$. We define the characteristic polynomial of a random $U \in G(n)$ as (note that for $G(n) = Sp(2n)$ the product goes up to $2n$)
\begin{equation}
p_{G(n)}(\theta) := \text{det} \left(I-e^{-i\theta}U\right) = \prod_{k = 1}^n (1-e^{i(\theta_k -\theta)}), \quad \theta \in [0,2\pi).
\end{equation}
Defined like this $p_{G(n)}(\cdot)$ is a random function on $[0,2\pi)$, and its zeroes are the phases $\theta_k$, $k = 1,...,n$ of the eigenvalues of $U$. Further we define (note again that for $G(n) = Sp(2n)$ the sum goes up to $2n$)
\begin{align}
\begin{split}
\log p_{G(n)}(\theta) := & \sum_{k = 1}^n \log (1-e^{i(\theta_k -\theta)}), 
\end{split}
\end{align}
with the branches on the RHS being the principal branches, such that for $\phi, \theta \in [0,2\pi)$ it holds that
\begin{equation}
\Im \log (1-e^{i(\phi - \theta)}) =  \begin{cases} - \frac{\pi}{2} + \frac{\phi-\theta}{2} & 0 \leq \theta < \phi < 2\pi \\ \frac{\pi}{2} + \frac{\phi-\theta}{2} & 0 \leq \phi \leq \theta < 2\pi \end{cases} \in \left( - \frac{\pi}{2}, \frac{\pi}{2} \right],
\end{equation}
with $\Im \log 0 := \pi/2$. Since
\begin{equation}
\log (1 - z) = - \sum_{k = 1}^\infty \frac{z^k}{k}
\end{equation}
for $|z| \leq 1$, where for $z = 1$ both sides equal $-\infty$, and by the identity $\log \det = \Tr \log$, we see that the Fourier expansion of $\log p_n$ is given as follows:
\begin{align}
\begin{split}
\log p_{G(n)}(\theta) =& - \sum_{k = 1}^{\infty} \frac{\Tr (U^k)}{k} e^{-ik\theta}.
\end{split}
\end{align}
By Theorem \ref{thm:DiaconisShahshahani} it follows that for any $\ell \in \mathbb{N}$ the random variables $k^{-1}\Tr (U^{k})$, $k = 1,...,l$, converge in distribution to independent complex Gaussians $A_k$, $k = 1,...,l$, whose real and imaginary parts are independent centered real Gaussians with variance $1/(2k)$. Thus we see that $\lim_{n \rightarrow \infty} \log p_{G(n)}(\theta)$ cannot exist for fixed $\theta$, as it would have to be a Gaussian with infinite variance. In fact, Hughes, Keating and O'Connell \cite{Hughes2005} proved that for any fixed $\theta$ it holds that 
\begin{align}
    \frac{\log p_{U(n)}(\theta)}{\sqrt{\log n}} \overset{D}{\implies} \mathcal{N}^{\mathbb{C}}(0,1).
\end{align}
In the same work Hughes, Keating and O'Connell also proved though that without any normalization $\log p_{U(n)}(\cdot)$ has a limit, however only when considered as a random variable taking values in the space of generalised functions $H_0^{-\epsilon}(S^1)$, $\epsilon > 0$, where for $s \in \mathbb{R}$ the space $H_0^{s}(S^1)$ is defined as
\begin{align}
    H_0^s(S^1) = \left\{ f(\theta) = \sum_{k \in \mathbb{Z}} f_k e^{ik\theta} : \sum_{k \in \mathbb{Z}} k^{2s} |f_k|^2 < \infty, f_0=0 \right\}.
\end{align}
The spaces $H_0^s(S^1)$ consist of (formal for $s < 0$) Fourier series whose Fourier coefficents decay at certain rates, and when endowed with the scalar products $\langle f, g \rangle_s := \sum_{k \in \mathbb{Z}\setminus \{0\}} f_k \overline{g_k}$, they are Hilbert spaces. 

\begin{theorem} [Hughes, Keating, O'Connell \cite{Hughes2005}] \label{thm:HughesKeatingOConnell1}
    For any $\epsilon > 0$, the sequence of pairs of fields $\left( \Re \log p_{U(n)}(\cdot), \Im \log p_{U(n)}(\cdot) \right)_{n \in \mathbb{N}}$ converges in distribution in $H_0^{-\epsilon}(S^1) \times H_0^{-\epsilon}(S^1)$ to the pair of generalized Gaussian fields $(\Re Z, \Im Z)$, where  
    \begin{align}
        Z(\theta) =& \sum_{k = 1}^\infty A_k e^{-ik\theta},
    \end{align}
    with $A_k$ being complex Gaussians whose real and imaginary parts are independent centered Gaussians with variance $1/(2k)$.
\end{theorem}
A straightforward formal calculation of the covariance kernel of $\Re Z$ and $\Im Z$ shows that the fields $\Re Z$ and $\Im Z$ are indeed log-correlated:
\begin{align}
    \mathbb{E} \left( \Re Z(\theta) \Re Z(\theta') \right) = \mathbb{E} \left( \Im Z(\theta) \Im Z(\theta') \right) = \frac{1}{2} \sum_{k = 1}^\infty \frac{1}{k} \cos(k(\theta - \theta')) = -\frac{1}{2} \log |e^{i\theta} - e^{i\theta'}|.
\end{align}

For various other ensembles the logarithm of the characteristic polynomial has been connected to a log-correlated Gaussian field, for example for the complex Ginibre ensemble, the limiting field has been identified as the Gaussian free field on $\mathbb{R}^2$ (conditioned to be analytic outside the unit disk) \cite{RV07}, and for the Gaussian unitary ensemble \cite{FKS13} the limiting field lives on $(-1,1)$ and has the covariance kernel $-\frac{1}{2} \log |x - y|$. For the orthogonal and symplectic groups the limiting field is given in the following theorem, which is the analogue to Theorem \ref{thm:HughesKeatingOConnell1} (to see that the fields $X$ and $\hat{X}$ are indeed log-correlated see Remark \ref{remark:cov}, in which their covariance kernels are computed):

\begin{theorem} \label{thm:Gaussian field1} (Assiotis, Keating \footnote{This result was first published in \cite{Forkel2021} with the kind agreement of Dr. Assiotis.})
Let $\mathcal{N}_j$, $j \in \mathbb{N}$, be independent real standard Gaussians, and let $\eta_j = 1_{j \text{ is even}}$, for $j \in \mathbb{N}$. Then for any $\epsilon > 0$, the sequence of pairs of fields $\left( \Re \log p_{O(n)}(\cdot), \Im \log p_{O(n)}(\cdot) \right)_{n \in \mathbb{N}}$ converges in distribution in $H^{-\epsilon}_0(S^1) \times H^{-\epsilon}_0(S^1)$ to the pair of generalized Gaussian fields $\left( X - x, \hat{X} - \hat{x} \right)$, where
\begin{align} 
\begin{split}
X(\theta) &= \frac{1}{2} \sum_{j = 1}^\infty \frac{1}{\sqrt{j}} \mathcal{N}_{j} \left(e^{-ij\theta} + e^{ij\theta}\right) = \sum_{j=1}^\infty \frac{1}{\sqrt{j}} \mathcal{N}_{j} \cos(j\theta),\\
\hat{X}(\theta) &= \frac{1}{2i} \sum_{j = 1}^\infty \frac{1}{\sqrt{j}} \mathcal{N}_{j} \left(e^{-ij\theta} - e^{ij\theta}\right) = -\sum_{j=1}^\infty \frac{1}{\sqrt{j}} \mathcal{N}_{j} \sin(j\theta),\\
x(\theta) &= \frac{1}{2} \sum_{j = 1}^\infty \frac{\eta_j}{j} \left(e^{-ij\theta} + e^{ij\theta}\right) = \sum_{j=1}^\infty \frac{\eta_j}{j} \cos(j\theta),\\
\hat{x}(\theta) &= \frac{1}{2i} \sum_{j = 1}^\infty \frac{\eta_j}{j} \left(e^{-ij\theta} - e^{ij\theta}\right) = - \sum_{j=1}^\infty \frac{\eta_j}{j} \sin(j\theta).
\end{split}
\end{align}
Similarly, for any $\epsilon > 0$, the sequence of pairs of fields $\left( \Re \log p_{Sp(2n)}(\cdot), \Im \log p_{Sp(2n)}(\cdot) \right)_{n \in \mathbb{N}}$ converges in distribution in $H^{-\epsilon}_0(S^1) \times H^{-\epsilon}_0(S^1)$ to the pair of generalized Gaussian fields $\left( X + x, \hat{X} + \hat{x} \right)$.
\end{theorem}

In Chapter \ref{chapter:UBM} we are concerned with the connection between the logarithm of the characteristic polynomial of unitary Brownian motion and the Gaussian free field on the cylinder. In the case of unitary Brownian motion there is an extra variable $t \in \mathbb{R}$, which we include in the notation as follows: 
\begin{align}
    p_{U(n)}(t,\theta) := \det \left (I-e^{-i\theta}U_t\right) = \prod_{k=1}^n (1-e^{i(\theta_k(t)-\theta)}), \quad (t,\theta) \in \mathbb{R} \times [0,2\pi),
\end{align}
where $\theta_k(t)$, $k = 1,...,n$, are the eigenangles of $U_t$. We thus consider $\log p_{U(n)}(\cdot, \cdot)$ as a random field on the cylinder $\mathbb{R} \times [0,2\pi)$. Since Haar measure is the equilibrium distribution of unitary Brownian motion it is natural to expect that the limit of $p_{U(n)}(\cdot_{\cdot})$ (in an appropriate function space) would be given by
\begin{align} \label{eqn:Z1}
    Z(t,\theta) = \sum_{k = 1}^\infty A_k(t) e^{ik\theta},
\end{align}
where $A_k(\cdot)$, $k \in \mathbb{N}$, are complex stationary Ornstein-Uhlenbeck processes defined by the SDEs
\begin{align}
    \text{d}A_k(t) = -k A_k(t) \text{d}t + \text{d}\left(W_k(t) + i\tilde{W}_k(t) \right),
\end{align}
with $A_k(0)$ being a complex Gaussian whose real and imaginary parts are independent Gaussians with variance $1/(2k)$, and $(W_k(t))_{t \geq 0}$, $(\tilde{W}_k(t))_{t \geq 0}$, $k \in \mathbb{N}$, denoting real standard Brownian motions. In Chapter \ref{chapter:UBM} we prove the following theorem:

\begin{theorem} \label{thm:FS}
For any $s \in [0,\frac{1}{2})$, $\epsilon > s$, and $T>0$, the sequence of random fields $ \left( \log p_{U(n)}(\cdot, \cdot \right)_{n \in \mathbb{N}}$ converges in distribution in the tensor product of Hilbert spaces $H^{s}([0,T]) \otimes  H^{-\epsilon}_0(S^1)$ (see Section \ref{section:Sobolev} for a definition) to the generalized Gaussian field $Z$ in (\ref{eqn:Z1}). \\
Furthermore, those regularity parameters $s$ and $-\epsilon$ are optimal, in the sense that for $s = 1/2$ or $s \geq \epsilon \geq 0$ the sequence $ \left( \log p_{U(n)}(\cdot, \cdot) \right)_{n \in \mathbb{N}}$ almost surely does not converge in $H^{s}([0,T]) \otimes  H^{-\epsilon}_0(S^1)$, and $Z$ is almost surely not an element of  $H^{s}([0,T]) \otimes  H^{-\epsilon}_0(S^1)$. 
\end{theorem}
As one can see by a covariance calculation, see Remark \ref{remark:covariance}, $\Re Z$ and $\Im Z$ are Gaussian free fields on the infinite cylinder $\mathbb{R} \times [0,2\pi)$. The identification of the Gaussian free field as the limit of the logarithm of the characteristic polynomial of unitary Brownian motion already follows from the much earlier work of Spohn \cite{Spohn1998}, but the convergence holds in a different way. We will expand in Section \ref{chapter:UBM} on how our result is related to Spohn's.\\

We conclude this section by mentioning that the maximum of the real and imaginary part of the logarithm of the characteristic polynomial of random matrices has attracted considerable interest, especially for random matrices from the unitary group, as in that case one can get analogous conjectures for the extreme values of the Riemann Zeta function on the critical line. Fyodorov-Hiary-Keating \cite{Fyodorov2012} made the very precise conjecture that the random variables 
\begin{align} \label{eqn:max}
    \max_{\theta \in [0,2\pi)} \log |p_{U(n)}(\theta)| - \log n + \frac{3}{4} \log \log n,
\end{align}
converge in distribution, as $n \rightarrow \infty$, towards a limiting random variable $R$, whose density is given by 
\begin{align}
    \mathbb{P}(R \in \text{d}x) = 4e^{2x} K_0(2e^x) \text{d}x,
\end{align}
where $K_\nu$ denotes the modified Bessel function of the second kind. This density is the density of the sum of two independent Gumbel random variables, as was later observed in \cite{SZ15}. As was pointed out in \cite{Fyodorov2012}, the constant $3/4$ in the subleading order is expected due to the asymptotically log-correlated structure of the field $\log |p_{U(n)}(\cdot)|$, and stands in contrast to the constant $1/4$ one obtains for the extreme values of short-range correlated random fields \cite{LLR83}. \\
The Fyodorov-Hiary-Keating conjecture has now been almost entirely proven: the first order term $\log n$ has been proven in \cite{Arguin2017}, the second order term $\frac{3}{4} \log \log n$ has been confirmed in \cite{PZ18}, tightness of the random variables in (\ref{eqn:max}) has been established in \cite{Chhaibi2018a}, and in \cite{PZ22} convergence has been proven towards a limiting random variable which is the sum of two independent random variables, one of which is Gumbel. The distribution of the second one is the last missing piece. Note that all of those results are also valid for the circular $\beta$-ensemble, denoted $C\beta E$, for a $\beta > 0$, whose joint eigenvalue density on the unit circle is given by (the case $\beta = 2$ corresponds to random unitary matrices, see Theorem \ref{thm:Weyl U})
\begin{align}
    \frac{1}{Z_{n,\beta}} \prod_{1 \leq j < k \leq n} |e^{i\theta_j} - e^{i\theta_k}|^\beta \text{d}\theta_1 \cdots \text{d}\theta_n.
\end{align}
For other ensembles there also exist conjectures and results: in \cite{FS16}, Simm and Fyodorov made a similar conjecture for the extreme values of the absolute value of the characteristic polynomial of the Gaussian unitary ensemble, and in \cite{Lam20}, Lambert did so for the complex Ginibre ensemble and proved the leading order term. Claeys, Fahs, Lambert and Webb \cite{Claeys2019} considered the extreme values of the centered eigenvalue counting function, (which is closely related to the imaginary part of the logarithm of the characteristic polynomial) of a certain class of random Hermitian matrices which includes the Gaussian unitary ensemble, and obtained the leading order term. This allowed them to prove bounds on how much the eigenvalues of those random Hermitian matrices deviate from certain deterministic location. %To prove their results Claeys et al. used the theory of Gaussian multiplicative chaos, which we now introduce.

\section{Gaussian Multiplicative Chaos}

The theory of Gaussian multiplicative chaos was initiated by Kahane in the context of turbulence in \cite{Kahane1985}, and has since found many applications ranging from finance to Liouville quantum gravity - see the review article \cite{Rhodes2013} and the references therein. The goal of the theory is to rigorously define and study the random measures (called Gaussian multiplicative chaos measures or GMC measures) 
\begin{align} \label{eqn:GMC def}
    M_\gamma(\text{d}x) = e^{\gamma h(x) - \frac{\gamma^2}{2} \mathbb{E}(h(x)^2)} \sigma (\text{d}x) = \frac{e^{\gamma h(x)}}{\mathbb{E}\left( e^{\gamma h(x)} \right)} \text{d}\sigma(x),
\end{align}
where $h$ is a real centered log-correlated Gaussian field on a locally compact metric space $(D,d)$, with covariance kernel $K(x,y) = \mathbb{E}(h(x)h(y))$, $\gamma$ is a constant, and $\sigma$ is a positive Radon measure on $D$. Since log-correlated fields cannot be defined pointwise they only exist as random distributions. Thus to make sense of the measures (\ref{eqn:GMC def}), one needs to first regularize $h$ in some way to obtain fields $h_\epsilon$, $\epsilon > 0$, which are defined pointwise and converge to $h$ in a suitable sense as $\epsilon \rightarrow 0$, then define the random measures
\begin{align}
    M_{\gamma,\epsilon}(\text{d}x) = e^{\gamma h_\epsilon(x) - \frac{\gamma^2}{2} \mathbb{E}(h_\epsilon(x)^2)} \sigma (\text{d}x) = \frac{e^{\gamma h_\epsilon(x)}}{\mathbb{E}\left( e^{\gamma h_\epsilon(x)} \right)} \text{d}x,
\end{align}
and finally set $M_\gamma = \lim_{\epsilon \rightarrow 0} M_{\gamma,\epsilon}$. Naturally the following questions arise:
\begin{itemize}
    \item In which sense does $M_{\gamma, \epsilon}$ converge to $M_\gamma$?
    \item Is the limit $M_\gamma$ independent of the choice of the regularizations $h_\epsilon$ of $h$?
    \item For which $\gamma > 0$ is the limiting field $M_{\gamma, \epsilon}$ non-trivial, i.e. not a.s. equal to the zero measure?
\end{itemize}

In Kahane's seminal work he proved the following theorem, which answers the first two questions for a certain class of covariance kernels $K$:
\begin{theorem}[Kahane \cite{Kahane1985}] \label{thm:Kahane1}
    Let $(D,d)$ be a locally compact metric space, and let $K:D\times D \rightarrow \mathbb{R} \cup \{\pm \infty\}$ be a positive definite covariance kernel that allows a decomposition
    \begin{align}
        K(x,y) = \sum_{k = 1}^{\infty} K_k(x,y), \quad x,y \in D,
    \end{align}
    where $K_k:D \times D \rightarrow \mathbb{R}$, $k \in \mathbb{N}$, are continuous and positive definite covariance kernels. Further let $g_k$, $k \in \mathbb{N}$, be independent (pointwise defined) Gaussian fields with covariance kernel $K_k$, and define the sequence of approximating fields $h_n = \sum_{k = 1}^n g_k$, $n \in \mathbb{N}$. Then given a Radon measure $\sigma$ on $D$ and $\gamma > 0$, the random measures 
    \begin{align}
        M_{\gamma,n}(\text{d}x) := e^{\gamma h_n(x) - \frac{\gamma^2}{2} \mathbb{E}(h_n(x)^2)} \sigma (\text{d}x),
    \end{align}
    converge almost surely in the space of Radon measures on $D$ (equipped with the the topology of weak convergence) to a random measure $M_\gamma$. If in addition $K_k(x,y) \geq 0$ for all $x,y \in D$ and $k \in \mathbb{N}$, then the distribution of $M_\gamma$ is independent of the decomposition of $K$ into $K_k$, $k \in \mathbb{N}$.
\end{theorem}

The proof of convergence uses a simple martingale argument, which we outline now, as this gives insight into how the Gaussian multiplicative chaos measures in Chapter \ref{chapter:GMC} are constructed: For any compact $A \subset D$ and $n \in \mathbb{N}$ it holds by Fubini that 
\begin{align} \label{eqn:expectation of mu k}
\mathbb{E}\left( M_{\gamma,n}(A) \right) = \int_A \sigma(\text{d} x),
\end{align}
and 
\begin{align}
\mathbb{E}\left( M_{\gamma,n}(A) | \sigma(g_1,...,g_{n-1}) \right) = M_{\gamma,n-1}(A),
\end{align} 
where $\sigma(g_1,...,g_{n-1})$ is the $\sigma$-algebra generated by $g_1,...,g_{n-1}$. Thus, since also for any $n \in \mathbb{N}$ the random variable $M_{\gamma,n}(A)$ is measurable w.r.t. $\sigma(g_1,...,g_{n-1})$, it follows that the sequence $\left(M_{\gamma,n}(A) \right)_{n \in \mathbb{N}}$ is a martingale. As a non-negative martingale the sequence converges a.s. to a random variable which will be denoted by $M_\gamma(A)$. One can then show that a.s. the map $A \mapsto M_\gamma(A)$ is a measure and it holds that a.s. $M_{\gamma,n} \xrightarrow{d} M_\gamma$ in the space of Radon measures on $D$, equipped with the topology of weak convergence. When constructing $M_\gamma$ like this, Kolmogorov's zero-one law also immediately implies that $\mathbb{P}(M_\gamma \text{ is the zero measure}) \in \{0,1\}$, since the event $\{M_\gamma \text{ is the zero measure} \}$ is independent of any finite number of $g_k$, $k \in \mathbb{N}$.\\

The answer to the third question, i.e. for which $\gamma > 0$ is the limiting measure $M_\gamma$ non-trivial, depends heavily on the specific kernel and the reference measure $\sigma$. For this question of non-triviality Kahane focused on the case of log-correlated Gaussian fields whose covariance kernel can be written as
\begin{align} \label{eqn:log kernel}
    K(x,y) = \log_+ \frac{T}{d(x,y)} + g(x,y),
\end{align}
where $\log_+ x := \max (\log x, 0)$ and $g$ is continuous and bounded on $D \times D$. While he proved more general results, we restrict ourselves to stating his result for the case that $D$ is an open subset of Euclidean space, which is the case we are concerned with.

\begin{theorem}[Kahane \cite{Kahane1985}] \label{thm:Kahane2}
    If the domain $D$ is an open subset of $\mathbb{R}^d$ for some $d \in \mathbb{N}$, equipped with the Euclidean distance, $\sigma$ is the Lebesgue measure, and $K$ takes the form (\ref{eqn:log kernel}), then 
    \begin{align}
        M_\gamma \text{ is non-trivial } \iff \gamma^2 < 2d.
    \end{align}
\end{theorem}
Usually the phase $\gamma^2 < 2d$ is referred to as the subcritical phase, the phase $\gamma^2 = 2d$ is referred to as the critical phase, and the phase $\gamma^2 > 2d$ is referred to as the supercritical phase. Constructed as above, a GMC measure in the critical and supercritical phases is almost surely the zero measure, however there exist ways to construct non-trivial critical and supercritical GMC measures, see for example \cite{Rhodes2013}. The phase $\gamma^2 < d$ is referred to as the $L^2$-phase, as here for any bounded measurable set $A \subset D$ it holds that \cite{Kahane1985}
\begin{align}
    \mathbb{E}\left( M_{\gamma}(A)^2 \right) < \infty,
\end{align}
and more generally the $m$-th moment of the total mass of $M_\gamma$ is finite if and only if $\gamma^2 < 2d/m$. \\%\todo{supported on thick points, relate to max of underlying field, corr of GMC}

Since the theory of GMC was initiated by Kahane his results have been generalized significantly. In particular the assumption that the covariance kernel can be decomposed into a sum of positive continuous covariance kernels is very restrictive and hard to check in practice. It is for example still unknown whether in the case $D = \mathbb{R}^3$ the kernel (\ref{eqn:log kernel}) allows a decomposition as in Theorem \ref{thm:Kahane1}. Rhodes and Vargas \cite{Robert2010} got rid of this assumption and regularized the Gaussian field $h$ by picking a mollifier function $\theta$ instead, i.e. a continuous non-negative and compactly supported function that integrates to $1$, and setting $h_\epsilon(x) = h * \theta_\epsilon$, where $*$ denotes convolution and $\theta_\epsilon(\cdot) := \epsilon^{-d} \theta(\cdot/\epsilon)$ for $\epsilon > 0$. They then showed then that the sequence $M_{\gamma,\epsilon}$ of GMC measures corresponding to the continuous Gaussian fields $h_\epsilon$, $\epsilon > 0$, converges to a limiting measure $M_\gamma$, in the sense that for all measurable $A \subset D$ the sequence $M_{\gamma,\epsilon}(A)$ converges to $M_\gamma(A)$ in distribution. Furthermore the limit is independent of the choice of the mollifier function. Shamov \cite{Shamov2016} further generalized their results by showing that for very general types $h_\epsilon$ of Gaussian regularizations of $h$, the approximating GMC measures $M_{\gamma,\epsilon}$ convergence to $M_{\gamma}$ weakly in probability, and that the limit is independent of the type of regularization. Furthermore he showed that the GMC measure $M_\gamma$ is measurable w.r.t. the underlying Gaussian field $h$. Using slightly less general assumptions than Shamov, but a very short and self-contained proof, Berestycki \cite{Berestycki2017} proved convergence in probability and uniqueness of the limit when regularising the field using mollifier functions, as well as non-triviality in the entire subcritical phase. \\

Since by Theorem \ref{thm:HughesKeatingOConnell1} it holds that 
\begin{align}
    (\Re \log p_{U(n)}(\cdot), \Im \log p_{U(n)}(\cdot)) \implies (\Re X, \Im X)
\end{align}
as $n \rightarrow \infty$, one would expect that for certain $\alpha > -1/2$ and $\beta \in i\mathbb{R}$ the random measures 
\begin{align}
    \mu_{n,\alpha,\beta}(\text{d}\theta) := \frac{|p_{U(n)}(\theta)|^{2\alpha_j} e^{2i\beta_j \Im \log p_{U(n)}(\theta)}}{\mathbb{E}\left( |p_{U(n)}(\theta)|^{2\alpha_j} e^{2i\beta_j \Im \log p_{U(n)}(\theta)} \right)} \text{d} \theta 
\end{align}
converge to GMC measures associated to the field $2\alpha \Re X + 2i\beta \Im X$ (see Figure 1.2 for two samples of $\Im \log p_{U(n)}(\cdot)$ and $\mu_{n,0,i/2}$). This was conjectured by Fyodorov, Hiary, and Keating in \cite{Fyodorov2012}, and first proven to be true by Webb in \cite{Webb2015}, thus establishing the first connection between random matrix theory and the theory of Gaussian multiplicative chaos. 
\begin{theorem}[Webb \cite{Webb2015}]\label{thm:Webb}
    For $G(n) = U(n)$, $\alpha > -1/2$ and $\alpha^2 - \beta^2 < 1/2$, as $n \rightarrow \infty$ the sequence of measures $\mu_{n,\alpha,\beta}$, $n \in \mathbb{N}$, converges in distribution in the space of Radon measures equipped with the topology of weak convergence to the (non-trivial) Gaussian multiplicative chaos measure 
    \begin{align}
    \begin{split}
        \mu_{\alpha,\beta}(\text{d}\theta) :=& \frac{e^{2\alpha \Re Z(\theta) + 2i\beta \Im Z(\theta)}}{\mathbb{E}\left(e^{2\alpha \Re Z(\theta) + 2i\beta \Im Z(\theta)} \right)} \text{d}\theta \\
        =& e^{2\alpha \Re Z(\theta) + 2i\beta \Im Z(\theta) - 2\alpha^2 \mathbb{E}(\Re Z(\theta)^2) + 2\beta^2 \mathbb{E}(\Im Z(\theta)^2)} \text{d}\theta \\
        \overset{D}{=} & e^{2\sqrt{\alpha^2 - \beta^2} \Re Z(\theta) - 2(\alpha^2 - \beta^2) \mathbb{E}(\Re Z(\theta)^2)} \text{d}\theta,
    \end{split}
    \end{align}
    where the complex Gaussian field $Z$ is given in Theorem \ref{thm:HughesKeatingOConnell1}. 
\end{theorem}
Here $\gamma = \sqrt{2\alpha^2 - 2\beta^2}$, thus Theorem \ref{thm:Webb} only shows convergence in the $L^2$-phase, which significantly simplifies the proof. Nikula, Saksman and Webb \cite{Nikula2020} proved convergence in the rest of the $L^1$ or subcritical phase in the case $\beta = 0$, but as they mention their proof technique can be modified in a straightforward manner to include the case $\beta \neq 0$. \\

\begin{figure} \label{figure:GFF GMC}
    \centering
    \begin{minipage}{.5\textwidth}
      \centering
    \includegraphics[scale=0.45]{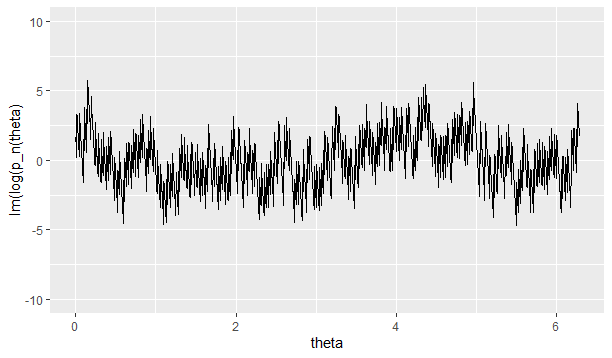}
    \includegraphics[scale=0.45]{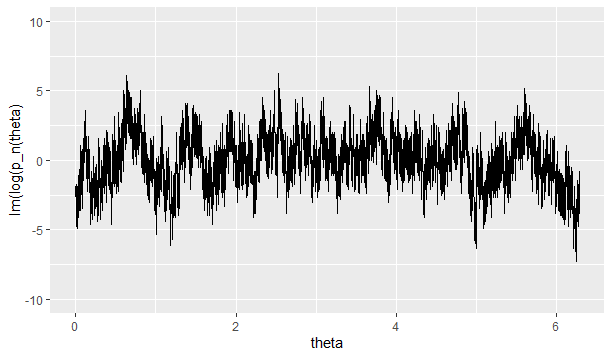}
    \end{minipage}%
    \begin{minipage}{.5\textwidth}
      \centering
    \includegraphics[scale=0.45]{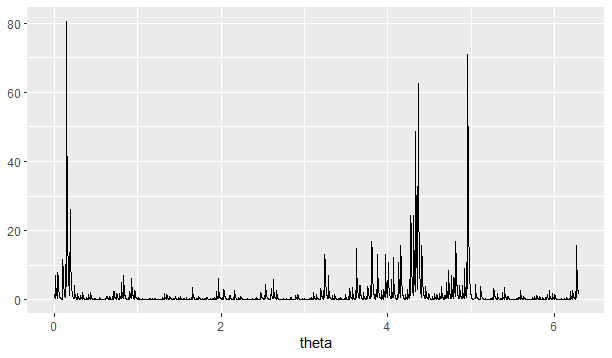}
       \includegraphics[scale=0.45]{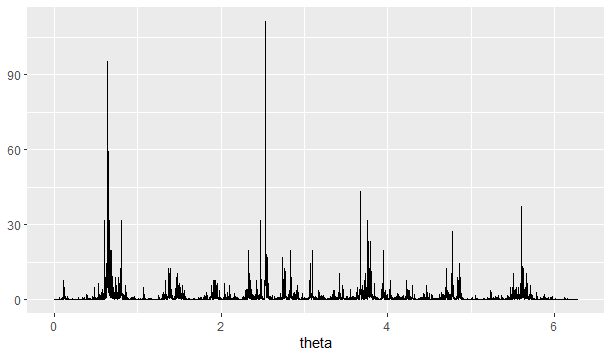}
    \end{minipage}
    \caption{A sample of $\Im \log p_{U(n)}(\theta)$ and its GMC measure $\mu_{n,0,i/2}$, for $n = 200$ (above) and $n = 500$ (below).}
\end{figure}
 
Since then the connection between the two fields has been extended to other random matrix ensembles. In \cite{Chhaibi2019}, Chhaibi and Najnudel proved convergence (in a different sense) of the characteristic polynomials of matrices drawn from the circular $\beta$-ensemble to a GMC measure on the unit circle. In \cite{Berestycki2018} Berestycki, Webb and Wong proved that, after suitable normalization, powers of the absolute value of the characteristic polynomial of a matrix from the Gaussian Unitary Ensemble converge to Gaussian multiplicative chaos measures on the real line. They proved this result in the $L^2$-phase, however it is likely to also hold in the whole $L^1$-phase. The analogous result for powers of the exponential of the imaginary part of the logarithm of the characteristic polynomial of the Gaussian Unitary Ensemble was proven in \cite{Claeys2019}, in the whole $L^1$-phase. In \cite{Lam20} Lambert connected the characteristic polynomial of the Ginibre ensemble to Gaussian multiplicative chaos measures associated to the Gaussian free field on the unit disk, however in a different sense. Recently Bourgade and Falconet \cite{Bourgade2022} have proved the corresponding result for unitary Brownian motion, that is they have shown the convergence 
\begin{align} \label{eqn:UBM GMC}
    \lim_{n \rightarrow \infty} \frac{|p_{U(n)}(\theta_t)|^{2\alpha}}{\mathbb{E}\left( |p_{U(n)}(\theta_t)|^{2\alpha} \right)} \text{d}\theta \text{d}t = \frac{e^{\gamma \Re Z(t,\theta)}}{\mathbb{E}\left( |e^{\gamma \Re Z(t,\theta)}| \right)} \text{d}\theta \text{d}t,
\end{align}
for all $\alpha$ in the subcritical phase, and as they remark a similiar result should hold for the imaginary part. This is the natural dynamical analogue to the results by Webb \cite{Webb2015} and Nikula, Saksman, and Webb \cite{Nikula2020}. The results by Bourgade and Falconet, and Lambert are special in that the measures considered are living on a two-dimensional space, instead of on a one-dimensional space as in all previously proven cases. Those results are also special in that the limiting GMC measure has a specific interpretation, since it is the GMC measure associated to the Gaussian free field, and thus equals the volume measure in two-dimensional Liouville quantum gravity. \\

Our purpose in Chapter \ref{chapter:GMC} is to prove the analogue to Webb's result for the other two classical compact groups, i.e. that 
\begin{align} \label{eqn:GMC convergence}
    \lim_{n \rightarrow \infty} \frac{|p_{G(n)}(\theta)|^{2\alpha} e^{2i\beta \Im \log p_{G(n)}(\theta)}}{\mathbb{E}\left( |p_{G(n)}(\theta)|^{2\alpha} e^{2i\beta \Im \log p_{G(n)}(\theta)} \right)} \text{d} \theta = \frac{e^{2\alpha X(\theta) + 2i\beta \hat{X}(\theta)}}{\mathbb{E}\left(e^{2\alpha X(\theta) + 2i\beta \hat{X}(\theta)} \right)} \text{d}\theta,
\end{align}
for $G(n) = O(n)$ and $G(n) = Sp(2n)$, and certain parameters $\alpha$ and $\beta$. The precise result is stated in Theorems \ref{thm:main} and \ref{thm:main2}. We thus complete the connection between the classical compact groups and Gaussian multiplicative chaos.

\section{Moments of moments} 

It is a natural question how much circle averages of the characteristic polynomial fluctuate, in particular what their moments are. Thus for $\alpha > -1/2$ and $m \in \mathbb{R}$ we define the moments of moments of $p_{G(n)}(\theta)$ by
\begin{align}
\begin{split}
    {\rm MoM}_{G(n)}(m,\alpha) :=& \mathbb{E}_{U \in G(n)} \left( \left( \frac{1}{2\pi} \int_0^{2\pi} |p_{G(n)}(\theta)|^{2\alpha} \text{d}\theta \right)^m \right),
\end{split}
\end{align}
where the term {\em moments of moments} refers to the fact that we first take a moment of the characteristic polynomial with respect to the spectral variable $\theta$, and then a moment with respect to the random matrix distribution. We are interested in the asymptotics of ${\rm MoM}_{G(n)}(m,\alpha)$ in the limit as $n \rightarrow \infty$, for fixed $\alpha$, $m$. \\
The maximum of the characteristic polynomial $\max_{\theta \in [0,2\pi)} |p_{G(n)}(\theta)|$ is connected to the moments of moments since ${\rm MoM}_{G(n)}(1/p,p/2) = \mathbb{E}_{ G(n)} \left( ||p_{G(n)}( \cdot)||_p \right)$, and since for large $p$ the $L^p$-norm approximates the $L^\infty$-norm $||p_{G(n)}(\cdot)||_{\infty} = \max_{\theta \in [0,2\pi)} |p_{G(n)}( \theta) |$. This suggests that the asymptotics of ${\rm MoM}_{G(n)}(m,\alpha)$ can be used to motivate conjectures for the maximum of the characteristic polynomials \cite{Fyodorov2012, FK14}. ${\rm MoM}_{G(n)}(m,\alpha)$ is also related to the $m$-th moment of the total mass of the Gaussian Multiplicative Chaos measures arising from the characteristic polynomial of random matrices from the classical compact groups, see \cite{Forkel2021, Nikula2020, Webb2015}. The moments of moments are also interesting because of their application to understanding the large values taken by the Riemann zeta function on the critical line and other $L$-functions \cite{Fyodorov2012, FK14} - see \cite{Bailey2022} for a review.\\

Recently there has been a good deal of attention given to ${\rm MoM}_{U(n)}(m,\alpha)$. Fyodorov, Hiary and Keating made conjectures on the large $n$ asymptotics in \cite{Fyodorov2012}, which were specified by Fyodorov and Keating in \cite{FK14}, and which were then supported by numerical computations and generalized in \cite{Fyodorov2018}. So far their conjectures have been proven only in cases where $m \in \mathbb{N}$, since then one can use Fubini's theorem to write the moments of moments as a multiple integral:
\begin{align} \label{eqn:Fubini1}
\begin{split}
{\rm MoM}_{G(n)}(m,\alpha) =& \int_0^{2\pi} \cdots \int_0^{2\pi} \mathbb{E}_{U \in G(n)} \left( \prod_{j = 1}^m |p_{G(n)}(\theta_j)|^{2\alpha} \right) \frac{\text{d}\theta_1}{2\pi} \cdots \frac{\text{d}\theta_m}{2\pi}. 
\end{split}
\end{align}
For $m = 2$ and $\alpha > -1/4$ the conjectured asymptotics of ${\rm MoM}_{U(n)}(m,\alpha)$ were proven by Claeys and Krasovsky, as an application of their calculation of the asymptotics of Toeplitz determinants with two merging singularities via a Riemann-Hilbert analysis \cite{Claeys2015} (Toeplitz determinants and Riemann-Hilbert problems will be introduced in the following Sections \ref{section:FH} and \ref{section:RH}). For $m,\alpha \in \mathbb{N}$ they were proven by Bailey and Keating using an approach based on exact identities for finite $n$ \cite{Bailey2019}. Using a combinatorial approach for $m, \alpha \in \mathbb{N}$, which involved representation theory and constrained Gelfand-Tsetlin patterns, Assiotis and Keating \cite{Assiotis2020} proved the same results as found in \cite{Bailey2019}, but with an alternative formula for the leading order coefficient (see also \cite{BarhoumiAndreani2020}, where the results in \cite{Bailey2019} and \cite{Assiotis2020} were rederived using another approach). This same combinatorial approach was then used by Assiotis, Bailey and Keating to prove asymptotic formulas for $G(n) = SO(2n)$ and $G(n) = Sp(2n)$, when $m, \alpha \in \mathbb{N}$ \cite{Assiotis2019}. For those asymptotic formulas Andrade and Best \cite{Andrade2022} then provided an alternative proof using the methods developed in \cite{Bailey2019}, again with an alternative formula for the leading order coefficient.  \\

Using the Riemann-Hilbert approach developed in \cite{Claeys2015}, Fahs \cite{Fahs2021} computed the asymptotics, up to a multiplicative constant, of Toeplitz determinants with arbitrarily many merging singularities. Using this Fahs then proved the asymptotic formula for ${\rm MoM}_{U(n)}(m,\alpha)$ for $m \in \mathbb{N}$ and general $\alpha > 0$, however without an explicit expression for the leading order coefficient:
\begin{theorem}[\cite{Fahs2021}] \label{thm:Fahs MoM}
    For $m \in \mathbb{N}$ and $\alpha > 0$, as $n \rightarrow \infty$:
    \begin{equation}
        {\rm MoM}_{U(n)}(m,\alpha) = \begin{cases} (1+o(1)) n^{m\alpha^2} \frac{G(1 + \alpha)^{2m}\Gamma(1 - m\alpha^2)}{G(1 + 2\alpha)^m \Gamma(1 - \alpha^2)^m}, & \alpha < \frac{1}{\sqrt{m}}, \\
        e^{\mathcal{O}(1)} n \log n & \alpha = \frac{1}{\sqrt{m}},\\
        e^{\mathcal{O}(1)} n^{m^2\alpha^2 + 1 - m}, & \alpha > \frac{1}{\sqrt{m}}, \end{cases}
    \end{equation}
    where $e^{\mathcal{O}(1)}$ denotes a function that is bounded and bounded away from $0$ as $n \rightarrow \infty$, and where $G(z)$ denotes the Barnes $G$-function.
\end{theorem}
Note that the phase transition happens exactly at the point $\alpha = 1/\sqrt{m}$, below which the $m$-th moment of the Gaussian multiplicative chaos measure $\mu_{\alpha,0}$ from Theorem \ref{thm:Webb} is finite. This is because of the relation 
\begin{align} \label{eqn:MoM GMC}
    {\rm MoM}_{U(n)}(m,\alpha) = (2\pi)^{-m} \mathbb{E}\left( |p_{U(n)}(0)|^{2\alpha} \right)^m \mathbb{E} \left( \left( \int_0^{2\pi} \frac{|p_{U(n)}(\theta)|^{2\alpha}}{\mathbb{E}(|p_{U(n)}(\theta)|^{2\alpha})} \text{d}\theta \right)^m \right).
\end{align}
The second expectation on the RHS converges to $\mathbb{E}(\mu_{\alpha,0}([0,2\pi)^m)$, which is finite if and only if $\alpha < 1/\sqrt{m}$, in which case it is given by the Fyodorov-Bouchaud formula \cite{FB08, Rem20}
\begin{align}
    \mathbb{E} \left( \left( \frac{1}{2\pi} \mu_{\alpha,0}([0,2\pi) \right)^m \right) = \frac{\Gamma(1 - m\alpha^2)}{\Gamma(1 - \alpha^2)^m}.
\end{align}
For $\alpha = 1/\sqrt{m}$ the second expectation on the RHS of (\ref{eqn:MoM GMC}) can be seen as the "critical" moment of $\mu_{\alpha,0}$, and it diverges like $\log n$. The explicit constant was found by Keating and Wong \cite{KW22} for $m \in \mathbb{N} \setminus \{1\}$, which allowed them to obtain an exact formula for the term $e^{\mathcal{O}(1)}$ in the critical phase in Theorem \ref{thm:Fahs MoM}. \\

Extending Fahs' Riemann-Hilbert approach and using his asymptotics of Toeplitz determinants, Claeys, Glesner, Minakov and Yang \cite{Claeys2022} computed asymptotics of Toeplitz+Hankel determinants with arbitrarily many merging singularities, up to a multiplicative constant. In Chapter \ref{chapter:MoM} we use those asymptotic formulae for Toeplitz+Hankel determinants to establish asymptotic formula for the moments of moments of characteristic polynomials of random matrices from the orthogonal and symplectic groups, up to the leading order coefficient. We obtain the leading order dependence of these moments of moments in the limit on the size of the random matrices when this gets large. Our results are as follows (see Chapter \ref{chapter:MoM}) for a definition of the constants $C^\pm(m,\alpha$)
\begin{theorem} \label{thm:MoM Sp intro}
Let $m \in \mathbb{N}$ and $\alpha > 0$. Then, as $n \rightarrow \infty$:
\begin{align}
\begin{split}
    {\rm MoM}_{Sp(2n)}(m,\alpha) = \begin{cases} (1+o(1)) (2n)^{m\alpha^2} C^-(m,\alpha), & \alpha < \frac{\sqrt{8m-3} - 1}{4m-2}, \\
    e^{\mathcal{O}(1)} n^{m\alpha^2} \log n & \alpha = \frac{\sqrt{8m-3} - 1}{4m-2},\\
    e^{\mathcal{O}(1)} n^{2(m\alpha)^2 + m\alpha -m}, & \alpha > \frac{\sqrt{8m-3} - 1}{4m-2}. \end{cases}
\end{split}
\end{align}
\end{theorem}

\begin{theorem} \label{thm:MoM O intro}
Let $G(n) \in \{ O(n), \, SO(n), \, SO^-(n) \}$ and $\alpha > 0$. For $m \in \mathbb{N} \setminus \{2\}$, as $n \rightarrow \infty$:
\begin{align}
\begin{split}
    {\rm MoM}_{G(n)}(m,\alpha) = \begin{cases} (1+o(1)) n^{m\alpha^2} C^+(m,\alpha), & \alpha < \frac{\sqrt{8m-3} + 1}{4m-2}, \\
    e^{\mathcal{O}(1)} n^{m\alpha^2} \log n & \alpha = \frac{\sqrt{8m-3} + 1}{4m-2},\\
    e^{\mathcal{O}(1)} n^{2(m\alpha)^2 - m\alpha -m}, & \alpha > \frac{\sqrt{8m-3} + 1}{4m-2}. \end{cases}
\end{split}
\end{align}
Moreover, as $n \rightarrow \infty$:
\begin{align}
    {\rm MoM}_{G(n)}(2,\alpha) = \begin{cases} (1 + o(1)) n^{2\alpha^2} C^+(2,\alpha), & \alpha < \frac{1}{\sqrt{2}}, \\
    e^{\mathcal{O}(1)} n^{2\alpha^2} \log n & \alpha = \frac{1}{\sqrt{2}},\\
    e^{\mathcal{O}(1)} n^{4\alpha^2 - 1}, & \alpha \in \left( \frac{1}{\sqrt{2}}, \frac{\sqrt{5} + 1}{4} \right), \\
    e^{\mathcal{O}(1)} n^{4\alpha^2 - 1} \log n, & \alpha = \frac{\sqrt{5} + 1}{4}, \\
    e^{\mathcal{O}(1)} n^{8\alpha^2 - 2\alpha - 2}, & \alpha > \frac{\sqrt{5} + 1}{4}. \end{cases}
\end{align}
\end{theorem}

We see that the critical points where phase transitions occur are significantly different for $U(n)$, $O(n)$ and $Sp(2n)$. This is due to the symmetries around $\pm 1$ in the eigenvalues of orthogonal and symplectic matrices. As will become apparent in our proofs of Theorems \ref{thm:MoM Sp intro} and \ref{thm:MoM O intro} in Chapter \ref{chapter:MoM}, the main contribution in the first phase, which we refer to as the subcritical phase, comes from integrating the integrand 
\begin{align*}
    \mathbb{E}_{U \in G(n)} \left( \prod_{j = 1}^m |p_{G(n)}(\theta_j)|^{2\alpha} \right) 
\end{align*}
in (\ref{eqn:Fubini1}) over the subset of $(0,2\pi)^m$ where $|e^{i\theta_j} - e^{i\theta_k}| > 1/n$ for $1 \leq j < k \leq m$, and $|e^{i\theta_j} \pm 1| > 1/n$ for $1 \leq j \leq m$. Except in the case where $m = 2$ and $G(n)$ is one of the orthogonal groups, the main contribution in the second and third phases, which we refer to as the critical and supercritical phases, respectively, comes from integrating over the subset of $(0,2\pi)^m$ where $|e^{i\theta_j} - 1| < 1/n$ for $1 \leq j \leq m$. In the special case where $m = 2$ and $G(n)$ is one of the orthogonal groups, the main contribution for $\alpha \in [\frac{1}{\sqrt{2}}, \frac{\sqrt{5} + 1}{4})$ comes from integrating over the subset of $(0,2\pi)^m$ where $|e^{i\theta_j} \pm 1| > 1/n$ for $1 \leq j \leq m$, and $|e^{i\theta_j} - e^{i\theta_k}| < 1/n$ for $1 \leq j < k \leq m$. This is the reason for the additional two phases in the asymptotics of ${\rm MoM}_{O(n)}(2,\alpha)$.\\

Note also that, unlike in the case $G(n) = U(n)$, the expectations $\mathbb{E}\left( |p_{G(n)}(\theta)|^{2\alpha} \right)$ are not independent of $\theta$ for $G(n) \in \{O(n), Sp(2n)\}$. This is again due to the symmetries around $\pm 1$ of the eigenvalues of orthogonal and symplectic matrices. Thus for $G(n) \in \{O(n), Sp(2n)\}$ there is no immediate relationship of the type (\ref{eqn:MoM GMC}), between ${\rm MoM}_{G(n)}(m,\alpha)$ and the $m$-th moment of the total mass of the GMC measures
\begin{align}
    \frac{e^{2\alpha X(\theta)}}{\mathbb{E}\left(e^{2\alpha X(\theta)} \right)} \text{d}\theta.
\end{align}

\section{Toeplitz, Hankel, Toeplitz+Hankel determinants with Fisher-Hartwig singularities} \label{section:FH}

Both to prove our Theorem \ref{thm:main} and Theorem \ref{thm:main2} in Chapter \ref{chapter:GMC}, and for the proofs of Theorem \ref{thm:MoM Sp(2n)} and Theorem \ref{thm:MoM O(n)} in Chapter \ref{chapter:MoM}, we need to understand the large $n$ asymptotics of expectations of the form
\begin{align} \label{eqn:FH expectation}
    \mathbb{E}_{U \in G(n)} \left( e^{\Tr V(U)} \prod_{j = 1}^m |p_{G(n)}(\theta_j)|^{2\alpha_j} e^{2i\beta_j \Im \log p_{G(n)}(\theta_j)} \right),
\end{align}
where $\alpha_j > -1/2$ and $\beta_j \in i\mathbb{R}$ (one can more generally consider complex $\alpha_j$ with $\Re \alpha_j > -1/2$ and $\beta_j \in \mathbb{C}$, but in this thesis those cases are not considered), and where $\Tr V(U) := \sum_{k = 1}^n V(z_k)$ for a given function $V$ and an $n \times n$ matrix with eigenvalues $z_1,...,z_n$. This is particularly obvious in the case of the moments of moments, see (\ref{eqn:Fubini1}).\\

Such expectations can be expressed in terms of Toeplitz, and Toeplitz+Hankel determinants, which are defined as follows:

\begin{definition}[Toeplitz, Toeplitz+Hankel determinants]\label{def:T+H}
For a function $f \in L^1(S^1)$ we define the Toeplitz determinants $D_n(f)$, and the Toeplitz+Hankel determinants $D_n^{T+H,\kappa}$, $\kappa = 1,2,3,4$, as follows:
\begin{align}
\begin{split}
D_n(f) =& \det \left( f_{j - k} \right)_{j,k = 0}^{n - 1}, \\
D_n^{T+H,1}(f) =& \det \left( f_{j-k} + f_{j+k} \right)_{j,k = 0}^{n-1}, \\
D_n^{T+H,2}(f) =& \det \left( f_{j-k} - f_{j+k+2} \right)_{j,k = 0}^{n-1}, \\
D_n^{T+H,3}(f) =& \det \left( f_{j-k} - f_{j+k+1} \right)_{j,k = 0}^{n-1}, \\
D_n^{T+H,4}(f) =& \det \left( f_{j-k} + f_{j+k+1} \right)_{j,k = 0}^{n-1},
\end{split}
\end{align}
where $f_j := \frac{1}{2\pi}\int_{0}^{2\pi} f(e^{i\theta}) e^{-ij\theta} \text{d}\theta$.
\end{definition}

It holds that
\begin{align}
    e^{\Tr V(U)} \prod_{j = 1}^m |p_{G(n)}(\theta_j)|^{2\alpha_j} e^{2i\beta_j \Im \log p_{G(n)}(\theta_j)} =& \det (f(U)), 
\end{align}
for the function
\begin{align} \label{eqn:f FH}
    f(z) =& e^{V(z)} \prod_{j = 1}^m |z - e^{i\theta_j}|^{2\alpha_j} e^{2i\beta_j \Im \log (z - e^{i\theta_j})},
\end{align}
where for an $n$ dimensional complex matrix $A$ with eigenvalues $z_1,...,z_n$ we set $\det (f(A)) = \prod_{k = 1}^n f(z_k)$. Thus by the following two famous identities one can express the expectations in (\ref{eqn:FH expectation}) in terms of Toeplitz and Toeplitz+Hankel determinants:
\begin{theorem}[Heine-Szeg\"{o}] \label{thm:Heine-Szego}
    For a function $f \in L^1(S^1)$ it holds that 
    \begin{align}
        \mathbb{E}_{U(n)}\left( \det (f(U)) \right) = D_n(f).
    \end{align}
\end{theorem}

\begin{theorem}\label{thm:Baik2001}(Theorem 2.2 in \cite{Baik2001})
Let $h(z)$ be any function on the unit circle such that for $\iota(e^{i\theta}) := h(e^{i\theta})h(e^{-i\theta})$ the integrals 
\begin{equation}
\iota_j = \frac{1}{2\pi} \int_0^{2\pi} \iota(e^{i\theta}) e^{-ij\theta} \, \text{d}\theta
\end{equation} 
are well-defined. Then for any $n \in \mathbb{N} \cup \{ 0\}$, with $D_n^{T+H,\kappa}$ defined in Definition \ref{def:T+H}, it holds that 
\begin{align}
\begin{split}
\mathbb{E}_{SO(2n)}\left( \text{det}(h(U)) \right) &= \frac{1}{2} D_n^{T+H,1}(\iota), \\
\mathbb{E}_{SO^-(2n)}\left( \text{det}(h(U)) \right) &=  h(1)h(-1) D_{n-1}^{T+H,2}(\iota), \\
\mathbb{E}_{SO(2n+1)}\left( \text{det}(h(U)) \right) &=  h(1) D_n^{T+H,3}(\iota), \\
\mathbb{E}_{SO^-(2n+1)}\left( \text{det}(h(U)) \right) &=  h(-1) D_n^{T+H,4}(\iota), \\
\mathbb{E}_{Sp(2n)}\left( \text{det}(h(U)) \right) &= D_n^{T+H,2}(\iota),
\end{split}
\end{align}
except that $\mathbb{E}_{SO(0)} \left( \det(h(U)) \right) = h(1)$.
\end{theorem} 

The study of the asymptotics of Toeplitz and Toeplitz+Hankel determinants has a long history and was initiated by Szeg\"{o}. The simplest case is the strong Szeg\"{o} limit theorem (see for example \cite{Simon2005} for its most general version), which states that for $f(z) = e^{V(z)} = \text{exp}\left( \sum_{k \in \mathbb{Z}} V_k z^k \right)$ 
\begin{equation}
D_n(e^V) = e^{nV_0}e^{\sum_{k = 1}^\infty kV_kV_{-k}}(1 + o(1)), \quad n \rightarrow \infty.
\end{equation}
Together with the Heine-Szeg\"{o} idendity this implies that for real-valued $V$, and Haar-distributed $U \in U(n)$, the asymptotic variance of the linear statistic $\Tr V(U)$ is given by $\sum_{k = 1}^\infty k|V_k|^2 = ||V||_{1/2}^2$. \\

The function inside a Toeplitz or Toeplitz+Hankel determinant is often referred to as the \textit{symbol} of the determinant, and roots and jumps of the type $f$ in (\ref{eqn:f FH}) has are called \textit{Fisher-Hartwig singularities}, after Fisher and Hartwig, who in \cite{Fisher1969} made precise conjectures on the asymptotics of Toeplitz determinants of symbols with such singularities. The asymptotics of $D_n(f)$ and $D_n^{T+H,\kappa}(f)$, for the symbol $f$ having Fisher-Hartwig singularities, were computed under various assumptions on the symbol $f$ \cite{Basor1978, Basor1979, Basor2001, Basor2002, Basor2002a, Basor2009, Boettcher1981, Boettcher1985, Boettcher1986, ES97, Ehrhardt2001, Widom1973}; see \cite{Deift2012} for a recent historical account and \cite{Deift2011} for the most general results on symbols with Fisher-Hartwig singularities. However all these results are only valid if the singularities of $f$ are bounded away from each other as $n \rightarrow \infty$.\\
In recent years advances were made on the asymptotics of Toeplitz and Toeplitz+Hankel determinants, when the symbol $f$ is allowed to vary with $n$, in particular when singularities merge as $n \rightarrow \infty$. Obtaining asymptotic formulas which are uniform w.r.t. the location of the singularities is necessary for our proofs of Theorem \ref{thm:main} and Theorem \ref{thm:main2} in Chapter \ref{chapter:GMC}, and for the proofs of Theorem \ref{thm:MoM O(n)} and Theorem \ref{thm:MoM Sp(2n)} in Chapter \ref{chapter:MoM}, since in both cases we want to integrate expectations of the form (\ref{eqn:FH expectation}) over the location of the singularities. \\
In \cite{Claeys2015} the asymptotics of $D_n(f)$ were computed for two merging singularities, and were related to a solution of the Painlev\'{e} V equation. Using the techniques in \cite{Claeys2015}, we establish new results on Toeplitz determinants in Theorems \ref{thm:T uniform} and \ref{thm:T, T+H extended} in Chapter \ref{chapter:FH}, that give the asymptotics when there are two conjugate pairs of merging singularities which are bounded away from $\pm 1$. Again the asymptotics are related to a Painlev\'{e} V equation, in fact the same one as in \cite{Claeys2015}. We use those results to prove Theorem \ref{thm:main} in Chapter \ref{chapter:GMC}. \\
In \cite{Fahs2021} the asymptotics of $D_n(f)$ for arbitrarily many singularities merging were computed up to a factor $e^{\mathcal{O}(1)}$ which is uniformly bounded and bounded away from $0$. The precise factor is believed to be related to higher-dimensional analogues of Painlev\'{e} equations. Using the techniques of Fahs in \cite{Fahs2021}, the asymptotics of $D_n^{T+H,\kappa}(f)$ for $f$ having arbitrarily many singularities merging were computed up to an $e^{\mathcal{O}(1)}$ factor in \cite{Claeys2022}. This result is stated in Theorem \ref{thm:T+H Claeys} in Chapter \ref{chapter:FH} and we use it to prove Theorem \ref{thm:main2} in Chapter \ref{chapter:GMC}. \\
For two conjugate pairs of merging singularities which are bounded away from $\pm 1$ our Theorem \ref{thm:T+H uniform} expresses that factor in terms of a Painlev\'{e} transcendent, which again is the same one as in \cite{Claeys2015}. \\

In a recent paper \cite{Bourgade2022}, Bourgade and Falconet gave the first dynamical extension of Fisher-Hartwig asymptotics, that is they computed the large $n$ asymptotics of expectations of the form
\begin{align}
    \mathbb{E}\left( e^{\sum_{s \in \mathcal{B}} \Tr f_s(U_s)} \prod_{z = t + i\theta \in \mathcal{A}} |\det(U_t - e^{i\theta}|^{\gamma_z} \right),
\end{align}
where the expectation is w.r.t. $U = (U_t)_{t \in \mathbb{R}}$, a unitary Brownian motion, $\mathcal{A} \subset \mathbb{R} \times [0,2\pi)$ and $\mathcal{B} \subset \mathbb{R}$ are finite subsets, $\gamma_z$ are real constants and the functions $f_s$ fulfill certain regularity conditions. They used those asymptotics to prove the convergence of the characteristic polynomial of unitary Brownian motion to a Gaussian multiplicative chaos measure associated to the Gaussian free field on the cylinder, i.e. the convergence in (\ref{eqn:UBM GMC}).\\ 

Closely related to Toeplitz and Toeplitz+Hankel determinants are Hankel determinants, i.e. determinants of the form 
\begin{align}
\begin{split}
\det \left( \int_I x^{j+k} w(x) \text{d}x \right),
\end{split}
\end{align}
where $I \subset \mathbb{R}$ is an interval on which the weight $w(x)$ is supported, which is of the form 
\begin{equation}
w(x) = e^{-nV(x)} e^{W(x)} \omega(x),
\end{equation}
where $W(x)$ is continuous, $\omega(x)$ has Fisher-Hartwig singularities and $V(x)$ is a potential (in case $I$ is unbounded $W$ and $V$ need to fulfill certain integrability conditions). \\
There are three canonical cases of Hankel determinants which appear as averages over Hermitian random matrix ensembles: 
\begin{equation}
\det \left( \int_I x^{j+k} w(x) \text{d}x \right) = \mathbb{E}_{\text{Herm}(n)} \left( \prod_{j = 1}^n e^{W(\lambda_j)} \omega(\lambda_j) \right), \quad ,
\end{equation}  
where $\lambda_1,...,\lambda_n \in \mathbb{R}$ are the eigenvalues of $\text{Herm}(n)$, and
\begin{itemize}
\item $I = \mathbb{R}$, $V(x) = 2x^2$ when $\text{Herm}(n)$ is the Gaussian Unitary Ensemble,
\item $I = [-1,\infty)$, $V(x) = 2(x+1)$ when $\text{Herm}(n)$ is the (shifted) Wishart Ensemble,
\item $I = [-1,1]$, $V(x) = 0$ when $\text{Herm}(n)$ is the (shifted) Jacobi Ensemble.
\end{itemize}
\cite{Its2008}, \cite{Garoni2005} and \cite{Krasovsky2007} are important early works on the asymptotics of Hankel determinants that have greatly contributed to the development of the theory. The most general results have been proven in \cite{Charlier2019} and \cite{Charlier2021}. \cite{Claeys2016} concerns the case when singularities are merging as $n \rightarrow \infty$. There are various formulas which relate Hankel determinants with $I=[-1,1]$ and $V(x) = 0$ to Toeplitz determinants and Toeplitz+Hankel determinants, which is how some of those formulas are proven, see for example \cite{Deift2011}. 

\section{Riemann-Hilbert Problems} \label{section:RH}
One way to study the asymptotics of Toeplitz and Toeplitz+Hankel determinants is to express them in terms of orthogonal polynomials on the unit circle, which in turn can expressed in terms of solutions to certain Riemann-Hilbert problems. That orthogonal polynomials can be represented as solutions to certain Riemann-Hilbert problems (RHPs) was first observed by Fokas, Its and Kitaev in \cite{Fokas1992} for polynomials on the real line, and then by Baik, Deift, and Johansson in \cite{Baik1999} for polynomials on the unit circle. One can then use the steepest descent method introduced by Deift and Zhou in \cite{DZ93}, to analyze those RHPs for orthogonals polynomials. This was first done by Bleher and Its in \cite{BI99} and Deift, Kriecherbauer, McLaughlin, Venakides, and Zhou in \cite{DKM99a, DKM99b}, motivated by questions of universality in random matrix theory. See for example \cite{Dei00} and \cite{Kui03} for an introduction to Riemann-Hilbert problems for orthogonal polynomials. We now follow those two works for a very brief introduction to what Riemann-Hilbert problems are and how they are useful for asymptotic analysis, without going into technical detail.\\

Roughly speaking Riemann-Hilbert problems are boundary value problems for analytic functions. Consider an oriented contour $\Sigma \subset \mathbb{C}$, where oriented means that any arc has a positive and a negative side to it. By convention the positive (resp. negative) side lies to the left (resp. right) of the contour when traversing in the direction of the orientation. For simplicity we assume $\Sigma$ to be a finite union of smooth curves in $\mathbb{C}$, with only finitely many points of intersections, and all intersections being transversal, which is the case for all contours encountered in Chapter \ref{chapter:FH}. Furthermore consider a function (subject to certain conditions for $z \rightarrow \infty$ or $z \rightarrow \{\text{intersection points}\}$)
\begin{align}
    V:\Sigma \setminus \{\text{intersection points} \} \rightarrow \text{GL}(n,\mathbb{C}).
\end{align}
A typical Riemann-Hilbert problem associated to $(\Sigma, V)$ then looks as follows: find a function $R:\mathbb{C} \setminus \Sigma \rightarrow \mathbb{C}^{n \times n}$ which fulfills the following conditions:
\begin{enumerate}[label=(\alph*)]
    \item $R$ is analytic, in the sense that every entry is analytic.
    \item The continuous boundary values of $R$ from the right and left side of $\Sigma$, denoted $R_+-$ and $R_-$, respectively, exist on $\Sigma \setminus \{ \text{intersection points} \}$, and are related by the jump condition
\begin{equation} 
R_+(z) = R_-(z) V(z), \quad z \in \Sigma \setminus \{ \text{intersection points} \}.
\end{equation}
    \item $R(z) = I + \mathcal{O}(1/z)$, uniformly as $z \rightarrow \infty$.
\end{enumerate}
Due to condition (c) the problem is normalized at infinity, which is necessary to guarantee uniqueness of a solution. Without this condition one could simply multiply a solution by a matrix of arbitrary entire functions from the left to obtain another solution. Note that one can also choose different normalization conditions for $R$ as $z \rightarrow \infty$, and one can also normalize at other sets of the complex plane, which will be the case for some Riemann-Hilbert problems we encounter in Chapter \ref{chapter:FH}. Further note that in case there are intersection points one needs to impose additional conditions at those to guarantee uniqueness of a solution, since otherwise one could multiply a solution from the left by for example matrix-valued functions of the form $I + M$, where $M$ is analytic in the complex plane apart from isolated singularities at the intersection points.\\

The reason why representing an object as a solution to a Riemann-Hilbert problem is powerful, is the following: suppose one has a sequence of Riemann-Hilbert problems of the type above, i.e. one has a sequence $(\Sigma^{(n)}, V^{(n)})_{n \in \mathbb{N}}$ of contours and jump matrices, and the jumps vanish as $n \rightarrow \infty$ in the sense that $\lim_{n \rightarrow \infty} ||V^{(n)} - I||_{\Sigma^{(n)}} = 0$ for some specified norm $|| \cdot ||_{\Sigma^{(n)}}$ on the set of measurable functions from $\Sigma^{(n)}$ to $\mathbb{C}^{n \times n}$. Typically one takes $|| \cdot ||_{\Sigma^{(n)}}$ to be one (or a sum) of the $L^p$ norms
\begin{align}
    \left( \int_{\Sigma^{(n)}} |V(z)|^p \text{d}z \right)^{1/p}, \quad 
\end{align}
where $| \cdot |$ denotes the Euclidean norm, or the corresponding $L^\infty$ norm. Then in many cases one can find a constant $C > 0$, independent of $n \in \mathbb{N}$ and $z \in \mathbb{C} \setminus \{\Sigma^{(n)}\}$, such that  
\begin{align} \label{eqn:R V}
    |R^{(n)}(z) - I| < C||V^{(n)} - I||_{\Sigma^{(n)}}. 
\end{align}
This allows us to compute very precise asymptotics for $R^{(n)}$ as $n \rightarrow \infty$, which are uniform in all of $\mathbb{C} \setminus \Sigma^{(n)}$. The exact statement depends on the case at hand, yet the proofs are very similar and are often omitted in the Riemann-Hilbert literature. See for example \cite{Dei00, Claeys2019, Fahs2021, Berestycki2018} for detailed proofs of such statements.\\

The typical procedure of a Riemann-Hilbert analysis is as then follows: 
\begin{enumerate}
    \item Represent the objects of interest in terms of solutions to a family of RHPs.
    \item Transform this family of RHPs in such a way that one obtains another family of RHPs whose jumps are small in the limit that one is interested in. The Deift-Zhou steepest descent method is used in this step. For a certain class of Riemann-Hilbert problems it allows us to construct a sequence of explicit transformations towards a Riemann-Hilbert problem with small jumps. 
    \item Use estimates of the form (\ref{eqn:R V}) to obtain uniform asymptotics for the solutions of the RHPs with small jumps.
    \item Reverse the transformation from the second step, to express the solutions to the original family of RHPs in terms of the solutions to the small-jump RHPs, and thus obtain uniform asymptotics of the solutions to the original family of RHPs.
    \item Extract the asymptotics of the objects of interest from the uniform asymptotics of the solutions of the original RHP.
\end{enumerate}

We now explain how Toeplitz determinants are related to orthogonal polynomials and Riemann-Hilbert problems: by the Heine-Szego identity and since $f > 0$ almost everywhere it holds that $D_{n}(f) \in (0,\infty)$ for all $n\in \mathbb{N}$. Thus for $n \in \mathbb{N}$ we can define the polynomials
\begin{align}
\begin{split}
\phi_n(z) &= \frac{1}{\sqrt{D_n(f) D_{n+1}(f)}} 
\left| 
\begin{array} {cccc} f_{0} & f_{-1} & \dots & f_{-n} \\
f_{1} & f_{0} & \dots & f_{-n+1} \\
\dots & \dots & & \dots \\
f_{n-1} & f_{n-2} & \dots & f_{-1} \\
1 & z & \dots & z^n 
\end{array}
\right| = \chi_n z^n + ..., \\
\hat{\phi}_n(z) &= \frac{1}{\sqrt{D_n(f) D_{n+1}(f)}} 
\left| 
\begin{array} {ccccc} f_{0} & f_{-1} & \dots & f_{-n+1} & 1 \\
f_{1} & f_{0} & \dots & f_{-n+2} & z \\
\dots & \dots & & \dots & \dots \\
f_{n} & f_{n-1} & \dots & f_{1} & z^n 
\end{array}
\right| = \chi_n z^n + ..., 
\end{split}
\end{align}
where the leading coefficient $\chi_n$ is given by 
\begin{equation}
\chi_n = \sqrt{\frac{D_n(f)}{D_{n+1}(f_{})}}.
\end{equation}
Further we set $D_0(f) = 0$ and $\phi_0(z) = \hat{\phi}_0(z) = \chi_0 = 1/\sqrt{D_1(f)}$. It is straightforward to verify that the above polynomials satisfy the orthogonality relations
\begin{align} \label{eqn:orthogonality}
\begin{split}
\frac{1}{2\pi} \int_0^{2\pi} \phi_n(e^{i\theta}) e^{-ik\theta} f(e^{i\theta}) \text{d}\theta =& \chi_n^{-1} \delta_{nk}, \\
\frac{1}{2\pi} \int_0^{2\pi} \hat{\phi}_n(e^{-i\theta}) e^{ik\theta} f(e^{i\theta}) \text{d}\theta =& \chi_n^{-1} \delta_{nk},
\end{split}
\end{align}
for $k = 0,1,...,n$, which means that they are orthonormal w.r.t. the weight $f$. Thus one can see that 
\begin{align}
    D_n(f) = \prod_{j = 0}^{n-1} \chi_j^{-2},
\end{align}
which implies that when knowing the large $n$ asymptotics of the leading order coefficients $\chi_n$ we can obtain information about the large $n$ asymptotics of the Toeplitz determinants $D_n(f)$. Furthermore the Toeplitz+Hankel determinants $D_n^{T+H, \kappa}(f)$ can be expressed in terms of the Toeplitz determinants $D_n(f)$ and $\phi_n(0)/\chi_n$, see Theorem \ref{lemma:connection between T and T+H} in Chapter \ref{chapter:FH}. \\ 

The characterization of orthogonal polynomials on the unit circle in terms of Riemann-Hilbert problems is stated in the following theorem. Since this theorem is of crucial importance for the analysis of orthogonal polynomials, and since the proof only uses elementary methods from complex analysis, we include the proof (note that the original theorem was only proven for smooth functions $f$, but the proof extends easily to functions with Fisher-Hartwig singularities). 

\begin{theorem} [Baik, Deift, Johansson \cite{Baik1999}] \label{thm:BaikDeiftJohansson}
Let $C$ denote the unit circle, oriented counterclockwise, and let $f$ be as in (\ref{eqn:f FH}). Then the matrix-valued function $Y(z) = Y(z;n)$ given by 
\begin{equation} \label{eqn:Y1}
Y(z) = \left( \begin{array}{cc} \chi_n^{-1} \phi_n(z) & \chi_n^{-1} \int_C \frac{\phi_n(\xi)}{\xi - z} \frac{f(\xi) \text{d}{\xi}}{2\pi i \xi^n} \\
- \chi_{n-1}z^{n-1} \hat{\phi}_{n-1}(z^{-1}) &  -\chi_{n-1} \int_C \frac{\hat{\phi}_{n-1}(\xi^{-1})}{\xi - z} \frac{f(\xi) \text{d}{\xi}}{2\pi i \xi} \end{array} \right)
\end{equation}
is the unique solution of the following Riemann-Hilbert problem.
\end{theorem}

\noindent \textbf{RH problem for} $Y(z) = Y(z;n)$

\begin{enumerate}[label=(\alph*)]
\item $Y:\mathbb{C}\setminus C \rightarrow \mathbb{C}^{2\times 2}$ is analytic, in the sense that every entry is analytic.

\item The continuous boundary values of $Y$ from inside the unit circle, denoted $Y_+$, and from outside, denoted $Y_-$, exist on $C\setminus \{ e^{i\theta_1},...,e^{i\theta_m} \}$, and are related by the jump condition
\begin{equation} \label{eqn:Y}
Y_+(z) = Y_-(z) \left( \begin{array} {cc} 1 & z^{-n} f(z) \\ 0 & 1 \end{array} \right), \quad z \in C\setminus \{e^{i\theta_1},...,e^{i\theta_m}\}.
\end{equation}

\item $Y(z) = (I + \mathcal{O}(1/z)) \left( \begin{array}{cc} z^n & 0 \\ 0 & z^{-n} \end{array} \right) = \left( \begin{array}{cc} z^n + \mathcal{O}(z^{n-1}) & \mathcal{O}(z^{-n-1}) \\ \mathcal{O}(z^{n-1}) & z^{-n} + \mathcal{O}(z^{-n-1}) \end{array} \right)$, as $ z \rightarrow \infty$.

\item As $z \rightarrow e^{i\theta_k}$, $z \in \mathbb{C}\setminus C$, $k= 1,...,m$, it holds that
\begin{equation}
Y(z) = \left( \begin{array}{cc} O(1) & O(1) + O(|z-z_k|^{2\alpha_k}) \\ O(1) & O(1) + O(|z-z_k|^{2\alpha_k}) \end{array}\right), \quad \text{if } \alpha_k \neq 0,
\end{equation}
and 
\begin{equation}
Y(z) = \left( \begin{array}{cc} O(1) & O(\log |z-z_k|) \\ O(1) & O(\log |z-z_k| ) \end{array}\right), \quad \text{if } \alpha_k = 0.
\end{equation}
\end{enumerate}

\begin{proof}
    We first show that $Y$ given in (\ref{eqn:Y1}) fulfills the RHP, and then show that this solution is unique. It is obvious that condition (a) is satisfied, i.e. that the entries of $Y$ are analytic for all $z \in \mathbb{C}\setminus C$, and that condition (d) is satisfied. Now from the Sokhotski–Plemelj theorem (see for example \cite{Gakhov66}[Section 4.2]) it follows that
    \begin{align}
    \begin{split}
        Y_{12}(z)_+ =& Y_{12}(z)_- + \chi_n^{-1}\phi_n(z)f(z)z^{-n} = Y_{12}(z)_- + Y_{11}(z)_-z^{-n}f(z), \\ 
        Y_{22}(z)_+ =& Y_{22}(z)_- - \chi_{n-1} \hat{\phi}_{n-1}(z^{-1})f(z)z^{-1} = Y_{22}(z)_- + Y_{21}(z)_- z^{-n}f(z),
    \end{split}
    \end{align}
    which implies that $Y$ satisfies condition (b). Condition (c) is obvious for $Y_{11}$ and $Y_{21}$. For $Y_{12}$ and $Y_{22}$ note that for $|z| > 1$ it holds that
    \begin{align}
        \frac{1}{\xi - z} = - \sum_{k = 0}^{\infty} \frac{\xi^k}{z^{k+1}}.
    \end{align}
    Using this, and the orthogonality relations in (\ref{eqn:orthogonality}), we see that condition (c) also holds for $Y_{12}$ and $Y_{22}$: as $z \rightarrow \infty$
    \begin{align}
    \begin{split}
        Y_{12}(z) =& \chi_n^{-1} \int_C \frac{\phi_n(\xi)}{\xi - z} \frac{f(\xi) \text{d}{\xi}}{2\pi i \xi^n} \\ 
        =& -\chi_n^{-1} \sum_{k = 0}^{\infty}  \frac{1}{2\pi i} \left( \int_C \phi_n(\xi) f(\xi) \xi^{k-n} \text{d}\xi \right) \frac{1}{z^{k +1}} \\ 
        =& -\chi_n^{-1} \sum_{k = 0}^{\infty}  \frac{1}{2\pi} \left( \int_0^{2\pi} \phi_n(e^{i\theta}) f(e^{i\theta}) e^{-i\theta(n - k - 1)} \text{d}\theta \right) \frac{1}{z^{k +1}} \\ 
        =& \mathcal{O}(z^{-n-1}) \\
        Y_{22}(z) =& -\chi_{n-1} \int_C \frac{\hat{\phi}_{n-1}(\xi^{-1})}{\xi - z} \frac{f(\xi) \text{d}{\xi}}{2\pi i \xi} \\ 
        =& \chi_{n-1} \sum_{k = 0}^{\infty} \frac{1}{2\pi i} \left( \int_C \hat{\phi}_{n-1}(\xi^{-1}) f(\xi) \xi^{k-1} \text{d}{\xi} \right) \frac{1}{z^{k+1}} \\ 
        =& \chi_{n-1} \sum_{k = 0}^{\infty} \frac{1}{2\pi} \left( \int_0^{2\pi} \hat{\phi}_{n-1}(e^{-i\theta}) f(e^{i\theta}) e^{ik\theta} \text{d}\theta \right) \frac{1}{z^{k+1}} \\ 
        =& z^{-n+1} + \mathcal{O}(z^{-n}).
    \end{split}
    \end{align}
    We now show that if a solution exists then it is unique. By condition (a) the function $\det Y : \mathbb{C} \setminus C \rightarrow \mathbb{C}$ is analytic and for $z \in C \setminus \{e^{i\theta_1},...,e^{i\theta_m}\}$, condition (b) implies that
    \begin{align}
        \left( \det Y \right)_+(z) = \det \left( Y_-(z) \left( \begin{array}{cc} 1 & z^{-n} f(z) \\ 0 & 1 \end{array} \right) \right) = \left( \det Y \right)_-(z).
    \end{align}
    Thus $\det$ is meromorphic on $\mathbb{C}$, with singularities at $e^{i\theta_1},...,e^{i\theta_m}$. Those singularities are all removable however since by condition (d) it holds that $\lim_{z \rightarrow e^{i\theta_j}} (z - e^{i\theta_j}) \det Y(z) = \lim_{z \rightarrow e^{i\theta_j}} \mathcal{O}(|z - e^{i\theta_j}|) + \mathcal{O}(|z - e^{i\theta_j}|^{1 + 2\alpha_j}) + \mathcal{O}(|z - e^{i\theta_j}| \log |z - e^{i\theta_j}|) = 0$, where in the last equality we used that $\alpha_j > -1/2$, $j = 1,...,m$. This implies that $\det Y$ extends to an entire function. By condition (c) it holds that $\det Y(z) \rightarrow 1$ as $z \rightarrow \infty$. Thus by Liouville's theorem it follows that $\det Y(z) = 1$ for all $z \in \mathbb{C}$, and $Y$ is invertible for all $z \in \mathbb{C}\setminus C$. Now suppose that there are two solutions $Y$ and $\tilde{Y}$. Then for any $C \setminus \{e^{i\theta_1},...,e^{i\theta_m}\}$ condition (b) implies that
    \begin{align}
        (Y\tilde{Y}^{-1})_+(z) = Y_-(z) \left( \begin{array}{cc} 1 & z^{-n} f(z) \\ 0 & 1 \end{array} \right) \left( \begin{array}{cc} 1 & - z^{-n} f(z) \\ 0 & 1 \end{array} \right) \tilde{Y}^{-1}_-(z) = (Y\tilde{Y}^{-1})_-(z),
    \end{align}
    and again by condition (d) the singularities of $Y\tilde{Y}^{-1}$ at the points $e^{i\theta_1},...,e^{i\theta_m}$ are removable. Thus $Y\tilde{Y}^{-1}$ extends analytically to the entire complex plane and every entry is an entire function, and since additionally by condition (c) it holds that $Y(z)\tilde{Y}^{-1}(z) \rightarrow I$ as $z \rightarrow \infty$, Liouville's theorem implies that $Y(z) = \tilde{Y}(z)$ for all $z \in \mathbb{C}\setminus C$. 
\end{proof}

We see that $Y(z;n)_{21}(0) = \chi_{n-1}^2$ and $Y(z;n)_{11}(z) = \chi_n^{-1} \phi_n(z) =: \Phi_n(z)$, thus if we know the asymptotics of $Y$, we know the asymptotics of $\Phi_n$, $\phi_n$ and $\chi_n$, from which we can then compute the asymptotics of the Toeplitz and Toeplitz+Hankel determinants $D_n(f)$, $D^{T+H, \kappa}_n(f)$, $\kappa = 1,...,4$. 

\section{Outline}

In Chapter \ref{chapter:FH} we prove new formulas for the asymptotics of Toeplitz and Toeplitz+Hankel determinants with two conjugate pairs of merging singularities, using the Deift-Zhou steepest descent method for Riemann-Hilbert problems for orthogonal polynomials. Those formulas are stated in Theorems \ref{thm:T uniform}, \ref{thm:T+H non-uniform}, and \ref{thm:T, T+H extended}. Chapter \ref{chapter:FH} is based on the Riemann-Hilbert analysis carried out in the paper "The Classical Compact Groups and Gaussian Multiplicative Chaos" \cite{Forkel2021}, which is joint work with Jon Keating.\\

In Chapter \ref{chapter:GMC} we use the formulae stated and proven in Chapter \ref{chapter:FH} to prove Theorems \ref{thm:main} and \ref{thm:main2}, which state that the convergence (\ref{eqn:GMC convergence}) holds, i.e. that powers of the exponential of the real and imaginary part of the logarithm of the characteristic polynomial of random orthogonal and symplectic matrices converge to certain Gaussian multiplicative chaos measures on the unit circle. Just like Chapter \ref{chapter:FH}, Chapter \ref{chapter:GMC} is based on the paper "The Classical Compact Groups and Gaussian Multiplicative Chaos" \cite{Forkel2021}, which is joint work with Jon Keating.\\

In Chapter \ref{chapter:MoM} we prove the asymptotic formulae for the moments of moments of random orthogonal and symplectic matrices stated in Theorems \ref{thm:MoM Sp intro} and \ref{thm:MoM O intro}. As key technical input we use Theorem \ref{thm:T+H Claeys} by Claeys et al. \cite{Claeys2022}, stated in Chapter \ref{chapter:FH}, which gives a formula for the asymptotics of Toeplitz+Hankel determinants with arbitrarily many merging Fisher-Hartwig singularities. Chapter \ref{chapter:MoM} is based on the paper "Moments of moments of the characteristic polynomials of random orthogonal and symplectic matrices" \cite{CFK23}, which is joint work with Tom Claeys and Jon Keating.\\

In Chapter \ref{chapter:UBM} we prove Theorem \ref{thm:FS}, which states that the logarithm of the characteristic polynomial of unitary Brownian motion converges in a certain tensor product of Sobolev spaces to a generalized Gaussian field whose real and imaginary parts are Gaussian free fields on the cylinder. We also prove a Wick-type identity, stated in Proposition \ref{prop:orbits}, which expresses expectations involving traces of products of random unitary and random Hermitian matrices, in terms of expectations involving only traces of random unitary matrices. Chapter \ref{chapter:UBM} is based on the paper "Convergence of the logarithm of the characteristic polynomial of unitary Brownian motion in Sobolev space" \cite{FS22}, which is joint work with Isao Sauzedde.

\chapter{Toeplitz, and Toeplitz+Hankel determinants with Fisher-Hartwig Singularities} \label{chapter:FH}

In this chapter, based on the Riemann-Hilbert analysis in \cite{Forkel2021}, which is joint work with Jon Keating, we state formulae, and establish new ones, for the asymptotics of Toeplitz and Toeplitz+Hankel determinants with Fisher-Hartwig singularities, which we will use in our proofs of Theorems \ref{thm:main} and \ref{thm:main2} in Chapter \ref{chapter:GMC} and Theorems \ref{thm:MoM Sp(2n)} and \ref{thm:MoM O(n)} in Chapter \ref{chapter:MoM}. The results we prove are for the asymptotics of Toeplitz and Toeplitz+Hankel determinants with two conjugate pairs of merging singularities. They are related to those in \cite{Claeys2015} and \cite{Deift2011} and in our proof we employ the Deift-Zhou steepest descent method for Riemann-Hilbert problems, in a way which is strongly influenced by these two papers.
 
\section{Context and Statement of Results}

\begin{definition} \label{def:FH}(\cite{Deift2011})
A function $f:S^1 \rightarrow \mathbb{C}$ is called a symbol with a fixed number of Fisher-Hartwig singularities if it has the following form:
\begin{equation}
f(z) = e^{V(z)} z^{\sum_{j=0}^{m} \beta_j} \prod_{j=0}^m |z-z_j|^{2\alpha_j} g_{z_j,\beta_j}(z) z_j^{-\beta_j}, \quad z = e^{i\theta}, \quad \theta \in [0,2\pi),
\end{equation}
for some $m = 0,1,...$, where 
\begin{align}
z_j = e^{i\theta_j}, \quad j = 0,...,m, \quad 0 = \theta_0 < \theta_1 < ... < \theta_m < 2\pi; \nonumber \\
g_{z_j,\beta_j}(z) = \begin{cases} e^{i\pi\beta_j} & 0 \leq \text{arg} \, z < \theta_j \\ e^{-i\pi\beta_j} & \theta_j \leq \text{arg} \, z < 2\pi \end{cases},\\
\Re \alpha_j > -1/2, \quad \beta_j \in \mathbb{C}, \quad j = 0,...,m, \nonumber 
\end{align}
and $V(z)$ is analytic in a neighborhood of the unit circle. A point $z_j$, $j=1,...,m$, is included if and only if either $\alpha_j \neq 0$ or $\beta_j \neq 0$, while always $z_0 = 1$, even if $\alpha_0 = \beta_0 = 0$.
\end{definition} 

Under these assumptions $V$ has a Laurent series, convergent in a neighborhood of the unit circle,
\begin{equation}
V(z) = \sum_{k = -\infty}^\infty V_kz^k, \quad V_k = \frac{1}{2\pi} \int_0^{2\pi} V(e^{i\theta})e^{-ik\theta}\text{d}\theta,
\end{equation} 
and the function $e^{V(z)}$ allows the standard Wiener-Hopf decomposition:
\begin{equation}
\hspace{-1cm} e^{V(z)} = b_+(z) b_0 b_-(z), \quad b_+(z) = e^{\sum_{k=1}^\infty V_kz^k}, \quad b_0 = e^{V_0}, \quad b_-(z) = e^{\sum_{k= -\infty}^{-1} V_k z^k}.
\end{equation}
Note that a symbol $f$ as in Definition \ref{def:FH} that is real on the unit circle and fulfills $f(e^{i\theta}) = f(e^{-i\theta})$ is of the following form:  
\begin{align} \label{eqn:FH even}
&f(z) \\
=& e^{V(z)} \prod_{j=0}^{r+1} |z - e^{i\theta_j}|^{2\alpha_j} |z - e^{-i\theta_j}|^{2\alpha_j} g_{e^{i\theta_j},\beta_j}(z) g_{e^{i(2\pi - \theta_j)}, - \beta_j}(z) e^{-i\theta_j\beta_j} e^{i(2\pi - \theta_j)\beta_j}, \nonumber 
\end{align}
where $r \in \mathbb{N} \cup \{0\}$, $0 = \theta_0 < \theta_1 < ... < \theta_r < \theta_{r+1} = \pi$, $\alpha_j > -1/2$, $\beta_j \in i\mathbb{R}$ for $j = 0,...,r+1$, and $\beta_0 = \beta_{r+1} = 0$, and where $V(e^{i\theta}) = V(e^{-i\theta})$ such that $V_k = V_{-k}$.\\
 
We recall the definition of the Toeplitz and Toeplitz+Hankel determinants we consider: 
\begin{align} \label{eqn:T+H def}
\begin{split}
D_n(f) &:= \det \left( f_{j-k} \right)_{j,k=0}^{n-1},\\
D_n^{T+H,1}(f) &:= \det \left( f_{j-k} + f_{j+k} \right)_{j,k = 0}^{n-1}, \\
D_n^{T+H,2}(f) &:= \det \left( f_{j-k} - f_{j+k+2} \right)_{j,k = 0}^{n-1}, \\
D_n^{T+H,3}(f) &:= \det \left( f_{j-k} - f_{j+k+1} \right)_{j,k = 0}^{n-1}, \\
D_n^{T+H,4}(f) &:= \det \left( f_{j-k} + f_{j+k+1} \right)_{j,k = 0}^{n-1},
\end{split}
\end{align}
where $f \in L^1(S^1)$ and
\begin{equation}
f_j = \frac{1}{2\pi} \int_0^{2\pi} f(e^{i\theta})e^{-ij\theta}\text{d}\theta.
\end{equation}

The following two theorems state the asymptotics of Toeplitz and Toeplitz+Hankel determinants with real-valued symbols, when the singularities are bounded away from each other: 

\begin{theorem}(Ehrhardt \cite{Ehrhardt2001}) \label{thm:T non-uniform}
Let $f$ be as in Definition \ref{def:FH}, with $\alpha_j \in \mathbb{R}$ and $\beta_j \in i\mathbb{R}$ for $j = 0,...,m$, and $V(S^1) \subset \mathbb{R}$, and let $G$ denote the Barnes $G$-function. Let $\epsilon > 0$, then as $n \rightarrow \infty$, uniformly for $\min_{j,k = 0,...,m} |z_j - z_k| \geq \epsilon$: 
\begin{align}
\begin{split}
\log D_n(f) =& n V_0 + \sum_{k = 0}^\infty k|V_k|^2 + \left( \log n \right) \sum_{j = 0}^m (\alpha_j^2 - \beta_j^2)  \\
&- \sum_{j=0}^m (\alpha_j - \beta_j) \left(\sum_{k = 1}^\infty V_k z_j^k\right) + (\alpha_j + \beta_j) \left( \sum_{k = 1}^\infty V_{-k} \overline{z_j}^k \right) \\
&+ \sum_{0 \leq j < k \leq m} 2(\beta_j\beta_k - \alpha_j\alpha_k) \log |z_j - z_k| + (\alpha_j\beta_k - \alpha_k \beta_j) \log \frac{z_k}{z_j e^{i\pi}} \\
&+ \sum_{j = 0}^m \log \frac{G(1+\alpha_j + \beta_j) G(1+\alpha_j - \beta_j)}{G(1+2\alpha_j)} \\
&+ o(1).
\end{split}
\end{align}
\end{theorem}
\begin{theorem} (Deift, Its, Krasovsky \cite{Deift2011}) \label{thm:T+H non-uniform}
Let $f(z)$ be as in (\ref{eqn:FH even}) and let $\epsilon > 0$. Then as $n \rightarrow \infty$, uniformly for $\min_{j,k = 0,...,r+1} |z_j - z_k| \geq \epsilon$:
\begin{align}
\begin{split} 
D_n^{T+H,\kappa}(f) =& e^{nV_0 + \frac{1}{2}\left( (\alpha_0 + \alpha_{r+1} + s'+ t' )V_0 - (\alpha_0+s')V(1) - (\alpha_{r+1}+t')V(-1) + \sum_{k=1}^\infty k V_k^2 \right)} \\
&\times \prod_{j=1}^r b_+(z_j)^{-\alpha_j+\beta_j} b_- (z_j)^{-\alpha_j - \beta_j} \\
&\times e^{-i \pi \left( \left( \alpha_0 + s' + \sum_{j=1}^r \alpha_j \right) \sum_{j=1}^r \beta_j + \sum_{1 \leq j < k \leq r}(\alpha_j \beta_k - \alpha_k \beta_j) \right)} \\
&\times 2^{(1-s'-t')n+q+\sum_{j=1}^r(\alpha_j^2-\beta_j^2)- \frac{1}{2}(\alpha_0+\alpha_{r+1}+s'+t')^2 + \frac{1}{2}(\alpha_0+\alpha_{r+1}+s'+t')} \\
&\times n^{\frac{1}{2}(\alpha_0^2+\alpha_{r+1}^2) + \alpha_0 s' + \alpha_{r+1} t' + \sum_{j=1}^r (\alpha_j^2 - \beta_j^2)}\\
&\times \prod_{1 \leq j < k \leq r} |z_j - z_k|^{-2(\alpha_j \alpha_k - \beta_j \beta_k)} |z_j - z_k^{-1}|^{-2(\alpha_j \alpha_k + \beta_j \beta_k)} \\ 
&\times \prod_{j=1}^r z_j^{2\tilde{A}\beta_j} |1-z_j^2|^{-(\alpha_j^2 + \beta_j^2)} |1-z_j|^{-2\alpha_j(\alpha_0+s')} |1+z_j|^{-2\alpha_j(\alpha_{r+1}+t')} \\
&\times \frac{\pi^{\frac{1}{2}(\alpha_0+\alpha_{r+1}+s'+t'+1)} G(1/2)^2}{G(1+\alpha_0+s') G(1 + \alpha_{r+1} + t')} \\
&\times \prod_{j = 1}^r \frac{G(1+\alpha_j + \beta_j)G(1+\alpha_j - \beta_j)}{G(1+2\alpha_j)} (1+o(1)),
\end{split}
\end{align} 
where $\log z_j = i\theta_j$, $b_+(z_j)^{-\alpha_j+\beta_j} = (-\alpha_j + \beta_j) \sum_{k = 1}^\infty V_k z_j^k$ and similarly for $b_-(z_j)$, $\tilde{A} = \frac{1}{2}(\alpha_0 + \alpha_{3} + s' + t') + \sum_{j=1}^2 \alpha_j$ and 
\begin{align} 
D_n^{T+H,1}(f) =& \det \left( f_{j-k} + f_{j+k} \right)_{j,k = 0}^{n-1}, \quad \text{with  } q = -2n + 2, \quad s' = t' = - \frac{1}{2}, \nonumber \\
D_n^{T+H,2}(f) =& \det \left( f_{j-k} - f_{j+k+2} \right)_{j,k = 0}^{n-1}, \quad \text{with  } q = 0, \quad s' = t' = \frac{1}{2}, \\
D_n^{T+H,3}(f) =& \det \left( f_{j-k} - f_{j+k+1} \right)_{j,k = 0}^{n-1}, \quad \text{with  } q = -n, \quad s' = \frac{1}{2}, \quad t' = -\frac{1}{2}, \nonumber \\
D_n^{T+H,4}(f) =& \det \left( f_{j-k} + f_{j+k+1} \right)_{j,k = 0}^{n-1}, \quad \text{with  } q = -n, \quad s' = -\frac{1}{2}, \quad t' = \frac{1}{2}. \nonumber 
\end{align}
\end{theorem}
\begin{remark}
The asymptotics in Theorem \ref{thm:T uniform} and Theorem \ref{thm:T+H uniform} were not stated in \cite{Deift2011} to hold uniformly when the singularities are bounded away from each other, but looking at their proofs one can quickly see that this is the case.
\end{remark}
Theorem 2.6 and Lemma 2.7 in \cite{Deift2011} relate Toeplitz+Hankel determinants to Toeplitz determinants and monic orthogonal polynomials, which is how Theorem \ref{thm:T+H non-uniform} is proved. This relation can be stated as follows:
\begin{lemma}[Deift, Its, Krasovsky] \label{lemma:connection between T and T+H}
Let $f$ be as in (\ref{eqn:FH even}). Then for all $n \in \mathbb{N}$:
\begin{align}
\begin{split}
D_n^{T+H,1}(f(z))^2 =& 4 \frac{\left( 1 + \Phi_{2n}(0) \right)^2}{\Phi_{2n}(1) \Phi_{2n}(-1)} D_{2n}(f(z)),\\
D_n^{T+H,2}(f(z))^2 =& \frac{\left( 1 + \Phi_{2n}(0) \right)^2}{\Phi_{2n}(1) \Phi_{2n}(-1)} D_{2n}\left( f(z) |1 - z|^2 |1 + z|^2 \right),\\
D_n^{T+H,3}(f(z))^2 =& \frac{1}{4} \frac{\left( 1 + \Phi_{2n}(0) \right)^2}{\Phi_{2n}(1) \Phi_{2n}(-1)} D_{2n}\left( f(z) |1 - z|^2 \right),\\
D_n^{T+H,4}(f(z))^2 =& \frac{1}{4}  \frac{\left( 1 + \Phi_{2n}(0) \right)^2}{\Phi_{2n}(1) \Phi_{2n}(-1)} D_{2n}\left( f(z) |1 + z|^2 \right),
\end{split}
\end{align}
where $\Phi_{n}(z) = z^n + ...$ are the monic orthogonal polynomials w.r.t. the symbols on the RHS, i.e. $f(z)$, $f(z)|1-z|^2 |1+z|^2$, $f(z)|1-z|^2$ and $f(z) |1+z|^2$.
\end{lemma}
We will use Lemma \ref{lemma:connection between T and T+H} in Section \ref{section:T+H} to prove our results on the asymptotics of Toeplitz+Hankel determinants with merging singularities, stated in Theorems \ref{thm:T+H uniform} and \ref{thm:T, T+H extended} below, using our results on the asymptotics of Toeplitz determinants with merging singularities, stated in Theorems \ref{thm:T uniform} and \ref{thm:T, T+H extended} below. In those theorems we consider the following class of symbols: let $\epsilon \in (0,\pi/2)$ and define
\begin{align} \label{eqn:f}
f_{p,t}(z) :=& e^{V(z)} z^{\sum_{j=0}^{5} \beta_j} \prod_{j=0}^5 |z - z_j|^{2\alpha_j} g_{z_j,\beta_j}(z) z_j^{-\beta_j}, \quad z = e^{i\theta}, \quad \theta \in [0,2\pi), \\
=& e^{V(z)} |z - 1|^{2\alpha_0} |z+1|^{2\alpha_3} \prod_{j=1}^2 |z - z_j|^{2\alpha_j} |z - \overline{z_j}|^{2\alpha_j} g_{z_j,\beta_j}(z) g_{\overline{z_j}, - \beta_j}(z) z_j^{-\beta_j} \overline{z_j}^{\beta_j}, \nonumber 
\end{align}
where
\begin{itemize}
\item
$z_0 = 1$, $z_1 = e^{i(p-t)}$, $z_2 = e^{i(p+t)}$, $z_3 = -1$, $z_4 = \overline{z_2} = e^{i(2\pi-p-t)}$, $z_5 = \overline{z_1} = e^{i(2\pi - p +t)}$, with $p \in (\epsilon, \pi - \epsilon)$, $0 < t < \epsilon$, 
\item $\alpha_j \in (-1/2,\infty)$ for $j = 0,...,5$, and $\alpha_1 = \alpha_5$, $\alpha_2 = \alpha_4$, 
\item $\beta_0 = \beta_3 = 0$, $\beta_1 = - \beta_5 \in i\mathbb{R}$, $\beta_2 = - \beta_4 \in i\mathbb{R}$,
\item $V(z)$ is real-valued on the unit circle, and satisfies $V(e^{i\theta}) = V(e^{-i\theta})$. 
\end{itemize}
To state our results on the uniform asymptotics of $D_n(f_{p,t})$ and $D_n^{T+H,\kappa}(f_{p,t})$ we further need the following theorem (not in its most general form) from \cite{Claeys2015}, which describes the relevant Painlev\'e transcendents: 

\begin{theorem} (Claeys, Krasovsky) \label{thm:painleve} 
Let $\alpha_1, \alpha_2, \alpha_1 + \alpha_2 > -\frac{1}{2}$, $\beta_1,\beta_2 \in i\mathbb{R}$ and consider the $\sigma$-form of the Painleve V equation
\begin{equation} \label{eqn:painleve}
s^2 \sigma^2_{ss} = (\sigma - s \sigma_s + 2 \sigma_s^2)^2 - 4(\sigma_s - \theta_1)(\sigma_s - \theta_2)(\sigma_s - \theta_3)(\sigma_s - \theta_4),
\end{equation}
where the parameters $\theta_1,\theta_2,\theta_3,\theta_4$ are given by 
\begin{align}
\begin{split}
\theta_1 =& - \alpha_2 + \frac{\beta_1 + \beta_2}{2}, \quad \theta_2 = \alpha_2 + \frac{\beta_1 + \beta_2}{2}, \\
\theta_3 =& \alpha_1 - \frac{\beta_1 - \beta_2}{2}, \quad \theta_4 = - \alpha_1 - \frac{\beta_1 + \beta_2}{2}.
\end{split}
\end{align} 
Then there exists a solution $\sigma(s)$ to (\ref{eqn:painleve}) which is real and free of poles for $s \in -i\mathbb{R}_+$, and which has the following asymptotic behavior along the negative imaginary axis:
\begin{align}
\begin{split}
\sigma(s) =& 2\alpha_1\alpha_2 - \frac{(\beta_1 + \beta_2)^2}{2} + \mathcal{O}(|s|^\delta), \quad s \rightarrow -i0_+,\\
\sigma(s) =& \frac{\beta_1 - \beta_2}{2} s - \frac{(\beta_1 - \beta_2)^2}{2} + \mathcal{O}(|s|^{-\delta}), \quad s \rightarrow -i\infty,
\end{split}
\end{align}
for some $\delta > 0$.
\end{theorem}

Our result on the uniform asymptotics of $D_n(f_{p,t})$ is then the following:
\begin{theorem} \label{thm:T uniform}
Let $f_{p,t}$ be as in (\ref{eqn:f}) with $\alpha_1 + \alpha_2 > -1/2$, and let $\sigma$ satisfy the conditions of Theorem \ref{thm:painleve}. Then we have the following large $n$ asymptotics, uniformly for $p \in (\epsilon, \pi - \epsilon)$ and $0 < t < t_0$, for a sufficiently small $t_0 \in (0,\epsilon)$:
%\begin{small}
\begin{align}
\log D_n(f_{p,t}) =& 2int(\beta_1 - \beta_2) + n V_0 + \sum_{k = 1}^\infty kV_k^2 + \log (n) \sum_{j = 0}^5 (\alpha_j^2 - \beta_j^2) \nonumber \\
&- \sum_{j=0}^5 (\alpha_j - \beta_j) \left(\sum_{k = 1}^\infty V_k z_j^k\right) + (\alpha_j + \beta_j) \left( \sum_{k = 1}^\infty V_k \overline{z_j}^k \right) \nonumber \\
&+ \sum_{\substack{0 \leq j < k \leq 5 \\ (j,k) \neq (1,2),(4,5)}} 2(\beta_j\beta_k - \alpha_j\alpha_k) \log |z_j - z_k| + (\alpha_j\beta_k - \alpha_k \beta_j) \log \frac{z_k}{z_j e^{i\pi}} \nonumber \\
&+4it(\alpha_1\beta_2 - \alpha_2\beta_1) \nonumber \\
&+ 2\int_0^{-2int} \frac{1}{s} \left( \sigma(s) - 2\alpha_1\alpha_2 + \frac{1}{2}(\beta_1+\beta_2)^2 \right) \text{d}s + 4\left( \beta_1 \beta_2 - \alpha_1\alpha_2 \right) \log \frac{\sin t}{nt} \nonumber \\ 
&+ \log \frac{G(1+\alpha_0)^2}{G(1+2\alpha_0)} + \log \frac{G(1+\alpha_3)^2}{G(1+2\alpha_3)} \nonumber \\
&+ \log \frac{G(1+\alpha_1+\alpha_2+\beta_1+\beta_2)^2 G(1+\alpha_1+\alpha_2-\beta_1-\beta_2)^2}{G(1+2\alpha_1+2\alpha_2)^2} \\
&+ o(1), \nonumber 
\end{align}
%\end{small}
where $\log \frac{z_k}{z_j e^{i\pi}} = i(\theta_k - \theta_j - \pi)$. 
\end{theorem}

Our result on the uniform asymptotics of $D_n^{T+H,\kappa}(f_{p,t})$, $\kappa = 1,...,4$, is the following, where we use the same notation as in Theorem \ref{thm:T+H non-uniform}:

\newpage
\begin{theorem}  \label{thm:T+H uniform}
Let $f_{p,t}$ be as in (\ref{eqn:f}) with $\alpha_1 + \alpha_2 > -1/2$, and let $\sigma$ satisfy the conditions of Theorem \ref{thm:painleve}. For $D_n^{T+H,\kappa}(f_{p,t})$ we get, as $n \rightarrow \infty$, uniformly in $p \in (\epsilon, \pi - \epsilon)$ and $0 < t < t_0$, for a sufficiently small $t_0 \in (0,\epsilon)$:
\begin{align} \label{eqn:T+H uniform}
D_n^{T+H,\kappa}(f_{p,t}) =& e^{2int(\beta_1 - \beta_2) + nV_0 + \frac{1}{2}\left( (\alpha_0 + \alpha_{3} + s'+ t' )V_0 - (\alpha_0+s')V(1) - (\alpha_{3}+t')V(-1) + \sum_{k=1}^\infty k V_k^2 \right)} \nonumber \\
&\times \prod_{j=1}^2 b_+(z_j)^{-\alpha_j+\beta_j} b_- (z_j)^{-\alpha_j - \beta_j} \nonumber e^{-i \pi \left( \alpha_0 + s' + \sum_{j=1}^2 \alpha_j \right) \sum_{j=1}^2 \beta_j } \nonumber \\
&\times 2^{(1-s'-t')n+q+\sum_{j=1}^2 (\alpha_j^2-\beta_j^2)- \frac{1}{2}(\alpha_0+\alpha_{3}+s'+t')^2 + \frac{1}{2}(\alpha_0+\alpha_{3}+s'+t')} \nonumber \\
&\times n^{\frac{1}{2}(\alpha_0^2+\alpha_{3}^2) + \alpha_0 s' + \alpha_{3} t'+ \sum_{j=1}^2 (\alpha_j^2 - \beta_j^2)} \\
&\times \left| \frac{\sin t}{2nt} \right|^{-2(\alpha_1 \alpha_2 - \beta_1 \beta_2)} |2\sin p|^{-2(\alpha_1 \alpha_2 + \beta_1 \beta_2)} e^{ \int_0^{-4int} \frac{1}{s} \left( \sigma(s) - 2\alpha_1\alpha_2 + \frac{1}{2}(\beta_1+\beta_2)^2 \right) \text{d}s } \nonumber \\
&\times \prod_{j=1}^2 z_j^{2\tilde{A}\beta_j} |1-z_j^2|^{-(\alpha_j^2 + \beta_j^2)} |1-z_j|^{-2\alpha_j(\alpha_0+s')} |1+z_j|^{-2\alpha_j(\alpha_{3}+t')} \nonumber \\
&\times \frac{\pi^{\frac{1}{2}(\alpha_0+\alpha_{3}+s'+t'+1)} G(1/2)^2}{G(1+\alpha_0+s') G(1 + \alpha_{3} + t')} \nonumber \\
&\times \frac{G(1 + \alpha_1 + \alpha_2 + \beta_1 + \beta_2) G(1 + \alpha_1 + \alpha_2 - \beta_1 - \beta_2)}{G(1 + 2 \alpha_1 + 2\alpha_2 )} (1+ o(1)). \nonumber 
\end{align} 
\end{theorem}

\begin{remark} One can probably get similar results if more generally one chooses complex $\alpha_j, \beta_j$ with $\Re(\alpha_j) > -1/2$, but to prove Theorem \ref{thm:main} in Chapter \ref{chapter:GMC}, which was the original motivation for computing the asymptotics in Theorem \ref{thm:T+H uniform}, this is not necessary. 
\end{remark} 

\begin{remark} 
The requirements $p \in (\epsilon, \pi - \epsilon)$, $t_0 \in (0,\epsilon)$ are necessary for us to be able to apply the proof techniques in \cite{Claeys2015}. The results there only hold for two merging singularities, while if $p \rightarrow 0,\pi$ we have 5 singularities merging at $\pm 1$, and if $t \rightarrow \epsilon$ we can have $p \pm t \rightarrow 0,\pi$ which means 3 singularities are merging at $\pm 1$.
\end{remark}  

\begin{remark}
Comparing the uniform asymptotics of $D_n(f_{p,t})$ in Theorem \ref{thm:T uniform} with the non-uniform asymptotics one gets from Theorem \ref{thm:T non-uniform}, one can see that the different expansions are related in the following way:
\begin{align} \label{eqn:relation}
&2\sum_{j=1}^2 \log \frac{G(1+\alpha_j + \beta_j) G(1 + \alpha_j - \beta_j)}{G(1+2\alpha_j)} - 2i\pi (\alpha_1 \beta_2 -\alpha_2\beta_1) + o(1)_{non-uniform} \nonumber \\
=& 2int(\beta_1 - \beta_2) + 2\int_0^{-2int} \frac{1}{s} \left( \sigma(s) - 2\alpha_1\alpha_2 + \frac{1}{2}(\beta_1+\beta_2)^2 \right) \text{d}s + 4(\beta_1\beta_2 - \alpha_1\alpha_2) \log \frac{1}{2nt} \nonumber \\
&+ 2\log \frac{G(1 + \alpha_1 + \alpha_2 + \beta_1 + \beta_2) G(1 + \alpha_1 + \alpha_2 - \beta_1 - \beta_2)}{G(1 + 2 \alpha_1 + 2\alpha_2 )} + o(1)_{uniform}.
\end{align}
This is exactly the same relationship as the one between the non-uniform and uniform expansions of $D_n(f_t)$ in \cite{Claeys2015} (see their (1.8), (1.24) and (1.26)). 
\end{remark}

\begin{remark}
The relationship between the uniform asymptotics of $D_n^{T+H,\kappa}(f_{p,t})$ in Theorem \ref{thm:T+H uniform} and the non-uniform asymptotics one gets from Theorem \ref{thm:T+H non-uniform} is given by (\ref{eqn:relation}), with both sides divided by $2$, and $n$ replaced by $2n$. This is because the uniform asymptotics of $D_n^{T+H,\kappa}(f_{p,t})$ are related to the uniform asymptotics of $D_{2n}(f_{p,t})^{1/2}$ (with added singularities at $\pm 1$) and $\Phi_n(\pm 1)^{1/2}$, $\Phi_n(0)$, by Lemma \ref{lemma:connection between T and T+H}. As will be argued in Section \ref{section:T+H} the asymptotics of $\Phi_{2n}(\pm 1)$, computed as in \cite{Deift2011} are unaffected by the merging of singularities away from $\pm 1$, and $\Phi_{2n}(0) = o(1)$ both when singularities merge or not. Thus only the asymptotics of $D_{2n}(f_{p,t})^{1/2}$ are different in the merging and non-merging regime, and their relationship is given by (\ref{eqn:relation}) with both sides divided by $2$, and $n$ replaced by $2n$.
\end{remark}  

Our last results corresponds to Theorem 1.11 in \cite{Claeys2015}. It extends Theorem \ref{thm:T non-uniform} for the symbol $f_{p,t}$, and Theorem \ref{thm:T+H non-uniform} in the case $r = 2$: 
\begin{theorem} \label{thm:T, T+H extended}
Let $\omega(x)$ be a positive, smooth function for $x$ sufficiently large, s.t.
\begin{equation}
\omega(x) \rightarrow \infty, \quad \omega(x) = o(x), \quad \text{as } x \rightarrow \infty.
\end{equation}
Then for any $t_0 \in (0,\epsilon)$ the expansion of $D_n(f_{p,t})$ one gets from Theorem \ref{thm:T non-uniform} holds uniformly in $p \in (\epsilon, \pi - \epsilon)$ and $\omega(n)/n < t < t_0$. For $r = 2$, the expansion of Theorem \ref{thm:T+H non-uniform} holds uniformly in $\theta_1,\theta_2 \in (\epsilon, \pi - \epsilon)$ for which $\omega(n)/n < |\theta_1 - \theta_2|$. 
\end{theorem}
Finally, we state the following result which gives asymptotics of Toeplitz+Hankel determinants which hold uniformly for arbitrarily many merging singularities, but only up to a multiplicative constant:
\begin{theorem} (Claeys, Glesner, Minakov, Yang \cite{Claeys2022}) \label{thm:T+H Claeys}
Let $f$ be as in (\ref{eqn:FH even}) with $r \in \mathbb{N}$, $\alpha_0 = \alpha_{r+1} = 0$ and $\alpha_j \geq 0$, $j = 1,...,r$. Then we have uniformly over the entire region $0 < \theta_1 < ... < \theta_r < \pi$, as $n \rightarrow \infty$,
\begin{align}
D_n^{T+H,1}(f) =& Fe^{nV_0} \prod_{j = 1}^r n^{\alpha_j^2 - \beta_j^2} \left( \sin \theta_j + \frac{1}{n} \right)^{\alpha_j - \alpha_j^2 - \beta_j^2} \times e^{\mathcal{O}(1)}, \nonumber \\
D_n^{T+H,2}(f) =& Fe^{nV_0} \prod_{j = 1}^r n^{\alpha_j^2 - \beta_j^2} \left( \sin \theta_j + \frac{1}{n} \right)^{-\alpha_j - \alpha_j^2 - \beta_j^2} \times e^{\mathcal{O}(1)},\\
D_n^{T+H,3}(f) =& Fe^{nV_0} \prod_{j = 1}^r n^{\alpha_j^2 - \beta_j^2} \left( \sin \frac{\theta_j}{2} + \frac{1}{n} \right)^{-\alpha_j - \alpha_j^2 - \beta_j^2} \left( \cos \frac{\theta_j}{2} + \frac{1}{n} \right)^{\alpha_j - \alpha_j^2 - \beta_j^2} \times e^{\mathcal{O}(1)}, \nonumber \\
D_n^{T+H,4}(f) =& Fe^{nV_0} \prod_{j = 1}^r n^{\alpha_j^2 - \beta_j^2} \left( \sin \frac{\theta_j}{2} + \frac{1}{n} \right)^{\alpha_j - \alpha_j^2 - \beta_j^2} \left( \cos \frac{\theta_j}{2} + \frac{1}{n} \right)^{-\alpha_j - \alpha_j^2 - \beta_j^2} \times e^{\mathcal{O}(1)}, \nonumber 
\end{align}
where 
\begin{align}
F = \prod_{1 \leq j < k \leq r } \left( \sin \left| \frac{\theta_j - \theta_k}{2} \right| + \frac{1}{n} \right)^{-2( \alpha_j\alpha_k - \beta_j \beta_k)} \left( \sin \left| \frac{\theta_j +\theta_k}{2} \right| + \frac{1}{n} \right)^{-2(\alpha_j \alpha_k + \beta_j \beta_k)}. 
\end{align}
Here $e^{\mathcal{O}(1)}$ denotes a function which is uniformly bounded and bounded away from $0$ as $n \rightarrow \infty$. 
\end{theorem}

\section{Riemann-Hilbert Problem for a System of Orthogonal Polynomials and a Differential Identity}
In this section and and the following Sections \ref{section:asymptotics of polynomials}, \ref{section:Toeplitz} and \ref{section:T+H} we prove Theorems \ref{thm:T uniform}, \ref{thm:T+H uniform} and \ref{thm:T, T+H extended}, by following the Riemann-Hilbert analysis of \cite{Claeys2015}. 

\subsection{RHP for Orthogonal Polynomials}
Let $f_{p,t}$ be as in (\ref{eqn:f}). We recall that by the Heine-Szeg\"{o} identity, and since $f_{p,t} > 0$ except at $z_0,...,z_5$, it holds that $D_{n}(f_{p,t}) \in (0,\infty)$ for all $n\in \mathbb{N}$. Thus we can define the polynomials $\phi_n(z) = \phi_n(z;p,t)$, $\hat{\phi}_n(z) = \hat{\phi}_n(z;p,t)$, by
\begin{align}
\begin{split}
\hspace{-0.5cm} \phi_n(z) &= \frac{1}{\sqrt{D_n(f_{p,t}) D_{n+1}(f_{p,t})}} 
\left| 
\begin{array} {cccc} f_{p,t,0} & f_{p,t,-1} & \dots & f_{p,t,-n} \\
f_{p,t,1} & f_{p,t,0} & \dots & f_{p,t,-n+1} \\
\dots & \dots & & \dots \\
f_{p,t,n-1} & f_{p,t,n-2} & \dots & f_{p,t,-1} \\
1 & z & \dots & z^n 
\end{array}
\right| = \chi_n z^n + ..., \\
\hspace{-0.5cm} \hat{\phi}_n(z) &= \frac{1}{\sqrt{D_n(f_{p,t}) D_{n+1}(f_{p,t})}} 
\left| 
\begin{array} {ccccc} f_{p,t,0} & f_{p,t,-1} & \dots & f_{p,t,-n+1} & 1 \\
f_{p,t,1} & f_{p,t,0} & \dots & f_{p,t,-n+2} & z \\
\dots & \dots & & \dots & \dots \\
f_{p,t,n} & f_{p,t,n-1} & \dots & f_{p,t,1} & z^n 
\end{array}
\right| = \chi_n z^n + ..., 
\end{split}
\end{align}
where the leading coefficient $\chi_n$ is given by 
\begin{equation}
\chi_n = \sqrt{\frac{D_n(f_{p,t})}{D_{n+1}(f_{p,t,})}}.
\end{equation}
The above polynomials satisfy the orthogonality relations
\begin{align}
\begin{split}
\frac{1}{2\pi} \int_0^{2\pi} \phi_n(e^{i\theta}) e^{-ik\theta} f_{p,t}(e^{i\theta}) \text{d}\theta =& \chi_n^{-1} \delta_{nk}, \\
\frac{1}{2\pi} \int_0^{2\pi} \hat{\phi}_n(e^{-i\theta}) e^{ik\theta} f_{p,t}(e^{i\theta}) \text{d}\theta =& \chi_n^{-1} \delta_{nk},
\end{split}
\end{align}
for $k = 0,1,...,n$, which implies that they are orthonormal w.r.t. the weight $f_{p,t}$. \\

Let $C$ denote the unit circle, oriented counterclockwise. By Theorem \ref{thm:BaikDeiftJohansson} it holds that the matrix-valued function $Y(z) = Y(z;n,p,t)$ given by 
\begin{equation}
Y(z) = \left( \begin{array}{cc} \chi_n^{-1} \phi_n(z) & \chi_n^{-1} \int_C \frac{\phi_n(\xi)}{\xi - z} \frac{f_{p,t}(\xi) \text{d}{\xi}}{2\pi i \xi^n} \\
- \chi_{n-1}z^{n-1} \hat{\phi}_{n-1}(z^{-1}) &  -\chi_{n-1} \int_C \frac{\hat{\phi}_{n-1}(\xi^{-1})}{\xi - z} \frac{f_{p,t}(\xi) \text{d}{\xi}}{2\pi i \xi} \end{array} \right)
\end{equation}
is the unique solution to the following Riemann-Hilbert problem:\\

\noindent \textbf{RH problem for} $Y$

\begin{enumerate}[label=(\alph*)]
\item $Y:\mathbb{C}\setminus C \rightarrow \mathbb{C}^{2\times 2}$ is analytic.

\item The continuous boundary values of $Y$ from inside the unit circle, denoted $Y_+$, and from outside, denoted $Y_-$, exist on $C\setminus \{ z_0,...,z_5 \}$, and are related by the jump condition
\begin{equation}
Y_+(z) = Y_-(z) \left( \begin{array} {cc} 1 & z^{-n} f_{p,t}(z) \\ 0 & 1 \end{array} \right), \quad z \in C\setminus \{z_0,...,z_5\}.
\end{equation}

\item $Y(z) = (I + \mathcal{O}(1/z)) \left( \begin{array}{cc} z^n & 0 \\ 0 & z^{-n} \end{array} \right)$, as $ z \rightarrow \infty$.

\item As $z \rightarrow z_k$, $z \in \mathbb{C}\setminus C$, $k= 0,...,5$, we have 
\begin{equation}
Y(z) = \left( \begin{array}{cc} \mathcal{O}(1) & \mathcal{O}(1) + \mathcal{O}(|z-z_k|^{2\alpha_k}) \\ \mathcal{O}(1) & \mathcal{O}(1) + \mathcal{O}(|z-z_k|^{2\alpha_k}) \end{array}\right), \quad \text{if } \alpha_k \neq 0,
\end{equation}
and 
\begin{equation}
Y(z) = \left( \begin{array}{cc} \mathcal{O}(1) & \mathcal{O}(\log |z-z_k|) \\ \mathcal{O}(1) & \mathcal{O}(\log |z-z_k| ) \end{array}\right), \quad \text{if } \alpha_k = 0.
\end{equation}
\end{enumerate}

We see that $Y(z;n,p,t)_{21}(0) = \chi_{n-1}^2$ and $Y(z;n,p,t)_{11}(z) = \chi_n^{-1} \phi_n(z) = \Phi_n(z)$, thus if we know the asymptotics of $Y$, we know the asymptotics of $\Phi_n$, $\phi_n$ and $\chi_n$.

\subsection{Differential Identity}

The Fourier coefficients are differentiable in $t$, thus $\log D_n(f_{p,t})$ is differentiable in $t$ for all $p \in (\epsilon, \pi - \epsilon)$ and $n\in \mathbb{N}$. We calculate:
\begin{align}
\begin{split}
&\frac{\partial}{\partial t} \log \left|z - e^{i(p-t)}\right|^{2\alpha_1} = \frac{\partial}{\partial t} \log \left| 2\sin \frac{\theta - (p-t)}{2} \right|^{2\alpha_1}\\
=& \alpha_1 \cot \frac{\theta - (p-t)}{2} = i\alpha_1 \frac{z + e^{i(p-t)}}{z - e^{i(p-t)}}.
\end{split}
\end{align}
Similarly we obtain 
\begin{align}
\begin{split}
\frac{\partial}{\partial t} \log \left|z - e^{i(p+t)} \right|^{2\alpha_2} &= - i\alpha_2 \frac{z + e^{i(p+t)}}{z - e^{i(p+t)}},\\
\frac{\partial}{\partial t} \log \left|z - e^{i(2\pi - (p+t))} \right|^{2\alpha_4} &= i\alpha_4 \frac{z + e^{i(2\pi - (p+t))}}{z - e^{i(2\pi - (p+t))}},\\
\frac{\partial}{\partial t} \log \left|z - e^{i(2\pi - (p-t)} \right|^{2\alpha_2} &= - i\alpha_5 \frac{z + e^{i(2\pi - (p-t))}}{z - e^{i(2\pi - (p-t))}}.
\end{split}
\end{align}
Therefore we get
\begin{align} \label{eqn:f diff}
\begin{split}
\frac{\partial f(z)}{\partial t} =& if(z) \sum_{k = 1,2,4,5} q_k \left( \alpha_k \frac{z+z_k}{z-z_k} + \beta_k \right)\\ 
=& if(z) \sum_{k = 1,2,4,5} q_k \left( \alpha_k + \beta_k + \frac{2\alpha_k z_k}{z - z_k} \right) \\
=& if(z) \sum_{k = 1,2,4,5} q_k \left( \beta_k + \frac{2\alpha_k z_k}{z - z_k} \right) \\
\end{split}
\end{align}
where $q_k = 1$ for $k = 1,4$ and $q_k = -1$ for $k = 2,5$. In the last line we used that $\sum_{k = 1,2,4,5} q_k \alpha_k = 0$.\\

Set $\tilde{Y}(z) = Y(z)$ in a neighborhood of $z_k$ if $\alpha_k > 0$. If $\alpha_k < 0$ the second column of $Y$ has a term of order $(z - z_k)^{2\alpha_k}$, which explodes as $z \rightarrow z_k$. We set $\tilde{Y}_{j1} = Y_{j1}$, $j = 1,2$, $\tilde{Y}_{j2} = Y_{j2} - c_j(z-z_k)^{2\alpha_k}$ in a neighborhood of $z_k$, with $c_j$ such that $\tilde{Y}$ is bounded in that neighborhood. Then we have

\begin{proposition}
Let $n \in \mathbb{N}$ and $\alpha_k \neq 0$ for $k = 1,2,4,5$. Then the following differential identity holds:
\begin{align} \label{eqn:diff id}
\begin{split}
\frac{1}{i} \frac{\text{d}}{\text{d}t} \log D_n(f_{p,t}) =& \sum_{k = 1,2,4,5} q_k \left( n\beta_k - 2\alpha_k z_k \left( \frac{\text{d}Y^{-1}}{\text{d}z} \tilde{Y} \right)_{22} (z_k) \right), 
\end{split}
\end{align}
with $q_k$ as above and $\left( \frac{\text{d}Y^{-1}}{\text{d}z} \tilde{Y} \right)_{22} (z_k) = \lim_{z \rightarrow z_k} \left(\frac{\text{d}Y^{-1}}{\text{d}z} \tilde{Y} \right)_{22} (z)$ with $z \rightarrow z_k$ non-tangentially to the unit circle.  
\end{proposition}

\noindent \textbf{Proof:} The proof for $\alpha_k \neq 0$, $k = 1,2,4,5$ works exactly like the proof of Proposition 2.1 in \cite{Claeys2015}. We have to modify their (2.16), which we replace with our (\ref{eqn:f diff}). The singularities at $\pm 1$ are independent of $p$ and $t$ and thus always stay within $f_{p,t}$. \qed

\begin{remark}
As in Remark 2.2 in \cite{Claeys2015} one can also get a differential identity for $\log D_n(f_{p,t})$ in the case where $\alpha_k = 0$ for some $k \in \{1,2,4,5\}$, by letting those $\alpha_k$'s go to zero in (\ref{eqn:f diff}), which is continuous in $\alpha_k$ on both sides. 
\end{remark}

\section{Aymptotics of the Orthogonal Polynomials} \label{section:asymptotics of polynomials}

\subsection{Normalization of the RHP}

Set 
\begin{equation}
T(z) = \begin{cases} Y(z) \left( \begin{array}{cc} z^{-n} & 0 \\ 0 & z^n \end{array} \right), & |z| >1, \\ Y(z), & |z| < 1. \end{cases}
\end{equation}
Then by the RH conditions for $Y$, we obtain the following RH condition for $T$:\\

\noindent \textbf{RH problem for $T$}

\begin{enumerate}[label=(\alph*)]
\item $T:\mathbb{C}\setminus C \rightarrow \mathbb{C}^{2\times 2}$ is analytic.

\item The continuous boundary values of $T$ from the inside, $T_+$, and from outside, $T_-$, of the unit circle exist on $C\setminus \{ z_0,...,z_5 \}$, and are related by the jump condition
\begin{equation}
T_+(z) = T_-(z) \left( \begin{array} {cc} z^n & f_{p,t}(z) \\ 0 & z^{-n} \end{array} \right), \quad z \in C\setminus \{z_0,...,z_5\}.
\end{equation}

\item $T(z) = I + \mathcal{O}(1/z), \quad \text{as } z \rightarrow \infty$.

\item As $z \rightarrow z_k$, $z \in \mathbb{C}\setminus C$, $k= 0,...,5$, we have 
\begin{equation}
T(z) = \left( \begin{array}{cc} \mathcal{O}(1) & \mathcal{O}(1) + \mathcal{O}(|z-z_k|^{2\alpha_k}) \\ \mathcal{O}(1) & \mathcal{O}(1) + \mathcal{O}(|z-z_k|^{2\alpha_k}) \end{array}\right), \quad \text{if } \alpha_k \neq 0,
\end{equation}
and 
\begin{equation}
T(z) = \left( \begin{array}{cc} \mathcal{O}(1) & \mathcal{O}(\log |z-z_k|) \\ \mathcal{O}(1) & \mathcal{O}(\log |z-z_k| ) \end{array}\right), \quad \text{if } \alpha_k = 0.
\end{equation}
\end{enumerate}

\subsection{Opening of the Lens}

Define the Szeg\"{o} function 
\begin{equation}
D(z) = \exp \left( \frac{1}{2\pi i} \int_C \frac{\log f_{p,t}(\xi)}{\xi - z} \text{d}\xi\right),
\end{equation}
which is analytic inside and outside of $C$ and satisfies
\begin{equation} \label{eqn:D pm f}
D_+(z) = D_-(z) f_{p,t}(z), \quad z \in C\setminus \{z_0,...,z_5\}.
\end{equation}
We have (see (4.9)-(4.10) in \cite{Deift2011}):
\begin{equation} \label{eqn:D in}
D(z) = e^{\sum_0^\infty V_j z^j} \prod_{k = 0}^5 \left( \frac{z - z_k}{z_ke^{i\pi}} \right)^{\alpha_k + \beta_k} =: D_{\text{in},p,t}(z), \quad |z| < 1,
\end{equation}
and
\begin{equation} \label{eqn:D out}
D(z) = e^{-\sum_{-\infty}^{-1} V_j z^j} \prod_{k = 0}^5 \left( \frac{z - z_k}{z} \right)^{-\alpha_k + \beta_k} =: D_{\text{out},p,t}(z), \quad |z| > 1,
\end{equation}
and thus
\begin{equation}
D_{\text{out},p,t}(z)^{-1} = e^{\sum_{-\infty}^{-1} V_j z^j} \prod_{k = 0}^5 \left( \frac{z - z_k}{z} \right)^{\alpha_k - \beta_k}.
\end{equation}
The branch of $(z-z_k)^{\alpha_k \pm \beta_k}$ is fixed by the condition that $\arg (z-z_k) = 2\pi$ on the line going from $z_k$ to the right parallel to the real axis, and the branch cut is the line $\theta = \theta_k$ going from $z = z_k = e^{i\theta}$ to infinity. For any $k$, the branch cut of the root $z^{\alpha_k - \beta_k}$ is the line $\theta = \theta_k$ from $z = 0$ to infinity, and $\theta_k < \arg z < \theta_k + 2\pi$. By (\ref{eqn:D pm f}) we have that 
\begin{equation}
f_{p,t}(e^{i\theta}) = D_{\text{in},p,t}(e^{i\theta}) D_{\text{out},p,t}(e^{i\theta})^{-1},
\end{equation}
and this function extends analytically to a neighborhood $\mathcal{S}$ of the unit circle with the 6 branch cuts $z_k\mathbb{R}^+ \cap \mathcal{S}$, $k = 0,...5$, which we orient away from zero. Then we obtain for the jumps of $f_{p,t}$:
\begin{align} \label{eqn:f jumps}
\begin{split}
f_{p,t+}(z) =& f_{p,t-}(z) e^{2\pi i (\alpha_j - \beta_j)}, \quad \text{on } z_j(0,1) \cap \mathcal{S},\\
f_{p,t+}(z) =& f_{p,t-}(z) e^{-2\pi i (\alpha_j + \beta_j)}, \quad \text{on } z_j(1,\infty) \cap \mathcal{S}.
\end{split}
\end{align}

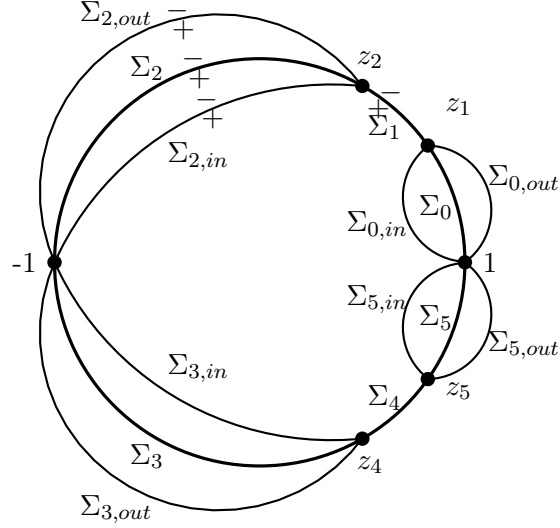
\begin{figure}[H] 
\centering
\begin{tikzpicture}[scale = 0.9]

\def\a{35} \def\b{60} \def\r{3}

\def\rsmallA{( (0.85*\r*sin(0.5*\a))^2 + (\r -0.85*\r*cos(0.5*\a))^2)^0.5}

\def\rsmallB{( (0.95*\r*sin(0.5*(\b-\a)))^2 + (\r - 0.95*\r*cos(0.5*(\b-\a)))^2)^0.5}

\def\rsmallC{( (0.4*\r*sin(0.5*(180-\b)))^2 + (\r - 0.4*\r*cos(0.5*(180-\b)))^2)^0.5}

\def\rbigC{( (0.7*\r*sin(0.5*(180-\b)))^2 + (\r + 0.7*\r*cos(0.5*(180-\b)))^2)^0.5}

\draw[name path=ellipse,black,very thick]
(0,0) circle[x radius = \r cm, y radius = \r cm];

\coordinate (1) at ({\r*cos(\a)}, {\r*sin(\a)});
\coordinate (2) at ({\r*cos(\b)}, {\r*sin(\b)});
\coordinate (4) at ({\r*cos(\b)}, {-\r*sin(\b)});
\coordinate (5) at ({\r*cos(\a)}, {-\r*sin(\a)});
	
\fill (1) circle (3pt) node[right,xshift=0.1cm,yshift=0.51cm] {$z_1$};
\fill (2) circle (3pt) node[above,xshift=0.1cm,yshift=0.1cm] {$z_2$};
\fill (4) circle (3pt) node[below,xshift=0.1cm,yshift=-0.1cm] {$z_4$};
\fill (5) circle (3pt) node[right,xshift=0.1cm,yshift=-0.15cm] {$z_5$};
\fill (\r,0) circle (3pt) node[right,xshift=0.1cm] {1};
\fill (-\r,0) circle (3pt) node[left,xshift=-0.1cm] {-1};

%From 1 to z1
\draw [black,thick,domain=-50:85] plot ({0.85*\r*cos(0.5*\a) + \rsmallA*cos(\x)},{0.85*\r*sin(0.5*\a) + \rsmallA*sin(\x)});

\draw [black,thick,domain=131:265] plot ({1.065*\r*cos(0.5*\a) + \rsmallA*cos(\x)},{1.065*\r*sin(0.5*\a) + \rsmallA*sin(\x)});

\fill ({1.35*\r*cos(0.5*\a)},{1.35*\r*sin(0.5*\a)}) node[] {$\Sigma_{0,out}$};
\fill ({0.9*\r*cos(0.5*\a)},{0.9*\r*sin(0.5*\a)}) node[] {$\Sigma_{0}$};
\fill ({0.6*\r*cos(0.5*\a)},{0.6*\r*sin(0.5*\a)}) node[] {$\Sigma_{0,in}$};

%From 1 to z5
\draw [black,thick,domain=-50:85] plot ({0.85*\r*cos(0.5*\a) + \rsmallA*cos(\x)},{0.85*\r*sin(-0.5*\a) + \rsmallA*sin(-\x)});

\draw [black,thick,domain=131:265] plot ({1.065*\r*cos(0.5*\a) + \rsmallA*cos(\x)},{1.065*\r*sin(-0.5*\a) + \rsmallA*sin(-\x)});

\fill ({1.35*\r*cos(0.5*\a)},{-1.35*\r*sin(0.5*\a)}) node[] {$\Sigma_{5,out}$};
\fill ({0.9*\r*cos(0.5*\a)},{-0.9*\r*sin(0.5*\a)}) node[] {$\Sigma_{5}$};
\fill ({0.6*\r*cos(0.5*\a)},{-0.6*\r*sin(0.5*\a)}) node[] {$\Sigma_{5,in}$};

%From z2 to -1
\draw [black,thick,domain=35:205] plot ({0.4*\r*cos(90+0.5*\b) + \rsmallC*cos(\x)},{0.4*\r*sin(90+0.5*\b) + \rsmallC*sin(\x)});

\draw [black,thick,domain=85:155] plot ({-0.7*\r*cos(90+0.5*\b) + \rbigC*cos(\x)},{-0.7*\r*sin(90+0.5*\b) + \rbigC*sin(\x)});

\fill ({1.4*\r*cos(0.5*\b+90)},{1.4*\r*sin(0.5*\b+90)}) node[] {$\Sigma_{2,out}$};
\fill ({1.1*\r*cos(0.5*\b+90)},{1.1*\r*sin(0.5*\b+90)}) node[] {$\Sigma_{2}$};
\fill ({0.6*\r*cos(0.5*\b+90)},{0.6*\r*sin(0.5*\b+90)}) node[] {$\Sigma_{2,in}$};

%+-
\fill ({0.74*\r*cos(0.6*\b+0.4*180)},{0.74*\r*sin(0.6*\b+0.4*180)}) node[] {$+$};
\fill ({0.84*\r*cos(0.6*\b+0.4*180)},{0.84*\r*sin(0.6*\b+0.4*180)}) node[] {$-$};

\fill ({0.95*\r*cos(0.6*\b+0.4*180)},{0.95*\r*sin(0.6*\b+0.4*180)}) node[] {$+$};
\fill ({1.05*\r*cos(0.6*\b+0.4*180)},{1.05*\r*sin(0.6*\b+0.4*180)}) node[] {$-$};

\fill ({1.2*\r*cos(0.6*\b+0.4*180)},{1.2*\r*sin(0.6*\b+0.4*180)}) node[] {$+$};
\fill ({1.3*\r*cos(0.6*\b+0.4*180)},{1.3*\r*sin(0.6*\b+0.4*180)}) node[] {$-$};

%From -1 to z4
\draw [black,thick,domain=35:205] plot ({0.4*\r*cos(90+0.5*\b) + \rsmallC*cos(\x)},{0.4*\r*sin(-(90+0.5*\b)) + \rsmallC*sin(-\x)});

\draw [black,thick,domain=85:155] plot ({-0.7*\r*cos(90+0.5*\b) + \rbigC*cos(\x)},{-0.7*\r*sin(-(90+0.5*\b)) + \rbigC*sin(-\x)});

\fill ({1.4*\r*cos(0.5*\b+90)},{-1.4*\r*sin(0.5*\b+90)}) node[] {$\Sigma_{3,out}$};
\fill ({1.1*\r*cos(0.5*\b+90)},{-1.1*\r*sin(0.5*\b+90)}) node[] {$\Sigma_{3}$};
\fill ({0.6*\r*cos(0.5*\b+90)},{-0.6*\r*sin(0.5*\b+90)}) node[] {$\Sigma_{3,in}$};

%From z1 to z2
\fill ({0.9*\r*cos(0.5*(\b+\a))},{0.9*\r*sin(0.5*(\b+\a))}) node[] {$\Sigma_{1}$};

%+-
\fill ({0.95*\r*cos(0.7*\b+0.3*\a)},{0.95*\r*sin(0.7*\b+0.3*\a)}) node[] {$+$};
\fill ({1.05*\r*cos(0.7*\b+0.3*\a)},{1.05*\r*sin(0.7*\b+0.3*\a)}) node[] {$-$};

%From z4 to z5
\fill ({0.9*\r*cos(0.5*(\b+\a))},{-0.9*\r*sin(0.5*(\b+\a))}) node[] {$\Sigma_{4}$};
\end{tikzpicture}
\caption{The jump contour $\Sigma_S$ of $S$.} \label{figure:S1}
\end{figure}

We factorize the jump matrix of $T$ as follows: 
\begin{equation}
\left( \begin{array}{cc} z^n & f_{p,t}(z) \\ 0 & z^{-n} \end{array} \right) = \left( \begin{array}{cc} 1 & 0 \\ z^{-n} f_{p,t}(z)^{-1} & 1 \end{array} \right) \left( \begin{array}{cc} 0 & f_{p,t}(z) \\ - f_{p,t}(z)^{-1} & 0 \end{array} \right) \left( \begin{array}{cc} 1 & 0 \\ z^n f_{p,t}(z)^{-1} & 1 \end{array} \right).
\end{equation}
We then fix a lens-shaped region as in Figure \ref{figure:S1} and define 
\begin{equation}
S(z) = \begin{cases} T(z), & \text{outside the lenses} \\ 
T(z) \left( \begin{array}{cc} 1 & 0 \\ z^{-n} f_{p,t}(z)^{-1} & 1 \end{array} \right), & \text{in the parts of the lenses outside the unit circle},\\
T(z) \left( \begin{array}{cc} 1 & 0 \\ -z^n f_{p,t}(z)^{-1} & 1 \end{array} \right), & \text{in the parts of the lenses inside the unit circle}.
\end{cases}
\end{equation}
The following RH conditions for $S$ can be verified directly:\\

\noindent \textbf{RH problem for $S$}
\begin{enumerate}[label=(\alph*)]
\item $S:\mathbb{C}\setminus \Sigma_S \rightarrow \mathbb{C}^{2\times 2}$ is analytic.\\

\item $S_+(z) = S_-(z) J_S(z)$ for $z \in \Sigma_S\setminus \{z_0,...,z_5\}$, where $J_S$ is given by  
\begin{equation}
\hspace{-1.5cm} J_S(z) = \begin{cases} \left( \begin{array}{cc} 1 & 0 \\ z^{-n} f_{p,t}(z)^{-1} & 1 \end{array} \right), & \text{on } \Sigma_{0,out} \cup \Sigma_{2,out} \cup \Sigma_{3,out} \cup \Sigma_{5,out}, \\
\left( \begin{array}{cc} 0 & f_{p,t}(z) \\ -f_{p,t}(z)^{-1} & 0 \end{array} \right), & \text{on } \Sigma_0 \cup \Sigma_2 \cup \Sigma_3 \cup \Sigma_5,\\
\left( \begin{array}{cc} 1 & 0 \\ z^n f_{p,t}(z)^{-1} & 1 \end{array} \right), & \text{on } \Sigma_{0,in} \cup \Sigma_{2,in} \cup \Sigma_{3,in} \cup \Sigma_{5,in}, \\
\left( \begin{array}{cc} z^n & f_{p,t}(z) \\ 0 & z^{-n} \end{array} \right), & \text{on } \Sigma_1 \cup \Sigma_4.
\end{cases}
\end{equation}

\item $S(z) = I + \mathcal{O}(1/z), \quad \text{as } z \rightarrow \infty$.\\

\item As $z \rightarrow z_k$ from outside the lenses, $k= 0,...,5$, we have 
\begin{equation}
S(z) = \left( \begin{array}{cc} \mathcal{O}(1) & \mathcal{O}(1) + \mathcal{O}(|z-z_k|^{2\alpha_k}) \\ \mathcal{O}(1) & \mathcal{O}(1) + \mathcal{O}(|z-z_k|^{2\alpha_k}) \end{array}\right), \quad \text{if } \alpha_k \neq 0,
\end{equation}
and
\begin{equation}
S(z) = \left( \begin{array}{cc} \mathcal{O}(1) & \mathcal{O}(\log |z-z_k|) \\ \mathcal{O}(1) & \mathcal{O}(\log |z-z_k| ) \end{array}\right), \quad \text{if } \alpha_k = 0.
\end{equation}
The behaviour of $S(z)$ as $z \rightarrow z_k$ from the other regions is obtained from these expressions by application of the appropriate jump conditions. 
\end{enumerate}

Fix $\delta_1,\delta_2 > 0$ such that the discs 
\begin{equation}
U_{\pm 1}: = \{z: |z-\pm 1| < \delta_1\}, \quad U_{\pm} := \{z: |z - e^{\pm ip}| < \delta_2\}
\end{equation}
are disjoint for any $p \in (\epsilon, \pi - \epsilon)$. Let $t_0 \in (0,\epsilon)$ such that $e^{i(p \pm t)} \in U_{+}$ and $e^{i(2\pi - (p \pm t))} \in U_{-}$ for one and hence for all $p \in (\epsilon, \pi - \epsilon)$. Then one observes that on the inner and out jump contours and outside of $U_1 \cup U_{-1} \cup U_{+} \cup U_{-}$ the jump matrix $J_S(z)$ converges to the identity matrix as $n\rightarrow \infty$, uniformly in $z$, $t < t_0$ and $p \in (\epsilon, \pi - \epsilon)$.

\subsection{Global Parametrix}

Define the function
\begin{align} \label{eqn:N}
\begin{split}
N(z) =& \begin{cases} \left( \begin{array}{cc} D_{\text{in},p,t}(z) & 0 \\ 0 & D_{\text{in},p,t}(z)^{-1} \end{array} \right) \left( \begin{array}{cc} 0 & 1 \\ -1 & 0 \end{array} \right), & |z| < 1, \\
\left( \begin{array}{cc} D_{\text{out},p,t}(z) & 0 \\ 0 & D_{\text{out},p,t}(z)^{-1} \end{array} \right), & |z| > 1. 
\end{cases}
\end{split}
\end{align}
One can easily verify that $N$ satisfies the following RH conditions:\\

\noindent \textbf{RH problem for $N$}
\begin{enumerate}[label=(\alph*)]
\item $N:\mathbb{C}\setminus{C} \rightarrow \mathbb{C}^{2\times 2}$ is analytic.\\

\item $N_+(z) = N_-(z)\left( \begin{array}{cc} 0 & f_{p,t}(z) \\ - f_{p,t}(z)^{-1} & 0 \end{array} \right)$ for $z \in C\setminus \{z_0,...,z_5\}$. \\

\item $N(z) = I + \mathcal{O}(1/z)$ as $z \rightarrow \infty$. 
\end{enumerate} 

\subsection{Local Parametrix near $\pm 1$} \label{section:local plus minus 1}

The local parametrix near $\pm 1$ are constructed in exactly the same way as in \cite{Deift2011}. We are looking for a solution of the following RHP:\\

\noindent \textbf{RH problem for $P_{\pm 1}(z)$}
\begin{enumerate}[label=(\alph*)]
\item $P_{\pm 1}: U_{\pm 1} \setminus \Sigma_S \rightarrow \mathbb{C}^{2\times 2}$ is analytic.

\item $P_{\pm 1}(z)_+ = P_{\pm 1}(z)_-J_S(z)$ for $z \in U_{\pm 1} \cap \Sigma_S$. 

\item As $z \rightarrow \pm 1$, $S(z)P_{\pm 1}(z)^{-1} = \mathcal{O}(1)$. 

\item $P_{\pm 1}$ satisfies the matching condition $P_{\pm 1}(z) N^{-1}(z) = I + o(1)$ as $n \rightarrow \infty$, uniformly in $z \in \partial U_{\pm 1}$, $p \in (\epsilon, \pi - \epsilon)$ and $0 < t < t_0$.
\end{enumerate}

$P_{\pm 1}$ is given by (4.15), (4.23), (4.24), (4.47)-(4.50) in \cite{Deift2011} and one can see from their construction that when all the other singularities are bounded away from $\pm 1$, then the matching condition is uniform in the location of the other singularities, i.e. holds uniformly in $p \in (\epsilon, \pi - \epsilon)$ and $0 < t < t_0$.

\subsection{$0 < t \leq \omega(n)/n$. Local Parametrices near $e^{\pm ip}$} \label{section:local 1}

Let $\omega(x)$ be a positive, smooth function for $x$ sufficiently large, s.t. 
\begin{equation}
\omega(x) \rightarrow \infty, \quad \omega(x) = o(x), \quad \text{as } x \rightarrow \infty.
\end{equation}
For $0 < t \leq 1/n$ and $1/n < t \leq \omega(n)/n$ we will construct local parametrices in $U_{\pm}$ which satisfy the same jump and growth conditions as $S$ inside $U_{\pm}$, and which match with the global parametrix $N$ on the boundaries $\partial U_{\pm}$ for large $n$. To be precise, we will construct $P_{\pm}$ satisfying the following conditions:\\

\noindent \textbf{RH problem for $P_{\pm}(z)$}
\begin{enumerate}[label=(\alph*)]
\item $P_{\pm}: U_{\pm} \setminus \Sigma_S \rightarrow \mathbb{C}^{2\times 2}$ is analytic.

\item $P_{\pm}(z)_+ = P_{\pm}(z)_-J_S(z)$ for $z \in U_{\pm} \cap \Sigma_S$. 

\item As $n \rightarrow \infty$, we have
\begin{equation}
P_{\pm}(z) N^{-1}(z) = (I+o(1)), \quad z \in \partial U_{\pm}, 
\end{equation}
uniformly for $p \in (\epsilon, \pi - \epsilon)$ and $0 < t < t_0$.

\item As $z \rightarrow z_k$, $S(z)P_{\pm}(z)^{-1} = \mathcal{O}(1)$, $k = 1,2$ for $+$ and $k = 4,5$ for $-$. 
\end{enumerate}

\subsubsection{RH problem for $\Phi_\pm$}
Define
\begin{align} \label{eqn:Phi Psi}
\begin{split}
\Phi_{+} (\zeta,s) =& \Psi_+(\zeta,s) \begin{cases} 1, & -1 < \Im < 1, \\
e^{\pi i(\alpha_2 - \beta_2)\sigma_3}, & \Im \zeta > 1, \\
e^{-\pi i(\alpha_1 - \beta_1)\sigma_3}, & \Im \zeta < - 1,
\end{cases}\\
\Phi_{-} (\zeta,s) =& \Psi_+(\zeta,s) \begin{cases} 1, & -1 < \Im < 1, \\
e^{\pi i(\alpha_1 + \beta_1)\sigma_3}, & \Im \zeta > 1, \\
e^{-\pi i(\alpha_2 + \beta_2)\sigma_3}, & \Im \zeta < - 1,
\end{cases}
\end{split}
\end{align}
where $\Psi_{+}(\zeta,s)$ equals $\Psi(\zeta,s)$, defined in Appendix \ref{appendix:Psi}, and $\Psi_-(\zeta,s)$ equals $\Psi(\zeta,s)$ with $(\alpha_1,\alpha_2,\beta_1,\beta_2)$ in the appendix changed to $(\alpha_4,\alpha_5,\beta_4,\beta_5) = (\alpha_2, \alpha_1, -\beta_2, -\beta_1)$. The RH conditions for $\Phi_\pm$ follow directly from the RHP for $\Psi$. \\ 

\begin{figure} [H]
\centering
\begin{tikzpicture}
%Contour 
\draw [thick] (0,-2) -- (0,2);
\draw [thick] (-3,-2) -- (3,-2);
\draw [thick] (-3,2) -- (3,2);
\draw [thick] (0,-2) -- (2,-4);
\draw [thick] (0,-2) -- (-2,-4);
\draw [thick] (0,2) -- (2,4);
\draw [thick] (0,2) -- (-2,4);

%Jump matrices
\fill (-0.3,1.7) node[] {$+i$};
\fill (-0.3,-1.7) node[] {$-i$};
\fill (-1,0) node[] {$\left( \begin{array}{cc} 0 & 1 \\ -1 & 1 \end{array} \right)$};
\fill (3,4) node[] {$\left( \begin{array}{cc} 1 & 1 \\ 0 & 1 \end{array} \right)$};
\fill (-3,4) node[] {$\left( \begin{array}{cc} 1 & 0 \\ -1 & 1 \end{array} \right)$};
\fill (4,2) node[] {$e^{\pi i(\alpha_2 + \beta_2) \sigma_3}$};
\fill (-4,2) node[] {$e^{\pi i(\alpha_2 - \beta_2) \sigma_3}$};
\fill (4,-2) node[] {$e^{\pi i(\alpha_1 + \beta_1) \sigma_3}$};
\fill (-4,-2) node[] {$e^{\pi i(\alpha_1 - \beta_1) \sigma_3}$};
\fill (3,-4) node[] {$\left( \begin{array}{cc} 1 & 1 \\ 0 & 1 \end{array} \right)$};
\fill (-3,-4) node[] {$\left( \begin{array}{cc} 1 & 0 \\ -1 & 1 \end{array} \right)$};

%Orientation of Contour
\node[fill=black,regular polygon, regular polygon sides=3,inner sep=1.pt, shape border rotate = -45] at (1,3) {};
\node[fill=black,regular polygon, regular polygon sides=3,inner sep=1.pt, shape border rotate = 45] at (-1,3) {};
\node[fill=black,regular polygon, regular polygon sides=3,inner sep=1.pt, shape border rotate = 0] at (0,1) {};
\node[fill=black,regular polygon, regular polygon sides=3,inner sep=1.pt, shape border rotate = 0] at (0,-1) {};
\node[fill=black,regular polygon, regular polygon sides=3,inner sep=1.pt, shape border rotate = 45] at (1,-3) {};
\node[fill=black,regular polygon, regular polygon sides=3,inner sep=1.pt, shape border rotate = -45] at (-1,-3) {};
\node[fill=black,regular polygon, regular polygon sides=3,inner sep=1.pt, shape border rotate = -90] at (1.5,2) {};
\node[fill=black,regular polygon, regular polygon sides=3,inner sep=1.pt, shape border rotate = -90] at (1.5,-2) {};
\node[fill=black,regular polygon, regular polygon sides=3,inner sep=1.pt, shape border rotate = -90] at (-1.5,2) {};
\node[fill=black,regular polygon, regular polygon sides=3,inner sep=1.pt, shape border rotate = -90] at (-1.5,-2) {};
\end{tikzpicture}
\caption{The jump contour $\Sigma$ and the jump matrices for $\Phi_+$. $\Phi_-$ has the same jump contour, while the jump matrices of $\Phi_-$ are given by replacing $(\alpha_1,\alpha_2, \beta_1,\beta_2)$ with $(\alpha_2,\alpha_1,-\beta_2,-\beta_1)$ in the jump matrices of $\Phi_+$.} \label{figure:Phi}
\end{figure}
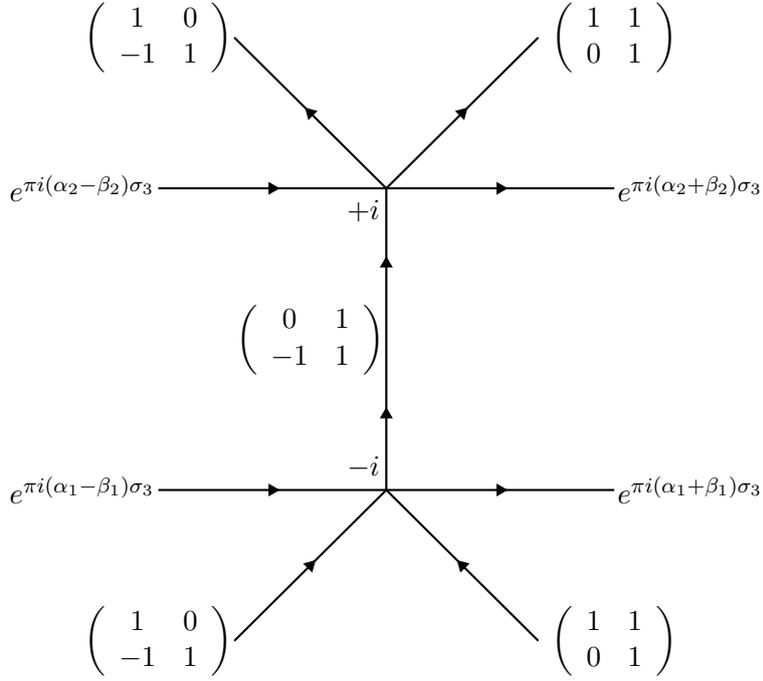

\newpage 
\noindent \textbf{RH Problem for $\Phi_\pm$}

\begin{enumerate}[label = (\alph*)]
\item $\Phi_\pm:\mathbb{C} \setminus \Sigma \rightarrow \mathbb{C}^{2\times 2}$ is analytic, where 
\begin{align}
\Sigma =& \cup_{k = 1}^9 \Sigma_k,  &&\Sigma_1 = i + e^{\frac{i\pi}{4}} \mathbb{R}_+, &&\Sigma_2 = i + e^{\frac{3i\pi}{4}\mathbb{R}_+} \nonumber \\
\Sigma_3 =& i - \mathbb{R}_+, &&\Sigma_4 = -i - \mathbb{R}_+, &&\Sigma_5 = -i + e^{- \frac{3i\pi}{4}}\mathbb{R}_+, \nonumber \\
\Sigma_6 =& -i + e^{\frac{i\pi}{4}} \mathbb{R}_+, &&\Sigma_7 = -i + \mathbb{R}_+, && \Sigma_8 = i + \mathbb{R}_+,\\
\Sigma_9 =& [-i,i], \nonumber
\end{align}
with the orientation chosen as in Figure \ref{figure:Phi} ("-" is always on the RHS of the contour).
\item $\Phi_+$ satisfies the jump conditions
\begin{equation}
\Phi_+(\zeta)_+ = \Phi_+(\zeta)_- V_k, \quad \zeta \in \Sigma_k,
\end{equation}
where 
\begin{align}
V_1 =& \left( \begin{array}{cc} 1 & 1 \\ 0 & 1 \end{array} \right), && V_2 = \left( \begin{array}{cc} 1 & 0 \\ - 1 & 1 \end{array} \right), \nonumber \\
V_3 =& e^{\pi i(\alpha_2 - \beta_2)\sigma_3}, && V_4 = e^{\pi i(\alpha_1 + \beta_1)\sigma_3}, \nonumber \\
V_5 =& \left( \begin{array}{cc} 1 & -1 \\ 0 & 1 \end{array} \right), && V_6 = \left( \begin{array}{cc} 1 & 1 \\ 0 & 1 \end{array} \right), \\
V_7 =& e^{\pi i(\alpha_2 + \beta_2)\sigma_3}, && V_8 = e^{\pi i(\alpha_1 + \beta_1)\sigma_3}, \nonumber \\
V_9 =& \left( \begin{array}{cc} 0 & 1 \\ -1 & 1 \end{array} \right). \nonumber
\end{align}
The jump conditions of $\Phi_-$ are given by replacing $(\alpha_1,\alpha_2,\beta_1,\beta_2)$ with $(\alpha_2,\alpha_1,-\beta_2,-\beta_1)$ in the jump matrices of $\Phi_+$.

\item We have in all regions:
\begin{equation} \label{eqn:Phi asymptotics}
\Phi_\pm(\zeta) = \left( I + \frac{\Psi_{\pm,1}}{\zeta} + \frac{\Psi_{\pm,2}}{\zeta^2} + \mathcal{O}(\zeta^{-3}) \right) \hat{P}^{(\infty)}_\pm (\zeta) e^{-\frac{is}{4} \zeta \sigma_3}, \quad \text{as } \zeta \rightarrow \infty,
\end{equation}
where 
\begin{align}
\begin{split}
\hat{P}^{(\infty)}_+(\zeta,s) =& P^{(\infty)}_+(\zeta,s) \begin{cases} 1, & -1 < \Im < 1, \\
e^{\pi i(\alpha_2 - \beta_2)\sigma_3}, & \Im \zeta > 1, \\
e^{-\pi i(\alpha_1 - \beta_1)\sigma_3}, & \Im \zeta < - 1,
\end{cases}\\
\hat{P}^{(\infty)}_{-} (\zeta,s) =& P^{(\infty)}_+(\zeta,s) \begin{cases} 1, & -1 < \Im < 1, \\
e^{\pi i(\alpha_1 + \beta_1)\sigma_3}, & \Im \zeta > 1, \\
e^{-\pi i(\alpha_2 + \beta_2)\sigma_3}, & \Im \zeta < - 1,
\end{cases}
\end{split}
\end{align}
with 
\begin{align}
\begin{split}
P_+^{(\infty)}(\zeta,s) =& \left( \frac{is}{2} \right)^{-(\beta_1+\beta_2)\sigma_3} (\zeta - i)^{-\beta_2 \sigma_3} (\zeta + i)^{-\beta_1 \sigma_3}, \\
P_-^{(\infty)}(\zeta,s) =& \left( \frac{is}{2} \right)^{(\beta_1+\beta_2)\sigma_3} (\zeta - i)^{\beta_1 \sigma_3} (\zeta + i)^{\beta_2 \sigma_3},
\end{split}
\end{align}
with the branches corresponding to the arguments between $0$ and $2\pi$, and where $s \in -i\mathbb{R}_+$.

\item $\Phi_\pm$ has singular behaviour near $\pm i$ which is inherited from $\Psi$. The precise conditions follow from (\ref{eqn:Phi Psi}), (\ref{eqn:F_1 neq 0}), (\ref{eqn:F_2 neq 0}), (\ref{eqn:F_1 0}) and (\ref{eqn:F_2 0}). 
\end{enumerate}

\subsubsection{$0 < t \leq \omega(n)/n$. Construction of a Local Parametrix near $e^{\pm ip}$ in terms of $\Phi_\pm$}

We choose $P_{\pm}$ as in (7.21) in \cite{Claeys2015}, i.e. 
\begin{equation} \label{eqn:P1}
P_{\pm}(z) = E_\pm (z) \Phi_{\pm} \left( \zeta;s \right) W_\pm(z), \quad \zeta = \frac{1}{t} \log \frac{z}{e^{\pm ip}}, \quad s = -2int,
\end{equation}
\begin{itemize}
\item where $\Im \log$ takes values in $(-\sigma, \sigma)$ for some $\sigma > 0$,
\item where $E_\pm$ is an analytic matrix-valued function in $U_{\pm}$, 
\item and where $W$ is given by 
\begin{align} \label{eqn:W}
\begin{split}
W_\pm(z) =& \begin{cases} -z^{\frac{n}{2}\sigma_3} f_{p,t}(z)^{-\frac{\sigma_3}{2}} \sigma_3, & \text{for } |z| < 1, \\
z^{\frac{n}{2}\sigma_3} f_{p,t}(z)^{\frac{\sigma_3}{2}} \sigma_1, & \text{for } |z| > 1,
\end{cases} \\
\sigma_1 =& \left( \begin{array}{cc} 0 & 1 \\ 1 & 0 \end{array} \right), \quad \sigma_3 = \left( \begin{array}{cc} 1 & 0 \\ 0 & -1 \end{array} \right).
\end{split}
\end{align} 
\end{itemize}
The singularities $z = e^{i(p \pm t)}$ for $+$ and $z = e^{i(2\pi - (p \mp t))}$ for $-$ correspond to the values $\zeta = \pm i$. The jumps of $W_\pm$ follow from (\ref{eqn:f jumps}):
\begin{align}
\begin{split}
W_\pm(z)_+ =& W_\pm(z)_- \left( \begin{array}{cc} 0 & f_{p,t}(z) \\ - f_{p,t}^{-1}(z) & 0 \end{array} \right), \quad z \in C, \\
W_\pm(z)_+ =& W_\pm(z)_- e^{-\pi i(\alpha_j - \beta_j)\sigma_3}, \quad z \in z_j(0,1),\\
W_\pm(z)_+ =& W_\pm(z)_- e^{\pi i(\alpha_j + \beta_j)\sigma_3}, \quad z \in z_j(1,\infty).
\end{split}
\end{align}

Choose $\Sigma_S$ such that $\frac{1}{t} \log \left( \frac{\Sigma_S \cap U_{\pm}}{e^{\pm ip}} \right) \subset \Sigma \cup i\mathbb{R}$, where $\Sigma$ is the contour of the RHP for $\Phi_\pm$, as shown in Figure \ref{figure:Phi}. Inside $U_{\pm}$ the combinded jumps of $W(z)$ and $\Phi$ are the same as the jumps of $S$:
\begin{align} \label{eqn:jumps}
P_\pm(z)_+ =& E_\pm(z) \Psi_\pm(z)_- e^{\pi i(\alpha_j - \beta_j)\sigma_3} W_\pm(z)_-e^{-\pi i(\alpha_j - \beta_j)\sigma_3} = P_\pm(z)_-, \quad z \in z_j(0,1), \nonumber \\
P_\pm(z)_+ =& E_\pm(z) \Psi_\pm(z)_- e^{\pi i(\alpha_j + \beta_j)\sigma_3} W_\pm(z)_-e^{\pi i(\alpha_j + \beta_j)\sigma_3} = P_\pm(z)_-, \quad z \in z_j(1,\infty), \nonumber \\
P_\pm(z)_+ =& E_\pm(z) \Psi_\pm(z) W_\pm(z)_- \left( \begin{array}{cc} 0 & f_{p,t}(z) \\ - f_{p,t}^{-1}(z) & 0 \end{array} \right) \nonumber \\
=& P_\pm(z)_- \left( \begin{array}{cc} 0 & f_{p,t}(z) \\ - f_{p,t}^{-1}(z) & 0 \end{array} \right), \quad z \in \Sigma_k \cap U_\pm, k = 0,2,3,5, \\  
P_\pm(z)_+ =& E_\pm(z) \Psi_\pm(z)_- \left( \begin{array}{cc} 0 & 1 \\ -1 & 1 \end{array} \right) W_\pm(z)_- \left( \begin{array}{cc} 0 & f_{p,t}(z) \\ - f_{p,t}^{-1}(z) & 0 \end{array} \right) \nonumber \\
=& P_\pm(z)_- \left( \begin{array}{cc} z^n & f_{p,t}(z) \\ 0 & z^{-n} \end{array} \right), \quad z \in \Sigma_k \cap U_\pm, k = 1,4, \nonumber
\end{align}
and
\begin{align}
\begin{split}
P_\pm(z)_+ =& E_\pm(z) \Psi_\pm(z)_- \left( \begin{array}{cc} 1 & 0 \\ -1 & 1 \end{array} \right) W_\pm(z) \\
=& P_\pm(z)_- \left( \begin{array}{cc} 1 & 0 \\ z^n f_{p,t}(z)^{-1} & 1 \end{array} \right), \quad z \in \Sigma_{k,in} \cap U_\pm, k = 0,2,3,5, \\
P_\pm(z)_+ =& E_\pm(z) \Psi_\pm(z)_- \left( \begin{array}{cc} 1 & 1 \\ 0 & 1 \end{array} \right) W_\pm(z) \\
=& P_\pm(z)_- \left( \begin{array}{cc} 1 & 0 \\ z^{-n} f_{p,t}(z)^{-1} & 1 \end{array} \right), \quad z \in \Sigma_{k,out} \cap U_\pm, k = 0,2,3,5. 
\end{split}
\end{align} 
By the condition (d) of the RHP for $S$, the singular behaviour of $W$ near $z_k$, $k = 1,2,4,5$ and condition (d) of the RHP for $\Phi_\pm$, the singularities of $S(z)P_{\pm}(z)^{-1}$ at $z_1,z_2$ for $+$, and at $z_4,z_5$ for $-$, are removable.\\

What remains is to choose $E_\pm$ such that the matching condition (c) of $P_\pm$ holds. Define
\begin{equation} \label{eqn:E}
E_\pm(z) = \sigma_1 (D_{in,p,t}(z) D_{out,p,t}(z))^{-\frac{1}{2}\sigma_3} e^{\mp ip \frac{n}{2} \sigma_3} \hat{P}_\pm^{(\infty)} (\zeta,s)^{-1},
\end{equation}
From (\ref{eqn:D in}) and (\ref{eqn:D out}) one quickly sees that the branch cuts and singularities of \\
$(D_{in,p,t}(z)D_{out,p,t}(z))^{-\frac{1}{2}\sigma_3}$ cancel out with those of $\hat{P}_\pm^{(\infty)}(z)^{-1}$, so that $E_\pm$ is analytic in $U_\pm$. In exactly the same way as in the proof of Proposition 7.1 in \cite{Claeys2015} one can now see that the matching condition (c) is satisfied, i.e. we get: 
\begin{proposition} \label{prop:matching}
As $n \rightarrow \infty$ we have
\begin{align} \label{eqn:matching1}
\begin{split}
P_\pm(z) N(z)^{-1} =& (I+\mathcal{O}(n^{-1})),
\end{split}
\end{align}
uniformly for $z \in \partial U_\pm$, $p \in (\epsilon, \pi - \epsilon)$ and $0 < t < t_0$ with $t_0$ sufficiently small.
\end{proposition}

\noindent \textbf{Proof:} Consider first the case where $c_0 \leq nt \leq C_0$, with some $c_0>0$ small and some $C_0>0$ large, which will be fixed below. Then $|\zeta| = |\frac{1}{t} \log \frac{z}{e^{\pm ip}}| > \delta n$ for $z  \in \partial U_{\pm}$, and $s = -2int$ remains bounded and bounded away from zero. Thus by (\ref{eqn:Phi Psi}), (\ref{eqn:P1}) and (\ref{eqn:Psi asymptotics}), as $n \rightarrow \infty$ 
\begin{equation}
P(z) N(z)^{-1} = E_\pm(z) (I + \mathcal{O}(n^{-1})) \hat{P}_\pm^{(\infty)}(\zeta,s) \left(\frac{z}{e^{\pm ip}} \right)^{-\frac{n}{2}\sigma_3} W_\pm(z) N(z)^{-1}, \quad z \in \partial U_\pm.  
\end{equation} 
Since the RHP for $\Psi$ is solvable for $c_0 \leq nt \leq C_0$, general properties of Painlev\'{e} RHPs imply that the error term is valid uniformly for $c_0 \leq nt \leq C_0$. By (\ref{eqn:N}) and (\ref{eqn:W}) we obtain
\begin{equation}
z^{-\frac{n}{2}\sigma_3} W_\pm(z) N(z)^{-1} = \left( D_{in,p,t}(z) D_{out,p,t}(z) \right)^{\frac{1}{2} \sigma_3} \sigma_1, \quad z \in U_\pm.
\end{equation}
Thus we have 
\begin{equation}
P_{\pm}(z) N(z)^{-1} = E_\pm(z)( I + (\mathcal{O}(n^{-1})) E_\pm(z)^{-1},
\end{equation}
and since one can quickly see that $E_\pm(z)$ is bounded uniformly for $z \in \partial U_\pm$, $p \in (\epsilon, \pi - \epsilon)$ and $c_0 \leq nt \leq C_0$, we get that (\ref{eqn:matching1}) holds uniformly $z \in \partial U_\pm$, $p \in (\epsilon, \pi - \epsilon)$ and $c_0 \leq nt \leq C_0$.

Now consider the case $C_0 < nt < \omega(n)$. In this case we cannot use the expansion (\ref{eqn:Psi asymptotics}) since the argument $s$ of $\Psi_1$ is not bounded. Instead we need to use the large $|s| = 2nt$ asymptotics for $\Psi$, which were computed in \cite{Claeys2015}[Section 5]. As is apparent from their (7.30) - (7.33), we have for $C_0$ sufficiently large
\begin{equation}
P_{\pm}(z) N(z)^{-1} = ( I + (\mathcal{O}(n^{-1})), \quad n \rightarrow \infty, \quad z \in \partial U_\pm,
\end{equation}
uniformly for $C_0/n < t < t_0$, $z \in \partial U_\pm$ and $p \in (\epsilon, \pi - \epsilon)$. 

If $nt < c_0$, we can use the small $|s|$ asymptotics for $\Psi(\zeta;s)$ for large values of $\zeta = \frac{1}{t} \log \frac{z}{e^{\pm ip}}$, as calculated in Section 6 of \cite{Claeys2015}. From their (7.34) and (7.35) we see that 
\begin{equation}
P_{\pm}(z) N(z)^{-1} = ( I + (\mathcal{O}(n^{-1})), \quad n \rightarrow \infty, \quad z \in \partial U_\pm,
\end{equation}
uniformly for $0 < t < C_0/n$, $z \in \partial U_\pm$ and $p \in (\epsilon, \pi - \epsilon)$. \qed

\subsubsection{$0 < t \leq \omega(n)/n$. Final Transformation}

\begin{figure}
\centering
\begin{tikzpicture}[scale = 1.2]

\def\a{35} \def\b{60} \def\r{3}

\def\rsmallA{( (0.85*\r*sin(0.5*\a))^2 + (\r -0.85*\r*cos(0.5*\a))^2)^0.5}

\def\rsmallB{( (0.95*\r*sin(0.5*(\b-\a)))^2 + (\r - 0.95*\r*cos(0.5*(\b-\a)))^2)^0.5}

\def\rsmallC{( (0.4*\r*sin(0.5*(180-\b)))^2 + (\r - 0.4*\r*cos(0.5*(180-\b)))^2)^0.5}

\def\rbigC{( (0.7*\r*sin(0.5*(180-\b)))^2 + (\r + 0.7*\r*cos(0.5*(180-\b)))^2)^0.5}

%U_1
\draw[name path=ellipse,black,very thick]
(\r,0) circle[x radius = 0.2*\r cm, y radius = 0.2*\r cm];

\fill ({1.3*\r},0) node[] {$U_1$};

%U_{-1}
\draw[name path=ellipse,black,very thick]
(-\r,0) circle[x radius = 0.2*\r cm, y radius = 0.2*\r cm];

\fill ({-1.3*\r},0) node[] {$U_{-1}$};

%U_+
\draw[name path=ellipse,black,very thick]
({\r*cos(0.5*(\a+\b))},{\r*sin(0.5*(\a+\b))}) circle[x radius = 0.3*\r cm, y radius = 0.3*\r cm];

\fill ({1.40*\r*cos(0.5*(\b+\a))},{1.40*\r*sin(0.5*(\b+\a))}) node[] {$U_+$};
\fill ({0.75*\r*cos(0.5*(\b+\a))},{0.75*\r*sin(0.5*(\b+\a))}) node[] {$-$}; 
\fill ({0.66*\r*cos(0.5*(\b+\a))},{0.65*\r*sin(0.5*(\b+\a))}) node[] {$+$}; 

%U_-
\draw[name path=ellipse,black,very thick]
({\r*cos(0.5*(\a+\b))},{-\r*sin(0.5*(\a+\b))}) circle[x radius = 0.3*\r cm, y radius = 0.3*\r cm];

\fill ({1.40*\r*cos(0.5*(\b+\a))},{-1.40*\r*sin(0.5*(\b+\a))}) node[] {$U_-$}; 

%From 1 to z1
\draw [black,thick,domain=-18:72] plot ({0.85*\r*cos(0.5*\a) + \rsmallA*cos(\x)},{0.85*\r*sin(0.5*\a) + \rsmallA*sin(\x)});

\draw [black,thick,domain=158:232] plot ({1.065*\r*cos(0.5*\a) + \rsmallA*cos(\x)},{1.065*\r*sin(0.5*\a) + \rsmallA*sin(\x)});

\fill ({1.35*\r*cos(0.5*\a)},{1.35*\r*sin(0.5*\a)}) node[] {$\Sigma_{0,out}$};
\fill ({0.6*\r*cos(0.5*\a)},{0.6*\r*sin(0.5*\a)}) node[] {$\Sigma_{0,in}$};

%From 1 to z5
\draw [black,thick,domain=-18:72] plot ({0.85*\r*cos(0.5*\a) + \rsmallA*cos(\x)},{0.85*\r*sin(-0.5*\a) + \rsmallA*sin(-\x)});

\draw [black,thick,domain=158:232] plot ({1.065*\r*cos(0.5*\a) + \rsmallA*cos(\x)},{1.065*\r*sin(-0.5*\a) + \rsmallA*sin(-\x)});

\fill ({1.35*\r*cos(0.5*\a)},{-1.35*\r*sin(0.5*\a)}) node[] {$\Sigma_{5,out}$};
\fill ({0.6*\r*cos(0.5*\a)},{-0.6*\r*sin(0.5*\a)}) node[] {$\Sigma_{5,in}$};

%From z2 to -1
\draw [black,thick,domain=42:190.5] plot ({0.4*\r*cos(90+0.5*\b) + \rsmallC*cos(\x)},{0.4*\r*sin(90+0.5*\b) + \rsmallC*sin(\x)});

\draw [black,thick,domain=88:148] plot ({-0.7*\r*cos(90+0.5*\b) + \rbigC*cos(\x)},{-0.7*\r*sin(90+0.5*\b) + \rbigC*sin(\x)});

\fill ({1.4*\r*cos(0.5*\b+90)},{1.4*\r*sin(0.5*\b+90)}) node[] {$\Sigma_{2,out}$};
\fill ({0.6*\r*cos(0.5*\b+90)},{0.6*\r*sin(0.5*\b+90)}) node[] {$\Sigma_{2,in}$};

%+-
\fill ({0.74*\r*cos(0.6*\b+0.4*180)},{0.74*\r*sin(0.6*\b+0.4*180)}) node[] {$+$};
\fill ({0.84*\r*cos(0.6*\b+0.4*180)},{0.84*\r*sin(0.6*\b+0.4*180)}) node[] {$-$};

\fill ({1.2*\r*cos(0.6*\b+0.4*180)},{1.2*\r*sin(0.6*\b+0.4*180)}) node[] {$+$};
\fill ({1.3*\r*cos(0.6*\b+0.4*180)},{1.3*\r*sin(0.6*\b+0.4*180)}) node[] {$-$};

%From -1 to z4
\draw [black,thick,domain=42:190.5] plot ({0.4*\r*cos(90+0.5*\b) + \rsmallC*cos(\x)},{0.4*\r*sin(-(90+0.5*\b)) + \rsmallC*sin(-\x)});

\draw [black,thick,domain=88:148] plot ({-0.7*\r*cos(90+0.5*\b) + \rbigC*cos(\x)},{-0.7*\r*sin(-(90+0.5*\b)) + \rbigC*sin(-\x)});

\fill ({1.4*\r*cos(0.5*\b+90)},{-1.4*\r*sin(0.5*\b+90)}) node[] {$\Sigma_{3,out}$};
\fill ({0.6*\r*cos(0.5*\b+90)},{-0.6*\r*sin(0.5*\b+90)}) node[] {$\Sigma_{3,in}$};
\end{tikzpicture}
\caption{The jump contour $\Sigma_R$ of $R$.} \label{figure:R1}
\end{figure}
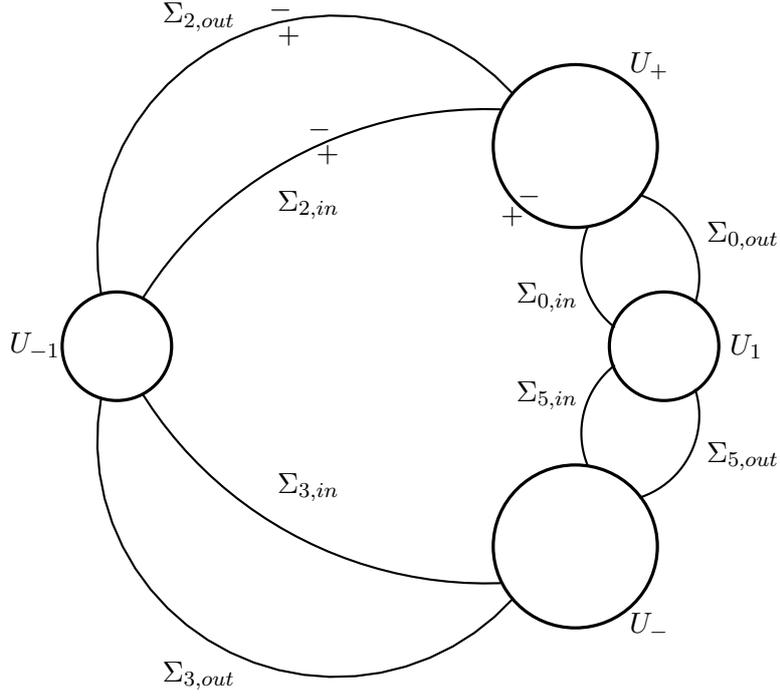

Define 
\begin{equation} \label{eqn:R}
R(z) = \begin{cases} S(z)N^{-1}(z) & z \in U_\infty \setminus \Sigma_S, \quad U_\infty := \mathbb{C}\setminus \text{local parametrices}, \\
S(z)P_{\pm }^{-1}(z), & z \in U_{\pm}\setminus \Sigma_S, \\
S(z)P_{\pm 1}^{-1}(z), & z \in U_{\pm 1} \setminus \Sigma_S.
\end{cases}
\end{equation}
Then $R$ solves the following RHP:\\
\newpage
\noindent \textbf{RH problem for $R$}
\begin{enumerate}[label=(\alph*)]
\item $R:\mathbb{C}\setminus \Sigma_R \rightarrow \mathbb{C}^{2\times 2}$ is analytic, where $\Sigma_R$ is shown in Figure \ref{figure:R1}\\

\item $R(z)$ has the following jumps: 
\begin{align}
R_+(z) &=  R_-(z)N(z)\left( \begin{array}{cc} 1 & 0 \\ f_{p,t}(z)^{-1}z^{-n} & 1 \end{array} \right) N(z)^{-1}, \quad z \in \Sigma_{j,out}, \, j = 0,2,3,5, \nonumber \\
R_+(z) &=  R_-(z)N(z)\left( \begin{array}{cc} 1 & 0 \\ f_{p,t}(z)^{-1}z^{n} & 1 \end{array} \right) N(z)^{-1}, \quad z \in \Sigma_{j,in}, \, j = 0,2,3,5, \nonumber \\
R_+(z) &=  R_-(z) P_{\pm}(z) N(z)^{-1}, \quad z \in \partial U_{\pm} \setminus \text{intersection points}, \\
R_+(z) &=  R_-(z) P_{\pm 1}(z) N(z)^{-1}, \quad z \in \partial U_{\pm 1} \setminus \text{intersection points}. \nonumber 
\end{align}

\item $R(z) = I + \mathcal{O}(1/z)$ as $z \rightarrow \infty$. 
\end{enumerate}

One quickly sees that uniformly in $z \in \Sigma_{j,out} \cup \Sigma_{j,in} \setminus \overline{U_\infty}$ we have
\begin{equation}
R_+(z) = R_-(z)(I+\mathcal{O}(e^{-\delta n}))
\end{equation}
for some $\delta > 0$ and uniformly in $p \in (\epsilon, \pi - \epsilon)$, $0 < t < t_0$.
By Proposition \ref{prop:matching} we have that 
\begin{equation}
R_+(z) = R_-(z) P_{\pm}(z) N(z)^{-1} = R_-(z)(I+\mathcal{O}(n^{-1})),
\end{equation}
uniformly for $z \in \partial U_{\pm}$, $p \in (\epsilon, \pi - \epsilon)$ and $0 < t < t_0$. Because of the matching condition (d) of $P_{\pm 1}$ we have 
\begin{equation}
R_+(z) = R_-(z) P_{\pm 1}(z) N(z)^{-1} = R_-(z)(I+\mathcal{O}(n^{-1})),
\end{equation}
uniformly for $z \in \partial U_{\pm 1}$, $p \in (\epsilon, \pi - \epsilon)$ and $0 < t < t_0$.\\

We see that we have a normalized RHP with small jumps, which by the standard theory on RHP implies that 
\begin{equation}
R(z) = I + \mathcal{O}(n^{-1}), \quad \frac{\text{d}R(z)}{\text{d}z} = \mathcal{O}(n^{-1}),
\end{equation}
as $n \rightarrow \infty$, uniformly for $z$ off the jump contour and uniformly in $p \in (\epsilon, \pi - \epsilon)$, $0 < t < t_0$.

\subsection{$\omega(n)/n < t < t_0$. Local Parametrices near $e^{\pm ip}$} \label{section:local 2}
We now transfer the construction from Section 7.5 in \cite{Claeys2015} to our setting in a completely straightforward manner. Although the parametrices $P_{\pm}$ from the previous section are valid for the whole region $0 < t < t_0$ we need to construct more explicit parametrices for the case $\omega(n)/n < t < t_0$ to get a simpler large $n$ expansion for $Y$, which is needed for the analysis in the next section. \\

In the case $\omega(n)/n < t < t_0$ $\zeta = \frac{1}{t} \log \frac{z}{e^{\pm ip}}$ is not necessarily large on $\partial U_{\pm}$. But we can construct a large $s = -int$ expansion for $Y$, as $|s| = nt$ is large. \\

We modify the $S$-RHP by now also opening up lenses around the arcs $(p-t,p+t)$ and $(2\pi - p - t, 2\pi - p + t)$, i.e. we choose the contour $\Sigma_S$ as in Figure \ref{figure:S2}.

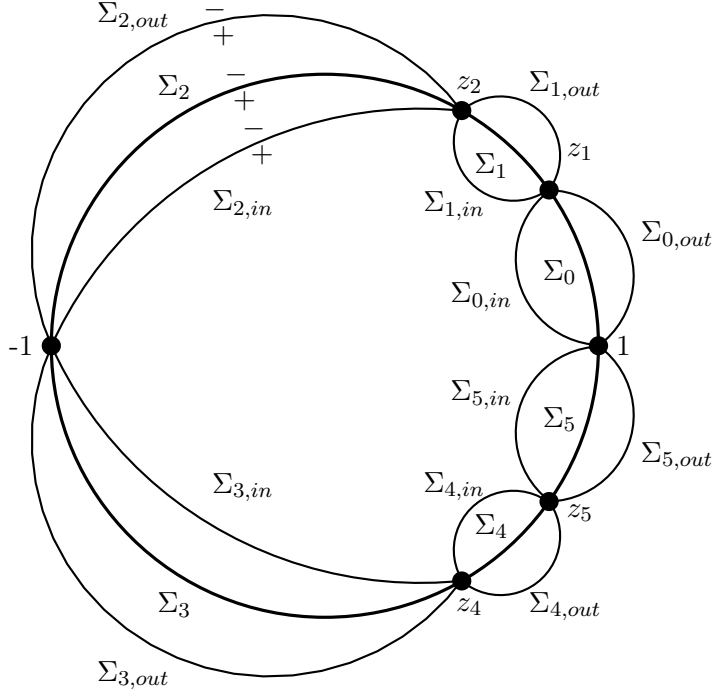
\begin{figure}[H] 
\centering
\begin{tikzpicture}[scale = 1.2]

\def\a{35} \def\b{60} \def\r{3}

\def\rsmallA{( (0.85*\r*sin(0.5*\a))^2 + (\r -0.85*\r*cos(0.5*\a))^2)^0.5}

\def\rsmallB{( (0.95*\r*sin(0.5*(\b-\a)))^2 + (\r - 0.95*\r*cos(0.5*(\b-\a)))^2)^0.5}

\def\rsmallC{( (0.4*\r*sin(0.5*(180-\b)))^2 + (\r - 0.4*\r*cos(0.5*(180-\b)))^2)^0.5}

\def\rbigC{( (0.7*\r*sin(0.5*(180-\b)))^2 + (\r + 0.7*\r*cos(0.5*(180-\b)))^2)^0.5}

\draw[name path=ellipse,black,very thick]
(0,0) circle[x radius = \r cm, y radius = \r cm];

\coordinate (1) at ({\r*cos(\a)}, {\r*sin(\a)});
\coordinate (2) at ({\r*cos(\b)}, {\r*sin(\b)});
\coordinate (4) at ({\r*cos(\b)}, {-\r*sin(\b)});
\coordinate (5) at ({\r*cos(\a)}, {-\r*sin(\a)});
	
\fill (1) circle (3pt) node[right,xshift=0.1cm,yshift=0.51cm] {$z_1$};
\fill (2) circle (3pt) node[above,xshift=0.1cm,yshift=0.1cm] {$z_2$};
\fill (4) circle (3pt) node[below,xshift=0.1cm,yshift=-0.1cm] {$z_4$};
\fill (5) circle (3pt) node[right,xshift=0.1cm,yshift=-0.15cm] {$z_5$};
\fill (\r,0) circle (3pt) node[right,xshift=0.1cm] {1};
\fill (-\r,0) circle (3pt) node[left,xshift=-0.1cm] {-1};

%From 1 to z1
\draw [black,thick,domain=-50:85] plot ({0.85*\r*cos(0.5*\a) + \rsmallA*cos(\x)},{0.85*\r*sin(0.5*\a) + \rsmallA*sin(\x)});

\draw [black,thick,domain=131:265] plot ({1.065*\r*cos(0.5*\a) + \rsmallA*cos(\x)},{1.065*\r*sin(0.5*\a) + \rsmallA*sin(\x)});

\fill ({1.35*\r*cos(0.5*\a)},{1.35*\r*sin(0.5*\a)}) node[] {$\Sigma_{0,out}$};
\fill ({0.9*\r*cos(0.5*\a)},{0.9*\r*sin(0.5*\a)}) node[] {$\Sigma_{0}$};
\fill ({0.6*\r*cos(0.5*\a)},{0.6*\r*sin(0.5*\a)}) node[] {$\Sigma_{0,in}$};

%From 1 to z5
\draw [black,thick,domain=-50:85] plot ({0.85*\r*cos(0.5*\a) + \rsmallA*cos(\x)},{0.85*\r*sin(-0.5*\a) + \rsmallA*sin(-\x)});

\draw [black,thick,domain=131:265] plot ({1.065*\r*cos(0.5*\a) + \rsmallA*cos(\x)},{1.065*\r*sin(-0.5*\a) + \rsmallA*sin(-\x)});

\fill ({1.35*\r*cos(0.5*\a)},{-1.35*\r*sin(0.5*\a)}) node[] {$\Sigma_{5,out}$};
\fill ({0.9*\r*cos(0.5*\a)},{-0.9*\r*sin(0.5*\a)}) node[] {$\Sigma_{5}$};
\fill ({0.6*\r*cos(0.5*\a)},{-0.6*\r*sin(0.5*\a)}) node[] {$\Sigma_{5,in}$};

%From z2 to -1
\draw [black,thick,domain=35:205] plot ({0.4*\r*cos(90+0.5*\b) + \rsmallC*cos(\x)},{0.4*\r*sin(90+0.5*\b) + \rsmallC*sin(\x)});

\draw [black,thick,domain=85:155] plot ({-0.7*\r*cos(90+0.5*\b) + \rbigC*cos(\x)},{-0.7*\r*sin(90+0.5*\b) + \rbigC*sin(\x)});

\fill ({1.4*\r*cos(0.5*\b+90)},{1.4*\r*sin(0.5*\b+90)}) node[] {$\Sigma_{2,out}$};
\fill ({1.1*\r*cos(0.5*\b+90)},{1.1*\r*sin(0.5*\b+90)}) node[] {$\Sigma_{2}$};
\fill ({0.6*\r*cos(0.5*\b+90)},{0.6*\r*sin(0.5*\b+90)}) node[] {$\Sigma_{2,in}$};

%+-
\fill ({0.74*\r*cos(0.6*\b+0.4*180)},{0.74*\r*sin(0.6*\b+0.4*180)}) node[] {$+$};
\fill ({0.84*\r*cos(0.6*\b+0.4*180)},{0.84*\r*sin(0.6*\b+0.4*180)}) node[] {$-$};

\fill ({0.95*\r*cos(0.6*\b+0.4*180)},{0.95*\r*sin(0.6*\b+0.4*180)}) node[] {$+$};
\fill ({1.05*\r*cos(0.6*\b+0.4*180)},{1.05*\r*sin(0.6*\b+0.4*180)}) node[] {$-$};

\fill ({1.2*\r*cos(0.6*\b+0.4*180)},{1.2*\r*sin(0.6*\b+0.4*180)}) node[] {$+$};
\fill ({1.3*\r*cos(0.6*\b+0.4*180)},{1.3*\r*sin(0.6*\b+0.4*180)}) node[] {$-$};

%From -1 to z4
\draw [black,thick,domain=35:205] plot ({0.4*\r*cos(90+0.5*\b) + \rsmallC*cos(\x)},{0.4*\r*sin(-(90+0.5*\b)) + \rsmallC*sin(-\x)});

\draw [black,thick,domain=85:155] plot ({-0.7*\r*cos(90+0.5*\b) + \rbigC*cos(\x)},{-0.7*\r*sin(-(90+0.5*\b)) + \rbigC*sin(-\x)});

\fill ({1.4*\r*cos(0.5*\b+90)},{-1.4*\r*sin(0.5*\b+90)}) node[] {$\Sigma_{3,out}$};
\fill ({1.1*\r*cos(0.5*\b+90)},{-1.1*\r*sin(0.5*\b+90)}) node[] {$\Sigma_{3}$};
\fill ({0.6*\r*cos(0.5*\b+90)},{-0.6*\r*sin(0.5*\b+90)}) node[] {$\Sigma_{3,in}$};

%From z1 to z2
\draw [black,thick,domain=-30:128] plot ({0.95*\r*cos(0.5*(\a+\b)) + \rsmallB*cos(\x)},{0.95*\r*sin(0.5*(\a+\b)) + \rsmallB*sin(\x)});

\draw [black,thick,domain=142:300] plot ({1.02*\r*cos(0.5*(\a+\b))+\rsmallB*cos(\x)},{1.02*\r*sin(0.5*(\a+\b)) + \rsmallB*sin(\x)});

\fill ({1.3*\r*cos(0.5*(\b+\a))},{1.3*\r*sin(0.5*(\b+\a))}) node[] {$\Sigma_{1,out}$};
\fill ({0.9*\r*cos(0.5*(\b+\a))},{0.9*\r*sin(0.5*(\b+\a))}) node[] {$\Sigma_{1}$};
\fill ({0.7*\r*cos(0.5*(\b+\a))},{0.7*\r*sin(0.5*(\b+\a))}) node[] {$\Sigma_{1,in}$};

%From z4 to z5
\draw [black,thick,domain=-30:128] plot ({0.95*\r*cos(0.5*(\a+\b)) + \rsmallB*cos(\x)},{-0.95*\r*sin(0.5*(\a+\b)) + \rsmallB*sin(-\x)});

\draw [black,thick,domain=142:300] plot ({1.02*\r*cos(0.5*(\a+\b))+\rsmallB*cos(\x)},{-1.02*\r*sin(0.5*(\a+\b)) + \rsmallB*sin(-\x)});

\fill ({1.3*\r*cos(0.5*(\b+\a))},{-1.3*\r*sin(0.5*(\b+\a))}) node[] {$\Sigma_{4,out}$};
\fill ({0.9*\r*cos(0.5*(\b+\a))},{-0.9*\r*sin(0.5*(\b+\a))}) node[] {$\Sigma_{4}$};
\fill ({0.7*\r*cos(0.5*(\b+\a))},{-0.7*\r*sin(0.5*(\b+\a))}) node[] {$\Sigma_{4,in}$};
\end{tikzpicture}
\caption{The modified jump contour $\Sigma_S$ of $S$ in the case $\omega(n)/n < t < t_0$. The difference compared to Figure \ref{figure:S1} is that here there are also lenses around the arcs $(p-t,p+t)$ and $(2\pi - p - t, 2\pi - p + t)$.} \label{figure:S2}
\end{figure}

The points $z_0 = 1$, $z_3 = -1$ we surround with disks $U_1$, $U_{-1}$, small enough such for all $p \in (\epsilon, \pi - \epsilon)$ and $\omega(n)/n < t < t_0$ they are disjoint with the neighborhoods $\tilde{\mathcal{U}}_1, \tilde{\mathcal{U}}_2, \tilde{\mathcal{U}}_4, \tilde{\mathcal{U}}_5$ defined in the next paragraph, and we take the same local parametrices $P_{\pm 1}$ in $U_1$, $U_{-1}$ as in Section \ref{section:local plus minus 1}.\\ 

Let $\mathcal{U}_1, \mathcal{U}_2$ be small non-intersecting disks around $\pm i$, those are the same neighborhoods as in Section 5 of \cite{Claeys2015}. We surround the points $z_1 = e^{i(p-t)}, z_2 = e^{i(p+t)}$ by small neighborhoods $\tilde{\mathcal{U}}_1, \tilde{\mathcal{U}}_2$, with $\tilde{\mathcal{U}}_1$ being the image of $\mathcal{U}_2$ under the inverse of the map $\zeta = \frac{1}{t} \log \frac{z}{e^{ip}}$, and $\tilde{\mathcal{U}}_2$ being the image of $\mathcal{U}_1$ under the same map. Similarly we surround $z_4 = e^{i(2\pi - (p+t))}$ by a small neighborhood $\tilde{\mathcal{U}}_4$, which is the image of $\mathcal{U}_2$ under the inverse of the map $\zeta = \frac{1}{t} \log \frac{z}{e^{-ip}}$, and $z_5 = e^{i(2\pi - (p-t))}$ we surround by $\tilde{\mathcal{U}}_5$ which is the image of $\mathcal{U}_1$ under the same map. Since the disks $\mathcal{U}_1$, $\mathcal{U}_2$ are fixed in the $\zeta$-plane, the neighborhoods $\tilde{\mathcal{U}}_1, \tilde{\mathcal{U}}_2, \tilde{\mathcal{U}}_4,\tilde{\mathcal{U}}_5$ contract in the $z$-plane if $t$ decreases with $n$.  \\

As global parametrix outside these neighborhoods we choose $N(z)$ as in the previous section. For $k = 1,2,4,5$ we choose the local parametrices in $\tilde{\mathcal{U}}_k$ as follows:
\begin{align} \label{eqn:P tilde}
\begin{split}
\tilde{P}_1(z) =& \tilde{E}_1(z) M^{(\alpha_1,\beta_1)}(nt(\zeta(z) + i)) \Omega_1(z) W_+(z), \quad \zeta = \frac{1}{t} \log \frac{z}{e^{ip}}, \\
\tilde{P}_2(z) =& \tilde{E}_2(z) M^{(\alpha_2,\beta_2)}(nt(\zeta(z) - i)) \Omega_2(z) W_+(z), \quad \zeta = \frac{1}{t} \log \frac{z}{e^{ip}}, \\
\tilde{P}_4(z) =& \tilde{E}_4(z) M^{(\alpha_4,\beta_4)}(nt(\zeta(z) + i)) \Omega_4(z) W_-(z), \quad \zeta = \frac{1}{t} \log \frac{z}{e^{-ip}}, \\
\tilde{P}_5(z) =& \tilde{E}_5(z) M^{(\alpha_5,\beta_5)}(nt(\zeta(z) - i)) \Omega_5(z) W_-(z), \quad \zeta = \frac{1}{t} \log \frac{z}{e^{-ip}},
\end{split}
\end{align}
where
\begin{align}
\tilde{E}_k(z) = \sigma_1 \left( D_{in,p,t}(z) D_{out,p,t}(z) \right)^{-\sigma_3/2} \Omega_k(z) (nt(\zeta \pm i))^{\beta_k \sigma_3} z_k^{-\frac{n}{2}\sigma_3},
\end{align}
with $+$ for $k = 1,4$ and $-$ for $k = 2,5$, where $M^{(\alpha_k,\beta_k)}(\lambda)$ is given in Appendix \ref{appendix:M} with $\alpha = \alpha_k, \beta = \beta_k$, where
\begin{align}
\Omega_k(z) = \begin{cases} e^{i\frac{\pi}{2}(\alpha_k - \beta_k)\sigma_3} , & \Im \zeta > 1 \\
e^{-i\frac{\pi}{2}(\alpha_k - \beta_k)\sigma_3} , & \Im \zeta < 1 \end{cases},
\end{align}
and where $W_\pm(z)$ is given in (\ref{eqn:W}). \\

By (\ref{eqn:N}) and (\ref{eqn:W}) we obtain
\begin{equation}
z^{-\frac{n}{2}\sigma_3} W_\pm(z) N(z)^{-1} = \left( D_{in,p,t}(z) D_{out,p,t}(z) \right)^{\frac{1}{2} \sigma_3} \sigma_1, \quad z \in U_\pm.
\end{equation}
Using the large argument expansion (\ref{eqn:M asymptotics}) for $M^{(\alpha_1,\beta_1)}(nt(\zeta + i))$ for $z \in \partial \tilde{\mathcal{U}}_1$, we see that
\begin{align} \label{eqn:matching tilde}
\begin{split}
\tilde{P}_1(z) N(z)^{-1} =& \tilde{E}_1(z) \left( I + \frac{M^{(\alpha_1,\beta_1)}}{nt(\zeta + i)} + \mathcal{O}((nt)^{-2} \right) \\
&\times (nt(\zeta + i))^{-\beta_1 \sigma_3} \left( \frac{z}{e^{i(p+t)}} \right)^{-\frac{n}{2}\sigma_3} \Omega_1(z) W_+(z) N(z)^{-1}\\
=& \tilde{E}_1(z) \left( I + \frac{M^{(\alpha_1,\beta_1)}}{nt(\zeta + i)} + \mathcal{O}((nt)^{-2} \right) \\
&\times (nt(\zeta + i))^{-\beta_1 \sigma_3}  z_1^{\frac{n}{2}\sigma_3} \Omega_1(z) \left( D_{in,p,t}(z) D_{out,p,t}(z) \right)^{\frac{1}{2} \sigma_3} \sigma_1 \\
=& \tilde{E}_1(z) \left( I + \mathcal{O}((nt)^{-1}) \right) \tilde{E}_1(z)^{-1} \\
=& \left( I + \mathcal{O}((nt)^{-1}) \right),
\end{split}
\end{align}
uniformly in $z \in \partial \tilde{\mathcal{U}}_1$, $p \in (\epsilon, \pi - \epsilon)$ and $\omega(n)/n < t < t_0$, since $\tilde{E}_1(z)$ is uniformly bounded for $\omega(n)/n < t < t_0$, $p \in (\epsilon, \pi - \epsilon)$ and $z \in \tilde{\mathcal{U}}_1$. Similarly one obtains that for $k = 2,4,5$
\begin{align}
\tilde{P}_k(z) N(z)^{-1} = \left( I + \mathcal{O}((nt)^{-1}) \right),
\end{align}
uniformly in $z \in \partial \tilde{\mathcal{U}}_k$, $p \in (\epsilon, \pi - \epsilon)$ and $\omega(n)/n < t < t_0$.\\

Choose $\Sigma_S$ such that $\frac{1}{t} \log \left( \frac{\Sigma_S}{z_k} \right) \subset \left( e^{\pm \frac{\pi i}{4}} \mathbb{R} \cup i\mathbb{R} \cup \mathbb{R} \right)$ in $\tilde{\mathcal{U}}_k$. Then one can easily verify, as in (\ref{eqn:jumps}),  that $\tilde{P}_k$ has the same jumps as $S$ in $\tilde{\mathcal{U}}_k$, so that $S(z) \tilde{P}_k(z)^{-1}$ is meromorphic in $\tilde{\mathcal{U}}_k$, with at most an isolated singulary at $z_k$. The singular behaviour of $S$ and $W_\pm$ near $z_k$, and of $M^{(\alpha_k,\beta_k)}$ near $0$ (given in (\ref{eqn:M at 0, neq 0}) and (\ref{eqn:M at 0, 0})), imply that $S(z) \tilde{P}_k^{-1}(z)$ is bounded at $z_k$, which shows that that $\tilde{P}_k$ is a parametrix for $S$ in $\tilde{\mathcal{U}}_k$ with the matching condition (\ref{eqn:matching tilde}) with $N(z)$ at $\partial \tilde{\mathcal{U}}_k$.  

\subsubsection{$\omega(n)/n < t < t_0$. Final Transformation}
\begin{figure}[H] 
\centering
\begin{tikzpicture}[scale = 1.2]

\def\a{35} \def\b{60} \def\r{3}

\def\rsmallA{( (0.85*\r*sin(0.5*\a))^2 + (\r -0.85*\r*cos(0.5*\a))^2)^0.5}

\def\rsmallB{( (0.95*\r*sin(0.5*(\b-\a)))^2 + (\r - 0.95*\r*cos(0.5*(\b-\a)))^2)^0.5}

\def\rsmallC{( (0.4*\r*sin(0.5*(180-\b)))^2 + (\r - 0.4*\r*cos(0.5*(180-\b)))^2)^0.5}

\def\rbigC{( (0.7*\r*sin(0.5*(180-\b)))^2 + (\r + 0.7*\r*cos(0.5*(180-\b)))^2)^0.5}

%U_1
\draw[name path=ellipse,black,very thick]
(\r,0) circle[x radius = 0.2*\r cm, y radius = 0.2*\r cm];

\fill ({1.3*\r},0) node[] {$U_1$};

%U_{-1}
\draw[name path=ellipse,black,very thick]
(-\r,0) circle[x radius = 0.2*\r cm, y radius = 0.2*\r cm];

\fill ({-1.3*\r},0) node[] {$U_{-1}$};

%U_+
\draw [black,thick,domain=3:92] plot ({0.95*\r*cos(0.5*(\a+\b)) + \rsmallB*cos(\x)},{0.95*\r*sin(0.5*(\a+\b)) + \rsmallB*sin(\x)});

\draw [black,thick,domain=187:265] plot ({1.02*\r*cos(0.5*(\a+\b))+\rsmallB*cos(\x)},{1.02*\r*sin(0.5*(\a+\b)) + \rsmallB*sin(\x)});

\draw[name path=ellipse,black,very thick]
({\r*cos(\a)},{\r*sin(\a)}) circle[x radius = 0.15*\r cm, y radius = 0.15*\r cm];

\fill ({1.25*\r*cos(\a)},{1.25*\r*sin(\a)}) node[] {$\tilde{U}_1$};
\fill ({0.9*\r*cos(\a)},{0.9*\r*sin(\a)}) node[] {$-$}; 
\fill ({0.8*\r*cos(\a)},{0.8*\r*sin(\a)}) node[] {$+$};

\draw[name path=ellipse,black,very thick]
({\r*cos(\b)},{\r*sin(\b)}) circle[x radius = 0.15*\r cm, y radius = 0.15*\r cm];
\fill ({1.25*\r*cos(\b)},{1.25*\r*sin(\b)}) node[] {$\tilde{U}_2$};

%U_-
\draw [black,thick,domain=3:92] plot ({0.95*\r*cos(0.5*(\a+\b)) + \rsmallB*cos(\x)},{-0.95*\r*sin(0.5*(\a+\b)) + \rsmallB*sin(-\x)});

\draw [black,thick,domain=187:265] plot ({1.02*\r*cos(0.5*(\a+\b))+\rsmallB*cos(\x)},{-1.02*\r*sin(0.5*(\a+\b)) + \rsmallB*sin(-\x)});

\draw[name path=ellipse,black,very thick]
({\r*cos(\a)},{-\r*sin(\a)}) circle[x radius = 0.15*\r cm, y radius = 0.15*\r cm];
\fill ({1.25*\r*cos(\a)},{-1.25*\r*sin(\a)}) node[] {$\tilde{U}_5$};

\draw[name path=ellipse,black,very thick]
({\r*cos(\b)},{-\r*sin(\b)}) circle[x radius = 0.15*\r cm, y radius = 0.15*\r cm];
\fill ({1.25*\r*cos(\b)},{-1.25*\r*sin(\b)}) node[] {$\tilde{U}_4$};

%From 1 to z1
\draw [black,thick,domain=-18:62.5] plot ({0.85*\r*cos(0.5*\a) + \rsmallA*cos(\x)},{0.85*\r*sin(0.5*\a) + \rsmallA*sin(\x)});

\draw [black,thick,domain=156:232] plot ({1.065*\r*cos(0.5*\a) + \rsmallA*cos(\x)},{1.065*\r*sin(0.5*\a) + \rsmallA*sin(\x)});

\fill ({1.35*\r*cos(0.5*\a)},{1.35*\r*sin(0.5*\a)}) node[] {$\Sigma_{0,out}$};
\fill ({0.6*\r*cos(0.5*\a)},{0.6*\r*sin(0.5*\a)}) node[] {$\Sigma_{0,in}$};

%From 1 to z5
\draw [black,thick,domain=-18:62.5] plot ({0.85*\r*cos(0.5*\a) + \rsmallA*cos(\x)},{0.85*\r*sin(-0.5*\a) + \rsmallA*sin(-\x)});

\draw [black,thick,domain=156:232] plot ({1.065*\r*cos(0.5*\a) + \rsmallA*cos(\x)},{1.065*\r*sin(-0.5*\a) + \rsmallA*sin(-\x)});

\fill ({1.35*\r*cos(0.5*\a)},{-1.35*\r*sin(0.5*\a)}) node[] {$\Sigma_{5,out}$};
\fill ({0.6*\r*cos(0.5*\a)},{-0.6*\r*sin(0.5*\a)}) node[] {$\Sigma_{5,in}$};

%From z2 to -1
\draw [black,thick,domain=46:190.5] plot ({0.4*\r*cos(90+0.5*\b) + \rsmallC*cos(\x)},{0.4*\r*sin(90+0.5*\b) + \rsmallC*sin(\x)});

\draw [black,thick,domain=90:148] plot ({-0.7*\r*cos(90+0.5*\b) + \rbigC*cos(\x)},{-0.7*\r*sin(90+0.5*\b) + \rbigC*sin(\x)});

\fill ({1.4*\r*cos(0.5*\b+90)},{1.4*\r*sin(0.5*\b+90)}) node[] {$\Sigma_{2,out}$};
\fill ({0.6*\r*cos(0.5*\b+90)},{0.6*\r*sin(0.5*\b+90)}) node[] {$\Sigma_{2,in}$};

%+-
\fill ({0.74*\r*cos(0.6*\b+0.4*180)},{0.74*\r*sin(0.6*\b+0.4*180)}) node[] {$+$};
\fill ({0.84*\r*cos(0.6*\b+0.4*180)},{0.84*\r*sin(0.6*\b+0.4*180)}) node[] {$-$};

\fill ({1.2*\r*cos(0.6*\b+0.4*180)},{1.2*\r*sin(0.6*\b+0.4*180)}) node[] {$+$};
\fill ({1.3*\r*cos(0.6*\b+0.4*180)},{1.3*\r*sin(0.6*\b+0.4*180)}) node[] {$-$};

%From -1 to z4
\draw [black,thick,domain=46:190.5] plot ({0.4*\r*cos(90+0.5*\b) + \rsmallC*cos(\x)},{0.4*\r*sin(-(90+0.5*\b)) + \rsmallC*sin(-\x)});

\draw [black,thick,domain=90:148] plot ({-0.7*\r*cos(90+0.5*\b) + \rbigC*cos(\x)},{-0.7*\r*sin(-(90+0.5*\b)) + \rbigC*sin(-\x)});

\fill ({1.4*\r*cos(0.5*\b+90)},{-1.4*\r*sin(0.5*\b+90)}) node[] {$\Sigma_{3,out}$};
\fill ({0.6*\r*cos(0.5*\b+90)},{-0.6*\r*sin(0.5*\b+90)}) node[] {$\Sigma_{3,in}$};
\end{tikzpicture}
\caption{The jump contour $\Sigma_{\tilde{R}}$ of $\tilde{R}$ in the case $\omega(n)/n < t < t_0$.} \label{figure:R2}
\end{figure}
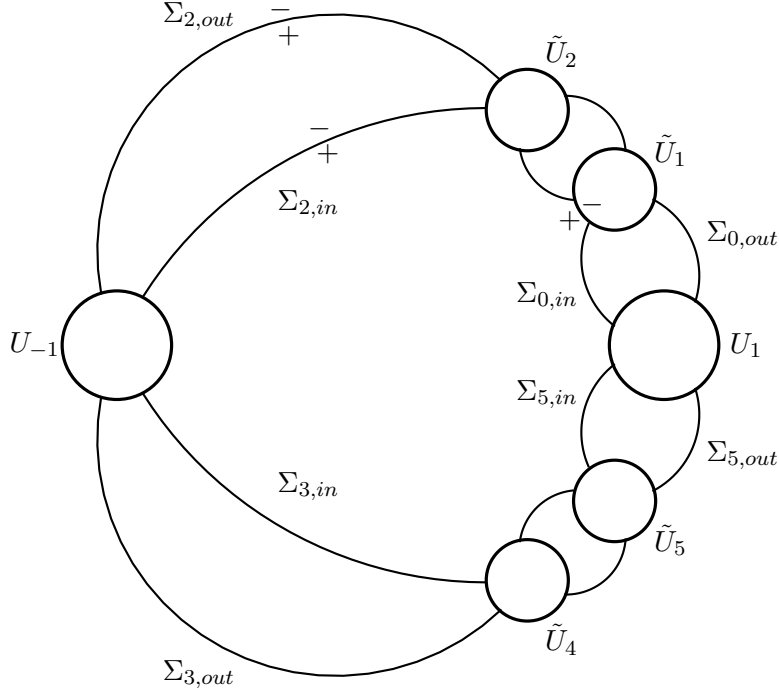

We transfer Section 7.5.1 in \cite{Claeys2015} to our case. Figure \ref{figure:R2} shows the contour chosen for the RHP of $\tilde{R}$, which we define as follows: 
\begin{align} \label{eqn:R2}
\tilde{R}(z) = \begin{cases} S(z) \tilde{P}_k(z)^{-1}, & z \in \tilde{\mathcal{U}}_k, \\
S(z)P_{\pm 1}(z)^{-1}, & z \in U_{\pm 1}, \\
S(z)N(z)^{-1}, & z \in \mathbb{C} \setminus \left( \overline{\tilde{\mathcal{U}}_1} \cup \overline{\tilde{\mathcal{U}}_2} \cup  \overline{\tilde{\mathcal{U}}_4} \cup 
\overline{\tilde{\mathcal{U}}_5} \cup \overline{U_1} \cup \overline{U_{-1}} \right). 
\end{cases}
\end{align}
Then $\tilde{R}$ is analytic, in particular has no jumps inside any of the local parametrices $\tilde{\mathcal{U}}_k$, $k  = 1,2,4,5$, $U_{\pm}$, or on the unit circle. On the rest of the lenses we can see that the jump matrix is $I + \mathcal{O}(e^{-\delta nt})$ for some $\delta>0$, uniformly in $p \in (\epsilon, \pi - \epsilon)$ and $\omega(n)/n < t < t_0$. Because of the matching condition (d) of $P_{\pm 1}$ we have as in the case $0 < t < \omega(n)/n$ that
\begin{equation}
\tilde{R}_+(z) = \tilde{R}_-(z) P_{\pm 1}(z) N(z)^{-1} = \tilde{R}_-(z)(I+\mathcal{O}(n^{-1})),
\end{equation}
uniformly for $z \in \partial U_{\pm 1}$, $p \in (\epsilon, \pi - \epsilon)$ and $0 < t < t_0$. Using (\ref{eqn:matching tilde}), we get that 
\begin{align} 
\tilde{R}_+(z) = \tilde{R}_-(z) \tilde{P}_k(z) N(z)^{-1} = \tilde{R}_-(z)(I+\mathcal{O}((nt)^{-1}),
\end{align}
uniformly for $z \in \partial \tilde{\mathcal{U}}_k$, $p \in (\epsilon, \pi - \epsilon)$ and $\omega(n)/n < t < t_0$. Finally we have that $\lim_{z \rightarrow \infty} \tilde{R}(z) = I$, which by standard theory for RHPs with small jumps and RHPs on contracting contours implies that 
\begin{align} \label{eqn:R2 asymptotics}
\tilde{R}(z) = I + \mathcal{O}((nt)^{-1}), \quad \frac{\text{d}\tilde{R}(z)}{\text{d}z} = \mathcal{O}((nt)^{-1}),
\end{align}
uniformly for $z$ off the jump contour of $\tilde{R}$, and uniformly in $p \in (\epsilon, \pi - \epsilon)$ and $\omega(n)/n < t < t_0$.

\section{Asymptotics of $D_n(f_{p,t})$} \label{section:Toeplitz}
This section is a transfer of Section 8 in \cite{Claeys2015} to our case.

\subsection{Asymptotics of the Differential Identity and Proof of Theorem \ref{thm:T, T+H extended}}

\begin{proposition}
Let $\alpha_1,\alpha_2,\alpha_1 + \alpha_2 > - \frac{1}{2}$, let $\sigma(s)$ be the solution to (\ref{eqn:painleve}) and let $\omega(x)$ be a positive, smooth function for $x$ sufficiently large, s.t.
\begin{equation}
\omega(x) \rightarrow \infty, \quad \omega(x) = o(x), \quad \text{as } x \rightarrow \infty.
\end{equation}
Then the following asymptotic expansion holds:
\begin{align} \label{eqn:diff_id_asymptotics}
\frac{1}{i} \frac{\text{d}}{\text{d}t} \log D_n(f_{p,t}) =& 2n(\beta_1 - \beta_2) + d_1(p,t;\alpha_0,\alpha_1,\beta_1,\alpha_2,\beta_2,\alpha_3) \\
&+ d_2(p,t;\alpha_0,\alpha_1,\beta_1,\alpha_2,\beta_2,\alpha_3) + d_3(p,t;\alpha_0,\alpha_1,\beta_1,\alpha_2,\beta_2,\alpha_3) + \epsilon_{n,p,t}, \nonumber 
\end{align}
where for the error term $\epsilon_{n,p,t}$
\begin{equation}
\left| \int_0^t \epsilon_{n,p,\tau} \text{d}\tau \right|= \mathcal{O}(\omega(n)^{-\delta}) = o(1),
\end{equation}
for some $\delta > 0$, uniformly in $p \in (\epsilon, \pi - \epsilon)$ and $0 < t < t_0$, and where 
\begin{align}
&d_1(p,t;\alpha_0,\alpha_1,\beta_1,\alpha_2,\beta_2,\alpha_3) \nonumber \\
=& 2\alpha_1 z_1 V'(z_1) - 2\alpha_2 z_2 V'(z_2) + 2(\beta_1 + \beta_2) \sum_{j = 1}^\infty jV_j \left( \cos j(p+t) - \cos j(p-t) \right) \nonumber \\
& -i(2\alpha_1^2 + \beta_1^2 + \beta_1 \beta_2) \frac{\cos p-t}{\sin p-t} + i(2\alpha_2^2 + \beta_2^2 + \beta_1\beta_2) \frac{ \cos p+t }{ \sin p+t } \nonumber \\
& + i(\beta_1^2-\beta_2^2) \frac{\cos p}{\sin p} \nonumber \\
& - 2i\alpha_1 \alpha_0 \frac{\cos \frac{p-t}{2}}{\sin \frac{p-t}{2}} + 2i\alpha_1 \alpha_3 \frac{\sin \frac{p-t}{2}}{\cos \frac{p-t}{2}} + 2i\alpha_2 \alpha_0 \frac{\cos \frac{p+t}{2}}{\sin \frac{p+t}{2}} - 2i\alpha_2 \alpha_3 \frac{\sin \frac{p+t}{2}}{\cos \frac{p+t}{2}} \nonumber \\
&d_2(p,t;\alpha_0,\alpha_1,\beta_1,\alpha_2,\beta_2,\alpha_3) \nonumber \\
=& \frac{2}{it} \sigma(s) + i (4\alpha_1 \alpha_2 - (\beta_1+\beta_2)^2) \frac{\cos t}{\sin t} \\
&d_3(p,t;\alpha_0,\alpha_1,\beta_1,\alpha_2,\beta_2,\alpha_3) \nonumber \\
=& 2\sigma_s \Bigg( - 2\sum_{j = 1}^\infty jV_j \left( \cos j(p-t) + \cos j(p+t) \right) - 2\sum_{j=0}^5 \alpha_j \nonumber \\
&- i\beta_1 \frac{\cos p-t}{\sin p-t} - i\beta_2 \frac{\cos p+t}{\sin p+t} + i(\beta_1 - \beta_2) \left( \frac{\cos t}{\sin t} - \frac{1}{t} \right) - i(\beta_1 + \beta_2) \frac{\cos p}{\sin p} \Bigg). \nonumber 
\end{align}
\end{proposition} 

\noindent \textbf{Proof:} The proof is analogous to the proof of Proposition 8.1 in \cite{Claeys2015}. As is done there, we assume below that $\alpha_k>0$, $k = 1,2,4,5$, for simplicity of notation. Once (\ref{eqn:diff_id_asymptotics}) is proven under this assumption, the case where $\alpha_k = 0$ for some $k$ then follows from the uniformity of the error terms in $\alpha_k$, $k = 1,2,4,5$. Extending to the case where $\alpha_k < 0$ for some $k$ is straightforward. We prove the proposition first in the regime $0 < t \leq \omega(n)/n$ and then in the regime $\omega(n)/n < t < t_0$.  \\

Using the transformation $Y \rightarrow T \rightarrow S$ inside the unit circle, outside the lenses, we can rewrite the differential identity (\ref{eqn:diff id}) in the form 
\begin{equation} \label{eqn:transform of diff id}
\frac{1}{i} \frac{\text{d}}{\text{d}t} \log  D_n(f_{p,t}) = \sum_{k = 1,2,4,5} q_k \left( n\beta_k + 2\alpha_k z_k \left( S^{-1} \frac{\text{d}S}{\text{d}z} \right)_{+,22} (z_k) \right), 
\end{equation}
with $q_k = 1$ for $k = 1,4$ and $q_k = -1$ for $k = 2,5$, and where the limit $z \rightarrow z_k$ is taken from the inside of the unit circle and outside the lenses.

\subsubsection{$0 < t \leq \omega(n)/n$} 
By (\ref{eqn:R}) we get 
\begin{equation}
S(z) = R(z) P_\pm(z), \quad z \in U_\pm,
\end{equation}
and thus 
\begin{align} \label{eqn:SS}
\begin{split}
\left( S^{-1} \frac{\text{d}S}{\text{d}z} \right)_{22} (z) =& \left( P_\pm^{-1} \frac{\text{d}P_\pm}{\text{d}z} \right)_{22} (z) + A_{n,p,t}(z), \quad z \in U_{e^{\pm ip}}, \\
A_{n,p,t}(z) =& \left( P_\pm^{-1}(z) R^{-1}(z) \frac{dR}{dz}(z) P_\pm(z) \right)_{22}.
\end{split}
\end{align}

Following exactly the same approach as on pages 60, 61 in \cite{Claeys2015}, we can use (\ref{eqn:Psi asymptotics}) and the small and large $|s|$ asymptotics from Sections 5, 6 of \cite{Claeys2015} to obtain that for $k = 1,2,4,5$ 
\begin{equation}
\int_0^t |A_{n,p,t}(z_k)| \text{d}t = o(\omega(n)^{-1}), \quad n \rightarrow \infty,
\end{equation}
uniformly in $0 < t \leq \omega(n)/n$ and $p \in (\epsilon, \pi - \epsilon)$, and thus also 
\begin{equation} \label{eqn:epsilon tilde}
\tilde{\epsilon}_{n,p,t} := 2 \sum_{k = 1,2,4,5} q_k \alpha_k z_k A_{n,p,t}(z_k) = \mathcal{O}(\omega(n)^{-1}),
\end{equation}
uniformly in $0 < t \leq \omega(n)/n$ and $p \in (\epsilon, \pi - \epsilon)$.\\ 

By (\ref{eqn:P1}) we see that
\begin{align} \label{eqn:PP1}
\begin{split}
&\left( P_\pm^{-1}\frac{\text{d}P_\pm}{\text{d}z} \right)_{22}(z) \\
=& -\frac{n}{2z} + \frac{1}{2}\frac{f'_{p,t}}{f_{p,t}}(z) + \left( \Phi_\pm^{-1} \frac{\text{d}\Phi_\pm}{\text{d}z} \right)_{22}(z) + \left( \Phi_\pm^{-1} E_\pm^{-1} \frac{\text{d}E_\pm}{\text{d}z} \Phi_\pm \right)_{22}(z),\\
\end{split}
\end{align}
with $z$ inside the unit circle and outside of the lenses of $\Sigma_S$. By (\ref{eqn:E}) we have for $z$ near $e^{ip}$:
\begin{align} \label{eqn:EE}
\begin{split} 
E_\pm(z)^{-1} \frac{\text{d} E_\pm(z)}{\text{d} z} =& h_\pm(z) \sigma_3, 
\end{split}
\end{align}
where 
\begin{align} \label{eqn:h_pm}
\begin{split}
h_\pm(z) =& \pm \frac{\beta_1}{z\log  \frac{z}{e^{\pm ip}} \pm itz} \pm \frac{\beta_2}{z\log  \frac{z}{e^{\pm ip}} \mp itz}  \\
&- \frac{1}{2} \sum_{j = 1}^\infty j V_j z^{j-1} + \frac{1}{2} \sum_{j=-1}^{-\infty} j V_j z^{j-1} - \sum_{j=0}^5 \frac{\beta_j}{z-z_j} - \frac{1}{2z} \sum_{j=0}^5 \alpha_j.
\end{split}
\end{align}

In the following equation we need the fact that
\begin{align} \label{eqn:1/log}
\begin{split}
\frac{1}{\log  (1+z)} =& \frac{1}{z} + \frac{1}{2} + \mathcal{O}(z), \quad z \rightarrow 0, \\ 
\frac{1}{\log  z - \log  z_k} =& \frac{1}{\log  \frac{z}{z_k}} = \frac{z_k}{z - z_k} + \frac{1}{2} + \mathcal{O}(|z-z_k|), \quad z \rightarrow z_k.
\end{split}
\end{align}
Let $k = 1,2$ and denote $k' = 1$ for $k = 2$, $k' = 2$ for $k = 1$. Putting together (\ref{eqn:PP1}) and (\ref{eqn:EE}), we obtain for $k = 1,2$ (as in (8.33) in \cite{Claeys2015})
\begin{align} \label{eqn:PPlimit}
&\left( P_+^{-1} \frac{\text{d}P_+}{\text{d}z} \right)_{22,+}(z_k) \nonumber \\
=& - \frac{n}{2z_k} + \lim_{z \rightarrow z_k} \left( \frac{1}{2} \frac{f'_{p,t}}{f_{p,t}}(z) - \frac{\alpha_k}{z \log  z - z \log  z_1} \right) \nonumber \\
&+ \frac{1}{tz_k} \left( F_{+,k'}^{-1} \frac{\text{d}F_{+,k'}}{\text{d}z} \right)_{22} \left( \frac{1}{t} \log  \frac{z_k}{e^{ip}} \right) + h_+(z_k) (F_{+,k'}^{-1}\sigma_3 F_{+,k'})_{22} \left( \frac{1}{t} \log  \frac{z_k}{e^{ip}} \right) \nonumber \\
=& - \frac{n}{2z_k} + \frac{1}{2} V'(z_k) + \sum_{j = 0}^5 \frac{\beta_j}{2z_k} + \sum_{j,j \neq k} \frac{\alpha_j}{z_k - z_j} - \frac{\alpha_j}{2z_k} \nonumber \\ 
& + \frac{1}{tz_k} \left( F_{+,k'}^{-1} \frac{\text{d}F_{+,k'}}{\text{d}z} \right)_{22} \left( \frac{1}{t} \log  \frac{z_k}{e^{ip}} \right) + h_+(z_k) (F_{+,k'}^{-1}\sigma_3 F_{+,k'})_{22} \left( \frac{1}{t} \log \frac{z_k}{e^{ip}} \right), \nonumber \\
&\left( P_-^{-1} \frac{\text{d}P_-}{\text{d}z} \right)_{22,+}(\overline{z_k}) \\
=& - \frac{n}{2\overline{z_k}} + \lim_{z \rightarrow \overline{z_k}} \left( \frac{1}{2} \frac{f'_{p,t}}{f_{p,t}}(z) - \frac{\alpha_k}{z \log z - z \log \overline{z_k}} \right) \nonumber \\
&+ \frac{1}{t\overline{z_k}} \left( F_{-,k}^{-1} \frac{\text{d}F_{-,k}}{\text{d}z} \right)_{22} \left( \frac{1}{t} \log \frac{\overline{z_k}}{e^{-ip}} \right) + h_-(\overline{z_2}) (F_{-,k}^{-1}\sigma_3 F_{-,k})_{22} \left( \frac{1}{t} \log \frac{\overline{z_k}}{e^{-ip}} \right) \nonumber \\
=& - \frac{n}{2\overline{z_k}} + \frac{1}{2} V'(\overline{z_k}) + \sum_{j = 0}^5 \frac{\beta_j}{2\overline{z_k}} + \sum_{j, z_j \neq \overline{z_k}} \frac{\alpha_j}{\overline{z_k} - z_j} - \frac{\alpha_j}{2\overline{z_k}} \nonumber \\ 
& + \frac{1}{t\overline{z_k}} \left( F_{-,k}^{-1} \frac{\text{d}F_{-,k}}{\text{d}z} \right)_{22} \left( \frac{1}{t} \log \frac{\overline{z_k}}{e^{-ip}} \right) + h_-(\overline{z_k}) (F_{-,k}^{-1}\sigma_3 F_{-,k})_{22} \left( \frac{1}{t} \log \frac{\overline{z_k}}{e^{-ip}} \right), \nonumber
\end{align}
where in the second equalities we used (\ref{eqn:1/log}), and where $F_{\pm,k}$ equal the functions $F_k$ defined in (\ref{eqn:F_1 neq 0}), (\ref{eqn:F_2 neq 0}), (\ref{eqn:F_1 0}) and (\ref{eqn:F_2 0}), with $(\alpha_1, \alpha_2, \beta_1, \beta_2)$ in the appendix replaced by $(\alpha_2,\alpha_1,-\beta_2,-\beta_1)$ in the $-$ case. When replacing $(\alpha_1,\alpha_2,\beta_1,\beta_2)$ in the Painlev\'{e} equation (3.52) in \cite{Claeys2015} with $(\alpha_2,\alpha_1,\beta_2,\beta_1)$ or $(\alpha_1,\alpha_2,-\beta_1,-\beta_2)$, we get the same Painlev\'{e} equation as in our Theorem \ref{thm:painleve}. Thus we can see that Propositions 3.1 and 3.2 in \cite{Claeys2015} become in our case:
\begin{proposition} \label{prop:F}
We have the identities
\begin{align}
\begin{split}
\alpha_2(F_{+,1}(i;s)^{-1}\sigma_3 F_{+,1}(i;s))_{22} =& -\sigma_s(s) + \frac{\beta_1+\beta_2}{2}, \\
\alpha_1(F_{+,2}(-i;s)^{-1}\sigma_3 F_{+,2}(-i;s))_{22} =& \sigma_s(s) + \frac{\beta_1+\beta_2}{2}, \\
\alpha_1(F_{-,1}(i;s)^{-1}\sigma_3 F_{-,1}(i;s))_{22} =& -\sigma_s(s) - \frac{\beta_1+\beta_2}{2}, \\ 
\alpha_2(F_{-,2}(-i;s)^{-1}\sigma_3 F_{-,2}(-i;s))_{22} =& \sigma_s(s) - \frac{\beta_1+\beta_2}{2}, 
\end{split}
\end{align}
and
\begin{align}
\begin{split}
\alpha_2 \left( F_{+,1}(i;s)^{-1} F_{+,1,\zeta}(i;s) \right)_{22} =& \frac{i}{4} \sigma(s) - \frac{i}{8}(\beta_1 + \beta_2)s + \frac{i}{8}(\beta_1+\beta_2)^2, \\
\alpha_1 \left(F_{+,2}(-i;s)^{-1} F_{+,2,\zeta}(-i;s) \right)_{22} =& -\frac{i}{4} \sigma(s) - \frac{i}{8}(\beta_1 + \beta_2)s - \frac{i}{8}(\beta_1+\beta_2)^2, \\
\alpha_1 \left(F_{-,1}(i;s)^{-1} F_{-,1,\zeta}(i;s) \right)_{22} =& \frac{i}{4} \sigma(s) + \frac{i}{8}(\beta_1 + \beta_2)s + \frac{i}{8}(\beta_1+\beta_2)^2, \\
\alpha_2 \left(F_{-,2}(-i;s)^{-1} F_{-,2,\zeta}(-i;s) \right)_{22} =& -\frac{i}{4} \sigma(s) + \frac{i}{8}(\beta_1 + \beta_2)s - \frac{i}{8}(\beta_1+\beta_2)^2.
\end{split}
\end{align}
\end{proposition}

Putting together (\ref{eqn:transform of diff id}), (\ref{eqn:SS}) and (\ref{eqn:PPlimit}), we obtain:
\begin{align} \label{eqn:diff_id_F}
&\frac{1}{i} \frac{\text{d}}{\text{d}t} \log D_n(f_{p,t}) \nonumber \\
=& \sum_{k = 1,2,4,5} q_k \left( n\beta_k + \alpha_k z_k V'(z_k) + \alpha_k \sum_{j=0}^5 \beta_j + 2 \alpha_k \sum_{j \neq k} \alpha_j \frac{z_k}{z_k - z_j} - \alpha_k \sum_{j\neq k} \alpha_j \right) \nonumber \\
&+ \sum_{k = 1}^2 (-1)^{k+1} \left( 2 \alpha_k z_k h_+(z_k) \left( F^{-1}_{+,k'} \sigma_3 F_{+,k'} \right)_{22} ((-1)^ki) + \frac{2}{t} \alpha_k \left( F_{+,k'}^{-1} F_{+,k'}' \right)_{22} ((-1)^ki) \right) \nonumber \\
&+ \sum_{k = 1}^2 (-1)^k \left( 2 \alpha_k \overline{z_k} h_-(\overline{z_k}) \left( F^{-1}_{-,k} \sigma_3 F_{-,k} \right)_{22} ((-1)^{k+1}i) + \frac{2}{t} \alpha_k \left( F_{-,k}^{-1} F_{-,k}' \right)_{22} ((-1)^{k+1}i) \right) \nonumber \\
&+ \tilde{\epsilon}_{n,p,t}, \nonumber \\
=& 2n(\beta_1 - \beta_2) + 2\alpha_1 z_1 V'(z_1) - 2\alpha_2 z_2 V'(z_2) \\
&+ 2 \sum_{k = 1,2,4,5} q_k \alpha_k \sum_{j \neq k} \alpha_j \frac{z_k}{z_k - z_j} \nonumber \\
&+ \sum_{k = 1}^2 (-1)^{k+1} \left( 2 \alpha_k z_k h_+(z_k) \left( F^{-1}_{+,k'} \sigma_3 F_{+,k'} \right)_{22} ((-1)^ki) + \frac{2}{t} \alpha_k \left( F_{+,k'}^{-1} F_{+,k'}' \right)_{22} ((-1)^ki) \right) \nonumber \\
&+ \sum_{k = 1}^2 (-1)^k \left( 2 \alpha_k \overline{z_k} h_-(\overline{z_k}) \left( F^{-1}_{-,k} \sigma_3 F_{-,k} \right)_{22} ((-1)^{k+1}i) + \frac{2}{t} \alpha_k \left( F_{-,k}^{-1} F_{-,k}' \right)_{22} ((-1)^{k+1}i) \right) \nonumber \\
&+ \tilde{\epsilon}_{n,p,t}, \nonumber 
\end{align}
where we used that since $\beta_1 = - \beta_5$, $\beta_2 = - \beta_4$, $\alpha_1 = \alpha_5$, $\alpha_2 = \alpha_4$ and $V(z) = V(\overline{z})$, it holds that: 
\begin{align}
\begin{split}
\sum_{k = 1,2,4,5} q_k \alpha_k \sum_{j \neq k} \alpha_j =& 0, \quad \sum_{j = 0}^5 \beta_j = 0, \quad zV'(z) = - \overline{z}V'(\overline{z}).
\end{split}
\end{align}
By Proposition \ref{prop:F} we get 
\begin{align} \label{eqn:FF}
&\sum_{k = 1}^2 (-1)^{k+1} \frac{2}{t} \alpha_k \left( F_{+,k'}^{-1} F_{+,k'}' \right)_{22} ((-1)^ki) + \sum_{k = 1}^2 (-1)^k \frac{2}{t} \alpha_k \left( F_{-,k}^{-1} F_{-,k}' \right)_{22} ((-1)^{k+1}i) \nonumber \\
=& \frac{2}{it} \sigma(s) - \frac{i}{t}(\beta_1 + \beta_2)^2, 
\end{align} 
and 
\begin{align} \label{eqn:FsigmaF}
&\sum_{k = 1}^2 (-1)^{k+1} 2 \alpha_k z_k h_+(z_k) \left( F^{-1}_{+,k'} \sigma_3 F_{+,k'} \right)_{22} ((-1)^ki) \nonumber \\
&+ \sum_{k = 1}^2 (-1)^k 2 \alpha_k \overline{z_k} h_-(\overline{z_k}) \left( F^{-1}_{-,k} \sigma_3 F_{-,k} \right)_{22} ((-1)^{k+1}i) \\
=& \left( 2\sigma_s(s) + \beta_1 + \beta_2 \right) \left( z_1 h_+(z_1) + \overline{z_1}h_-(\overline{z_1}) \right) + \left( 2\sigma_s(s) - \beta_1 - \beta_2 \right) \left( z_2 h_+(z_2) + \overline{z_2}h_-(\overline{z_2}) \right)  \nonumber \\
=& 2\sigma_s \left( z_1 h_+(z_1) + \overline{z_1}h_-(\overline{z_1}) + z_2 h_+(z_2) + \overline{z_2}h_-(\overline{z_2}) \right) \nonumber \\
&+ (\beta_1 + \beta_2) \left( z_1 h_+(z_1) + \overline{z_1}h_-(\overline{z_1}) - z_2 h_+(z_2) - \overline{z_2}h_-(\overline{z_2}) \right). \nonumber 
\end{align}
Then (\ref{eqn:diff_id_F}), (\ref{eqn:FF}) and (\ref{eqn:FsigmaF}) result in 
\begin{align} \label{eqn:diff_id_sum_hh}
\begin{split}
\frac{1}{i} \frac{\text{d}}{\text{d}t} \log D_n(f_{p,t}) =& 2n(\beta_1 - \beta_2) + 2\alpha_1 z_1 V'(z_1) - 2\alpha_2 z_2 V'(z_2) \\
&+ 2 \sum_{k = 1,2,4,5} q_k \alpha_k \sum_{j \neq k} \alpha_j \frac{z_k}{z_k - z_j} \\
&+ \frac{2}{it} \sigma(s) - \frac{i}{t}(\beta_1 + \beta_2)^2 \\
&+ 2\sigma_s \left( z_1 h_+(z_1) + \overline{z_1}h_-(\overline{z_1}) + z_2 h_+(z_2) + \overline{z_2}h_-(\overline{z_2}) \right) \\
&+ (\beta_1 + \beta_2) \left( z_1 h_+(z_1) + \overline{z_1}h_-(\overline{z_1}) - z_2 h_+(z_2) - \overline{z_2}h_-(\overline{z_2}) \right) \\
&+ \tilde{\epsilon}_{n,p,t}.
\end{split}
\end{align}

\noindent With $h_+(z)$ and $h_-(z)$ given in (\ref{eqn:h_pm}), and using (\ref{eqn:1/log}), we obtain: 
\begin{align} \label{eqn:hz}
\begin{split}
h_+(z_1) z_1 =& - \frac{1}{2} \sum_{j = 1}^\infty j V_j (z_1^{j} + \overline{z_1}^j) - \sum_{j\neq 1} \frac{\beta_jz_1}{z_1-z_j} - \frac{1}{2} \sum_{j=0}^5 \alpha_j + \frac{\beta_1}{2} - \frac{\beta_{2}}{2it}, \\
h_+(z_2) z_2 =& - \frac{1}{2} \sum_{j = 1}^\infty j V_j (z_2^{j} + \overline{z_2}^j) - \sum_{j\neq 2} \frac{\beta_jz_2}{z_2-z_j} - \frac{1}{2} \sum_{j=0}^5 \alpha_j + \frac{\beta_2}{2} + \frac{\beta_{1}}{2it}, \\
h_-(\overline{z_1}) \overline{z_1} =& - \frac{1}{2} \sum_{j = 1}^\infty j V_j (z_1^{j} + \overline{z_1}^j) - \sum_{j\neq 5} \frac{\beta_j z_5}{z_5-z_j} - \frac{1}{2} \sum_{j=0}^5 \alpha_j - \frac{\beta_1}{2} - \frac{\beta_{2}}{2it}, \\
h_-(\overline{z_2}) \overline{z_2} =& - \frac{1}{2} \sum_{j = 1}^\infty j V_j (z_2^{j} + \overline{z_2}^j) - \sum_{j\neq 4} \frac{\beta_j z_4}{z_4-z_j} - \frac{1}{2} \sum_{j=0}^5 \alpha_j - \frac{\beta_2}{2} + \frac{\beta_{1}}{2it}. 
\end{split}
\end{align}
We thus see that 
\begin{align} \label{eqn:hh}
\begin{split}
&h_+(z_1)z_1 + \overline{z_1} h_-(\overline{z_1}) \\
=& - 2\sum_{j = 1}^\infty jV_j \cos j(p-t) - \sum_{j=0}^5 \alpha_j - \frac{\beta_2}{it} + \beta_1 \left( \frac{z_1 + \overline{z_1}}{z_1 - \overline{z_1}} \right) \\
&+ \beta_2 \left( - \frac{z_1}{z_1 - z_2} + \frac{z_1}{z_1 - \overline{z_2}} - \frac{\overline{z_1}}{\overline{z_1} - z_2} + \frac{\overline{z_1}}{\overline{z_1} - \overline{z_2}} \right) \\
=& - 2\sum_{j = 1}^\infty jV_j \cos j(p-t) - \sum_{j=0}^5 \alpha_j - \frac{\beta_2}{it} - i\beta_1 \frac{\cos p-t}{\sin p-t} - i\beta_2\frac{\cos t}{\sin t} \\
&- i\beta_2 \frac{\cos p}{\sin p}, \\
&h_+(z_2)z_2 + \overline{z_2}h_-(\overline{z_2}) \\
=& - 2\sum_{j = 1}^\infty jV_j \cos j(p+t) - \sum_{j=0}^5 \alpha_j  + \frac{\beta_1}{it} + \beta_2 \left( \frac{z_2 + \overline{z_2}}{z_2 - \overline{z_2}} \right) \\
&+ \beta_1 \left( - \frac{z_2}{z_2 - z_1} + \frac{z_2}{z_2 - \overline{z_1}} - \frac{\overline{z_2}}{\overline{z_2} - z_1} + \frac{\overline{z_2}}{\overline{z_2} - \overline{z_1}} \right) \\
=& - 2\sum_{j = 1}^\infty jV_j \cos (p+t) - \sum_{j=0}^5 \alpha_j  + \frac{\beta_1}{it} - i\beta_2 \frac{\cos p+t}{\sin p+t} + i\beta_1 \frac{\cos t}{\sin t} \\
& - i\beta_1 \frac{\cos p}{\sin p}. 
\end{split}
\end{align}
Further we calculate
\begin{align} \label{eqn:sum_sin}
\sum_{k = 1,5} q_k \alpha_k \sum_{j \neq k} \alpha_j \frac{z_k}{z_k - z_j} 
=& \alpha_1 \alpha_0 \left( \frac{z_1}{z_1 - 1} - \frac{\overline{z_1}}{\overline{z_1} - 1} \right) + \alpha_1^2 \left( \frac{z_1}{z_1 - \overline{z_1}} - \frac{\overline{z_1}}{\overline{z_1} - z_1} \right) \nonumber \\
&+ \alpha_1 \alpha_2 \left( \frac{z_1}{z_1 - z_2} + \frac{z_1}{z_1 - \overline{z_2}} - \frac{\overline{z_1}}{\overline{z_1} - z_2} - \frac{\overline{z_1}}{\overline{z_1} - \overline{z_2}} \right) \nonumber \\
&+ \alpha_1 \alpha_3 \left( \frac{z_1}{z_1 + 1} - \frac{\overline{z_1}}{\overline{z_1} + 1} \right), \nonumber \\
- \sum_{k = 2,4} q_k \alpha_k \sum_{j \neq k} \alpha_j \frac{z_k}{z_k - z_j} =& \alpha_2 \alpha_0 \left( \frac{z_2}{z_2 - 1} - \frac{\overline{z_2}}{\overline{z_2} - 1} \right) \\
&+ \alpha_2 \alpha_1 \left( \frac{z_2}{z_2 - z_1} + \frac{z_2}{z_2 - \overline{z_1}} - \frac{\overline{z_2}}{\overline{z_2} - z_1} - \frac{\overline{z_2}}{\overline{z_2} - \overline{z_1}} \right) \nonumber \\
&+ \alpha_2^2 \left( \frac{z_2}{z_2 - \overline{z_2}} - \frac{\overline{z_2}}{\overline{z_2} - z_2} \right) + \alpha_2 \alpha_3 \left( \frac{z_2}{z_2 + 1} - \frac{\overline{z_2}}{\overline{z_2} + 1} \right), \nonumber \\
\sum_{k = 1,2,4,5} q_k \alpha_k \sum_{j \neq k} \alpha_j \frac{z_k}{z_k - z_j} =& -i\alpha_1^2 \frac{\cos p-t}{\sin p-t} + i\alpha_2^2 \frac{ \cos p+t }{ \sin p+t } + 2i \alpha_1 \alpha_2 \frac{\cos t}{\sin t} \nonumber \\
&+ \alpha_1 \alpha_0 \frac{z_1 + 1}{z_1 - 1} + \alpha_1 \alpha_3 \frac{z_1 - 1}{z_1 + 1} - \alpha_2 \alpha_0 \frac{z_2 + 1}{z_2 - 1} - \alpha_2 \alpha_3 \frac{z_2 - 1}{z_2 + 1} \nonumber \\
=& -i\alpha_1^2 \frac{\cos p-t}{\sin p-t} + i\alpha_2^2 \frac{ \cos p+t }{ \sin p+t } + 2i \alpha_1 \alpha_2 \frac{\cos t}{\sin t} \nonumber \\
&- i\alpha_1 \alpha_0 \frac{\cos \frac{p-t}{2}}{\sin \frac{p-t}{2}} + i\alpha_1 \alpha_3 \frac{\sin \frac{p-t}{2}}{\cos \frac{p-t}{2}} + i\alpha_2 \alpha_0 \frac{\cos \frac{p+t}{2}}{\sin \frac{p+t}{2}} - i\alpha_2 \alpha_3 \frac{\sin \frac{p+t}{2}}{\cos \frac{p+t}{2}}. \nonumber 
\end{align}
Putting together (\ref{eqn:diff_id_sum_hh}), (\ref{eqn:hh}) and (\ref{eqn:sum_sin}) we get that uniformly in $p \in (\epsilon, \pi - \epsilon)$ and $0 < t \leq \omega(n)/n$:
\begin{align} \label{eqn:diff_id_final}
\begin{split}
&\frac{1}{i} \frac{\text{d}}{\text{d}t} \log D_n(f_{p,t}) \\
=& 2n(\beta_1 - \beta_2) + 2\alpha_1 z_1 V'(z_1) - 2\alpha_2 z_2 V'(z_2) \\
& -2i\alpha_1^2 \frac{\cos p-t}{\sin p-t} + 2i\alpha_2^2 \frac{ \cos p+t }{ \sin p+t } + 4i \alpha_1 \alpha_2 \frac{\cos t}{\sin t} \\
&- 2i\alpha_1 \alpha_0 \frac{\cos \frac{p-t}{2}}{\sin \frac{p-t}{2}} + 2i\alpha_1 \alpha_3 \frac{\sin \frac{p-t}{2}}{\cos \frac{p-t}{2}} + 2i\alpha_2 \alpha_0 \frac{\cos \frac{p+t}{2}}{\sin \frac{p+t}{2}} - 2i\alpha_2 \alpha_3 \frac{\sin \frac{p+t}{2}}{\cos \frac{p+t}{2}} \\
&+ \frac{2}{it} \sigma(s) - \frac{i}{t}(\beta_1 + \beta_2)^2 \\
&+ 2\sigma_s(s) \Bigg( - 2\sum_{j = 1}^\infty jV_j \left( \cos j(p-t) + \cos j(p+t) \right) - 2\sum_{j=0}^5 \alpha_j + \frac{\beta_1 - \beta_2}{it} \\
&- i\beta_1 \frac{\cos p-t}{\sin p-t} - i\beta_2 \frac{\cos p+t}{\sin p+t} + i(\beta_1 - \beta_2) \frac{\cos t}{\sin t} - i(\beta_1 + \beta_2) \frac{\cos p}{\sin p} \Bigg) \\
&+ (\beta_1 + \beta_2) \Bigg( 2\sum_{j = 1}^\infty jV_j \left( \cos j(p+t) - \cos j(p-t) \right) - \frac{\beta_1 + \beta_2}{it} \\
&- i\beta_1 \frac{\cos p-t}{\sin p-t} + i\beta_2 \frac{\cos p+t}{\sin p+t} - i(\beta_1 + \beta_2) \frac{\cos t}{\sin t} + i(\beta_1 - \beta_2) \frac{\cos p}{\sin p} \Bigg) \\
&+ \tilde{\epsilon}_{n,p,t}. 
\end{split}
\end{align}
Simplifying further and setting $\epsilon_{n,p,t} = \tilde{\epsilon}_{n,p,t}$ for $0 < t \leq \omega(n)/n$ we obtain (\ref{eqn:diff_id_asymptotics}) for $0 < t \leq \omega(n)/n$.

\subsubsection{$\omega(n)/n < t < t_0$}
For $\omega(n)/n < t < t_0$, $z \in \tilde{\mathcal{U}}_k$, we obtain instead of (\ref{eqn:SS}):
\begin{align} \label{eqn:SS2}
\begin{split}
\left( S^{-1}(z) \frac{\text{d}S(z)}{\text{d}z} \right)_{22} =& \left( \tilde{P}_k(z)^{-1} \frac{\text{d}\tilde{P}_k(z)}{\text{d}z} \right)_{22} + A_{n,p,t}, \\
A_{n,p,t}(z) =& \left( \tilde{P}_k(z)^{-1} \tilde{R}^{-1}(z) \frac{\text{d}\tilde{R}(z)}{\text{d}z} \tilde{P}_k(z) \right)_{22},
\end{split}
\end{align}
with $\tilde{R}(z)$ given in (\ref{eqn:R2}). From (\ref{eqn:P tilde}) and (\ref{eqn:W}) it follows that for $|z| < 1$
\begin{align} \label{eqn:PP2}
\begin{split}
\left( \tilde{P}_k(z)^{-1} \frac{\text{d}\tilde{P}_k(z)}{\text{d}z} \right)_{22} =& -\frac{n}{2z} + \frac{1}{2} \frac{f_{p,t}'(z)}{f_{p,t}(z)} + \left( M_k(z)^{-1} \frac{\text{d}M_k(z)}{\text{d}z} \right)_{22} \\
&+ \tilde{h}_1(z) \left( M_k(z)^{-1} \sigma_3 M_k(z) \right)_{22}, \\
A_{n,p,t}(z_k) =& \lim_{z \rightarrow z_k} A_{n,p,t}(z) \\
=& \lim_{z \rightarrow z_k} \left( M_k^{-1} \tilde{E}_k^{-1} \tilde{R}^{-1} \frac{\text{d}\tilde{R}}{\text{d}z} \tilde{E}_k M_k \right)_{22}(z),
\end{split}
\end{align}
where 
\begin{align}
\begin{split}
M_1(z) =& M^{(\alpha_1,\beta_1)}(nt(\frac{1}{t} \log \frac{z}{e^{ip}} + i)), \\
M_2(z) =& M^{(\alpha_2,\beta_2)}(nt(\frac{1}{t} \log \frac{z}{e^{ip}} - i)), \\
M_4(z) =& M^{(\alpha_4,\beta_4)}(nt(\frac{1}{t} \log \frac{z}{e^{-ip}} + i)), \\
M_5(z) =& M^{(\alpha_5,\beta_5)}(nt(\frac{1}{t} \log \frac{z}{e^{-ip}} - i)), 
\end{split}
\end{align}
and where
\begin{align} \label{eqn:hhtilde}
\begin{split}
\tilde{h}_1(z)\sigma_3 =& \tilde{E}_1(z)^{-1} \frac{\text{d} \tilde{E}_1(z)}{\text{d}z}, \quad \tilde{h}_1(z) = h_+(z) - \frac{\beta_2}{z \log \frac{z}{e^{ip}} - itz}, \\ 
\tilde{h}_2(z)\sigma_3 =& \tilde{E}_2(z)^{-1} \frac{\text{d} \tilde{E}_2(z)}{\text{d}z}, \quad \tilde{h}_2(z) = h_+(z) - \frac{\beta_1}{z \log \frac{z}{e^{ip}} + itz}, \\ 
\tilde{h}_4(z)\sigma_3 =& \tilde{E}_4(z)^{-1} \frac{\text{d} \tilde{E}_4(z)}{\text{d}z}, \quad \tilde{h}_4(z) = h_-(z) + \frac{\beta_1}{z \log \frac{z}{e^{-ip}} - itz}, \\ 
\tilde{h}_5(z)\sigma_3 =& \tilde{E}_5(z)^{-1} \frac{\text{d} \tilde{E}_5(z)}{\text{d}z}, \quad \tilde{h}_5(z) = h_-(z) + \frac{\beta_2}{z \log \frac{z}{e^{ip}} + itz}, 
\end{split}
\end{align}
with $h_+(z),h_-(z)$ given in (\ref{eqn:h_pm}). By (\ref{eqn:R2 asymptotics}) and the fact that $\tilde{E}_k$ is uniformly bounded for $p \in (\epsilon, \pi - \epsilon)$, $\omega(n)/n < t < t_0$, and $z \in \tilde{U}_k$, we see that 
\begin{align}
A_{n,p,t}(z_k) = \mathcal{O}((nt)^{-1}) = \mathcal{O}(\omega(n)^{-1})
\end{align}
uniformly in $p \in (\epsilon, \pi - \epsilon)$, $\omega(n)/n < t < t_0$, and thus also 
\begin{equation} \label{eqn:epsilon tilde 2}
\tilde{\epsilon}_{n,p,t} := 2 \sum_{k = 1,2,4,5} q_k \alpha_k z_k A_{n,p,t}(z_k) = \mathcal{O}(\omega(n)^{-1}),
\end{equation}
uniformly in $0 < t < \omega(n)/n$ and $p \in (\epsilon, \pi - \epsilon)$. This implies that as $n \rightarrow \infty$
\begin{align} \label{eqn:error2}
\int_{\omega(n)/n}^{t_0} |\tilde{\epsilon}_{n,p,t}|\text{d}t = o(\omega(n)^{-1}), 
\end{align}
uniformly in $p \in (\epsilon, \pi - \epsilon)$. 

From (8.41) and (8.42) in \cite{Claeys2015} we can see that for $z \rightarrow z_k$ inside the unit circle and outside of the lenses of $\Sigma_S$ we have 
\begin{align}
\begin{split}
\left( M_1^{-1} \frac{\text{d}M_1}{\text{d}z} \right) (z) =& - \left( \frac{\beta_1}{2\alpha_1} + \frac{\alpha_1}{n \log \frac{z}{e^{ip}} + int} \right) \frac{n}{z} + o(1),\\
\left( M_2^{-1} \frac{\text{d}M_2}{\text{d}z} \right) (z) =& - \left( \frac{\beta_2}{2\alpha_2} + \frac{\alpha_2}{n \log \frac{z}{e^{ip}} - int} \right) \frac{n}{z} + o(1),\\
\left( M_4^{-1} \frac{\text{d}M_4}{\text{d}z} \right) (z) =& - \left( \frac{-\beta_2}{2\alpha_2} + \frac{\alpha_2}{n \log \frac{z}{e^{-ip}} + int} \right) \frac{n}{z} + o(1),\\
\left( M_5^{-1} \frac{\text{d}M_5}{\text{d}z} \right) (z) =& - \left( \frac{-\beta_1}{2\alpha_1} + \frac{\alpha_1}{n \log \frac{z}{e^{-ip}} - int} \right) \frac{n}{z} + o(1),
\end{split}
\end{align}
and in the same limit,
\begin{align}
\begin{split}
\left( M_k \sigma_3 M \right)_{22}(z_k) =& \frac{\beta_k}{\alpha_k}.
\end{split}
\end{align}
Together with (\ref{eqn:hz}) and (\ref{eqn:hhtilde}) we obtain (again in the same limit)
\begin{align}
&z_k \left( M^{-1}_k \frac{\text{d}M_k}{\text{d}z} \right)_{22}(z) + z_k\tilde{h}_k(z_k) \left( M_k^{-1} \sigma_3 M_k \right)_{22} (z_k) \\
=& \left( \frac{\beta_k}{2\alpha_k} + \frac{\alpha_k}{n \log \frac{z}{z_k}} \right)n + \frac{\beta_k}{\alpha_k} \left(  -\frac{1}{2} \sum_{j = 1}^\infty j V_j (z_k^{j} + \overline{z_k}^j) - \sum_{j,j\neq k} \frac{\beta_jz_k}{z_k-z_j} - \frac{1}{2} \sum_{j=0}^5 \alpha_j + \frac{\beta_k}{2} \right) \nonumber \\
& +o(1). \nonumber 
\end{align}
Combining this with (\ref{eqn:1/log}) and (\ref{eqn:PP2}) we get
\begin{align}
\begin{split}
&2\alpha_kz_k \left( \tilde{P}_k^{-1} \frac{\text{d}\tilde{P}_k}{\text{d}z} \right)_{+,22}(z_k) \\
=& - n(\alpha_k + \beta_k) + \alpha_k z_k V'(z_k) + 2 \alpha_k \sum_{j,j \neq k} \frac{\alpha_jz_k}{z_k - z_j} - \alpha_k \sum_{j,j\neq k} \alpha_j \\
& +2\beta_k \left(  -\frac{1}{2} \sum_{j = 1}^\infty j V_j (z_k^{j} + \overline{z_k}^j) - \sum_{j,j\neq k} \frac{\beta_jz_k}{z_k-z_j} - \frac{1}{2} \sum_{j=0}^5 \alpha_j + \frac{\beta_k}{2} \right).
\end{split}
\end{align}
Together with (\ref{eqn:transform of diff id}), (\ref{eqn:SS2}) and (\ref{eqn:epsilon tilde 2}) we obtain
\begin{align} \label{eqn:diff id final2}
\begin{split}
\frac{1}{i} \frac{\text{d}}{\text{d}t} \log D_n(f_{p,t}) =& S(p,t;\alpha_0,\alpha_1,\beta_1,\alpha_2, \beta_2,\alpha_3) + \tilde{\epsilon}_{n,p,t},
\end{split}
\end{align}
where 
\begin{align}
\begin{split}
&S(p,t;\alpha_0,\alpha_1,\beta_1,\alpha_2, \beta_2,\alpha_3) \\
=& 2\sum_{k=1}^2 (-1)^{k+1} \left( (\alpha_k- \beta_k) \left(\sum_{j = 1}^\infty jV_j z_k^j\right) - (\alpha_k + \beta_k) \left( \sum_{j = 1}^\infty jV_{j} \overline{z_k}^j \right) \right) \\  
&- 2 (\beta_1 - \beta_2) \sum_{k = 0}^5 \alpha_k \\
&-2i(\alpha_1^2 + \beta_1^2) \frac{\cos (p-t)}{\sin (p-t)} + 2i(\alpha_2^2 + \beta_2^2) \frac{\cos(p+t)}{\sin(p+t)} + 4i(\alpha_1\alpha_2 - \beta_1 \beta_2) \frac{\cos t}{\sin t} \\
&- 2i\alpha_1 \alpha_0 \frac{\cos \frac{p-t}{2}}{\sin \frac{p-t}{2}} + 2i\alpha_1 \alpha_3 \frac{\sin \frac{p-t}{2}}{\cos \frac{p-t}{2}} + 2i\alpha_2 \alpha_0 \frac{\cos \frac{p+t}{2}}{\sin \frac{p+t}{2}} - 2i\alpha_2 \alpha_3 \frac{\sin \frac{p+t}{2}}{\cos \frac{p+t}{2}}.
\end{split}
\end{align}

Now we compare this expression to (\ref{eqn:diff_id_final}), obtained for $0 < t \leq \omega(n)/n$. Consider (\ref{eqn:diff_id_sum_hh}) for large $ s = -2int$ and without the error term. Substituting there the asymptotics of $\sigma(s)$ from Theorem \ref{thm:painleve} and using (\ref{eqn:sum_sin}) we see that
\begin{align}
&\frac{1}{i} \frac{\text{d}}{\text{d}t} \log D_n(f_{p,t}) \nonumber \\
=& 2\alpha_1 z_1 V'(z_1) - 2\alpha_2 z_2 V'(z_2) \nonumber \\
& -2i\alpha_1^2 \frac{\cos p-t}{\sin p-t} + 2i\alpha_2^2 \frac{ \cos p+t }{ \sin p+t } + 4i \alpha_1 \alpha_2 \frac{\cos t}{\sin t} \\
&- 2i\alpha_1 \alpha_0 \frac{\cos \frac{p-t}{2}}{\sin \frac{p-t}{2}} + 2i\alpha_1 \alpha_3 \frac{\sin \frac{p-t}{2}}{\cos \frac{p-t}{2}} + 2i\alpha_2 \alpha_0 \frac{\cos \frac{p+t}{2}}{\sin \frac{p+t}{2}} - 2i\alpha_2 \alpha_3 \frac{\sin \frac{p+t}{2}}{\cos \frac{p+t}{2}} \nonumber \\
& + \frac{1}{it}4 \beta_1 \beta_2 + 2\beta_1\left( h_+(z_1)z_1 + h_-(\overline{z_1}) \overline{z_1} \right) - 2\beta_2 \left( h_+(z_2)z_2 + h_-(\overline{z_2}) \overline{z_2} \right) \nonumber \\
&+ \Theta_{n,p,t}, \nonumber 
\end{align}
where $\Theta_{n,p,t}$ arises from the error term in the asymptotics of $\sigma(s)$, and becomes of order $\omega(n)^{-\delta}$ after integration w.r.t. $t$, i.e.
\begin{equation}
\left| \int_{\omega(n)/n}^t \Theta_{n,p,\tau} \text{d}\tau \right| = \mathcal{O}(\omega(n)^{-\delta}),
\end{equation}
uniformly in $p \in (\epsilon, \pi - \epsilon)$ and $\omega(n)/n < t < t_0$. Using (\ref{eqn:hh}) now we see that  
\begin{align}
\begin{split}
n(\beta_1 - \beta_2) + d_1 + d_2 + d_3 = S + \Theta_{n,p,t}, \quad \omega(n)/n < t < t_0.
\end{split}
\end{align}
Thus when setting 
\begin{align}
\begin{split}
\epsilon_{n,p,t} =& \tilde{\epsilon}_{n,p,t} + \Theta_{n,p,t}, \quad \omega(n)/n < t < t_0,
\end{split}
\end{align}
we see that (\ref{eqn:diff_id_asymptotics}) remains valid also in the region $\omega(n)/n < t < t_0$, where the smallness of the error terms follows from (\ref{eqn:epsilon tilde 2}). \qed

\begin{remark} Integrating (\ref{eqn:diff id final2}) from $t$ to $t_0$ with $\omega(n)/n < t < t_0$ and using Theorem \ref{thm:T non-uniform} for the expansion of $D_n(f_{p,t_0})$, we get the same expansion for $D_{n}(f_{p,t})$ that Theorem \ref{thm:T non-uniform} gives. The error term is then $\mathcal{O}(\omega(n)^{-1})$ and uniform for $p \in (\epsilon, \pi - \epsilon)$ and $\omega(n)/n < t < t_0$. Thus we have proven the statement on Toeplitz determinants in Theorem \ref{thm:T, T+H extended}.   
\end{remark}

\subsection{Integration of the Differential Identity}

We now integrate (\ref{eqn:diff_id_asymptotics}), where we use exactly the same approach as in Section 8.2 of \cite{Claeys2015}. We obtain
\begin{align}
&\int_0^t d_1(p,\tau;\alpha_0,\alpha_1,\beta_1,\alpha_2,\beta_2,\alpha_3) \text{d}\tau \nonumber \\
=& +2i\alpha_1 \left( V(e^{i(p-t)}) - V(e^{ip}) \right) + 2i\alpha_2 \left( V(e^{i(p+t)}) - V(e^{ip}) \right) \\
&+ 2(\beta_1+\beta_2) \sum_{j = 1}^\infty V_j \left( \sin j(p+t) + \sin j(p-t) - 2 \sin jp\right) \nonumber \\
& + i\left( 2\alpha_1^2 + \beta_1^2 + \beta_1\beta_2 \right) \log \frac{\sin p-t}{\sin p} +i \left( 2\alpha_2^2 + \beta_2^2 + \beta_1\beta_2 \right) \log \frac{\sin p+t}{\sin p} + it(\beta_1^2 - \beta_2^2) \frac{\cos p}{\sin p} \nonumber \\
& + 4i\alpha_1 \alpha_0 \log \frac{\sin \frac{p-t}{2}}{\sin \frac{p}{2}} - 4i\alpha_1 \alpha_3 \log \frac{\cos \frac{p-t}{2}}{\cos \frac{p}{2}} + 4i\alpha_2 \alpha_0 \log \frac{ \sin \frac{p+t}{2} }{ \sin \frac{p}{2}} - 4i\alpha_2 \alpha_3 \log \frac{ \cos \frac{p+t}{2} }{ \cos \frac{p}{2} }, \nonumber 
\end{align}
and 
\begin{align}
&\int_0^t d_2(p,\tau;\alpha_0,\alpha_1,\beta_1,\alpha_2,\beta_2,\alpha_3) \text{d}\tau \\
=&- 2i\int_0^{-2int} \frac{1}{s} \left( \sigma(s) - 2\alpha_1\alpha_2 + \frac{1}{2}(\beta_1+\beta_2)^2 \right) \text{d}s + i\left( (\beta_1+\beta_2)^2 - 4\alpha_1\alpha_2 \right) \log \frac{\sin t}{t}, \nonumber 
\end{align}
and 
\begin{align}
&\int_0^t d_3(p,\tau;\alpha_0,\alpha_1,\beta_1,\alpha_2,\beta_2,\alpha_3) \text{d}\tau \nonumber \\
=& (\beta_1 - \beta_2) \Bigg[ 2\sum_{j = 1}^\infty V_j \left( \sin j(p-t) - \sin j(p+t) \right) - 2t \sum_{j=0}^5 \alpha_j \\
& + i\beta_1 \log \frac{\sin p-t}{\sin p} - i\beta_2 \log \frac{\sin p+t}{\sin p} + i(\beta_1 - \beta_2) \log  \frac{\sin t}{t} - it(\beta_1 + \beta_2) \frac{\cos p}{\sin p} \Bigg]. \nonumber 
\end{align}
Putting things together we see that, uniformly in $p \in (\epsilon, \pi - \epsilon)$ and $0 < t < t_0$, 
\begin{align} 
&\log D_n(f_{p,t}) \nonumber \\
=& \log D_n(f_{p,0}) + 2int (\beta_1 - \beta_2) \nonumber \\
&- 2\alpha_1 \left( V(e^{i(p-t)}) - V(e^{ip}) \right) - 2\alpha_2 \left( V(e^{i(p+t)}) - V(e^{ip}) \right) \nonumber \\
&+ 4i \sum_{j = 1}^\infty V_j \left( \beta_ 1\sin j(p-t) + \beta_2 \sin j(p+t) - (\beta_1+\beta_2) \sin jp\right) \\
&- 2\left( \alpha_1^2 + \beta_1^2 \right) \log \frac{\sin p-t}{\sin p} - 2\left( \alpha_2^2 + \beta_2^2 \right) \log \frac{\sin p+t}{\sin p} \nonumber \\
& - 4\alpha_1 \alpha_0 \log \frac{\sin \frac{p-t}{2}}{\sin \frac{p}{2}} - 4\alpha_1 \alpha_3 \log \frac{\cos \frac{p-t}{2}}{\cos \frac{p}{2}} - 4\alpha_2 \alpha_0 \log \frac{ \sin \frac{p+t}{2} }{ \sin \frac{p}{2}} - 4i\alpha_2 \alpha_3 \log \frac{ \cos \frac{p+t}{2} }{ \cos \frac{p}{2} } \nonumber \\
&- 2i\int_0^{-2int} \frac{1}{s} \left( \sigma(s) - 2\alpha_1\alpha_2 + \frac{1}{2}(\beta_1+\beta_2)^2 \right) \text{d}s \nonumber \\
&+ 4\left( \beta_1 \beta_2 - \alpha_1\alpha_2 \right) \log \frac{\sin t}{t} \nonumber \\ 
&- 2it(\beta_1 - \beta_2) \sum_{j = 0}^5 \alpha_j \nonumber \\
&+ o(1). \nonumber 
\end{align}
To calculate the asymptotics of $D_n(f_{p,0})$ we use Theorem \ref{thm:T non-uniform} and get
\begin{align}
\begin{split}
&\log D_n(f_{p,0}) \\
=& n V_0 + \sum_{k = 1}^\infty kV_k^2 + \log n \left( \alpha_0^2 + \alpha_3^2 + 2(\alpha_1 + \alpha_2)^2 - 2(\beta_1 + \beta_2)^2 \right) \\
&- \alpha_0 \log \left( b_+(1)b_-(1) \right) - \alpha_3 \log \left( b_+(-1)b_-(-1) \right) \\
&- 2(\alpha_1 + \alpha_2 - \beta_1 - \beta_2) \log b_+(e^{ip}) - 2(\alpha_1 + \alpha_2 + \beta_1 + \beta_2) \log b_+(e^{-ip}) \\
&- 2\alpha_0\alpha_3 \log 2 - 4\alpha_0(\alpha_1+\alpha_2) \log 2 \sin \frac{p}{2} - 4\alpha_3(\alpha_1+\alpha_2) \log 2 \cos \frac{p}{2} \\
&- 2((\beta_1+\beta_2)^2 + (\alpha_1+\alpha_2)^2) \log 2\sin p  \\
&+ 2i\alpha_0(\beta_1+\beta_2)(p - \pi) + 2i\alpha_3(\beta_1+\beta_2) p  \\
&+ 2i(\alpha_1+\alpha_2)(\beta_1+\beta_2)(2p - \pi) \\
&+ \log \frac{G(1+\alpha_0)^2}{G(1+2\alpha_0)} + \log \frac{G(1+\alpha_3)^2}{G(1+2\alpha_3)} \\
&+ \log \frac{G(1+\alpha_1+\alpha_2+\beta_1+\beta_2)^2 G(1+\alpha_1+\alpha_2-\beta_1-\beta_2)^2}{G(1+2\alpha_1+2\alpha_2)^2}\\
&+ o(1),
\end{split}
\end{align}
uniformly in $p \in (\epsilon, \pi - \epsilon)$. Combining the two last equations we obtain
\begin{align} \label{eqn:ln_D_sin}
&\log D_n(f_{p,t}) \nonumber \\
=& 2int (\beta_1 - \beta_2) + n V_0 + \sum_{k = 1}^\infty kV_k^2 + \log n \sum_{j = 0}^5 (\alpha_j^2 - \beta_j^2) \nonumber \\
&- \sum_{j=0}^5 (\alpha_j - \beta_j) \left(\sum_{k = 1}^\infty V_k z_j^k\right) + (\alpha_j + \beta_j) \left( \sum_{k = 1}^\infty V_k \overline{z_j}^k \right) \nonumber \\
&- 2\left( \alpha_1^2 + \beta_1^2 \right) \log 2\sin (p-t) - 2\left( \alpha_2^2 + \beta_2^2 \right) \log 2\sin (p+t) \nonumber \\
&- 4(\alpha_1\alpha_2 + \beta_1\beta_2) \log 2\sin p - 2\alpha_0\alpha_3 \log 2 \nonumber \\
& - 4\alpha_1 \alpha_0 \log 2\sin \frac{p-t}{2} - 4\alpha_1 \alpha_3 \log 2\cos \frac{p-t}{2} - 4\alpha_2 \alpha_0 \log 2\sin \frac{p+t}{2} - 4\alpha_2 \alpha_3 \log 2\cos \frac{p+t}{2} \nonumber \\
&+ 2\int_0^{-2int} \frac{1}{s} \left( \sigma(s) - 2\alpha_1\alpha_2 + \frac{1}{2}(\beta_1+\beta_2)^2 \right) \text{d}s + 4\left( \beta_1 \beta_2 - \alpha_1\alpha_2 \right) \log \frac{\sin t}{nt} \nonumber \\ 
&+ 2i\alpha_0\beta_1(p-t-\pi) + 2i\alpha_0\beta_2(p+t-\pi) + 2i\alpha_3\beta_1 (p-t) + 2i\alpha_3\beta_2(p+t) \nonumber \\
&+ 2i\alpha_1\beta_1(2p-2t-\pi) + 2i\alpha_2\beta_1(2p-2t-\pi) + 2i\alpha_1\beta_2 (2p+2t-\pi) + 2i\alpha_2\beta_2(2p+2t-\pi) \nonumber \\
&+ \log \frac{G(1+\alpha_0)^2}{G(1+2\alpha_0)} + \log \frac{G(1+\alpha_3)^2}{G(1+2\alpha_3)} \nonumber \\
&+ \log \frac{G(1+\alpha_1+\alpha_2+\beta_1+\beta_2)^2 G(1+\alpha_1+\alpha_2-\beta_1-\beta_2)^2}{G(1+2\alpha_1+2\alpha_2)^2} \\
&+ o(1), \nonumber 
\end{align}
uniformly in $p \in (\epsilon, \pi - \epsilon)$ and $0 < t < t_0$. We note that
\begin{align}
&\sum_{0 \leq j < k \leq 5, (j,k) \neq (1,2),(4,5)} (\alpha_j\beta_k - \alpha_k \beta_j) \log \frac{z_k}{z_j e^{i\pi}} \\
=& 2i\alpha_0\beta_1 (p-t-\pi) + 2i\alpha_0\beta_2(p+t-\pi) + 2i\alpha_3\beta_1(p-t) + 2i\alpha_3\beta_2(p+t) \nonumber \\
&+ 2i(\alpha_2\beta_1 + \alpha_1\beta_2)(2p - \pi) + 2i\alpha_1\beta_1(2p - 2t -\pi) + 2i\alpha_2\beta_2(2p + 2t -\pi). \nonumber 
\end{align}
Thus we finally get that
\begin{align} \label{eqn:ln_D_uniform}
&\log D_n(f_{p,t}) \nonumber \\
=& 2int(\beta_1 - \beta_2) + n V_0 + \sum_{k = 1}^\infty kV_k^2 + \log n \sum_{j = 0}^5 (\alpha_j^2 - \beta_j^2) \nonumber \\
&- \sum_{j=0}^5 (\alpha_j - \beta_j) \left(\sum_{k = 1}^\infty V_k z_j^k\right) + (\alpha_j + \beta_j) \left( \sum_{k = 1}^\infty V_k \overline{z_j}^k \right) \nonumber \\
&+ \sum_{0 \leq j < k \leq 5, (j,k) \neq (1,2),(4,5)} 2(\beta_j\beta_k - \alpha_j\alpha_k) \log |z_j - z_k| + (\alpha_j\beta_k - \alpha_k \beta_j) \log \frac{z_k}{z_j e^{i\pi}} \nonumber \\
&+4it(\alpha_1\beta_2 - \alpha_2\beta_1) \nonumber \\
&+ 2\int_0^{-2int} \frac{1}{s} \left( \sigma(s) - 2\alpha_1\alpha_2 + \frac{1}{2}(\beta_1+\beta_2)^2 \right) \text{d}s + 4\left( \beta_1 \beta_2 - \alpha_1\alpha_2 \right) \log \frac{\sin t}{nt} \nonumber \\ 
&+ \log \frac{G(1+\alpha_0)^2}{G(1+2\alpha_0)} + \log \frac{G(1+\alpha_3)^2}{G(1+2\alpha_3)} \nonumber \\
&+ \log \frac{G(1+\alpha_1+\alpha_2+\beta_1+\beta_2)^2 G(1+\alpha_1+\alpha_2-\beta_1-\beta_2)^2}{G(1+2\alpha_1+2\alpha_2)^2} \\
&+ o(1), \nonumber 
\end{align}
uniformly in $p \in (\epsilon, \pi - \epsilon)$ and $0 < t < t_0$, which proves Theorem \ref{thm:T uniform}.

\section{Asymptotics of $\Phi_n(0)$, $\Phi_n(\pm 1)$ and $D_n^{T+H,\kappa}(f_{p,t})$} \label{section:T+H}

Tracing back the transformations of Section \ref{section:asymptotics of polynomials} we get  
\begin{align}
\begin{split}
\Phi_n(0) &= Y^{(n)}_{11}(0) = T(0)_{11} = S(0)_{11} = \left( R(0)N(0) \right)_{11} = - R_{12}(0) D_{\text{in},p,t}(0)^{-1} \\
&= - R_{12}(0) e^{-V_0} = o(1),
\end{split}
\end{align}
uniformly in $p \in (\epsilon, \pi - \epsilon)$ and $0 < t < t_0$. \\

The asymptotics of $\Phi_n(1)$ and $\Phi_n(-1)$ can be calculated exactly as in Chapter 7 of \cite{Deift2011}. As is apparent from (7.13) there, these asymptotics are uniform for all other singularities bounded away from $\pm 1$. Thus when using Lemma \ref{lemma:connection between T and T+H} to calculate from the asymptotics of $D_n(f_{p,t})$ given in Theorem \ref{thm:T, T+H extended}, $\Phi_n(\pm 1)$ and $\Phi_n(0)$, the asymptotics of the corresponding Toeplitz+Hankel determinants, then uniformity of the error terms in $p,t$ is preserved. This proves the statement on Toeplitz+Hankel determinants in Theorem \ref{thm:T, T+H extended}. 

Further, since we know the relation (\ref{eqn:relation}) between the asymptotic expansion of $D_n(f_{p,t})$ which is uniform for $0 < t < t_0$, and the asymptotic expansion for $t > t_0$, Lemma \ref{lemma:connection between T and T+H} immediately gives the relationship between the expansions of $D_n^{T+H,\kappa}(f_{p,t})$ in the two regimes $0 < t < t_0$ and $t > t_0$. In view of the way the uniform asymptotics of $D_n^{T+H,\kappa}(f_{p,t})$ are derived from the uniform asymptotics of $D_{2n}(f_{p,t})^{1/2}$, $\Phi_n(\pm 1)^{1/2}$ and $\Phi_n(0)$, the relationship between the $0 < t < t_0$ asymptotics of $D_n^{T+H,\kappa}(f_{p,t})$ and the $t > t_0$ asymptotics one gets from Theorem \ref{thm:T+H uniform} is given by (\ref{eqn:relation}), with both sides divided by $2$, and $n$ replaced by $2n$. This proves Theorem \ref{thm:T+H uniform}.

\begin{subappendices}

\section{Riemann-Hilbert Problem for $\Psi$} \label{appendix:Psi}
This appendix is a mostly verbatim transfer from the beginning of Section 3 of \cite{Claeys2015}. We include it here to make our account self-contained.  We use $\Psi$ to construct local parametrices for the RHP for the orthogonal polynomials in Section \ref{section:local 1}. We always assume that $\alpha_1,\alpha_2 > -\frac{1}{2}$ and $\beta_1,\beta_2 \in i \mathbb{R}$ (in \cite{Claeys2015} also the more general case of $\alpha_1,\alpha_2,\beta_1,\beta_2 \in \mathbb{C}$ was considered). \\

\noindent \textbf{RH Problem for $\Psi$}
\begin{enumerate}[label=(\alph*)]
\item $\Psi:\mathbb{C} \setminus \Gamma \rightarrow \mathbb{C}^{2\times 2}$ is analytic, where 
\begin{align}
\Gamma =& \cup_{k = 1}^7 \Gamma_k,  &&\Gamma_1 = i + e^{\frac{i\pi}{4}} \mathbb{R}_+, &&\Gamma_2 = i + e^{\frac{3i\pi}{4}\mathbb{R}_+} \nonumber \\
\Gamma_3 =& -i + e^{\frac{5i\pi}{4}} \mathbb{R}_+, &&\Gamma_4 = -i + e^{\frac{7i\pi}4} \mathbb{R}_+, &&\Gamma_5 = -i + \mathbb{R}_+, \\
\Gamma_6 =& i + \mathbb{R}_+, &&\Gamma_7 = [-i,i], \nonumber
\end{align}
with the orientation chosen as in Figure \ref{figure:Psi} ("-" is always on the RHS of the contour).
\item $\Psi$ satisfies the jump conditions
\begin{equation}
\Psi(\zeta)_+ = \Psi(\zeta)_- J_k, \quad \zeta \in \Gamma_k,
\end{equation}
where 
\begin{align}
J_1 =& \left( \begin{array}{cc} 1 & e^{2\pi i(\alpha_2 - \beta_2)} \\ 0 & 1 \end{array} \right), && J_2 = \left( \begin{array}{cc} 1 & 0 \\ - e^{-2\pi i(\alpha_2 - \beta_2)} & 1 \end{array} \right), \nonumber \\
J_3 =& \left( \begin{array}{cc} 1 & 0 \\ -e^{2\pi i (\alpha_1 - \beta_1)} & 1 \end{array} \right), && J_4 = \left( \begin{array}{cc} 1 & e^{-2\pi i(\alpha_1 - \beta_1)} \\ 0 & 1 \end{array} \right) \\
J_5 =& e^{2\pi i \beta_1 \sigma_3}, && J_6 = e^{2\pi i \beta_2 \sigma_3}, \nonumber \\
J_7 =& \left( \begin{array}{cc} 0 & 1 \\ -1 & 1 \end{array} \right). \nonumber 
\end{align}

\item We have in all regions:
\begin{equation} \label{eqn:Psi asymptotics}
\Psi(\zeta) = \left( I + \frac{\Psi_1}{\zeta} + \frac{\Psi_2}{\zeta^2} + \mathcal{O}(\zeta^{-3}) \right) P^{(\infty)} (\zeta) e^{-\frac{is}{4} \zeta \sigma_3}, \quad \text{as } \zeta \rightarrow \infty,
\end{equation}
where 
\begin{equation}
P^{(\infty)}(\zeta) = \left( \frac{is}{2} \right)^{-(\beta_1 + \beta_2)\sigma_3} (\zeta - i)^{-\beta_2 \sigma_3} (\zeta + i)^{-\beta_1 \sigma_3},
\end{equation}
with the branches corresponding to the arguments between $0$ and $2\pi$, and where $s \in -i\mathbb{R}_+$.

\item The functions $F_1$ and $F_2$ defined in (\ref{eqn:F_1 neq 0}), (\ref{eqn:F_2 neq 0}), (\ref{eqn:F_1 0}) and (\ref{eqn:F_2 0}) below are analytic functions of $\zeta$ at $i$ and $-i$ respectively.
\end{enumerate}

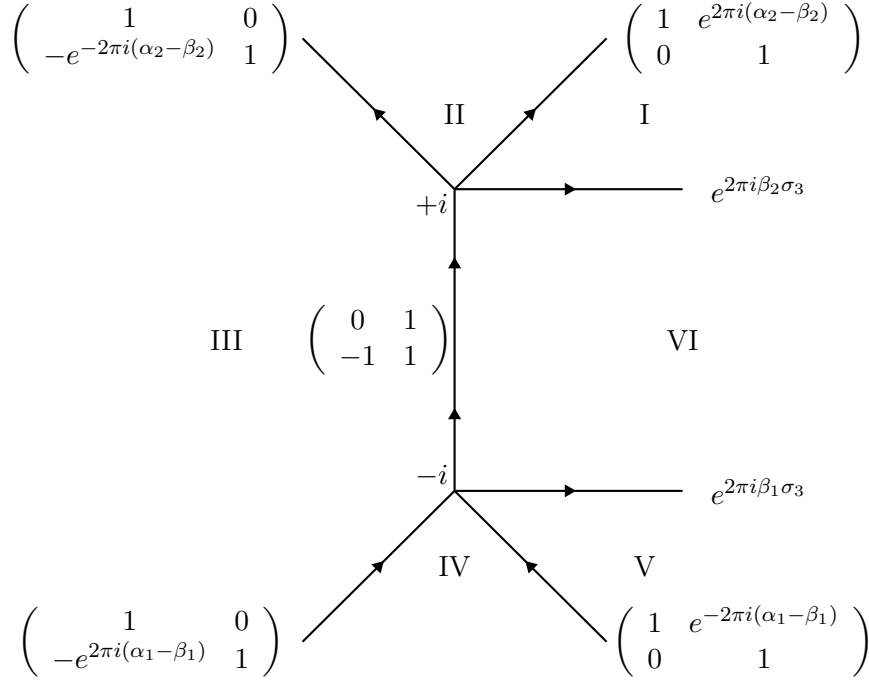
\begin{figure} [H]
\centering
\begin{tikzpicture}
%Contour 
\draw [thick] (0,-2) -- (0,2);
\draw [thick] (0,-2) -- (3,-2);
\draw [thick] (0,2) -- (3,2);
\draw [thick] (0,-2) -- (2,-4);
\draw [thick] (0,-2) -- (-2,-4);
\draw [thick] (0,2) -- (2,4);
\draw [thick] (0,2) -- (-2,4);

%Jump matrices
\fill (-0.3,1.8) node[] {$+i$};
\fill (-0.3,-1.8) node[] {$-i$};
\fill (-1,0) node[] {$\left( \begin{array}{cc} 0 & 1 \\ -1 & 1 \end{array} \right)$};
\fill (3.8,4) node[] {$\left( \begin{array}{cc} 1 & e^{2\pi i(\alpha_2 - \beta_2)} \\ 0 & 1 \end{array} \right)$};
\fill (-4,4) node[] {$\left( \begin{array}{cc} 1 & 0 \\ -e^{-2\pi i(\alpha_2 - \beta_2)} & 1 \end{array} \right)$};
\fill (4,2) node[] {$e^{2\pi i\beta_2 \sigma_3}$};
\fill (4,-2) node[] {$e^{2\pi i\beta_1 \sigma_3}$};
\fill (3.8,-4) node[] {$\left( \begin{array}{cc} 1 & e^{-2\pi i(\alpha_1 - \beta_1)} \\ 0 & 1 \end{array} \right)$};
\fill (-4,-4) node[] {$\left( \begin{array}{cc} 1 & 0 \\ -e^{2\pi i(\alpha_1 - \beta_1)} & 1 \end{array} \right)$};

%Orientation of Contour
\node[fill=black,regular polygon, regular polygon sides=3,inner sep=1.pt, shape border rotate = -45] at (1,3) {};
\node[fill=black,regular polygon, regular polygon sides=3,inner sep=1.pt, shape border rotate = 45] at (-1,3) {};
\node[fill=black,regular polygon, regular polygon sides=3,inner sep=1.pt, shape border rotate = 0] at (0,1) {};
\node[fill=black,regular polygon, regular polygon sides=3,inner sep=1.pt, shape border rotate = 0] at (0,-1) {};
\node[fill=black,regular polygon, regular polygon sides=3,inner sep=1.pt, shape border rotate = 45] at (1,-3) {};
\node[fill=black,regular polygon, regular polygon sides=3,inner sep=1.pt, shape border rotate = -45] at (-1,-3) {};
\node[fill=black,regular polygon, regular polygon sides=3,inner sep=1.pt, shape border rotate = -90] at (1.5,2) {};
\node[fill=black,regular polygon, regular polygon sides=3,inner sep=1.pt, shape border rotate = -90] at (1.5,-2) {};

%Regions
\fill (0,3) node[] {II};
\fill (0,-3) node[] {IV};
\fill (3,0) node[] {VI};
\fill (-3,0) node[] {III};
\fill (2.5,3) node[] {I};
\fill (2.5,-3) node[] {V};
\end{tikzpicture}
\caption{The jump contour and jump matrices of $\Psi$.} \label{figure:Psi}
\end{figure}

The solution $\Psi = \Psi(\zeta;s)$ to this RHP not only depends on the complex variable $\zeta$, but also on the complex parameter $s \in -i\mathbb{R}_+$. Without the additional condition (d) on the behaviour of $\Psi$ near the points $\pm i$, the RHP wouldn't have a unique solution. If $2 \alpha_2 \notin \mathbb{N} \cup \{0\}$, define $F_1(\zeta,s)$ by the equations
\begin{equation} \label{eqn:F_1 neq 0}
\Psi(\zeta;s) = F_1(\zeta,s)(\zeta - i)^{\alpha_2\sigma_3} G_j, \quad \zeta \in \text{region } j,
\end{equation}
where $j \in \{I,II,III,VI\}$, and where $(\zeta - i)^{\alpha_2 \sigma_3}$ is taken with the branch cut on $i + e^{\frac{3\pi i}{4}} \mathbb{R}_+$, with the argument of $\zeta - i$ between $-5\pi/4$ and $3\pi/4$. The matrices $G_j$ are piecewise constant matrices consistent with the jump relations; they are given by 
\begin{align}
\begin{split}
G_{III} =& \left( \begin{array}{cc} 1 & g \\ 0 & 1 \end{array} \right), \quad g = - \frac{1}{2 i \sin (2\alpha_2)}(e^{2\pi i\alpha_2} - e^{-2\pi i \beta_2}),\\
G_{VI} =& G_{III} J_7^{-1}, \quad G_I = G_{VI}J_6, \quad G_{II} = G_I J_1.
\end{split}
\end{align}
It is straighforward to verify that $F_1$ has no jumps near $i$, and it is thus meromorphic in a neighborhood of $i$, with possibly an isolated singularity at $i$. 

Similarly, for $\zeta$ near $-i$, if $2\alpha_1 \notin \mathbb{N} \cup \{0\}$, we define $F_2$ by the equations
\begin{equation} \label{eqn:F_2 neq 0}
\Psi(\zeta;s) = F_2(\zeta,s)(\zeta + i)^{\alpha_1\sigma_3} H_j, \quad \zeta \in \text{region } j,
\end{equation}
where $j \in \{III,,IV,V,VI\}$, where $(\zeta + i)^{\alpha_1 \sigma_3}$ is taken with the branch cut on $-i + e^{\frac{5\pi i}{4}} \mathbb{R}_+$, with the argument of $\zeta + i$ between $-3\pi/4$ and $5\pi/4$, and where the matrices $H_j$ are piecewise constant matrices consistent with the jump relations; they are given by 
\begin{align} \label{eqn:H_j}
\begin{split}
H_{III} =& \left( \begin{array}{cc} 1 & h \\ 0 & 1 \end{array} \right), \quad h = - \frac{1}{2 i \sin (2\alpha_1)}(e^{2\pi i\beta_1} - e^{-2\pi i \alpha_1}),\\
H_{IV} =& G_{III} J_3^{-1}, \quad H_V = G_{IV}J_4^{-1}, \quad H_{VI} = H_V J_5.
\end{split}
\end{align}
Similarly as at $i$, one shows using the jump conditions for $\Psi$ that $F_2$ is  meromorphic near $-i$ with a possible singularity at $-i$. \\

If $2\alpha_2 \in \mathbb{N} \cup \{0\}$, the constant $g$ and the matrices $G_j$ are ill-defined, and we need a different definition of $F_1$:
\begin{align} \label{eqn:F_1 0}
\Psi(\zeta;s) = F_1(\zeta;s) (\zeta - i)^{\alpha_2 \sigma_3} \left( \begin{array}{cc} 1 & g_{int} \log (\zeta - i) \\ 0 & 1 \end{array} \right) G_j, \quad \zeta \in \text{region } j,
\end{align}
where 
\begin{align}
g_{int} = \frac{e^{-2\pi i \beta_2} - e^{2\pi i \alpha_2}}{2\pi i e^{2\pi i \alpha_2}},
\end{align}
and $G_{III} = I$, and the other $G_j$'s are defined as above by applying the appropriate jump conditions. Thus defined, $F_1$ has no jumps in a neighborhood of $i$. Similarly, if $2 \alpha_1 \in \mathbb{N} \cup \{0\}$, we define $F_2$ by the expression:
\begin{align} \label{eqn:F_2 0}
\Psi(\zeta;s) = F_2(\zeta;s) (\zeta + i)^{\alpha_1\sigma_3} \left( \begin{array}{cc} 1 & \frac{e^{-2\pi i \alpha_1} - e^{2\pi i \beta_1}}{2\pi i e^{2\pi i \alpha_1}} \\ 0 & 1 \end{array} \right) H_j, \quad \zeta \in \text{region } j, 
\end{align}
with $H_{III} = I$, and the other $H_j$'s expressed via $H_{III}$ as in (\ref{eqn:H_j}). Then $F_2$ has no jumps near $-i$. \\

Given parameters $s,\alpha_1,\alpha_2,\beta_1,\beta_2$, the uniqueness of the function $\Psi$ which satisfies RH conditions (a) - (d) can be proven by standard arguments. \\

In Section 3 of \cite{Claeys2015} it was shown that for $\alpha_1,\alpha_2, \alpha_1 + \alpha_2 > -\frac{1}{2}$ and $\beta_1,\beta_2 \in i\mathbb{R}$, the RHP is solvable for any $s \in -i\mathbb{R}_+$. Furthermore they analyzed the RHP asymptotically as $s \rightarrow -i\infty$ and $s \rightarrow -i0_+$. 

\begin{remark}
We have to be careful what their $(\alpha_1,\alpha_2,\beta_1,\beta_2)$ correspond to in our case, when using $\Psi$ from \cite{Claeys2015}. In their paper $\alpha_1,\beta_1$ correspond to the singularity left of the merging point and $\alpha_2,\beta_2$ correspond to the singularity to the right of the merging point, while for us in the case + for example, $\alpha_1, \beta_1$ are right and $\alpha_2,\beta_2$ are left.
\end{remark}

\section{Riemann-Hilbert Problem for $M$} \label{appendix:M}
This appendix is a mostly verbatim transfer of Section 4 of \cite{Claeys2015}. We include it here to make our account self-contained.  Let $\alpha > -\frac{1}{2}$ and $\beta \in i\mathbb{R}$. In Section 4.2.1 of \cite{Claeys2011}, see also \cite{Deift2011, Its2008, MMFS10}, a function $M = M^{(\alpha,\beta)}$ was constructed explicitly in terms of the confluent hypergeometric function, which solves the following RH problem:\\

\noindent \textbf{RH Problem for $M$}
\begin{enumerate}[label=(\alph*)]
\item $M: \mathbb{C}\setminus \left( e^{\pm \frac{\pi i }{4}} \mathbb{R} \cup \mathbb{R}_+ \right) \rightarrow \mathbb{C}^{2 \times 2}$ is analytic, 
\item $M$ has continuous boundary values on $e^{\pm \frac{\pi i }{4}} \mathbb{R} \cup \mathbb{R}_+ \setminus \{0\}$ related by the conditions:
\begin{align}
M(\lambda)_+ =& M(\lambda)_- \left( \begin{array}{cc} 1 & e^{\pi i (\alpha - \beta)} \\ 0 & 1 \end{array} \right), && \lambda \in e^{\frac{i\pi}{4}} \mathbb{R}_+, \nonumber \\
M(\lambda)_+ =& M(\lambda)_- \left( \begin{array}{cc} 1 & 0 \\ -e^{-\pi i (\alpha - \beta)} & 0 \end{array} \right), && \lambda \in e^{\frac{3i\pi}{4}} \mathbb{R}_+, \nonumber \\
M(\lambda)_+ =& M(\lambda)_- \left( \begin{array}{cc} 1 & 0 \\ e^{\pi i (\alpha - \beta)} & 0 \end{array} \right), && \lambda \in e^{\frac{5i\pi}{4}} \mathbb{R}_+, \\
M(\lambda)_+ =& M(\lambda)_- \left( \begin{array}{cc} 1 & -e^{-\pi i (\alpha - \beta)} \\ 0 & 1 \end{array} \right), && \lambda \in e^{\frac{7i\pi}{4}} \mathbb{R}_+, \nonumber \\
M(\lambda)_+ =& M(\lambda)_-(\lambda) e^{2\pi i \beta \sigma_3}, && \lambda \in \mathbb{R}_+, \nonumber 
\end{align}
where all the rays of the jump contour are oriented away from the origin. 

\item Furthermore, in all sectors,
\begin{equation} \label{eqn:M asymptotics}
M(\lambda) = (I + M_1 \lambda^{-1} + \mathcal{O}(\lambda^{-2})) \lambda^{-\beta \sigma_3} e^{-\frac{1}{2} \lambda \sigma_3}, \quad \text{as } \lambda \rightarrow \infty,
\end{equation}
where $0 < \arg \lambda < 2\pi$, and 
\begin{equation}
M_1 = M_1^{(\alpha,\beta)} = \left( \begin{array}{cc} \alpha^2 - \beta^2 & - e^{-2\pi i \beta} \frac{\Gamma(1 + \alpha - \beta)}{\Gamma(\alpha + \beta)} \\ e^{2\pi i \beta} \frac{\Gamma(1+\alpha+\beta)}{\Gamma(\alpha - \beta)} & - \alpha^2 + \beta^2 \end{array} \right).
\end{equation}
\end{enumerate}

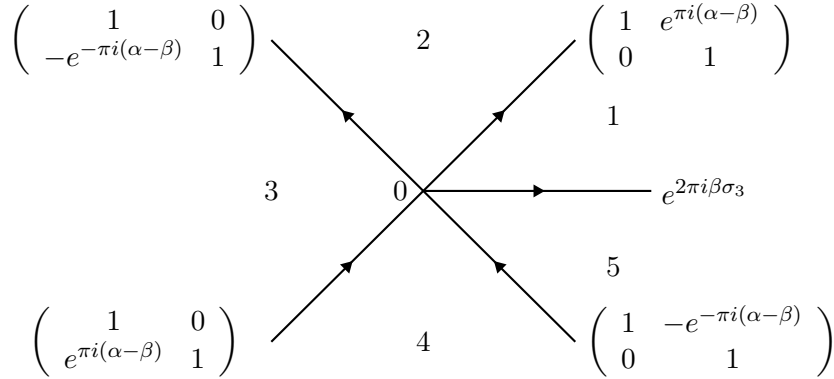
\begin{figure} [H]
\centering
\begin{tikzpicture}
%Contour 
\draw [thick] (0,0) -- (3,0);
\draw [thick] (0,0) -- (2,-2);
\draw [thick] (0,0) -- (-2,-2);
\draw [thick] (0,0) -- (2,2);
\draw [thick] (0,0) -- (-2,2);

%Jump matrices
\fill (-0.3,0) node[] {$0$};
\fill (3.5,2) node[] {$\left( \begin{array}{cc} 1 & e^{\pi i(\alpha - \beta)} \\ 0 & 1 \end{array} \right)$};
\fill (-3.8,2) node[] {$\left( \begin{array}{cc} 1 & 0 \\ -e^{-\pi i(\alpha - \beta)} & 1 \end{array} \right)$};
\fill (3.7,0) node[] {$e^{2\pi i\beta \sigma_3}$};
\fill (3.8,-2) node[] {$\left( \begin{array}{cc} 1 & -e^{-\pi i(\alpha - \beta)} \\ 0 & 1 \end{array} \right)$};
\fill (-3.8,-2) node[] {$\left( \begin{array}{cc} 1 & 0 \\ e^{\pi i(\alpha - \beta)} & 1 \end{array} \right)$};

%Orientation of Contour
\node[fill=black,regular polygon, regular polygon sides=3,inner sep=1.pt, shape border rotate = -45] at (1,1) {};
\node[fill=black,regular polygon, regular polygon sides=3,inner sep=1.pt, shape border rotate = 45] at (-1,1) {};
\node[fill=black,regular polygon, regular polygon sides=3,inner sep=1.pt, shape border rotate = 45] at (1,-1) {};
\node[fill=black,regular polygon, regular polygon sides=3,inner sep=1.pt, shape border rotate = -45] at (-1,-1) {};
\node[fill=black,regular polygon, regular polygon sides=3,inner sep=1.pt, shape border rotate = -90] at (1.5,0) {};

%Regions
\fill (0,2) node[] {2};
\fill (0,-2) node[] {4};
\fill (-2,0) node[] {3};
\fill (2.5,1) node[] {1};
\fill (2.5,-1) node[] {5};
\end{tikzpicture}
\caption{The jump contour and jump matrices of $M$.} \label{figure:M}
\end{figure}

We use $M$ to construct local parametrices for the RHP of the orthogonal polynomials in Section \ref{section:local 2}. For that we also need the local behaviour of $M$ at zero in region 3, i.e. the region between the lines $e^{\frac{3\pi i}{4}} \mathbb{R}_+$ and $e^{\frac{5\pi i}{4}} \mathbb{R}_+$. Write $M \equiv M^{(3)}$ in this region. It is known (see Section 4.2.1 of \cite{Claeys2011}) that $M^{(3)}$ can be written in the form 
\begin{align} \label{eqn:M at 0, neq 0}
M^{(3)}(\lambda) = L(\lambda)\lambda^{\alpha\sigma_3} \tilde{G}_3, \quad 2\alpha \notin \mathbb{N} \cup \{0\},
\end{align}
with the branch of $\lambda^{\pm \alpha}$ chosen with $0 < \arg \lambda < 2\pi$. Here
\begin{align}
\begin{split}
    L(\lambda) =& e^{-\lambda/2} \displaystyle \Big( \begin{array}{c} e^{-i\pi(\alpha+\beta)} \frac{\Gamma(1+\alpha - \beta)}{\Gamma(1+2\alpha)} \varphi(\alpha+\beta,1+2\alpha,\lambda) \\
    -e^{-i\pi(\alpha-\beta)} \frac{\Gamma(1+\alpha + \beta)}{\Gamma(1+2\alpha)} \varphi(1+\alpha+\beta,1+2\alpha,\lambda) \end{array} \\
    &\hspace{3.5cm} \begin{array}{c} e^{i\pi(\alpha-\beta)} \frac{\Gamma(2\alpha)}{\Gamma(\alpha+\beta)} \varphi(-\alpha+\beta,1-2\alpha,\lambda) \\
    e^{i\pi(\alpha+\beta)} \frac{\Gamma(2\alpha)}{\Gamma(\alpha-\beta)} \varphi(1-\alpha+\beta,1-2\alpha,\lambda) \\
    \end{array} \Big),
\end{split}
\end{align}
is an entire function, with
\begin{align}
\varphi(a,c;z) = 1 + \sum_{n = 1}^\infty \frac{a(a+1) \cdots (a + n - 1)}{c(c+1) \cdots (c + n - 1)} \frac{z^n}{n!}, \quad c \neq 0,-1,-2,...,
\end{align}
and $\tilde{G}_3$ is the constant matrix 
\begin{align}
\tilde{G}_3 = \left( \begin{array}{cc} 1 & \tilde{g} \\ 0 & 1 \end{array} \right), \quad \tilde{g} = - \frac{\sin \pi(\alpha+\beta)}{\sin 2\pi \alpha}.
\end{align}
If $2 \alpha$ is an integer, we have 
\begin{align} \label{eqn:M at 0, 0}
M^{(3)}(\lambda) =& \tilde{L}(\lambda) \lambda^{\alpha \sigma_3} \left( \begin{array}{cc} 1 & m(\lambda) \\ 0 & 1 \end{array} \right), \\
m(\lambda) =& \frac{(-1)^{2\alpha + 1}}{\pi} \sin \pi(\alpha + \beta) \log (\lambda e^{-i\pi}), \nonumber 
\end{align}
where $\tilde{L}(\lambda)$ is analytic at $0$, and the branch of the logarithm corresponds to the argument of $\lambda$ between $0$ and $2\pi$. 

\end{subappendices}
\chapter{The Classical Compact Groups and Gaussian Multiplicative Chaos} \label{chapter:GMC}

In \cite{Hughes2005}, Hughes, Keating and O'Connell proved Theorem \ref{thm:HughesKeatingOConnell1}, which states that the real and the imaginary part of the logarithm of the characteristic polynomial of a random unitary matrix converges jointly to a pair of Gaussian fields on the unit circle. Using results on Toeplitz determinants with merging Fisher-Hartwig singularities due to Claeys and Krasovsky in \cite{Claeys2015}, Webb proved Theorem \ref{thm:Webb}, which states that powers of the exponential of the real and imaginary part of the logarithm of the characteristic polynomial of a random unitary matrix converge, when suitably normalized, to Gaussian multiplicative chaos measures on the unit circle. Webb proved the result only in the so-called $L^2$-phase, that is those powers for which the second moment of the total mass of the limiting GMC measure is finite. In \cite{Nikula2020} the result was extended to the whole $L^1$- or subcritical phase, i.e. was proven to also hold for the whole set of powers for which the limiting GMC measure is non-trivial.\\ 

This chapter is based on joint work with Jon Keating \cite{Forkel2021}, and its purpose is to extend Webb's result to the other classical compact groups, i.e. to the orthogonal and symplectic groups. Our starting point is Theorem \ref{thm:Gaussian field1}, due to Assiotis and Keating, concerning the convergence of the real and imaginary parts of the logarithm of the characteristic polynomials of random orthogonal or symplectic matrices to a pair of Gaussian fields on the unit circle. This is the analogous result to Theorem \ref{thm:HughesKeatingOConnell1} by Hughes, Keating and O'Connell \cite{Hughes2005}. We then complete the connection between the classical compact groups and Gaussian multiplicative chaos, by showing that for the orthogonal and symplectic groups we get statements similar to Theorem \ref{thm:Webb} by Webb. Using the same approach as Webb in \cite{Webb2015} we prove our results only in the $L^2$-phase, i.e. when the limiting GMC measure's total mass has a finite second moment. We believe that using the techniques in \cite{Nikula2020} one can extend our results to hold more generally in the whole $L^1$-phase, however extra care is needed since the covariance function of the underlying Gaussian field has singularities not just on the diagonal but also on the antidiagonal, in particular at $\pm 1$ the field has special behaviour. \\
Due to this special behaviour around $\pm 1$ we first consider the case where we restrict all involved measures to $(\epsilon, \pi - \epsilon) \cup (\pi + \epsilon, 2\pi - \epsilon)$, i.e. we exclude small neighborhoods around $\pm 1$. In this case we can use our Theorems \ref{thm:T+H uniform} and \ref{thm:T, T+H extended}, which give formulas for the uniform asymptotics of Toeplitz+Hankel determinants with up to two pairs of merging singularities which are all bounded away from $\pm 1$. \\
To prove convergence on the full unit circle we need also to know the uniform asymptotics of Toeplitz+Hankel determinants with 3 or 5 singularities merging at $\pm 1$, for which we use Theorem \ref{thm:T+H Claeys} by Claeys, Glesner, Minakov and Yang, which gives a formula for the uniform asymptotics of Toeplitz+Hankel determinants with arbitrarily many merging singularities, up to a multiplicative constant \cite{Claeys2022}. Using their formula allows us to extend our analysis to around $\pm 1$, and so to cover the full unit circle, however for a slightly smaller set of powers than when neighborhoods around $\pm 1$ are excluded.

\section{Statement of Results and Strategy of Proof}

Recall that for $G(n) \in \left\{ U(n), \, O(n), \, SO(n), \, SO^-(n), \, Sp(2n) \right\}$ we define the characteristic polynomial of $U \in G(n)$ as a function on the unit circle, where all its zeroes lie:
\begin{equation}
p_{G(n)}(\theta) := \text{det} \left(I-e^{-i\theta}U\right) = \prod_{k = 1}^n (1 - e^{-i(\theta - \theta_k)}), \quad \theta \in [0,2\pi),
\end{equation} 
where $e^{i\theta_k}$, $k = 1,...,n$ are the eigenvalues of $U$ (and where the product is up to $2n$ in the case $G(n) = Sp(2n)$). 

\begin{definition} For $n \in \mathbb{N}$, $\alpha \in \mathbb{R}$, $\beta \in i\mathbb{R}$ and $\theta \in [0,2\pi)$ let (from which $G(n)$ $U$ is taken will be clear from context)
\begin{equation}
f_{n,\alpha,\beta}(\theta) = |p_{G(n)}(\theta)|^{2\alpha} e^{2i\beta \Im \log p_{G(n)}(\theta)},
\end{equation}
where (with the sum being up to $2n$ for $G(n) = Sp(2n)$)
\begin{equation}
\Im \log p_{G(n)}(\theta):= \sum_{k = 1}^n \Im \log (1-e^{i(\theta_k -\theta)}),
\end{equation} 
with the branches on the RHS being the principal branches, such that 
\begin{equation}
\Im \log (1-e^{i(\theta_k - \theta)}) \in \left( - \frac{\pi}{2}, \frac{\pi}{2} \right],
\end{equation}
where $\Im \log 0 := \pi/2$. We now let $U$ be Haar-distributed on $G(n)$ and define the random Radon measures $\mu_{G(n),\alpha,\beta}$ on $S^1 \sim [0,2\pi)$ by 
\begin{equation}
\mu_{G(n),\alpha,\beta}(\text{d}\theta) = \frac{f_{n,\alpha,\beta}(\theta) }{\mathbb{E}\left(f_{n,\alpha,\beta}(\theta)\right)}\text{d}\theta.
\end{equation}
\end{definition}

Let $\left(\mathcal{N}_j\right)_{j \in \mathbb{N}}$ be a sequence of independent standard (real) normal random variables and denote 
\begin{equation}
\eta_j := 1_{j \text{ is even}}.
\end{equation}
We recall the following result from \cite{Diaconis1994, Diaconis2001}:

\begin{theorem} \label{thm:traces} (Diaconis and Shahshahani, Diaconis and Evans)
If $U_n$ is Haar distributed  on $O(n)$ we have for any fixed $k$:
\begin{equation}
\left( \text{Tr}(U_n), \frac{1}{\sqrt{2}} \text{Tr}(U_n^2), ... , \frac{1}{\sqrt{k}} \text{Tr}(U_n^k) \right) \xrightarrow[n\rightarrow \infty]{d} \left( \mathcal{N}_1 + \eta_1, \mathcal{N}_2 + \frac{\eta_2}{\sqrt{2}}, ... , \mathcal{N}_k + \frac{\eta_k}{\sqrt{k}} \right).
\end{equation}
Similarly, if $U_n$ is Haar distributed on $Sp(2n)$, we have for any fixed $k \in \mathbb{N}$:
\begin{equation}
\left( \text{Tr}(U_n), \frac{1}{\sqrt{2}} \text{Tr}(U_n^2), ... , \frac{1}{\sqrt{k}} \text{Tr}(U_n^k) \right) \xrightarrow[n\rightarrow \infty]{d} \left( \mathcal{N}_1 - \eta_1, \mathcal{N}_2 - \frac{\eta_2}{\sqrt{2}}, ... , \mathcal{N}_k - \frac{\eta_k}{\sqrt{k}} \right).
\end{equation}
Finally we have the bound
\begin{equation} \label{eqn:trace bound}
\mathbb{E}_{U_n}\left( \left( \text{Tr}(U_n^k)\right)^2 \right) \leq \text{const} \min\{k,n\},
\end{equation}
where const is independent of $k$ and $n$. 
\end{theorem}

Using this result, Assiotis and Keating have proved the following theorem\footnote{This result was first published in \cite{Forkel2021} with the kind agreement of Dr. Assiotis.}, which identifies the limits of the real and imaginary parts of $\log p_{O(n)} (\cdot)$ and $\log p_{Sp(2n)} (\cdot)$ as log-correlated Gaussian fields on the unit circle. Their result was already stated as Theorem \ref{thm:Gaussian field1} in Chapter \ref{chapter:Intro} and is the analogue of the result by Hughes, Keating, and O'Connell \cite{Hughes2005} for the unitary group, stated as Theorem \ref{thm:HughesKeatingOConnell1} in Chapter \ref{chapter:Intro}:

\begin{theorem} \label{thm:Gaussian field} (Assiotis, Keating)
Let $p_{O(n)}(\cdot)$ be the characteristic polynomial of a random $U \in O(n)$, w.r.t. Haar measure. Then for any $\epsilon > 0$ the pair of fields $\left( \Re \log p_{O(n)}(\cdot), \Im \log p_{O(n)}(\cdot) \right)$ converges in distribution in $H^{-\epsilon}_0(S^1) \times H^{-\epsilon}_0(S^1)$ to the pair of Gaussian fields $\left( X - x, \hat{X} - \hat{x} \right)$, where
\begin{align} 
\begin{split}
X(\theta) &= \frac{1}{2} \sum_{j = 1}^\infty \frac{1}{\sqrt{j}} \mathcal{N}_{j} \left(e^{-ij\theta} + e^{ij\theta}\right) = \sum_{j=1}^\infty \frac{1}{\sqrt{j}} \mathcal{N}_{j} \cos(j\theta),\\
\hat{X}(\theta) &= \frac{1}{2i} \sum_{j = 1}^\infty \frac{1}{\sqrt{j}} \mathcal{N}_{j} \left(e^{-ij\theta} - e^{ij\theta}\right) = -\sum_{j=1}^\infty \frac{1}{\sqrt{j}} \mathcal{N}_{j} \sin(j\theta),\\
x(\theta) &= \frac{1}{2} \sum_{j = 1}^\infty \frac{\eta_j}{j} \left(e^{-ij\theta} + e^{ij\theta}\right) = \sum_{j=1}^\infty \frac{\eta_j}{j} \cos(j\theta),\\
\hat{x}(\theta) &= \frac{1}{2i} \sum_{j = 1}^\infty \frac{\eta_j}{j} \left(e^{-ij\theta} - e^{ij\theta}\right) = - \sum_{j=1}^\infty \frac{\eta_j}{j} \sin(j\theta).
\end{split}
\end{align}
Similarly, for $U \in Sp(2n)$, the pair of fields $\left( \Re \log p_{Sp(2n)}(\cdot), \Im \log p_{Sp(2n)}(\cdot) \right)$ converges in distribution in $H^{-\epsilon}_0(S^1) \times H^{-\epsilon}_0(S^1)$ to the pair of Gaussian fields $\left( X + x, \hat{X} + \hat{x} \right)$, for any $\epsilon > 0$.
\end{theorem}

\begin{remark}\label{remark:cov} Doing a formal calculation one see that $X$ and $\hat{X}$ are log-correlated fields, i.e. their covariance kernel blows up logarithmically as $\theta \rightarrow \theta'$ (and also as $\theta \rightarrow -\theta'$):
\begin{align} \label{eqn:cov 1}
\begin{split}
\mathbb{E}\left( X(\theta)X(\theta') \right) &= \sum_{j=1}^\infty \frac{\cos(j\theta)\cos(j\theta')}{j}\\ &= 
\frac{1}{2} \left( \sum_{j=1}^\infty \frac{\cos(j(\theta+\theta'))}{j} + \sum_{j=1}^\infty \frac{\cos(j(\theta-\theta'))}{j} \right) \\
&= \frac{1}{4} \sum_{j=1}^\infty \frac{1}{j} \left( e^{ij(\theta+\theta')} + e^{-ij(\theta+\theta')} + e^{ij(\theta-\theta')} + e^{-ij(\theta - \theta')} \right) \\
&= - \frac{1}{2} \left( \log |e^{i\theta} - e^{i\theta'}| + \log |e^{i\theta} - e^{-i\theta'}| \right).
\end{split}
\end{align}
Similarly one sees that 
\begin{align} \label{eqn:cov 2}
\begin{split}
\mathbb{E}\left( \hat{X}(\theta)\hat{X}(\theta') \right) =& - \frac{1}{2} \left( \log |e^{i\theta} - e^{i\theta')}| - \log |e^{i\theta} - e^{-i\theta'}| \right),\\
\mathbb{E}\left( X(\theta)\hat{X}(\theta') \right) =& \frac{1}{2} \left( \Im \log (1-e^{i(\theta+\theta')}) - \Im \log (1-e^{i(\theta-\theta')})\right).
\end{split}
\end{align}
\end{remark} 

For $\alpha \in \mathbb{R}$ and $\beta \in i\mathbb{R}$ we define the centered log-correlated Gaussian field
\begin{align} \label{eqn:Y}
\begin{split}
    Y_{\alpha,\beta}(\theta) =& 2\alpha X(\theta) + 2i\beta \hat{X}(\theta) \\
    =& \sum_{j=1}^\infty \frac{\mathcal{N}_j}{\sqrt{j}} \left( 2\alpha \cos(j\theta) - 2i\beta \sin(j\theta) \right).
\end{split}
\end{align}
By (\ref{eqn:cov 1}) and \ref{eqn:cov 2}) we see that its covariance kernel is given by 
\begin{align} \label{eqn:Cov Y}
&\text{Cov}(Y_{\alpha,\beta} (\theta),Y_{\alpha,\beta} (\theta')) \\
=& -2(\alpha^2-\beta^2) \log |e^{i\theta}-e^{i\theta'}| - 2(\alpha^2+\beta^2) \log |e^{i\theta}-e^{-i\theta'}| +4i\alpha \beta \Im \log (1-e^{i(\theta+\theta')}) \nonumber \\
=& \sum_{j=1}^\infty \frac{1}{j} \left( 2(\alpha^2 - \beta^2) \cos(j(\theta-\theta')) + 2(\alpha^2 + \beta^2) \cos(j(\theta+\theta')) - 4i \alpha \beta \sin(j(\theta+\theta')) \right). \nonumber
\end{align}
Motivated by Theorem \ref{thm:Gaussian field}, and analogous to the unitary case stated in Theorem \ref{thm:Webb}, one expects that for large $n$ the random measures $\mu_{G(n),\alpha,\beta}$ behave like the Gaussian multiplicative chaos measure 
\begin{align}
\mu_{\alpha,\beta}(\text{d}\theta) := \frac{ e^{Y_{\alpha,\beta}(\theta)} } {\mathbb{E}(e^{Y_{\alpha,\beta}(\theta)})} \text{d}\theta = e^{Y_{\alpha,\beta}(\theta) - \frac{1}{2}\mathbb{E}(Y_{\alpha,\beta}(\theta)^2)}\text{d}\theta.
\end{align} 
Since $Y_{\alpha,\beta}$ is not defined pointwise due to it being log-correlated one has to first define what is meant by $\mu_{\alpha,\beta}(\text{d}\theta)$. Even though the covariance kernel of $Y_{\alpha,\beta}$ does not just blow up on the diagonal, but also on the anti-diagonal, its covariance kernel can still be expressed as an infinite sum of continuous covariance kernels. Thus one can still use Theorem \ref{thm:Kahane2} to define $\mu_{\alpha,\beta}$ as the large $k$ limit of
\begin{align} \label{eqn:mu k}
    \mu_{\alpha,\beta}^{(k)}:= \frac{e^{Y_{\alpha,\beta}^{(k)}(\theta)}}{\mathbb{E}\left(e^{Y_{\alpha,\beta}^{(k)}(\theta)}\right)} \text{d}\theta = e^{Y_{\alpha,\beta}^{(k)}(\theta) - \frac{1}{2}\mathbb{E}(Y^{(k)}_{\alpha,\beta}(\theta)^2)} \text{d}\theta,
\end{align}
where the limit is almost sure in distribution, and where for $k \in \mathbb{N}$ 
\begin{align}
\begin{split}
Y_{\alpha,\beta}^{(k)} (\theta) &:= \sum_{j=1}^k \frac{\mathcal{N}_j}{\sqrt{j}} \left( 2\alpha \cos(j\theta) - 2i\beta \sin(j\theta) \right).
\end{split}
\end{align}

\begin{remark}
Since the covariance kernel of $Y_{\alpha,\beta}$ has logarithmic singularities not just on the diagonal $\theta = \theta'$ but also on the anti-diagonal $\theta = - \theta'$, one cannot directly use Theorem \ref{thm:Kahane2} to determine for which $\alpha,\beta$ the GMC measure $\mu_{\alpha,\beta}$ is non-trivial, i.e. not almost surely the zero measure. However, when restricting the measure to the arc $(\epsilon, \pi - \epsilon)$, for some $\epsilon > 0$, the covariance kernel can be written in the form required in Theorem \ref{thm:Kahane2}, which implies that when restricting $\mu_{\alpha,\beta}$ to $(\epsilon, \pi - \epsilon)$, it is non-trivial if and only if $\alpha^2 - \beta^2 < 1$. In particular the random variable $\mu_{\alpha,\beta}((\epsilon, \pi - \epsilon))$ is not almost surely zero if and only if $\alpha^2 - \beta^2 < 1$, and its construction does not depend on whether one restricts $\mu_{\alpha,\beta}$ to $(\epsilon, \pi - \epsilon)$ (see the proof sketch below Theorem \ref{thm:Kahane1}). Thus also when considered as a measure on the full circle, $\mu_{\alpha, \beta}$ is non-trivial if and only if $\alpha^2 - \beta^2 < 1$. 
\end{remark}

\begin{remark}
Since near $\pm 1$ one gets both the blowup from the diagonal as well as the blowup from the anti-diagonal, the set of parameters $\alpha$ and $\beta$ for which the moments of the total mass of $\mu_{\alpha,\beta}$ are finite, depend on whether one defines the measure on the full unit circle or whether one excludes small neighborhoods around $\pm 1$. This is part of the reason for the different sets of parameters $\alpha$ and $\beta$ in our Theorems \ref{thm:main} and \ref{thm:main2} below, as for our proofs to work we need the second moment of the total mass of $\mu_{\alpha,\beta}$ to be finite.
\end{remark}

For $\epsilon \in (0,\pi/2)$ define $I_{\epsilon}:= (\epsilon, \pi - \epsilon) \cup (\pi + \epsilon, 2\pi - \epsilon)$. Then our first result in this chapter is the following:
\begin{theorem} \label{thm:main}
Let $\alpha^2 - \beta^2 < 1/2$ and $\alpha > - 1/4$. Further let $G(n) \in \{O(n),Sp(2n)\}$. When restricting the random measures $\mu_{G(n),\alpha,\beta}$ and $\mu_{\alpha,\beta}$ to $I_\epsilon$, then as $n \rightarrow \infty$, for any fixed $\epsilon>0$, the sequence $\left(\mu_{G(n),\alpha,\beta}\right)_{n \in \mathbb{N}}$ converges weakly to $\mu_{\alpha,\beta}$ in the space of Radon measures on $I_\epsilon$ equipped with the topology of weak convergence, i.e. for any bounded $F:\left\{\text{Radon measures on } I_\epsilon \right\} \rightarrow \mathbb{R}$ for which $F(\mu_n) \rightarrow F(\mu)$ whenever $\mu_n \xrightarrow{d} \mu$, it holds that 
\begin{equation}
\mathbb{E}\left( F(\mu_{G(n),\alpha,\beta}) \right) \xrightarrow{n \rightarrow \infty} \mathbb{E} \left( F(\mu_{\alpha,\beta}) \right).
\end{equation}
\end{theorem}  

\begin{remark}
    Note that the limit of $\mu_{G(n),\alpha,\beta}$ is the same for both $G(n) = O(n)$ and $G(n) = Sp(2n)$ since by our normalisation procedure the deterministic means $\pm x$ and $\pm \hat{x}$ of the limiting fields get cancelled out.
\end{remark}

The specialisation of the formulas in \cite{Claeys2022} to our situation is stated in Theorem \ref{thm:T+H Claeys}. Using these, we can prove our second result in this chapter, which extends Theorem \ref{thm:main} to the full circle, but for a slightly smaller set of parameters $\alpha,\beta$. The reason for the different sets of parameters $\alpha, \beta$ is explained in Remark \ref{remark:parameters}. 

\begin{theorem} \label{thm:main2}
Let $\alpha^2 - \beta^2 < 1/2$ and $0 \leq \alpha < 1/2$. Further let $G(n) \in \{O(n),Sp(2n)\}$. Then the sequence of random measures $\left( \mu_{G(n),\alpha,\beta}\right)_{n \in \mathbb{N}}$ converges weakly to $\mu_{\alpha,\beta}$ in the space of Radon measures on $[0,2\pi)$ equipped with the topology of weak convergence. 
\end{theorem}

\noindent \textbf{Proof strategy:} Let $I$ denote either $I_\epsilon$ or $[0,2\pi)$. We first remark that by Theorem 4.2. in \cite{Kallenberg1976}, weak convergence of $\mu_{G(n),\alpha,\beta}$ to $\mu_{\alpha,\beta}$ in the space of Radon measures on $I$ equipped with the topology of weak convergence is equivalent to
\begin{equation}\label{eqn:conv equiv}
\int_{I} g(\theta) \mu_{G(n),\alpha,\beta}(\text{d}\theta) \xrightarrow{d} \int_{I} g(\theta) \mu_{\alpha,\beta}(\text{d}\theta),
\end{equation}
as $n \rightarrow \infty$, for any bounded continuous non-negative function $g$ on $I$. \\

\noindent To prove (\ref{eqn:conv equiv}) we use the following theorem: 
\begin{theorem} [Theorem 4.28 in \cite{Kallenberg2001}] \label{thm:approximation} For $k,n \in \mathbb{N}$ let $\xi$, $\xi_n$, $\eta^k$ and $\eta^k_n$ be random variables with values in a metric space $(S,\rho)$ such that $\eta_n^k \xrightarrow{d} \eta^k$ as $n\rightarrow \infty$ for any fixed $k$, and also $\eta^k \xrightarrow{d} \xi$ as $k\rightarrow \infty$. Then $\xi_n \xrightarrow{d} \xi$ holds under the further condition
\begin{equation}
\lim_{k\rightarrow \infty} \limsup_{n\rightarrow \infty} \mathbb{E}\left( \rho(\eta_n^k, \xi_n) \wedge 1 \right) = 0.
\end{equation}
\end{theorem}

Our setting corresponds to $S = \mathbb{R}$, $\rho = | \cdot |$, and 
\begin{align}
\begin{split}
\xi &= \int_{I} g(\theta) \mu_{\alpha,\beta}(\text{d}\theta), \quad \xi_n = \int_{I} g(\theta) \mu_{G(n),\alpha,\beta}(\text{d}\theta), \\
\eta^k &= \int_{I} g(\theta) \mu^{(k)}_{\alpha,\beta}(\text{d}\theta), \quad 
\eta^k_n = \int_{I} g(\theta) \mu^{(k)}_{n,\alpha,\beta}(\text{d}\theta),
\end{split}
\end{align}
where $\mu^{(k)}_{n,\alpha,\beta}$ will now be defined by truncating the Fourier series of $\log f_{n,\alpha,\beta}$. We see that
\begin{align}
\begin{split}
\log f_{n,\alpha,\beta}(\theta) &= -\sum_{j=1}^\infty \frac{1}{j} \left( (\alpha + \beta)\text{Tr}(U_n^j)e^{-ij\theta} + (\alpha - \beta)\text{Tr}(U_n^{-j})e^{ij\theta} \right) \\
&= - \sum_{j=1}^\infty \frac{\text{Tr}(U^j_n)}{j} \left( 2\alpha \cos(j\theta) - 2i\beta \sin (j\theta) \right),
\end{split}
\end{align}
where we used that for $U \in O(n)$ or $Sp(2n)$ we have $\text{Tr}(U^{-j}_n) = \text{Tr}(U^j_n)$.

\begin{definition}\label{def:mu_n^(k)} For $k,n \in \mathbb{N}$, $\alpha \in \mathbb{R}$, $\beta \in i\mathbb{R}$ and $\theta \in I$, let
\begin{equation}
f_{n,\alpha,\beta}^{(k)}(\theta) = e^{-\sum_{j=1}^k \frac{\text{Tr}(U^j_n)}{j} \left(2\alpha \cos(j\theta) - 2i\beta \sin(j\theta)\right)}, 
\end{equation}
and 
\begin{equation}
\mu_{G(n),\alpha,\beta}^{(k)}(\text{d}\theta) = \frac{f_{n,\alpha,\beta}^{(k)}(\theta)}{\mathbb{E}(f_{n,\alpha,\beta}^{(k)}(\theta))} \text{d}\theta.
\end{equation}
\end{definition}

In order to apply Theorem \ref{thm:approximation} to verify (\ref{eqn:conv equiv}), which then implies our main results, we thus need to examine the following three limits: for any bounded continuous non-negative function on $I$ 
\begin{align} \label{eqn:first limit} 
\lim_{k \rightarrow \infty} \int_{I} g(\theta) \mu^{(k)}_{\alpha,\beta} (\text{d}\theta) \overset{d}{=}& \int_{I} g(\theta) \mu_{\alpha,\beta} (\text{d}\theta), 
\end{align}
\begin{align} \label{eqn:second limit}
\lim_{n \rightarrow \infty} \int_{I} g(\theta) \mu_{G(n),\alpha,\beta}^{(k)} (\text{d}\theta) \overset{d}{=}& \int_{I} g(\theta) \mu^{(k)}_{\alpha,\beta} (\text{d}\theta),
\end{align}
and
\begin{align} \label{eqn:third limit}
\lim_{k\rightarrow \infty} \limsup_{n\rightarrow \infty} \mathbb{E}\left( \left| \int_{I} g(\theta) \mu_{G(n),\alpha,\beta}^{(k)} (\text{d}\theta) - \int_{I} g(\theta) \mu_{G(n),\alpha,\beta} (\text{d}\theta) \right| \wedge 1 \right) = 0.
\end{align}
The first limit (\ref{eqn:first limit}) follows immediately from the definition of $\mu_{\alpha,\beta}$ as the large $k$ limit of $\mu_{\alpha,\beta}^{(k)}$, defined in (\ref{eqn:mu k}): almost surely 
\begin{equation}
\lim_{k\rightarrow \infty} \mu^{(k)}_{\alpha,\beta} \overset{d}{=} \mu_{\alpha,\beta},
\end{equation}
so in particular almost surely (and thus also in distribution)
\begin{equation}
\lim_{k \rightarrow \infty} \int_{I} g(\theta) \mu^{(k)}_{\alpha,\beta} (\text{d}\theta) = \int_{I} g(\theta) \mu_{\alpha,\beta} (\text{d}\theta).
\end{equation}
The second limit (\ref{eqn:second limit}) will be proved in Section \ref{section:second limit}, using previously established results on the asymptotics of Toeplitz+Hankel determinants.\\

\noindent To show that the third limit (\ref{eqn:third limit}) holds, we will prove the following lemma in Section \ref{section:L2}:
\begin{lemma}[The $L^2$-limit]\label{lemma:L2 limit} 
Let $\alpha^2 - \beta^2 <1/2$. Further let $0 \leq \alpha < 1/2$ in the case $I = [0,2\pi)$, and $\alpha > -1/4$ in the case $I = I_\epsilon$. Then for any bounded continuous non-negative function $g$ on $I$ the following expectation goes to zero, as first $n \rightarrow \infty$ and then $k \rightarrow \infty$:
\begin{align} \label{eqn:L2 limit}
\begin{split}
&\mathbb{E} \left( \left( \int_{I} g(\theta) \mu_{G(n),\alpha,\beta}(\text{d}\theta) - \int_{I} g(\theta) \mu_{G(n),\alpha,\beta}^{(k)}(\text{d}\theta) \right)^2 \right) \\
=&\int_{I} \int_{I} g(\theta)g(\theta') \frac{\mathbb{E}\left(f_{n,\alpha,\beta}^{(k)}(\theta) f_{n,\alpha,\beta}^{(k)}(\theta') \right)}{\mathbb{E}\left(f_{n,\alpha,\beta}^{(k)}(\theta)\right)\mathbb{E}\left(f_{n,\alpha,\beta}^{(k)}(\theta')\right)} \text{d}\theta \text{d}\theta'\\
&-2 \int_{I} \int_{I} g(\theta)g(\theta') \frac{\mathbb{E}\left(f_{n,\alpha,\beta}^{(k)}(\theta) f_{n,\alpha,\beta}(\theta') \right)}{\mathbb{E}\left(f_{n,\alpha,\beta}^{(k)}(\theta)\right)\mathbb{E}\left(f_{n,\alpha,\beta}(\theta')\right)} \text{d}\theta \text{d}\theta' \\
&+ \int_{I} \int_{I} g(\theta)g(\theta') \frac{\mathbb{E}\left(f_{n,\alpha,\beta}(\theta) f_{n,\alpha,\beta}(\theta') \right)}{\mathbb{E}\left(f_{n,\alpha,\beta}(\theta)\right)\mathbb{E}\left(f_{n,\alpha,\beta}(\theta')\right)} \text{d}\theta \text{d}\theta'.
\end{split}
\end{align}
\end{lemma}
All the expectations inside the integrals can be expressed as (sums of) Toeplitz+Hankel determinants (see (\ref{eqn:T+H def}) and Theorem \ref{thm:average}). We will prove Lemma \ref{lemma:L2 limit} explicitly only in the case $I = [0,2\pi)$, $\alpha^2 - \beta^2 <1/2$, $0 \leq \alpha < 1/2$, using Theorem \ref{thm:T+H Claeys} below. \\

The proof in the case $I = I_\epsilon$, $\alpha^2 - \beta^2 < 1/2$, $\alpha > -1/4$, works in almost exactly the same way. Instead of Theorem \ref{thm:T+H Claeys}, which only holds for $\alpha \geq 0$, we use new results on the uniform asymptotics of Toeplitz+Hankel determinants of symbols with two pairs of merging singularities bounded away from $\pm 1$, which are stated in Theorems \ref{thm:T+H uniform} and \ref{thm:T, T+H extended}, in the next section. To the best of our knowledge these results have not previously been set out and we believe them to be of independent interest.

\section{The $L^2$-Limit} \label{section:L2}
In this section we prove Lemma \ref{lemma:L2 limit}, for which we need to compute asymptotics of all the expectation terms in the integrals, which hold uniformly in $\theta$ and $\theta'$ even as $\theta \rightarrow \theta'$. We recall the Baik-Rains identity (already stated as Theorem \ref{thm:Baik2001} in Chapter \ref{chapter:Intro}), which allows us to express all those expectations as Toeplitz+Hankel determinants:

\begin{theorem}\label{thm:average}(Theorem 2.2 in \cite{Baik2001})
Let $h(z)$ be any function on the unit circle such that for $\iota(e^{i\theta}) := h(e^{i\theta})h(e^{-i\theta})$ the integrals 
\begin{equation}
\iota_j = \frac{1}{2\pi} \int_0^{2\pi} \iota(e^{i\theta}) e^{-ij\theta} \, \text{d}\theta
\end{equation} 
are well-defined. Then with $D_n^{T+H,\kappa}$ defined as in (\ref{eqn:T+H def}) we have
\begin{align}
\begin{split}
\mathbb{E}_{SO(2n)}\left( \text{det}(h(U)) \right) &= \frac{1}{2} D_n^{T+H,1}(\iota), \\
\mathbb{E}_{SO^-(2n)}\left( \text{det}(h(U)) \right) &=  h(1)h(-1) D_{n-1}^{T+H,2}(\iota), \\
\mathbb{E}_{SO(2n+1)}\left( \text{det}(h(U)) \right) &=  h(1) D_n^{T+H,3}(\iota), \\
\mathbb{E}_{SO^-(2n+1)}\left( \text{det}(h(U)) \right) &=  h(-1) D_n^{T+H,4}(\iota), \\
\mathbb{E}_{Sp(2n)}\left( \text{det}(h(U)) \right) &= D_n^{T+H,2}(\iota),
\end{split}
\end{align}
except that $\mathbb{E}_{SO(0)} \left( \det(h(U)) \right) = h(1)$.
\end{theorem} 

We define, for $\phi \in [0,2\pi)$ and $\theta, \theta' \in [0,2\pi)$:
\begin{align} \label{eqn:sigma hat}
\begin{split}
\hat{\sigma}_{1,\theta,\theta'}(e^{i\phi}) &= e^{-\sum_{j=1}^k \frac{2}{j} \Re \left( (\alpha - \beta)(e^{ij\theta} + e^{ij\theta'}) \right) e^{ij\phi} },\\
\hat{\sigma}_{2,\theta,\theta'}(e^{i\phi}) &= e^{-\sum_{j=1}^k \frac{2}{j} \Re \left( (\alpha - \beta)e^{ij\theta} \right) e^{ij\phi}} |e^{i\phi}-e^{i\theta'}|^{2\alpha} e^{2i\beta \Im \log (1-e^{i(\phi-\theta')})}, \\
\hat{\sigma}_{3,\theta,\theta'}(e^{i\phi}) &= |e^{i\phi}-e^{i\theta}|^{2\alpha} e^{2i\beta \Im \log (1-e^{i(\phi-\theta)})} |e^{i\phi}-e^{i\theta'}|^{2\alpha} e^{2i\beta \Im \log (1-e^{i(\phi-\theta')})}, \\
\hat{\sigma}_{4,\theta}(e^{i\phi}) &= e^{-\sum_{j=1}^k \frac{2}{j} \Re \left( (\alpha - \beta)e^{ij\theta} \right) e^{ij\phi}}, \\
\hat{\sigma}_{5,\theta}(e^{i\phi}) &= |e^{i\phi}-e^{i\theta}|^{2\alpha} e^{2i\beta \Im \log (1-e^{i(\phi-\theta)})}, 
\end{split}
\end{align}
where the branch of the logarithm is the principal one (so in particular $\Im \log (1-e^{i(\phi-\theta)}) \in (-\pi/2,\pi/2]$). Then we have
\begin{align}
\begin{split}
f_{n,\alpha,\beta}^{(k)}(\theta)f_{n,\alpha,\beta}^{(k)}(\theta') &= \prod_{\ell = 1}^n \hat{\sigma}_{1,\theta,\theta'}(e^{i\theta_\ell}), \\
f_{n,\alpha,\beta}^{(k)}(\theta)f_{n,\alpha,\beta}(\theta') &= \prod_{\ell = 1}^n \hat{\sigma}_{2,\theta,\theta'}(e^{i\theta_\ell}), \\
f_{n,\alpha,\beta}(\theta)f_{n,\alpha,\beta}(\theta') &= \prod_{\ell = 1}^n \hat{\sigma}_{3,\theta,\theta'}(e^{i\theta_\ell}), \\
f_{n,\alpha,\beta}^{(k)}(\theta) &= \prod_{\ell = 1}^n \hat{\sigma}_{4,\theta}(e^{i\theta_\ell}), \\
f_{n,\alpha,\beta}(\theta) &= \prod_{\ell = 1}^n \hat{\sigma}_{5,\theta}(e^{i\theta_\ell}), 
\end{split}
\end{align}
where the product is up to $2n$ for $U \in Sp(2n)$. Further we define 
\begin{align}
\sigma_{1,\theta,\theta'}(e^{i\phi}) =& \hat{\sigma}_{1,\theta,\theta'}(e^{i\phi})\hat{\sigma}_{1,\theta,\theta'}(e^{-i\phi}) \nonumber \\ 
=& e^{-\sum_{j=1}^k \frac{2}{j} \Re \left((\alpha - \beta)(e^{ij\theta}+e^{ij\theta'}) \right) (e^{ij\phi} + e^{-ij\phi})}, \nonumber \\
\sigma_{2,\theta,\theta'}(e^{i\phi}) =& \hat{\sigma}_{2,\theta,\theta'}(e^{i\phi})\hat{\sigma}_{2,\theta,\theta'}(e^{-i\phi}) \nonumber \\
=& e^{-\sum_{j=1}^k \frac{2}{j} \Re \left((\alpha - \beta)e^{ij\theta} \right) (e^{ij\phi} + e^{-ij\phi})} \nonumber \\
&\times |e^{i\theta'}-e^{i\phi}|^{2\alpha} e^{2i\beta \Im \log (1-e^{i(\phi-\theta')})} |e^{i\theta'}-e^{-i\phi}|^{2\alpha} e^{2i\beta \Im \log (1-e^{i(-\phi-\theta')})}, \nonumber \\
\sigma_{3,\theta,\theta'}(e^{i\phi}) =& \hat{\sigma}_{3,\theta,\theta'}(e^{i\phi})\hat{\sigma}_{3,\theta,\theta'}(e^{-i\phi}) \\
=& |e^{i\theta}-e^{i\phi}|^{2\alpha} e^{2i\beta \Im \log (1-e^{i(\phi-\theta)})} |e^{i\theta'}-e^{i\phi}|^{2\alpha} e^{2i\beta \Im \log (1-e^{i(\phi-\theta')})} \nonumber \\
&\times |e^{i\theta}-e^{-i\phi}|^{2\alpha} e^{2i\beta \Im \log (1-e^{i(-\phi-\theta)})} |e^{i\theta'}-e^{-i\phi}|^{2\alpha} e^{2i\beta \Im \log (1-e^{i(-\phi-\theta')})}, \nonumber \\
\sigma_{4,\theta}(e^{i\phi}) =& \hat{\sigma}_{4,\theta}(e^{i\phi}) \hat{\sigma}_{4,\theta}(e^{-i\phi}) \nonumber \\ 
=&  e^{-\sum_{j=1}^k \frac{2}{j} \Re \left((\alpha - \beta)e^{ij\theta} \right) (e^{ij\phi} + e^{-ij\phi})}, \nonumber \\
\sigma_{5,\theta}(e^{i\phi}) =& \hat{\sigma}_{5,\theta}(e^{i\phi})\hat{\sigma}_{5,\theta}(e^{-i\phi}) \nonumber \\ 
=& |e^{i\theta}-e^{i\phi}|^{2\alpha} e^{2i\beta \Im \log (1-e^{i(\phi-\theta)})} |e^{i\theta}-e^{-i\phi}|^{2\alpha} e^{2i\beta \Im \log (1-e^{i(-\phi-\theta)})}. \nonumber  
\end{align}
Applying Theorem \ref{thm:average}, we obtain 
\begin{align} \label{eqn:averages sigma}
&\mathbb{E}_{O(2n)} \left( f_{2n,\alpha, \beta}^{(k)}(\theta) f_{2n,\alpha,\beta}^{(k)}(\theta') \right) \nonumber \\
=& \frac{1}{2} \mathbb{E}_{SO(2n)} \left( f_{2n,\alpha, \beta}^{(k)}(\theta) f_{2n,\alpha,\beta}^{(k)}(\theta') \right) + \frac{1}{2} \mathbb{E}_{SO^-(2n)} \left( f_{2n,\alpha, \beta}^{(k)}(\theta) f_{2n,\alpha,\beta}^{(k)}(\theta') \right) \nonumber \\
=& \frac{1}{4} D_n^{T+H,1}(\sigma_{1,\theta,\theta'}) + \frac{1}{2} \hat{\sigma}_{1,\theta,\theta'}(1)\hat{\sigma}_{1,\theta,\theta'}(-1) D_{n-1}^{T+H,2}(\sigma_{1,\theta,\theta'}), \nonumber \\
&\mathbb{E}_{O(2n+1)} \left( f_{2n+1,\alpha, \beta}^{(k)}(\theta) f_{2n+1,\alpha,\beta}^{(k)}(\theta') \right) \\
=& \frac{1}{2} \mathbb{E}_{SO(2n+1)} \left( f_{2n+1,\alpha, \beta}^{(k)}(\theta) f_{2n+1,\alpha,\beta}^{(k)}(\theta') \right) + \frac{1}{2} \mathbb{E}_{SO^-(2n+1)} \left( f_{2n+1,\alpha, \beta}^{(k)}(\theta) f_{2n+1,\alpha,\beta}^{(k)}(\theta') \right) \nonumber \\
=& \frac{1}{2} \hat{\sigma}_{1,\theta,\theta'}(1) D_n^{T+H,3}(\sigma_{1,\theta,\theta'}) + \frac{1}{2} \hat{\sigma}_{1,\theta,\theta'}(-1) D_n^{T+H,4}(\sigma_{1,\theta,\theta'}), \nonumber \\
&\mathbb{E}_{Sp(2n)} \left(  f_{n,\alpha, \beta}^{(k)}(\theta) f_{n,\alpha,\beta}^{(k)}(\theta') \right) = D_n^{T+H,2}(\sigma_{1,\theta,\theta'}), \nonumber
\end{align}
and similarly for $\sigma_{2,\theta,\theta'}$, ..., $\sigma_{5,\theta}$.\\

The symbols $\sigma_{1,\theta,\theta'}$, ..., $\sigma_{5,\theta}$ can be written as symbols with Fisher-Hartwig singularities, i.e. as in Definition \ref{def:FH}. Due to our choice of logarithm, we have for $\phi, \theta \in [0,2\pi)$:
\begin{align}
\begin{split}
\Im \log (1-e^{i(\phi-\theta)}) =& \begin{cases} -\frac{\pi}{2} + \frac{\phi-\theta}{2} & 0 \leq \theta < \phi < 2\pi \\ \frac{\pi}{2} + \frac{\phi-\theta}{2} & 0 \leq \phi \leq \theta < 2\pi \end{cases}, \\
\Im \log (1-e^{i(-\phi-\theta)}) =& \Im \log (1-e^{i(2\pi-\theta - \phi)})\\
=&\begin{cases} -\frac{\pi}{2} + \frac{2\pi - \theta - \phi}{2} & 0 \leq \phi < 2\pi - \theta \leq 2\pi \\ \frac{\pi}{2} + \frac{2\pi - \theta - \phi}{2} & 0 < 2\pi - \theta \leq \phi < 2\pi \end{cases},
\end{split}
\end{align}
which implies that we can write
\begin{align}
&\sigma_{2,\theta,\theta'}(e^{i\phi}) \nonumber \\
=&  e^{-\sum_{j=1}^k \frac{2}{j} \Re \left((\alpha-\beta)e^{ij\theta} \right) (e^{ij\phi} + e^{-ij\phi})} \nonumber \\
&\times |e^{i\phi}-e^{i\theta'}|^{2\alpha} e^{i\beta (\phi - \theta')} g_{e^{i\theta'},\beta}(e^{i\phi}) |e^{i\phi}-e^{-i\theta'}|^{2\alpha} e^{-i\beta (\phi - (2\pi - \theta'))} g_{e^{i(2\pi - \theta')},-\beta}(e^{i\phi}), \nonumber \\
&\sigma_{3,\theta,\theta'}(e^{i\phi}) \\
=& |e^{i\phi}-e^{i\theta}|^{2\alpha} e^{i\beta (\phi - \theta)} g_{e^{i\theta},\beta}(e^{i\phi}) |e^{i\phi}-e^{-i\theta}|^{2\alpha} e^{-i\beta (\phi - (2\pi - \theta))} g_{e^{i(2\pi - \theta)},-\beta}(e^{i\phi}) \nonumber \\
& \times |e^{i\phi}-e^{i\theta'}|^{2\alpha} e^{i\beta (\phi - \theta')} g_{e^{i\theta'},\beta}(e^{i\phi}) |e^{i\phi}-e^{-i\theta'}|^{2\alpha} e^{-i\beta (\phi - (2\pi - \theta'))}, g_{e^{i(2\pi - \theta')},-\beta}(e^{i\phi}), \nonumber \\
&\sigma_{5,\theta}(e^{i\phi}) \nonumber \\
=& |e^{i\phi}-e^{i\theta}|^{2\alpha} e^{i\beta (\phi - \theta)} g_{e^{i\theta},\beta}(e^{i\phi}) |e^{i\phi}-e^{-i\theta}|^{2\alpha} e^{-i\beta (\phi - (2\pi - \theta'))} g_{e^{i(2\pi - \theta')},-\beta}(e^{i\phi}). \nonumber 
\end{align}

Theorem \ref{thm:T+H non-uniform} gives asymptotic formulae for the Toeplitz+Hankel determinants of $\sigma_{1,\theta,\theta'}$, $\sigma_{2,\theta,\theta'}, \sigma_{3,\theta,\theta'}, \sigma_{4,\theta}, \sigma_{5,\theta}$, which are uniform when all of the singularities that appear are bounded away from each other. Since $\sigma_{1,\theta,\theta'}$ and $\sigma_{4,\theta}$ do not have any singularities the asymptotics of their Toeplitz+Hankel determinants are uniform in $\theta, \theta' \in [0,2\pi)$. To obtain asymptotics for the Toeplitz+Hankel determinants of $\sigma_{2,\theta,\theta'}$, $\sigma_{3,\theta,\theta'}$ and $\sigma_{5,\theta}$ that are also uniform when singularities merge we use Theorem \ref{thm:T+H Claeys}. \\

\noindent When applying Theorem \ref{thm:T+H uniform} we always have $\alpha_0 = \alpha_{r+1} = 0$, thus we get
\begin{equation}
\frac{\pi^{\frac{1}{2}(\alpha_0+\alpha_{r+1}+s'+t'+1)} G(1/2)^2}{G(1+\alpha_0+s') G(1 + \alpha_{r+1} + t')} = 1,
\end{equation}  
for any choices of $s',t' \in \{+ 1/2, - 1/2\}$. Further one has to be careful that $z_1$, $z_2$ always correspond to the singularities in the upper half circle, i.e. $\arg z_1, \arg z_2 \in (0, \pi)$. The asymptotics of the Toeplitz+Hankel determinants of $\sigma_{1,\theta,\theta'},\sigma_{2,\theta,\theta'},\sigma_{4,\theta},\sigma_{5,\theta'}$ obtained with Theorem \ref{thm:T+H uniform} are then as follows:

\begin{itemize}
\item For $f(e^{i\phi}) = \sigma_{1,\theta,\theta'}(e^{i\phi})$ we have $r = 0$, $\alpha_0 = \alpha_1 =0$ and
\begin{equation}
V(z) = -\sum_{j=1}^k \frac{2}{j} \Re \left((\alpha-\beta)(e^{ij\theta}+e^{ij\theta'}) \right) (e^{ij\phi} + e^{-ij\phi}).
\end{equation} 
Thus we obtain
\begin{align} \label{eqn:sigma1}
\begin{split}
D_n^{T+H,\kappa}(\sigma_{1,\theta,\theta'}) =& e^{\sum_{j=1}^k \frac{2}{j}   \Re \left((\alpha-\beta)(e^{ij\theta}+e^{ij\theta'}) \right)^2} \\
&\times e^{2\sum_{j=1}^k \frac{s'+(-1)^j t'}{j} \Re \left( (\alpha-\beta)(e^{ij\theta}+e^{ij\theta'}) \right)}\\
&\times 2^{(1-s'-t')n+q - \frac{1}{2}(s'+t')^2 + \frac{1}{2}(s'+t')} (1+o(1)),
\end{split}
\end{align}
uniformly for $\theta,\theta' \in [0,2\pi)$.

\item For $f(e^{i\phi}) = \sigma_{2,\theta,\theta'}(e^{i\phi})$ we have $r = 1$, $\alpha_0 = \alpha_2 = 0$, 
\begin{equation}
z_1 = \begin{cases} e^{i\theta'} & 0 < \theta' < \pi \\ e^{i(2\pi - \theta')} & \pi < \theta' < 2\pi \end{cases}, \quad \beta_1 = \begin{cases} \beta & 0 < \theta' < \pi \\ -\beta & \pi < \theta' < 2\pi \end{cases},
\end{equation}
$\alpha_1 = \alpha$, and 
\begin{equation}
V(z) = -\sum_{j=1}^k \frac{2}{j} \Re \left((\alpha-\beta)e^{ij\theta} \right) (e^{ij\phi} + e^{-ij\phi}).
\end{equation} 
Thus we obtain:
\begin{align} \label{eqn:sigma2}
D_n^{T+H,\kappa}(\sigma_{2,\theta,\theta'}) =& (2n)^{(\alpha^2-\beta^2)} e^{\sum_{j=1}^k \frac{2}{j}   \Re \left((\alpha-\beta) e^{ij\theta} \right)^2} e^{\sum_{j=1}^k \frac{4}{j} \Re \left( (\alpha-\beta)e^{ij\theta} \right)\Re \left( (\alpha-\beta)e^{ij\theta'} \right)} \nonumber \\
&\times e^{-i\pi \alpha \beta_1} z_1^{2\alpha \beta_1} |1- e^{2i\theta'}|^{-(\alpha^2+\beta^2)} \frac{G(1+\alpha+\beta) G(1+\alpha-\beta)}{G(1+2\alpha)} \nonumber \\
&\times e^{2\sum_{j=1}^k \frac{s'+(-1)^j t'}{j} \Re \left( (\alpha-\beta)e^{ij\theta} \right)} 2^{(1-s'-t')n+q - \frac{1}{2}(s'+t')^2 + \frac{1}{2}(s'+t')} \nonumber \\
&\times e^{-i\pi s' \beta_1} z_1^{\beta_1(s'+t')}  |1-e^{i\theta'}|^{-2\alpha s'} |1+e^{i\theta'}|^{-2\alpha t'} \\
&\times (1+o(1)), \nonumber 
\end{align}
uniformly for $\theta,\theta' \in [0,2\pi)$ s.t. $e^{i\theta'}$ stays bounded away from $\pm 1$.  

\item For $f(e^{i\phi}) = \sigma_{4,\theta}(e^{i\phi})$ we have $r = 0$, $\alpha_0 = \alpha_1 = 0$ and 
\begin{equation}
V(z) = -\sum_{j=1}^k \frac{2}{j} \Re \left((\alpha-\beta)e^{ij\theta} \right) (e^{ij\phi} + e^{-ij\phi}).
\end{equation} 
Thus we obtain:
\begin{align} \label{eqn:sigma4}
\begin{split}
D_n^{T+H,\kappa}(\sigma_{4,\theta}) =& e^{\sum_{j=1}^k \frac{2}{j} \Re \left((\alpha-\beta) e^{ij\theta} \right)^2} e^{2\sum_{j=1}^k \frac{s'+(-1)^j t'}{j} \Re \left( (\alpha-\beta)e^{ij\theta} \right)} \\
&\times 2^{(1-s'-t')n+q - \frac{1}{2}(s'+t')^2 + \frac{1}{2}(s'+t')} (1+o(1)),	
\end{split}
\end{align}
uniformly for $\theta \in [0,2\pi)$.  

\item For $f(e^{i\phi}) = \sigma_{5,\theta'}(e^{i\phi})$ with $e^{i\theta'} \neq \pm 1$, we have $r = 1$, $\alpha_0 = \alpha_2 = 0$, 
\begin{equation} \label{eqn:sigma5 decomposition}
z_1 = \begin{cases} e^{i\theta'} & 0 < \theta' < \pi \\ e^{i(2\pi - \theta')} & \pi < \theta' < 2\pi \end{cases}, \quad \beta_1 = \begin{cases} \beta & 0 < \theta' < \pi \\ -\beta & \pi < \theta' < 2\pi \end{cases}
\end{equation}
$\alpha_1 = \alpha$, and $V = 0$. Thus we obtain:
\begin{align} \label{eqn:sigma5}
&D_n^{T+H,\kappa}(\sigma_{5,\theta'}) \nonumber \\
=& (2n)^{(\alpha^2-\beta^2)} e^{-i\pi \alpha \beta_1} z_1^{2\alpha \beta_1}  |1- e^{2i\theta'}|^{-(\alpha^2+\beta^2)} \frac{G(1+\alpha+\beta) G(1+\alpha-\beta)}{G(1+2\alpha)} \nonumber \\
&\times 2^{(1-s'-t')n+q - \frac{1}{2}(s'+t')^2 + \frac{1}{2}(s'+t')} e^{-i\pi s' \beta_1} z_1^{\beta_1(s'+t')}  |1-e^{i\theta'}|^{-2\alpha s'} |1+e^{i\theta'}|^{-2\alpha t'} \nonumber \\
&\times (1+o(1)),
\end{align}
uniformly for $\theta' \in [0,2\pi)$ s.t. $e^{i\theta'}$ stays bounded away from $\pm 1$.    
\end{itemize}

For $f(e^{i\phi}) = \sigma_{3,\theta,\theta'}(e^{i\phi})$ we have $r = 2$, $\alpha_0 = \alpha_3 = V(z) = 0$ and $z_1,z_2,\alpha_1,\alpha_2,\beta_1,\beta_2$ chosen according to the following decomposition (always $\alpha_1 = \alpha_2 = \alpha$):
\begin{align} \label{eqn:decomposition}
[0,2\pi)^2 =& \cup_{j=1}^8 J_j \cup \left\{ (\theta, \theta') \in [0,2\pi)^2:\theta = \theta', \text{ or } \theta = 2\pi - \theta', \text{ or } \theta,\theta' \in \{0,\pi\} \right\}, \nonumber \\
J_1 =& \left\{ (\theta, \theta') \in (0, \pi) \times (0, \pi): \theta < \theta' \right\}, \nonumber \\
&\text{for } (\theta, \theta') \in J_1: \quad z_1 = e^{i\theta}, z_2 = e^{i\theta'}, \beta_1 = \beta_2 = \beta. \nonumber \\ 
J_2 =& \left\{ (\theta, \theta') \in (0, \pi) \times (0, \pi): \theta' < \theta \right\}, \nonumber \\
&\text{for } (\theta, \theta') \in J_2: \quad z_1 = e^{i\theta'}, z_2 = e^{i\theta}, \beta_1 = \beta_2 = \beta. \nonumber \\ 
J_3 =& \left\{ (\theta, \theta') \in (0, \pi) \times (\pi, 2\pi): \theta < 2\pi - \theta' \right\}, \nonumber  \\
&\text{for } (\theta, \theta') \in J_3: \quad z_1 = e^{i\theta}, z_2 = e^{i(2\pi - \theta')}, \beta_1 = -\beta_2 = \beta. \nonumber \\ 
J_4 =& \left\{ (\theta, \theta') \in (\pi, 2\pi) \times (0, \pi): \theta' < 2\pi - \theta \right\},  \nonumber \\
&\text{for } (\theta, \theta') \in J_4: \quad z_1 = e^{i\theta'}, z_2 = e^{i(2\pi - \theta)}, \beta_1 = - \beta_2 = \beta. \\ 
J_5 =& \left\{ (\theta, \theta') \in (\pi, 2\pi) \times (\pi, 2\pi): 2\pi - \theta < 2\pi - \theta' \right\},  \nonumber \\
&\text{for } (\theta, \theta') \in J_5: \quad z_1 = e^{i(2\pi - \theta)}, z_2 = e^{i(2\pi - \theta')}, \beta_1 = \beta_2 = -\beta. \nonumber \\ 
J_6 =& \left\{ (\theta, \theta') \in (\pi, 2\pi) \times (\pi, 2\pi): 2\pi - \theta' < 2\pi - \theta \right\},  \nonumber \\
&\text{for } (\theta, \theta') \in J_6: \quad z_1 = e^{i(2\pi - \theta')}, z_2 = e^{i(2\pi - \theta)}, \beta_1 = \beta_2 = -\beta. \nonumber \\ 
J_7 =& \left\{ (\theta, \theta') \in (\pi, 2\pi) \times (0, \pi): 2\pi - \theta < \theta' \right\},  \nonumber \\
&\text{for } (\theta, \theta') \in J_7: \quad z_1 = e^{i(2\pi - \theta)}, z_2 = e^{i\theta'}, \beta_1 = -\beta_2 = -\beta. \nonumber \\ 
J_8 =& \left\{ (\theta, \theta') \in (0, \pi) \times (\pi, 2\pi): 2\pi - \theta' < \theta \right\}, \nonumber \\
&\text{for } (\theta, \theta') \in J_8: \quad z_1 = e^{i(2\pi - \theta')}, z_2 = e^{i\theta}, \beta_1 = -\beta_2 = -\beta. \nonumber 
\end{align}

\noindent In this notation we obtain by Theorem \ref{thm:T+H uniform} 
\begin{align} \label{eqn:sigma3 extended} 
&D_n^{T+H,\kappa}(\sigma_{3,\theta,\theta'}) \nonumber \\
=& (2n)^{2(\alpha^2-\beta^2)} e^{-i\pi \alpha \left( \beta_1 + \beta_2 + 2 \beta_2 \right)} |e^{i\theta} - e^{i\theta'}|^{-2(\alpha^2-\beta^2)} |e^{i\theta} - e^{-i\theta'}|^{-2(\alpha^2+\beta^2)} \nonumber \\
&\times z_1^{4\beta_1 \alpha} z_2^{4\beta_2\alpha} |1- e^{2i\theta}|^{-(\alpha^2+\beta^2)} |1- e^{2i\theta'}|^{-(\alpha^2+\beta^2)} \frac{G(1+\alpha+\beta)^2 G(1+\alpha-\beta)^2}{G(1+2\alpha)^2} \nonumber \\
&\times 2^{(1-s'-t')n+q - \frac{1}{2}(s'+t')^2 + \frac{1}{2}(s'+t')} e^{-i\pi s'(\beta_1 + \beta_2)} z_1^{\beta_1(s'+t')} z_2^{\beta_2(s'+t')} \\
&\times |1-e^{i\theta}|^{-2\alpha s'} |1+e^{i\theta}|^{-2\alpha t'} |1-e^{i\theta'}|^{-2\alpha s'} |1+e^{i\theta'}|^{-2\alpha t'} (1+o(1)), \nonumber 
\end{align}
uniformly for $\theta,\theta' \in [0,2\pi)$, s.t. $e^{i\theta}, e^{i\theta'}, e^{-i\theta}, e^{-i\theta'}$
stay bounded away from each other and $\pm 1$.\\

In the following sections we use the asymptotics obtained in this section to compute the asymptotics of the quotients of expectations that appear in Lemma \ref{lemma:L2 limit}. 

\subsection{The Symplectic Case} 

By (\ref{eqn:averages sigma}), (\ref{eqn:sigma1}), (\ref{eqn:sigma4}) with $\kappa = 2$, $q=0$, $s' = t' = \frac{1}{2}$, we get 
\begin{align} \label{eqn:symplectic1}
\begin{split}
&\frac{\mathbb{E}_{Sp(2n)} \left( f_{n,\alpha,\beta}^{(k)}(\theta)f_{n,\alpha,\beta}^{(k)}(\theta') \right)}{\mathbb{E}_{Sp(2n)} \left( f_{n,\alpha,\beta}^{(k)}(\theta) \right)\mathbb{E}_{Sp(2n)} \left(f_{n,\alpha,\beta}^{(k)}(\theta') \right)} = \frac{D_n^{T+H,2}(\sigma_{1,\theta,\theta'})}{D_n^{T+H,2}(\sigma_{4,\theta})D_n^{T+H,2}(\sigma_{4,\theta'})} \\
=& e^{\sum_{j=1}^k \frac{4}{j} \Re \left((\alpha - \beta)e^{ij\theta}\right) \Re \left((\alpha - \beta)e^{ij\theta'}\right)} (1+o(1)),
\end{split}
\end{align}
uniformly for $\theta,\theta' \in [0,2\pi)$.\\

By (\ref{eqn:averages sigma}), (\ref{eqn:sigma2}), (\ref{eqn:sigma4}) and  (\ref{eqn:sigma5}) with $\kappa = 2$, $q = 0$, $s' = t' = \frac{1}{2}$, we get
\begin{align} \label{eqn:symplectic2}
\begin{split}
&\frac{\mathbb{E}_{Sp(2n)} \left( f_{n,\alpha,\beta}^{(k)}(\theta)f_{n,\alpha,\beta}(\theta') \right)}{\mathbb{E}_{Sp(2n)} \left( f_{n,\alpha,\beta}^{(k)}(\theta) \right)\mathbb{E}_{Sp(2n)} \left(f_{n,\alpha,\beta}(\theta') \right)} = \frac{D_n^{T+H,2}(\sigma_{2,\theta,\theta'})}{D_n^{T+H,2}(\sigma_{4,\theta})D_n^{T+H,2}(\sigma_{5,\theta'})} \\
=& e^{\sum_{j=1}^k \frac{4}{j} \Re \left((\alpha - \beta)e^{ij\theta}\right) \Re \left((\alpha -\beta)e^{ij\theta'}\right)} (1+o(1)),
\end{split}
\end{align}
uniformly for $\theta,\theta' \in [0,2\pi)$ s.t. $e^{i\theta'}$ stays bounded away from $\pm 1$. \\

By (\ref{eqn:averages sigma}), (\ref{eqn:sigma5}), (\ref{eqn:sigma3 extended}), with $\kappa = 2$, $q = 0$ and $s' = t' = \frac{1}{2}$, and $z_1,z_2,\beta_1,\beta_2$ chosen as in (\ref{eqn:decomposition}), a quick calculation results in
\begin{align} \label{eqn:symplectic3}
\begin{split}
&\frac{\mathbb{E}_{Sp(2n)} \left( f_{n,\alpha,\beta}(\theta)f_{n,\alpha,\beta}(\theta') \right)}{\mathbb{E}_{Sp(2n)} \left( f_{n,\alpha,\beta}(\theta) \right)\mathbb{E}_{Sp(2n)} \left(f_{n,\alpha,\beta}(\theta') \right)} = \frac{D_n^{T+H,2}(\sigma_{3,\theta,\theta'})}{D_n^{T+H,2}(\sigma_{5,\theta})D_n^{T+H,2}(\sigma_{5,\theta'})} \\
=& |e^{i\theta} - e^{i\theta'}|^{-2(\alpha^2-\beta^2)} |e^{i\theta} - e^{-i\theta'}|^{-2(\alpha^2+\beta^2)} z_1^{2\alpha\beta_1} z_2^{2\alpha \beta_2} e^{-2\pi i \alpha \beta_2} (1+o(1)) \\
=& |e^{i\theta} - e^{i\theta'}|^{-2(\alpha^2-\beta^2)} |e^{i\theta} - e^{-i\theta'}|^{-2(\alpha^2 + \beta^2)} e^{4i\alpha\beta \Im \log (1-e^{i(\theta+\theta')})} (1+o(1)),
\end{split}
\end{align}
uniformly for $\theta,\theta' \in [0,2\pi)$, s.t. $e^{i\theta}, e^{i\theta'}, e^{-i\theta}, e^{-i\theta'}$
stay bounded away from each other and $\pm 1$.

\subsection{The Odd Orthogonal Case}
In the odd orthogonal case we always have $q = -n$, and $s' + t' = 0$, which implies that
\begin{equation}
2^{(1-s'-t')n+q - \frac{1}{2}(s'+t')^2 + \frac{1}{2}(s'+t')} = 1.
\end{equation}
We also note that by (\ref{eqn:sigma hat}) we have 
\begin{align} \label{eqn:sigma hat pm}
\begin{split}
\hat{\sigma}_{1,\theta,\theta'}(\pm 1) =& e^{-2\sum_{j=1}^k \frac{(\pm 1)^j}{j} \Re \left( (\alpha-\beta)\left( e^{ij\theta} + e^{ij\theta'} \right) \right)}, \\
\hat{\sigma}_{2,\theta,\theta'}(\pm 1) =& e^{-2\sum_{j=1}^k \frac{(\pm 1)^j}{j} \Re \left( (\alpha-\beta) e^{ij\theta} \right)} |1 \mp e^{i\theta'}|^{2\alpha} e^{i\beta(\pi - \theta')} g_{e^{i\theta'},\beta}(\pm 1), \\
\hat{\sigma}_{3,\theta,\theta'}(1) =& |1 - e^{i\theta}|^{2\alpha} |1 - e^{i\theta'}|^{2\alpha} e^{i\beta(\pi - \theta)} e^{i\beta(\pi - \theta')}, \\
\hat{\sigma}_{3,\theta,\theta'}(-1) =& |1 + e^{i\theta}|^{2\alpha} |1 + e^{i\theta'}|^{2\alpha} e^{i\beta(\pi - \theta)} e^{i\beta(\pi - \theta')} e^{-i\pi(\beta_1 + \beta_2)}, \\
\hat{\sigma}_{4,\theta}(\pm 1) =& e^{-2\sum_{j=1}^k \frac{(\pm 1)^j}{j} \Re \left( (\alpha-\beta) e^{ij\theta} \right)}, \\
\hat{\sigma}_{5,\theta'}(1) =& |1 - e^{i\theta'}|^{2\alpha} e^{i\beta(\pi - \theta')},\\
\hat{\sigma}_{5,\theta'}(-1) =& |1 + e^{i\theta'}|^{2\alpha} e^{i\beta(\pi - \theta')} e^{-i\pi \beta_1},
\end{split}
\end{align}
where $\beta_1,\beta_2$ are chosen as in (\ref{eqn:decomposition}) for $\hat{\sigma}_{3,\theta,\theta'}$, and as in (\ref{eqn:sigma5 decomposition}) for $\hat{\sigma}_{5,\theta'}$. Thus by (\ref{eqn:averages sigma}), (\ref{eqn:sigma1}) and (\ref{eqn:sigma hat pm}) we get 
\begin{align} \label{eqn:odd sigma1}
&\mathbb{E}_{O(2n+1)}\left( f_{2n+1,\alpha,\beta}^{(k)}(\theta)f_{2n+1,\alpha,\beta}^{(k)}(\theta') \right) \nonumber \\
=& \frac{1}{2} \hat{\sigma}_{1,\theta,\theta'}(1) D_n^{T+H,3}(\sigma_{1,\theta,\theta'}) + \frac{1}{2} \hat{\sigma}_{1,\theta,\theta'}(-1) D_n^{T+H,4}(\sigma_{1,\theta,\theta'}) \nonumber \\
=& \frac{1}{2} e^{\sum_{j=1}^k \frac{2}{j}   \Re \left((\alpha-\beta)(e^{ij\theta}+e^{ij\theta'}) \right)^2} \nonumber \\
&\times \Bigg( e^{-2\sum_{j=1}^k \frac{1}{j} \Re \left( (\alpha-\beta)\left( e^{ij\theta} + e^{ij\theta'} \right) \right)} e^{\sum_{j=1}^k \frac{1-(-1)^j}{j} \Re \left( (\alpha-\beta)(e^{ij\theta}+e^{ij\theta'}) \right)} \\
&+ e^{-2\sum_{j=1}^k \frac{(-1)^j}{j} \Re \left( (\alpha-\beta)\left( e^{ij\theta} + e^{ij\theta'} \right) \right)} e^{-\sum_{j=1}^k \frac{1-(-1)^j}{j} \Re \left( (\alpha-\beta)(e^{ij\theta}+e^{ij\theta'}) \right)} \Bigg) (1+o(1)) \nonumber \\
=& e^{\sum_{j=1}^k \frac{2}{j}   \Re \left((\alpha-\beta)(e^{ij\theta}+e^{ij\theta'}) \right)^2} e^{-\sum_{j=1}^k \frac{1+(-1)^j}{j} \Re \left( (\alpha-\beta)(e^{ij\theta}+e^{ij\theta'}) \right)}(1+o(1)), \nonumber 
\end{align}
uniformly for $\theta,\theta' \in [0,2\pi)$. \\

Similarly we obtain from (\ref{eqn:averages sigma}), (\ref{eqn:sigma2}, (\ref{eqn:sigma4}), (\ref{eqn:sigma5}) and (\ref{eqn:sigma hat pm}), that
\begin{align} \label{eqn:odd sigma4}
\mathbb{E}_{O(2n+1)}\left( f_{2n+1,\alpha,\beta}^{(k)}(\theta) \right) = e^{\sum_{j=1}^k \frac{2}{j}   \Re \left((\alpha-\beta)e^{ij\theta} \right)^2} e^{-\sum_{j=1}^k \frac{1+(-1)^j}{j} \Re \left( (\alpha-\beta) e^{ij\theta}\right)} (1+o(1)),
\end{align}
uniformly for $\theta,\theta' \in [0,2\pi)$, and
\begin{align} \label{eqn:odd sigma25}
\mathbb{E}_{O(2n+1)}\left( f_{2n+1,\alpha,\beta}^{(k)}(\theta)f_{2n+1,\alpha,\beta}(\theta') \right) =& (2n)^{(\alpha^2-\beta^2)} e^{\sum_{j=1}^k \frac{2}{j}   \Re \left((\alpha-\beta) e^{ij\theta} \right)^2} \nonumber \\
&\times e^{\sum_{j=1}^k \frac{4}{j} \Re \left( (\alpha-\beta)e^{ij\theta} \right)\Re \left( (\alpha-\beta)e^{ij\theta'} \right)} \nonumber \\
&\times e^{-i\pi \alpha \beta_1} z_1^{2\alpha \beta_1} |1- e^{2i\theta'}|^{-(\alpha^2+\beta^2)} \nonumber \\
&\times \frac{G(1+\alpha+\beta) G(1+\alpha-\beta)}{G(1+2\alpha)} \\
&\times e^{-\sum_{j=1}^k \frac{1+(-1)^j}{j} \Re \left( (\alpha-\beta)e^{ij\theta} \right)} \nonumber \\
&\times e^{i\beta(\pi - \theta')} e^{-\frac{i \pi}{2} \beta_1} |1-e^{i\theta'}|^{\alpha} |1+e^{i\theta'}|^{\alpha} (1+o(1)), \nonumber \\
\mathbb{E}_{O(2n+1)}\left( f_{2n+1,\alpha,\beta}(\theta') \right) =& (2n)^{(\alpha^2-\beta^2)} e^{-i\pi \alpha \beta_1} z_1^{2\alpha \beta_1} |1- e^{2i\theta'}|^{-(\alpha^2+\beta^2)} \nonumber \\
&\times \frac{G(1+\alpha+\beta) G(1+\alpha-\beta)}{G(1+2\alpha)} \nonumber \\
&\times e^{i\beta(\pi - \theta')} e^{-\frac{i \pi}{2} \beta_1}
|1-e^{i\theta'}|^{\alpha} |1+e^{i\theta'}|^{\alpha} (1+o(1)), \nonumber 
\end{align}
uniformly for $\theta,\theta' \in [0,2\pi)$, s.t. $e^{i\theta'}$ stays bounded away from $\pm 1$, where
\begin{equation}
z_1 = \begin{cases} e^{i\theta'} & 0 < \theta' < \pi \\ e^{i(2\pi - \theta')} & \pi < \theta' < 2\pi \end{cases}, \quad \beta_1 = \begin{cases} \beta & 0 < \theta' < \pi \\ -\beta & \pi < \theta' < 2\pi \end{cases}.
\end{equation}

From (\ref{eqn:averages sigma}), (\ref{eqn:sigma3 extended}) and (\ref{eqn:sigma hat pm}) we obtain
\begin{align} \label{eqn:odd sigma3 extended}
& \mathbb{E}_{O(2n+1)}\left( f_{2n+1,\alpha,\beta}(\theta)f_{2n+1,\alpha,\beta}(\theta') \right) \nonumber \\
=& (2n)^{2(\alpha^2-\beta^2)} e^{-i\pi \alpha \left( \beta_1 + \beta_2 + 2 \beta_2 \right)} \nonumber \\
&\times |e^{i\theta} - e^{i\theta'}|^{-2(\alpha^2-\beta^2)} |e^{i\theta} - e^{-i\theta'}|^{-2(\alpha^2+\beta^2)} z_1^{4\beta_1\alpha} z_2^{4\beta_2\alpha} |1- e^{2i\theta}|^{-(\alpha^2+\beta^2)} |1- e^{2i\theta}|^{-(\alpha^2+\beta^2)} \nonumber \\
&\times \frac{G(1+\alpha+\beta)^2 G(1+\alpha-\beta)^2}{G(1+2\alpha)^2} e^{i\beta(\pi - \theta)} e^{i\beta(\pi - \theta')} e^{-\frac{i\pi}{2}(\beta_1 + \beta_2)} \\
&\times |1-e^{i\theta}|^{\alpha} |1+e^{i\theta}|^{\alpha} |1-e^{i\theta'}|^{\alpha} |1+e^{i\theta'}|^{\alpha} (1+o(1)), \nonumber 
\end{align}
uniformly for $\theta,\theta' \in [0,2\pi)$, s.t. $e^{i\theta}, e^{i\theta'}, e^{-i\theta}, e^{-i\theta'}$
stay bounded away from each other and $\pm 1$, and where $z_1,z_2,\beta_1,\beta_2$ are chosen as in (\ref{eqn:decomposition}).\\

Combining (\ref{eqn:odd sigma1}), (\ref{eqn:odd sigma4}) and (\ref{eqn:odd sigma25}), we obtain
\begin{align} \label{eqn:odd orthogonal1}
\begin{split}
&\frac{\mathbb{E}_{O(2n+1)} \left( f_{2n+1,\alpha,\beta}^{(k)}(\theta)f_{2n+1,\alpha,\beta}^{(k)}(\theta') \right)}{\mathbb{E}_{O(2n+1)} \left( f_{2n+1,\alpha,\beta}^{(k)}(\theta) \right)\mathbb{E}_{O(2n+1)} \left(f_{2n+1,\alpha,\beta}^{(k)}(\theta') \right)} \\
=& e^{\sum_{j=1}^k \frac{4}{j} \Re \left((\alpha - \beta)e^{ij\theta}\right) \Re \left((\alpha - \beta)e^{ij\theta'}\right)} (1+o(1)),
\end{split}
\end{align}
uniformly for $\theta,\theta' \in [0,2\pi)$, and 
\begin{align} \label{eqn:odd orthogonal2}
\begin{split}
&\frac{\mathbb{E}_{O(2n+1)} \left( f_{2n+1,\alpha,\beta}^{(k)}(\theta) f_{2n+1,\alpha,\beta}(\theta') \right)}{\mathbb{E}_{O(2n+1)} \left( f_{2n+1,\alpha,\beta}^{(k)}(\theta) \right)\mathbb{E}_{O(2n+1)} \left(f_{2n+1,\alpha,\beta}(\theta') \right)} \\
=& e^{\sum_{j=1}^k \frac{4}{j} \Re \left((\alpha - \beta)e^{ij\theta}\right) \Re \left((\alpha - \beta)e^{ij\theta'}\right)} (1+o(1)), 
\end{split}
\end{align}
uniformly for $\theta,\theta' \in [0,2\pi)$, s.t. $e^{i\theta'}$ stays bounded away from $\pm 1$. \\

By (\ref{eqn:odd sigma25}) and (\ref{eqn:odd sigma3 extended}) we obtain
\begin{align} \label{eqn:odd orthogonal3}
\begin{split}
&\frac{\mathbb{E}_{O(2n+1)} \left( f_{2n+1,\alpha,\beta}(\theta)f_{2n+1,\alpha,\beta}(\theta') \right)}{\mathbb{E}_{O(2n+1)} \left( f_{2n+1,\alpha,\beta}(\theta) \right)\mathbb{E}_{O(2n+1)} \left(f_{2n+1,\alpha,\beta}(\theta') \right)} \\
=& |e^{i\theta} - e^{i\theta'}|^{-2(\alpha^2-\beta^2)} |e^{i\theta} - e^{-i\theta'}|^{-2(\alpha^2+\beta^2)} z_1^{2\alpha\beta_1} z_2^{2\alpha \beta_2} e^{-2\pi i \alpha \beta_2} (1+o(1)), \\
=& e^{4i\alpha\beta \Im \log (1-e^{i(\theta+\theta')})} |e^{i\theta} - e^{i\theta'}|^{-2(\alpha^2-\beta^2)} |e^{i\theta} - e^{-i\theta'}|^{-2(\alpha^2 + \beta^2)} (1+o(1)),
\end{split}
\end{align}
uniformly for $\theta,\theta' \in [0,2\pi)$, s.t. $e^{i\theta}, e^{i\theta'}, e^{-i\theta}, e^{-i\theta'}$
stay bounded away from each other and $\pm 1$.

\subsection{The Even Orthogonal Case}
In the same way as in the odd orthogonal case one can use (\ref{eqn:averages sigma}), (\ref{eqn:sigma1}) - (\ref{eqn:sigma3 extended}) and (\ref{eqn:sigma hat pm}) to obtain 
\begin{align} \label{eqn:even orthogonal1}
\begin{split}
&\frac{\mathbb{E}_{O(2n)} \left( f_{2n,\alpha,\beta}^{(k)}(\theta)f_{2n,\alpha,\beta}^{(k)}(\theta') \right)}{\mathbb{E}_{O(2n)} \left( f_{2n,\alpha,\beta}^{(k)}(\theta) \right)\mathbb{E}_{O(2n)} \left(f_{2n,\alpha,\beta}^{(k)}(\theta') \right)} = e^{\sum_{j=1}^k \frac{4}{j} \Re \left((\alpha - \beta)e^{ij\theta}\right) \Re \left((\alpha - \beta)e^{ij\theta'}\right)} (1+o(1)),
\end{split}
\end{align}
uniformly for $\theta,\theta \in [0,2\pi)$, and 
\begin{align}
\begin{split} \label{eqn:even orthogonal2}
&\frac{\mathbb{E}_{O(2n)} \left( f_{2n,\alpha,\beta}^{(k)}(\theta) f_{2n,\alpha,\beta}(\theta') \right)}{\mathbb{E}_{O(2n)} \left( f_{2n,\alpha,\beta}^{(k)}(\theta) \right)\mathbb{E}_{O(2n)} \left(f_{2n,\alpha,\beta}(\theta') \right)} = e^{\sum_{j=1}^k \frac{4}{j} \Re \left((\alpha - \beta)e^{ij\theta}\right) \Re \left((\alpha - \beta)e^{ij\theta'}\right)} (1+o(1)),
\end{split}
\end{align}
uniformly for $\theta, \theta' \in [0,2\pi)$, s.t. $e^{i\theta'}$ stays bounded away from $\pm 1$, and 
\begin{align} \label{eqn:even orthogonal3}
\begin{split}
&\frac{\mathbb{E}_{O(2n)} \left( f_{2n,\alpha,\beta}(\theta)f_{2n,\alpha,\beta}(\theta') \right)}{\mathbb{E}_{O(2n)} \left( f_{2n,\alpha,\beta}(\theta) \right)\mathbb{E}_{O(2n)} \left(f_{2n,\alpha,\beta}(\theta') \right)} \\
=& e^{4i\alpha\beta \Im \log (1-e^{i(\theta+\theta')})} |e^{i\theta} - e^{i\theta'}|^{-2(\alpha^2-\beta^2)} |e^{i\theta} - e^{-i\theta'}|^{-2(\alpha^2 + \beta^2)} (1+o(1)),
\end{split}
\end{align}
uniformly for $\theta,\theta' \in [0,2\pi)$, s.t. $e^{i\theta}, e^{i\theta'}, e^{-i\theta}, e^{-i\theta'}$
stay bounded away from each other and $\pm 1$.\\

In Section \ref{section:second limit} we will also need that
\begin{align} \label{eqn:even sigma4}
\begin{split}
&\mathbb{E}_{O(2n)}\left( f_{2n,\alpha,\beta}^{(k)}(\theta) \right) = e^{- 2\sum_{j=1}^k \frac{\eta_j}{j} \Re\left( (\alpha-\beta)e^{ij\theta} \right) + \sum_{j=1}^k \frac{2}{j} \Re\left((\alpha-\beta)e^{ij\theta}\right)^2}(1+o(1)),
\end{split}
\end{align}
uniformly in $\theta \in [0,2\pi)$. 

\subsection{Pulling the large-$n$ limit inside the integral}
Using the asymptotics computed in the previous sections and Theorem \ref{thm:T+H Claeys}, we follow the proof of Corollary 2.1 in \cite{Fahs2021} to show that we can pull $\lim_{n \rightarrow \infty}$ inside the integral, i.e. that we can use the non-uniform asymptotics of the integrand:
\begin{lemma} \label{lemma:L2 limit Claeys}
Let the expectations be over $O(n)$ or $Sp(2n)$, and $\alpha^2 - \beta^2 < 1/2$ and $0 \leq \alpha < 1/2$. Then 
\begin{align} \label{eqn:L2 limit Claeys1}
\lim_{n \rightarrow \infty} &\int_0^{2\pi} \int_0^{2\pi} g(\theta) g(\theta') \frac{\mathbb{E} \left( f_{n,\alpha,\beta}(\theta) f_{n,\alpha,\beta}(\theta') \right)} {\mathbb{E} \left( f_{n,\alpha,\beta}(\theta) \right)\mathbb{E} \left( f_{n,\alpha,\beta}(\theta') \right)} \text{d}\theta \text{d}\theta' \\
=&\int_0^{2\pi} \int_0^{2\pi} g(\theta)g(\theta') e^{4i\alpha\beta \Im \log (1-e^{i(\theta+\theta')})} |e^{i\theta} - e^{i\theta'}|^{-2(\alpha^2-\beta^2)} |e^{i\theta} - e^{-i\theta'}|^{-2(\alpha^2 + \beta^2)} \text{d}\theta \text{d}\theta', \nonumber 
\end{align}
and 
\begin{align} \label{eqn:L2 limit Claeys2}
\begin{split}
\lim_{n \rightarrow \infty} &\int_0^{2\pi} \int_0^{2\pi} g(\theta) g(\theta') \frac{\mathbb{E} \left( f^{(k)}_{n,\alpha,\beta}(\theta) f_{n,\alpha,\beta}(\theta') \right)} {\mathbb{E} \left( f^{(k)}_{n,\alpha,\beta}(\theta) \right)\mathbb{E} \left( f_{n,\alpha,\beta}(\theta') \right)} \text{d}\theta \text{d}\theta' \\
=&\int_0^{2\pi} \int_0^{2\pi} g(\theta)g(\theta') e^{\sum_{j=1}^k \frac{4}{j} \Re \left((\alpha - \beta)e^{ij\theta}\right) \Re \left((\alpha -\beta)e^{ij\theta'}\right)} \text{d}\theta \text{d}\theta'.
\end{split}
\end{align}
\end{lemma}
\begin{remark} \label{remark:parameters}
Note that the limit in (\ref{eqn:L2 limit Claeys1}) is finite since $2(\alpha^2 + \beta^2) \leq 2(\alpha^2 - \beta^2) < 1$ and $2(\alpha^2 + \beta^2) + 2(\alpha^2 - \beta^2) = 4\alpha^2 < 1$. If we restrict the integrals to $I_\epsilon = (\epsilon, \pi - \epsilon) \cup (\pi + \epsilon, 2\pi - \epsilon)$ then we only need that $2(\alpha^2 + \beta^2) \leq 2(\alpha^2 - \beta^2) < 1$ for the integral to be finite, since $|e^{i\theta} - e^{i\theta'}|$ and $|e^{i\theta} - e^{-i\theta'}|$ can only go to zero simultaneously when both $\theta, \theta'$ approach $\{0,\pi,2\pi\}$. Further, if we restrict to $I_\epsilon$ we can rely on Theorems \ref{thm:T+H uniform} and \ref{thm:T, T+H extended} for the proof, which only require $\alpha > -1/4$, and we will not need Theorem \ref{thm:T+H Claeys} anymore, which requires $\alpha \geq 0$. Thus when restricting to $I_\epsilon$, Lemma \ref{lemma:L2 limit Claeys} holds for the larger set of parameters $\alpha^2 - \beta^2 < 1/2$ and $\alpha > - 1/4$. 
\end{remark} 
\noindent \textbf{Proof of Lemma \ref{lemma:L2 limit Claeys}:} By Theorems \ref{thm:T+H Claeys} and \ref{thm:average} we have 
\begin{align} \label{eqn:quotients Sp(2n)}
\begin{split}
&\frac{\mathbb{E}_{Sp(2n)} \left( f_{n,\alpha,\beta}(\theta) f_{n,\alpha,\beta}(\theta') \right)} {\mathbb{E}_{Sp(2n)} \left( f_{n,\alpha,\beta}(\theta) \right)\mathbb{E}_{Sp(2n)} \left( f_{n,\alpha,\beta}(\theta') \right)} = \frac{D_n^{T+H,2}(\sigma_{3,\theta,\theta'}) }{ D_n^{T+H,2}(\sigma_{5,\theta}) D_n^{T+H,2}(\sigma_{5,\theta'})} \\
=& e^{\mathcal{O}(1)} \left( \sin \left| \frac{\theta - \theta'}{2} \right| + \frac{1}{n} \right)^{-2( \alpha^2 - \beta^2)} \left( \sin \left| \frac{\theta +\theta'}{2} \right| + \frac{1}{n} \right)^{-2(\alpha^2 + \beta^2)}, 
\end{split}
\end{align}
as $n\rightarrow \infty$, uniformly in (Lebesgue almost all) $(\theta,\theta') \in [0,2\pi)^2$. By the same theorems we get 
\begin{align}
\begin{split}
&\mathbb{E}_{O(2n+1)} \left( f_{2n+1,\alpha,\beta}(\theta) f_{2n+1,\alpha,\beta}(\theta') \right) \\
=& \frac{1}{2} \left( \hat{\sigma}_{3,\theta,\theta'}(1) D_n^{T+H,3}(\sigma_{3,\theta,\theta'}) + \hat{\sigma}_{3,\theta,\theta'}(-1)D_n^{T+H,4}(\sigma_{3,\theta,\theta'}) \right) \\
=& \frac{1}{2} F_{\sigma_{3,\theta,\theta'}} n^{2(\alpha^2 - \beta^2)} e^{\mathcal{O}(1)} \Bigg(  \prod_{j = 1}^2 \left( \frac{ \sin \frac{|\theta_j|}{2} + \frac{1}{n} }{ \sin^2 \frac{|\theta_j|}{2} \left( \cos \frac{|\theta_j|}{2} + \frac{1}{n} \right)} \right)^{-\alpha} \\
& \hspace{4.5cm} + \prod_{j = 1}^2 \left( \frac{ \cos \frac{|\theta_j|}{2} + \frac{1}{n} }{ \cos^2 \frac{|\theta_j|}{2} \left( \sin \frac{|\theta_j|}{2} + \frac{1}{n} \right)} \right)^{-\alpha} \Bigg),
\end{split}
\end{align}
and
\begin{align}
&\mathbb{E}_{O(2n+1)} \left( f_{2n+1,\alpha,\beta}(\theta) \right) \\
=& \frac{1}{2} F_{\sigma_{5,\theta}} n^{\alpha^2 - \beta^2} e^{\mathcal{O}(1)} \left( \left( \frac{ \sin \frac{|\theta|}{2} + \frac{1}{n} }{ \sin^2 \frac{|\theta|}{2} \left( \cos \frac{|\theta|}{2} + \frac{1}{n} \right)} \right)^{-\alpha} + \left( \frac{ \cos \frac{|\theta|}{2} + \frac{1}{n} }{ \cos^2 \frac{|\theta|}{2} \left( \sin \frac{|\theta|}{2} + \frac{1}{n} \right)} \right)^{-\alpha} \right), \nonumber 
\end{align}
as $n \rightarrow \infty$, uniformly for (Lebesgue almost all) $(\theta,\theta') \in [0,2\pi)^2$, where $\hat{\sigma}_{1,\theta,\theta'}$,...,$\hat{\sigma}_{5,\theta}$ are defined in (\ref{eqn:sigma hat}). Thus we can see that also for the expectations over $O(2n+1)$ we get
\begin{align} \label{eqn:quotients O(2n+1)}
\begin{split}
&\frac{ \mathbb{E}_{O(2n+1)} \left( f_{2n+1,\alpha,\beta}(\theta) f_{2n+1,\alpha,\beta}(\theta') \right) }{ \mathbb{E}_{O(2n+1)} \left( f_{2n+1,\alpha,\beta}(\theta) \right) \mathbb{E}_{O(2n+1)} \left( f_{2n+1,\alpha,\beta}(\theta') \right)} \\
=& e^{\mathcal{O}(1)} \left( \sin \left| \frac{\theta - \theta'}{2} \right| + \frac{1}{n} \right)^{-2( \alpha^2 - \beta^2)} \left( \sin \left| \frac{\theta +\theta'}{2} \right| + \frac{1}{n} \right)^{-2(\alpha^2 + \beta^2)}, 
\end{split}
\end{align}
as $n\rightarrow \infty$, uniformly in (Lebesgue almost all) $(\theta,\theta') \in [0,2\pi)^2$. Similarly, Theorem \ref{thm:T+H Claeys} gives
\begin{align} \label{eqn:quotients O(2n)}
\begin{split}
&\frac{ \mathbb{E}_{O(2n)} \left( f_{2n,\alpha,\beta}(\theta) f_{2n,\alpha,\beta}(\theta') \right) }{ \mathbb{E}_{O(2n)} \left( f_{2n,\alpha,\beta}(\theta) \right) \mathbb{E}_{O(2n)} \left( f_{2n,\alpha,\beta}(\theta') \right)} \\
=& e^{\mathcal{O}(1)} \left( \sin \left| \frac{\theta - \theta'}{2} \right| + \frac{1}{n} \right)^{-2( \alpha^2 - \beta^2)} \left( \sin \left| \frac{\theta +\theta'}{2} \right| + \frac{1}{n} \right)^{-2(\alpha^2 + \beta^2)}, 
\end{split}
\end{align}
as well as 
\begin{align} \label{eqn:quotients}
\begin{split}
&\frac{ \mathbb{E}_{Sp(2n)} \left( f^{(k)}_{n,\alpha,\beta}(\theta) f_{n,\alpha,\beta}(\theta') \right) }{ \mathbb{E}_{Sp(2n)} \left( f^{(k)}_{n,\alpha,\beta}(\theta) \right) \mathbb{E}_{Sp(2n)} \left( f_{n,\alpha,\beta}(\theta') \right)} = e^{\mathcal{O}(1)},\\
&\frac{ \mathbb{E}_{O(2n+1)} \left( f^{(k)}_{2n+1,\alpha,\beta}(\theta) f_{2n+1,\alpha,\beta}(\theta') \right) }{ \mathbb{E}_{O(2n+1)} \left( f^{(k)}_{2n+1,\alpha,\beta}(\theta) \right) \mathbb{E}_{O(2n+1)} \left( f_{2n+1,\alpha,\beta}(\theta') \right)} = e^{\mathcal{O}(1)},\\
&\frac{ \mathbb{E}_{O(2n)} \left( f^{(k)}_{2n,\alpha,\beta}(\theta) f_{2n,\alpha,\beta}(\theta') \right) }{ \mathbb{E}_{O(2n)} \left( f^{(k)}_{2n,\alpha,\beta}(\theta) \right) \mathbb{E}_{O(2n)} \left( f_{2n,\alpha,\beta}(\theta') \right)} = e^{\mathcal{O}(1)},
\end{split}
\end{align}
as $n\rightarrow \infty$, uniformly in (Lebesgue almost all) $(\theta,\theta') \in [0,2\pi)^2$.\\

Now, for a given measureable subset $R \subset [0,2\pi)^2$, we denote
\begin{align}
\begin{split}
L_\epsilon(R) =& \int_R \left( \sin \left| \frac{\theta - \theta'}{2} \right| + \epsilon \right)^{-2( \alpha^2 - \beta^2)} \left( \sin \left| \frac{\theta +\theta'}{2} \right| + \epsilon \right)^{-2(\alpha^2 + \beta^2)} \text{d}\theta \text{d}\theta',\\
K_\epsilon(R) =& \int_R \left( \sin \left| \frac{\theta - \theta'}{2} \right| + \epsilon \right)^{-2( \alpha^2 - \beta^2)} \text{d}\theta \text{d}\theta'.
\end{split}
\end{align}
In the case $\alpha^2 + \beta^2 > 0$ we have $L_\epsilon(R) < L_0(R) < \infty$ for any $\epsilon > 0$ (since $2(\alpha^2 - \beta^2), 2(\alpha^2 + \beta^2), 4\alpha^2 < 1$), while in the case $\alpha^2 + \beta^2 \leq 0$ we have $K_\epsilon(R) < K_0(R) < \infty$ for any $\epsilon > 0$. For $\eta > 0$ we define 
\begin{align}
\begin{split}
R_1(\eta) =& \left\{ (\theta,\theta') \in [0,2\pi)^2: \sin \frac{|\theta - \theta'|}{2}, \sin \frac{|\theta + \theta'|}{2} > \mu \right\} \\
R_2(\eta) =& R_1(\eta)^c.
\end{split}
\end{align}
It follows by (\ref{eqn:quotients Sp(2n)}), (\ref{eqn:quotients O(2n+1)}) and (\ref{eqn:quotients O(2n)}) that for any $\eta > 0$ there exists a $C > 0$ and $N_0 \in \mathbb{N}$ such that 
\begin{align}
\begin{split}
\int_{R_2(\eta)} g(\theta) g(\theta') \frac{\mathbb{E} \left( f_{n,\alpha,\beta}(\theta) f_{n,\alpha,\beta}(\theta') \right)} {\mathbb{E} \left( f_{n,\alpha,\beta}(\theta) \right)\mathbb{E} \left( f_{n,\alpha,\beta}(\theta') \right)} \leq \begin{cases} C L_0(R_2(\eta)), & \alpha^2 + \beta^2 > 0, \\
C K_0(R_2(\eta)), & \alpha^2 + \beta^2 \leq 0, \end{cases}
\end{split}
\end{align}
for $n > N_0$. Fix $\delta > 0$. Since $L_0(R_2(\eta)), K_0(R_2(\eta) \rightarrow 0$ as $\eta \rightarrow 0$, it follows that there exists an $\eta_0 > 0$ and an $N_0 \in \mathbb{N}$ such that
\begin{align} \label{eqn:bound R_2}
\begin{split}
\int_{R_2(\eta)} g(\theta) g(\theta') \frac{\mathbb{E} \left( f_{n,\alpha,\beta}(\theta) f_{n,\alpha,\beta}(\theta') \right)} {\mathbb{E} \left( f_{n,\alpha,\beta}(\theta) \right)\mathbb{E} \left( f_{n,\alpha,\beta}(\theta') \right)} < \delta/2
\end{split}
\end{align}
for $n > N_0$ and $\eta < \eta_0$.  \\

Using (\ref{eqn:symplectic3}), (\ref{eqn:odd orthogonal3}) and (\ref{eqn:even orthogonal3}), we get that for any fixed $\eta > 0$ it holds that
\begin{align}
\begin{split}
&\int_{R_1(\eta)} g(\theta) g(\theta') \frac{\mathbb{E} \left( f_{n,\alpha,\beta}(\theta) f_{n,\alpha,\beta}(\theta') \right)} {\mathbb{E} \left( f_{n,\alpha,\beta}(\theta) \right)\mathbb{E} \left( f_{n,\alpha,\beta}(\theta') \right)} = (1+ o(1)) L_0(R_1(\eta)) \\
=& (1+o(1)) L_0([0,2\pi)^2) - (1+o(1)) L_0(R_2(\eta)).
\end{split}
\end{align}
We pick $\eta < \eta_0$ such that $I_0(R_2(\eta)) < \delta/4$, then we have
\begin{align}
\begin{split}
&\left| \int_{R_1(\eta)} g(\theta) g(\theta') \frac{\mathbb{E} \left( f_{n,\alpha,\beta}(\theta) f_{n,\alpha,\beta}(\theta') \right)} {\mathbb{E} \left( f_{n,\alpha,\beta}(\theta) \right)\mathbb{E} \left( f_{n,\alpha,\beta}(\theta') \right)} - L_0([0,2\pi)^2) \right| \\
=& \left| L_0(R_2(\eta)) +o(1) \right| < \delta/4 + o(1).
\end{split}
\end{align}
Together with (\ref{eqn:bound R_2}) we obtain that there exists an $N \in \mathbb{N}$ such that 
\begin{align}
\begin{split}
\left| \int_0^{2\pi} \int_0^{2\pi} g(\theta) g(\theta') \frac{\mathbb{E} \left( f_{n,\alpha,\beta}(\theta) f_{n,\alpha,\beta}(\theta') \right)} {\mathbb{E} \left( f_{n,\alpha,\beta}(\theta) \right)\mathbb{E} \left( f_{n,\alpha,\beta}(\theta') \right)} - L_0([0,2\pi)^2) \right| < \delta
\end{split}
\end{align}
for all $n > N$. Since $\delta >0 $ is arbitrary, this shows the first part of (\ref{eqn:L2 limit Claeys1}). (\ref{eqn:L2 limit Claeys2}) follows from (\ref{eqn:quotients}) in a similar way. \qed 

\subsection{Proof of Lemma \ref{lemma:L2 limit}}
Now we have all the ingredients necessary to prove Lemma \ref{lemma:L2 limit}. We will only prove it for $I = [0,2\pi)$, the proof for $I = I_\epsilon = (\epsilon, \pi - \epsilon) \cup (\pi + \epsilon, 2\pi + \epsilon)$ is completely analogous and relies on the fact that Lemma \ref{lemma:L2 limit Claeys} also holds for $I = I_\epsilon$, $\alpha^2 - \beta^2 < 1/2$ and $\alpha > -1/4$, as explained in Remark \ref{remark:parameters}.\\ 

\noindent \textbf{Proof of Lemma \ref{lemma:L2 limit}:} From Lemma \ref{lemma:L2 limit Claeys} and (\ref{eqn:L2 limit}), (\ref{eqn:symplectic1}), (\ref{eqn:odd orthogonal1}) and (\ref{eqn:even orthogonal1}) it follows that 
\begin{align} 
& \lim_{n\rightarrow \infty} \mathbb{E} \left( \left( \int_0^{2\pi} g(\theta) \mu_{G(n),\alpha,\beta}(\text{d}\theta) - \int_0^{2\pi} g(\theta) \mu_{G(n),\alpha,\beta}^{(k)}(\text{d}\theta) \right)^2 \right) \nonumber \\
=& \int_0^{2\pi} \int_0^{2\pi} g(\theta)g(\theta') e^{4i\alpha\beta \Im \log (1-e^{i(\theta+\theta')})} |e^{i\theta} - e^{i\theta'}|^{-2(\alpha^2-\beta^2)} |e^{i\theta} - e^{-i\theta'}|^{-2(\alpha^2 + \beta^2)} \text{d}\theta \text{d}\theta' \nonumber \\
& - \int_0^{2\pi} \int_0^{2\pi} g(\theta)g(\theta') e^{\sum_{j=1}^k \frac{4}{j} \Re \left((\alpha - \beta)e^{ij\theta}\right) \Re \left((\alpha - \beta)e^{ij\theta'}\right)}  \text{d}\theta \text{d}\theta'.
\end{align}
We see that
\begin{align}\label{eqn:bounded in L2 - 1}
&4\Re \left((\alpha - \beta)e^{ij\theta}\right) \Re \left((\alpha - \beta)e^{ij\theta'}\right) \nonumber \\
=& (\alpha - \beta)^2 e^{ij(\theta+\theta')} + (\alpha + \beta)^2 e^{-ij(\theta+\theta')} + (\alpha^2- \beta^2)(e^{ij(\theta-\theta')} + e^{-ij(\theta-\theta')})\\
=& 2(\alpha^2-\beta^2) \cos(j(\theta-\theta')) + 2(\alpha^2+\beta^2) \cos(j(\theta+\theta')) - 2 \alpha \beta \left(e^{ij(\theta+\theta')} - e^{-ij(\theta + \theta')}\right).  \nonumber 
\end{align}
Since
\begin{align}
\begin{split}
\log |e^{i\theta} - e^{i\theta'}| =& - \sum_{j=1}^\infty \frac{1}{j} \cos(j(\theta-\theta')), \\
\log |e^{i\theta} - e^{-i\theta'}| =& - \sum_{j=1}^\infty \frac{1}{j} \cos(j(\theta+\theta')),
\end{split}
\end{align}
and 
\begin{align}
\begin{split}
-\sum_{j=1}^\infty \frac{1}{j}(e^{ij(\theta+\theta')} - e^{-ij(\theta+\theta')}) =& \log  (1-e^{i(\theta+\theta')}) - \log (1-e^{-i(\theta+\theta')}) \\
=& 2i\Im \log (1-e^{i(\theta+\theta')}),
\end{split}
\end{align}
we see that
\begin{align} \label{eqn:series for covariance function}
\begin{split}
&e^{\sum_{j=1}^\infty \frac{1}{j} \Re \left((\alpha + i\beta)e^{ij\theta}\right) \Re \left((\alpha + i\beta)e^{ij\theta'}\right)} \\
=& |e^{i\theta} - e^{i\theta'}|^{-2(\alpha^2-\beta^2)} |e^{i\theta} - e^{-i\theta'}|^{-2(\alpha^2+\beta^2)} e^{4i\alpha \beta \Im \log (1-e^{i(\theta+\theta')})}.
\end{split}
\end{align}
Because $\mathbb{E}\left( (...)^2\right) \geq 0$ it holds that
\begin{align} \label{eqn:bounded in L2 - 2}
&\int_{0}^{2\pi} \int_{0}^{2\pi} g(\theta)g(\theta') |e^{i\theta} - e^{i\theta'}|^{-2(\alpha^2-\beta^2)} |e^{i\theta} - e^{-i\theta'}|^{-2(\alpha^2+\beta^2)} e^{4i\alpha \beta \Im \log (1-e^{i(\theta+\theta')})} \text{d}\theta \text{d}\theta' \nonumber \\
\geq &\limsup_{k\rightarrow \infty} \int_{0}^{2\pi} \int_{0}^{2\pi} g(\theta)g(\theta') e^{\sum_{j=1}^k \frac{4}{j} \Re \left((\alpha - \beta)e^{ij\theta}\right) \Re \left((\alpha - \beta)e^{ij\theta'}\right)}  \text{d}\theta \text{d}\theta'.
\end{align}
Now we use that $g$ is non-negative to apply Fatou's lemma to get the other inequality, which finishes the proof. \qed 

\section{Proof of the Second Limit} \label{section:second limit}

In this section we prove (\ref{eqn:second limit}) for $I = [0,2\pi)$, i.e. that for any fixed $k \in \mathbb{N}$ and bounded continuous function $g:[0,2\pi) \rightarrow \mathbb{R}$ it holds that
\begin{equation}
\int_0^{2\pi} g(\theta) \mu_{G(n),\alpha,\beta}^{(k)}(\text{d}\theta) \xrightarrow{d} \int_0^{2\pi} g(\theta) \mu_{\alpha,\beta}^{(k)}(\text{d}\theta),
\end{equation}
as $n\rightarrow \infty$, where $\mu_{G(n),\alpha,\beta}^{(k)}$ is defined in Definition \ref{def:mu_n^(k)} and $\mu^{(k)}_{\alpha,\beta}$ is defined in (\ref{eqn:mu k}). For $I = I_\epsilon = (\epsilon, \pi - \epsilon) \cup (\pi + \epsilon, 2\pi - \epsilon)$ the proof is exactly the same.   \\

We consider the function $F:\mathbb{R}^k \rightarrow \mathbb{R}$,
\begin{equation}
F(z_1,...,z_k) = \int_0^{2\pi} \frac{g(\theta) e^{-2\sum_{j=1}^k \frac{z_j}{\sqrt{j}} \left( \alpha \cos(j\theta) - i\beta \sin(j\theta) \right)}}{e^{\pm 2\sum_{j=1}^k \frac{\eta_j}{j} \Re\left( (\alpha-\beta)e^{ij\theta} \right) + \sum_{j=1}^k \frac{2}{j} \Re\left((\alpha-\beta)e^{ij\theta}\right)^2}} \text{d}\theta, 
\end{equation}
which is continuous since the integrand is continuous in $z_1,...,z_n$ and $\theta$, and bounded in $\theta$ for any fixed $z_1,...,z_n$. Then we have, with $\pm$ corresponding to symplectic/orthogonal:
\begin{align}
\int_0^{2\pi} g(\theta) \mu_{G(n),\alpha,\beta}^{(k)}(\text{d}\theta) =& \int_0^{2\pi} \frac{g(\theta) e^{-2\sum_{j=1}^k \frac{\text{Tr}(U_n^j)}{\sqrt{j}} \left( \alpha \cos(j\theta) - i\beta \sin(j\theta) \right)}}{\mathbb{E}(f_{n,\alpha,\beta}^{(k)}(\theta))} \text{d}\theta \nonumber \\
=& \frac{1}{1+o(1)} \int_0^{2\pi} \frac{g(\theta) e^{-2\sum_{j=1}^k \frac{\text{Tr}(U_n^j)}{\sqrt{j}} \left( \alpha \cos(j\theta) - i\beta \sin(j\theta) \right)}}{e^{\pm 2\sum_{j=1}^k \frac{\eta_j}{j} \Re\left( (\alpha-\beta)e^{ij\theta} \right) + \sum_{j=1}^k \frac{2}{j} \Re\left((\alpha-\beta)e^{ij\theta}\right)^2}} \text{d}\theta \nonumber \\
&\xrightarrow{d} \int_0^{2\pi} \frac{g(\theta) e^{- 2\sum_{j=1}^k \frac{\mathcal{N}_j \mp \frac{\eta_j}{\sqrt{j}}}{\sqrt{j}} \left( \alpha \cos(j\theta) - i\beta \sin(j\theta) \right)}}{e^{\pm 2\sum_{j=1}^k \frac{\eta_j}{j} \Re\left( (\alpha-\beta)e^{ij\theta} \right) + \sum_{j=1}^k \frac{2}{j} \Re\left((\alpha-\beta)e^{ij\theta}\right)^2}} \text{d}\theta \nonumber \\
\overset{d}{=}& \int_0^{2\pi} g(\theta) e^{ 2\sum_{j=1}^k \frac{\mathcal{N}_j}{\sqrt{j}} \left( \alpha \cos(j\theta) - i\beta \sin(j\theta) \right) -\sum_{j=1}^k \frac{2}{j} \Re\left((\alpha-\beta)e^{ij\theta}\right)^2} \text{d}\theta \nonumber \\
=& \int_0^{2\pi} g(\theta) \mu_{\alpha,\beta}^{(k)}(\text{d}\theta),  
\end{align} 
where in the second equality we used (\ref{eqn:sigma4}), (\ref{eqn:odd sigma4}), and (\ref{eqn:even sigma4}), where the convergence in distribution follows from Theorem \ref{thm:traces} and the continuous mapping theorem, and where the penultimate equality follows from the fact that $- \mathcal{N}_j \overset{d}{=} \mathcal{N}_j$.

%\section*{Acknowledgements}
%Our work was supported by ERC Advanced Grant 740900 (LogCorRM).  We are most grateful to Theo Assiotis for his kind permission to state Theorem \ref{thm:Gaussian field} here and to set out its proof in Appendix \ref{appendix:Gaussian fields}, as well as for many extremely helpful discussions.  We are also most grateful to Tom Claeys, Gabriel Glesner, Alexander Minakov and Meng Yang for having kindly shared with us their work in progress and for having communicated to us one of their results from \cite{Claeys et al}, prior to posting it on the arXiv, which we quote in Theorem \ref{thm:T+H Claeys}. We make use of this result to prove our Theorem \ref{thm:main2}. Further we thank Mo Dick Wong for helpful comments and suggestions. Finally we thank two anonymous referees for their careful reading, helpful remarks and suggestions.

\chapter{Moments of Moments of the Characteristic Polynomial of Random Orthogonal and Symplectic Matrices} \label{chapter:MoM}

In this chapter, which is based on joint work with Tom Claeys and Jon Keating \cite{CFK23}, we use the uniform asymptotics of Toeplitz+Hankel determinants from Theorem \ref{thm:T+H Claeys}, due to Claeys et al. \cite{Claeys2022}, to establish formulae, stated in Theorems \ref{thm:MoM Sp(2n)} and \ref{thm:MoM O(n)} below, for the asymptotics of the moments of the moments of the characteristic polynomials of random orthogonal and symplectic matrices, as the matrix-size tends to infinity. Those formulae are analogous to the ones for random unitary matrices, found by Fahs \cite{Fahs2021}, and stated in Theorem \ref{thm:Fahs}. A key feature of the formulae we derive is that the phase transitions in the moments
of moments are seen to depend on the symmetry group in question in a
significant way. 

\section{Statement of Results}
We recall that for $G(n) \in \left\{ U(n), \, O(n), \, SO(n), \, SO^-(n), \, Sp(2n) \right\}$ we define the characteristic polynomial of a Haar distributed $U \in G(n)$ by
\begin{equation}
p_{G(n)}(\theta) := \text{det} \left(I-e^{-i\theta}U\right), \quad \theta \in [0,2\pi),
\end{equation}
and that for $\alpha > -1/2$ and $m \in \mathbb{R}$ we define the moments of moments of $p_{G(n)}(\theta)$ by
\begin{align}\label{def:MoM}
\begin{split}
    {\rm MoM}_{G(n)}(m,\alpha) :=& \mathbb{E}_{ G(n)} \left( \left( \frac{1}{2\pi} \int_0^{2\pi} |p_{G(n)}(\theta)|^{2\alpha} \text{d}\theta \right)^m \right),
\end{split}
\end{align}
where $\mathbb{E}_{ G(n)}$ denotes the expectation with respect to the normalized Haar measure on $G(n)$. Furthermore we recall Fahs' asymptotic formula for ${\rm MoM}_{U(n)}(m,\alpha)$ (this result was already stated as Theorem \ref{thm:Fahs} in Chapter \ref{chapter:Intro}, but we restate it here for convenience):
\begin{theorem}[\cite{Fahs2021}] \label{thm:Fahs}
    For $m \in \mathbb{N}$ and $\alpha > 0$, as $n \rightarrow \infty$:
    \begin{equation}
        {\rm MoM}_{U(n)}(m,\alpha) = \begin{cases} (1+o(1)) n^{m\alpha^2} \frac{G(1 + \alpha)^{2m}\Gamma(1 - m\alpha^2)}{G(1 + 2\alpha)^m \Gamma(1 - \alpha^2)^m}, & \alpha < \frac{1}{\sqrt{m}}, \\
        e^{\mathcal{O}(1)} n^{m\alpha^2} \log n & \alpha = \frac{1}{\sqrt{m}},\\
        e^{\mathcal{O}(1)} n^{m^2\alpha^2 + 1 - m}, & \alpha > \frac{1}{\sqrt{m}}, \end{cases}
    \end{equation}
    where $e^{\mathcal{O}(1)}$ denotes a function that is bounded and bounded away from $0$ as $n \rightarrow \infty$, and where $G(z)$ denotes the Barnes $G$-function.
\end{theorem}

For $m\in\mathbb N$ and $\alpha>0$, we define the constants
\begin{align} \label{eqn:C pm}
\begin{split}
C^\pm(m,\alpha) := \frac{G(1+\alpha)^{2m}}{G(1+2\alpha)^m} \frac{4^{-\alpha^2m^2 \pm \alpha m}}{\pi^m} \prod_{j = 0}^{m-1} \frac{\Gamma(1 - \alpha^2 - j\alpha^2) \Gamma\left( \frac{1 - \alpha^2 \pm \alpha}{2} - j\alpha^2 \right)^2}{\Gamma(1 - \alpha^2) \Gamma \left( 1 \pm \alpha - \alpha^2(m + j) \right)},
\end{split}
\end{align}
where $\Gamma$ denotes Euler's Gamma function and $G$ denotes the Barnes $G$-function, satisfying $G(z+1)=\Gamma(z)G(z)$ and $G(0)=0$, $G(1)=1$. Note that although $C^\pm(m,\alpha)$ has poles for certain values of $\alpha$, it is well defined if 
\begin{equation}\label{eq:alphasubcrit}\alpha<\min\left\{\frac{1}{\sqrt{m}}, \frac{\sqrt{8m-3}\pm 1}{4m-2}\right\}=\begin{cases}\frac{1}{\sqrt{m}}&\mbox{if $m=2$ and $\pm=+$,}\\
\frac{\sqrt{8m-3}\pm 1}{4m-2}&\mbox{otherwise,}\end{cases}\end{equation} since in this case we have
that $1-\alpha^2-j\alpha^2$ and $ \frac{1 - \alpha^2 \pm \alpha}{2} - j\alpha^2$ are positive for $j=0,\ldots, m-1$, such that the poles of $\Gamma$ are avoided.\\

Our main results of this chapter are then as follows (recall that they were already stated as Theorems \ref{thm:MoM Sp intro} and \ref{thm:MoM O intro} in Chapter \ref{chapter:Intro}):

\begin{theorem} \label{thm:MoM Sp(2n)}
Let $m \in \mathbb{N}$ and $\alpha > 0$. Then, as $n \rightarrow \infty$:
\begin{align}
\begin{split}
    {\rm MoM}_{Sp(2n)}(m,\alpha) = \begin{cases} (1+o(1)) (2n)^{m\alpha^2} C^-(m,\alpha), & \alpha < \frac{\sqrt{8m-3} - 1}{4m-2}, \\
    e^{\mathcal{O}(1)} n^{m\alpha^2} \log n & \alpha = \frac{\sqrt{8m-3} - 1}{4m-2},\\
    e^{\mathcal{O}(1)} n^{2(m\alpha)^2 + m\alpha -m}, & \alpha > \frac{\sqrt{8m-3} - 1}{4m-2}. \end{cases}
\end{split}
\end{align}
\end{theorem}

\begin{theorem} \label{thm:MoM O(n)}
Let $G(n) \in \{ O(n), \, SO(n), \, SO^-(n) \}$ and $\alpha > 0$. For $m \in \mathbb{N} \setminus \{2\}$, as $n \rightarrow \infty$:
\begin{align}
\begin{split}
    {\rm MoM}_{G(n)}(m,\alpha) = \begin{cases} (1+o(1)) n^{m\alpha^2} C^+(m,\alpha), & \alpha < \frac{\sqrt{8m-3} + 1}{4m-2}, \\
    e^{\mathcal{O}(1)} n^{m\alpha^2} \log n & \alpha = \frac{\sqrt{8m-3} + 1}{4m-2},\\
    e^{\mathcal{O}(1)} n^{2(m\alpha)^2 - m\alpha -m}, & \alpha > \frac{\sqrt{8m-3} + 1}{4m-2}. \end{cases}
\end{split}
\end{align}
Moreover, as $n \rightarrow \infty$:
\begin{align}
    {\rm MoM}_{G(n)}(2,\alpha) = \begin{cases} (1 + o(1)) n^{2\alpha^2} C^+(2,\alpha), & \alpha < \frac{1}{\sqrt{2}}, \\
    e^{\mathcal{O}(1)} n^{2\alpha^2} \log n & \alpha = \frac{1}{\sqrt{2}},\\
    e^{\mathcal{O}(1)} n^{4\alpha^2 - 1}, & \alpha \in \left( \frac{1}{\sqrt{2}}, \frac{\sqrt{5} + 1}{4} \right), \\
    e^{\mathcal{O}(1)} n^{4\alpha^2 - 1} \log n, & \alpha = \frac{\sqrt{5} + 1}{4}, \\
    e^{\mathcal{O}(1)} n^{8\alpha^2 - 2\alpha - 2}, & \alpha > \frac{\sqrt{5} + 1}{4}. \end{cases}
\end{align}
\end{theorem}

\begin{remark}
Since $\mathbb{E}_{ O(n)}(f(U)) = \mathbb{E}_{ SO(n)}(f(U))/2 + \mathbb{E}_{ SO^-(n)}(f(U))/2$ for any measurable function $f:O(n) \rightarrow \mathbb{R}$, the above result for $G(n)=O(n)$ is in fact a simple consequence of the results for $G(n)=SO(n)$ and $G(n)=SO^-(n)$, and therefore we can restrict ourselves to $G(n) \in \{SO(n), \, SO^-(n), \, Sp(2n) \}$ in what follows.
\end{remark}

\begin{remark}We will refer to the situations where $\alpha$ is small enough such that \eqref{eq:alphasubcrit} holds as the subcritical regimes or phases. For $m\neq 2$, or when both $m=2$ and $G(n)=Sp(2n)$, there is the unique critical value $\alpha=\frac{\sqrt{8m-3} + 1}{4m-2}$ for $G(n) \in \{ O(n), \, SO(n), \, SO^-(n) \}$ or $\alpha=\frac{\sqrt{8m-3} - 1}{4m-2}$ for $G(n) = Sp(2n)$, and we will refer to values of $\alpha$ larger than the critical value as the supercritical regimes. When both $m=2$ and $G(n) \in \{ O(n), \, SO(n), \, SO^-(n) \}$, then there are the two critical values $\alpha=\frac{1}{\sqrt{2}}$ and $\alpha=\frac{\sqrt{5}+1}{4}$, and we speak of the intermediate regime if $\alpha$ lies in between the two critical values, and of the supercritical regime for $\alpha>\frac{\sqrt{5}+1}{4}$.
\end{remark}

\begin{remark}For $m, \alpha \in \mathbb{N}$, our results are consistent with the results obtained by Assiotis, Bailey and Keating \cite{Assiotis2019}, and Andrade and Best \cite{Andrade2022}. 
\end{remark}

\section{Proof Strategy and Outline}
For $G(n) \in \{SO(n), \, SO^-(n), \, Sp(2n) \}$, it follows from \eqref{def:MoM} and Fubini's theorem (recall that $m\in\mathbb N$) that
\begin{align} \label{eqn:Fubini}
\begin{split}
{\rm MoM}_{G(n)}(m,\alpha) =& \int_0^{2\pi} \cdots \int_0^{2\pi} \mathbb{E}_{ G(n)} \left( \prod_{j = 1}^m |p_{G(n)}(\theta_j)|^{2\alpha} \right) \frac{\text{d}\theta_1}{2\pi} \cdots \frac{\text{d}\theta_m}{2\pi} \\
=& \int_0^{\pi} \cdots \int_0^{\pi} \mathbb{E}_{ G(n)} \left( \prod_{j = 1}^m |p_{G(n)}(\theta_j)|^{2\alpha} \right) \frac{\text{d}\theta_1}{\pi} \cdots \frac{\text{d}\theta_m}{\pi}.
\end{split}
\end{align}
For the second equality above, we used the fact that
$p_{G(n)}(-\theta) = \overline{p_{G(n)}(\theta)}$, which holds since the eigenvalues of orthogonal and symplectic matrices are $\pm 1$ or appear in complex conjugate pairs. \\

For $\theta_1,\ldots,\theta_m \in (0,\pi)$ we define the symbols  
\begin{equation}
f_m^{(\alpha)}(z) =  \prod_{j = 1}^m |z - e^{i\theta_j}|^{2\alpha} |z - e^{-i\theta_j}|^{2\alpha}.
\end{equation}
Then by the Baik-Rains identity \cite{Baik2001}[Theorem 2.2], stated in Theorem \ref{thm:Baik2001}, we see that the averages in the integrand in (\ref{eqn:Fubini}) can be expressed as determinants of Toeplitz+Hankel matrices:
\begin{align} \label{eqn:Baik-Rains}
\begin{split}
\mathbb{E}_{ SO(2n)} \left( \prod_{j = 1}^m |p_{SO(2n)}(\theta_j)|^{2\alpha} \right) =& \frac{1}{2} D_n^{T+H,1} \left( f_m^{(\alpha)} \right), \\
\mathbb{E}_{ SO^-(2n)} \left( \prod_{j = 1}^m |p_{SO^-(2n)}(\theta_j)|^{2\alpha} \right) =& D_{n-1}^{T+H,2} \left( f_m^{(\alpha)} \right) \prod_{j = 1}^m \left( 2 \sin \theta_j \right)^{2\alpha}, \\
\mathbb{E}_{ SO(2n+1)} \left( \prod_{j = 1}^m |p_{SO(2n+1)}(\theta_j)|^{2\alpha} \right) =&  D_n^{T+H,3} \left( f_m^{(\alpha)} \right) \prod_{j = 1}^m \left( 2 \sin \frac{\theta_j}{2} \right)^{2\alpha}, \\
\mathbb{E}_{ SO^-(2n+1)} \left( \prod_{j = 1}^m |p_{SO^-(2n+1)}(\theta_j)|^{2\alpha} \right) =&  D_n^{T+H,4} \left(f_m^{(\alpha)} \right) \prod_{j = 1}^m \left( 2 \cos \frac{\theta_j}{2} \right)^{2\alpha}, \\
\mathbb{E}_{ Sp(2n)} \left( \prod_{j = 1}^m |p_{Sp(2n)}(\theta_j)|^{2\alpha} \right) =& D_n^{T+H,2} \left( f_m^{(\alpha)} \right),
\end{split}
\end{align}
where we recall Definition \ref{def:T+H}, i.e. that for a function $f \in L^1(S^1)$
\begin{align} 
\begin{split}
D_n^{T+H,1}(f) &:= \det \left( f_{j-k} + f_{j+k} \right)_{j,k = 0}^{n-1}, \\
D_n^{T+H,2}(f) &:= \det \left( f_{j-k} - f_{j+k+2} \right)_{j,k = 0}^{n-1}, \\
D_n^{T+H,3}(f) &:= \det \left( f_{j-k} - f_{j+k+1} \right)_{j,k = 0}^{n-1}, \\
D_n^{T+H,4}(f) &:= \det \left( f_{j-k} + f_{j+k+1} \right)_{j,k = 0}^{n-1}, 
\end{split}
\end{align}
where $f_j$ is the $j$-th Fourier coefficient:
\begin{align}
\begin{split}
f_j &= \frac{1}{2\pi} \int_0^{2\pi} f(e^{i\theta})e^{-ij\theta}\text{d}\theta.
\end{split}
\end{align}
On the right hand side of \eqref{eqn:Baik-Rains}, the Toeplitz+Hankel determinants account for the contribution of the complex conjugate pairs of eigenvalues of $U$, while the extra factors are contributions
from the fixed eigenvalues at $\pm 1$. 

Recall that uniform asymptotics of Toeplitz+Hankel determinants, including the case when singularities are allowed to merge, were computed in \cite[Theorem 2.2]{Claeys2022}, up to an $e^{\mathcal{O}(1)}$ factor. These results are stated here in Theorem \ref{thm:T+H Claeys}, and applied to our symbols $f_m^{(\alpha)}$ imply that uniformly over the entire region $0 < \theta_1 < \cdots < \theta_m < \pi$, as $n \rightarrow \infty$,
\begin{align} \label{eqn:uniform asymptotics}
D_n^{T+H,1}(f_m^{(\alpha)}) =& e^{O(1)} n^{m\alpha^2} F_n(\theta_1,\ldots, \theta_m)\prod_{j = 1}^m  \left( 2\sin \theta_j + \frac{1}{n} \right)^{- \alpha^2 + \alpha}, \nonumber \\
D_n^{T+H,2}(f_m^{(\alpha)}) =& e^{O(1)} n^{m\alpha^2} F_n(\theta_1,\ldots, \theta_m)\prod_{j = 1}^m \left( 2\sin \theta_j + \frac{1}{n} \right)^{- \alpha^2 - \alpha} ,\\
D_n^{T+H,3}(f_m^{(\alpha)}) =& e^{O(1)} n^{m\alpha^2} F_n(\theta_1,\ldots, \theta_m)\prod_{j = 1}^m \left( 2\sin \frac{\theta_j}{2} + \frac{1}{n} \right)^{- \alpha^2 - \alpha} \left( 2\cos \frac{\theta_j}{2} + \frac{1}{n} \right)^{- \alpha^2 + \alpha} , \nonumber \\
D_n^{T+H,4}(f_m^{(\alpha)}) =& e^{O(1)} n^{m\alpha^2} F_n(\theta_1,\ldots, \theta_m)\prod_{j = 1}^m \left( 2\sin \frac{\theta_j}{2} + \frac{1}{n} \right)^{- \alpha^2 + \alpha} \left( 2\cos \frac{\theta_j}{2} + \frac{1}{n} \right)^{- \alpha^2 - \alpha} , \nonumber
\end{align}
where 
\begin{align}
F_n(\theta_1,\ldots, \theta_m) = \prod_{1 \leq j < k \leq m } \left( 2\sin \left| \frac{\theta_j - \theta_k}{2} \right| + \frac{1}{n} \right)^{-2\alpha^2} \left( 2\sin \left| \frac{\theta_j +\theta_k}{2} \right| + \frac{1}{n} \right)^{-2\alpha^2}. 
\end{align}

Let $H(n) \in \left\{ SO(2n), \, SO^-(2n), \, SO(2n+1), \, SO^-(2n+1), \, Sp(2n) \right\}$. Combining (\ref{eqn:Fubini}), (\ref{eqn:Baik-Rains}), and (\ref{eqn:uniform asymptotics}), we see that
\begin{align} \label{eqn:MoM I}
\begin{split}
    {\rm MoM}_{H(n)}(m,\alpha) =& e^{\mathcal{O}(1)} n^{m\alpha^2} I_{H(n)}(\alpha, (0,\pi)^m), 
\end{split}
\end{align}
where for a measurable subset $R \subset (0,\pi)^m$
\begin{align} \label{eqn:I}
I_{SO(2n)}(\alpha,R) :=& \int_R F_n(\theta_1,\ldots, \theta_m) \prod_{j = 1}^m \left( 2\sin \theta_j + \frac{1}{n} \right)^{- \alpha^2 + \alpha} \frac{\text{d}\theta_1}{\pi} \cdots \frac{\text{d}\theta_m}{\pi}, \nonumber \\
I_{SO^-(2n)}(\alpha,R) :=& \int_R F_n(\theta_1,\ldots, \theta_m) \prod_{j = 1}^m \left( 2\sin \theta_j + \frac{1}{n} \right)^{- \alpha^2 - \alpha} \left( 2\sin \theta_j \right)^{2\alpha} \frac{\text{d}\theta_1}{\pi} \cdots \frac{\text{d}\theta_m}{\pi}, \nonumber \\
I_{SO(2n+1)}(\alpha,R) :=& \int_R F_n(\theta_1,\ldots, \theta_m) \\
&\hspace{-1.2cm} \times \prod_{j = 1}^m \left( 2\sin \frac{\theta_j}{2} + \frac{1}{n} \right)^{- \alpha^2 - \alpha} \left( 2\cos \frac{\theta_j}{2} + \frac{1}{n} \right)^{- \alpha^2 + \alpha} \left( 2\sin \frac{\theta_j}{2} \right)^{2\alpha} \frac{\text{d}\theta_1}{\pi} \cdots \frac{\text{d}\theta_m}{\pi}, \nonumber \\
I_{SO^-(2n+1)}(\alpha,R) :=& \int_R F_n(\theta_1,\ldots, \theta_m) \nonumber \\
&\hspace{-1.2cm}\times \prod_{j = 1}^m \left( 2\sin \frac{\theta_j}{2} + \frac{1}{n} \right)^{- \alpha^2 + \alpha} \left( 2\cos \frac{\theta_j}{2} + \frac{1}{n} \right)^{- \alpha^2 - \alpha} \left( 2\cos \frac{\theta_j}{2} \right)^{2\alpha} \frac{\text{d}\theta_1}{\pi} \cdots \frac{\text{d}\theta_m}{\pi}, \nonumber \\
I_{Sp(2n)}(\alpha,R) :=& \int_R F_n(\theta_1,\ldots, \theta_m) \prod_{j = 1}^m \left( 2\sin \theta_j + \frac{1}{n} \right)^{- \alpha^2 - \alpha} \frac{\text{d}\theta_1}{\pi} \cdots \frac{\text{d}\theta_m}{\pi}. \nonumber
\end{align}
For the proofs of the subcritical regimes in Theorem \ref{thm:MoM Sp(2n)} and Theorem \ref{thm:MoM O(n)}, we will show in Section \ref{section:subcritical} that for $R = (0,\pi)^m$ the above integrals converge as $n\to\infty$ to Selberg-type integrals which can be evaluated explicitly.

In the critical, intermediate, and supercritical regimes the integrals $I_{H(n)}(\alpha,(0,\pi)^m)$ diverge, and we need to prove optimal lower and upper bounds for them, up to an $e^{\mathcal{O}(1)}$ term. \\

To obtain lower bounds we use the inequalities
\begin{align}
    \left( 2\sin \left| \frac{\theta_j \pm \theta_k}{2} \right| + \frac{1}{n} \right)^{-2\alpha^2} \geq & \left(\frac{n}{3}\right)^{2\alpha^2}, \nonumber \\ 
    \left( 2\sin \frac{\theta_j}{2} + \frac{1}{n} \right)^{-\alpha^2 \pm \alpha} \geq \left( 2\sin \theta_j + \frac{1}{n} \right)^{-\alpha^2 \pm \alpha} \geq & \left(\frac{n}{3}\right)^{\alpha^2 \mp \alpha}, \quad \text{for } -\alpha^2 \pm  \alpha < 0, \nonumber \\
    \left( 2\sin \theta_j + \frac{1}{n} \right)^{-\alpha^2 \pm \alpha} \geq \left( 2\sin \frac{\theta_j}{2} + \frac{1}{n} \right)^{-\alpha^2 \pm \alpha} \geq & n^{\alpha^2 \mp \alpha}, \quad \text{for } -\alpha^2 \pm \alpha \geq 0,  
\end{align}
valid for $0<\theta_j,\theta_k < 1/n$, and the fact that
\begin{align}
    \int_0^{1/n} (2\sin \theta_j)^{2\alpha} \text{d}\theta_j \geq \int_0^{1/n} \left( 2\sin \frac{\theta_j}{2} \right)^{2\alpha} \text{d}\theta_j \geq (2/\pi)^{2\alpha} \int_{0}^{1/n} \theta_j^{2\alpha} \text{d}\theta_j = \frac{(2/\pi)^{2\alpha}}{2\alpha + 1} n^{-2\alpha - 1},
\end{align}
to obtain the inequalities
\begin{align} \label{eqn:supercritical lower bound}
I_{H(n)}(\alpha,(0,\pi)^m)\geq 
I_{H(n)}(\alpha,(0,1/n)^m)\geq 
cn^{2m(m-1)\alpha^2 - m(1 -\alpha^2 \pm \alpha)},
\end{align}
with $\pm = +$ for $H(n)$ one of the orthogonal ensembles, and $\pm = -$ if $H(n)=Sp(2n)$.
Together with (\ref{eqn:MoM I}) this provides us with the required lower bounds in the supercritical phases in Theorems \ref{thm:MoM Sp(2n)} and \ref{thm:MoM O(n)}, except when both $m=2$ and $H(n)\in \{ SO(2n), \, SO^-(2n),$ 
$SO(2n+1), \, SO^-(2n+1)\}$. \\

Observe also that the lower bound diverges as $n \to \infty$ if and only if 
\begin{align} 
2m(m-1)\alpha^2 - m(1 -\alpha^2 \pm \alpha) > 0 \iff \alpha > \frac{\sqrt{8m-3} \pm 1}{4m-2}.
\end{align}

To prove the lower bound $c \log n$ for $I_{H(n)}(\alpha,(0,\pi)^m)$ in the critical phase $\alpha = \frac{\sqrt{8m-3} \pm 1}{4m-2}$, we define the sets 
\begin{align}
    B_n(\ell) := (0,\ell/n)^m\setminus (0,(\ell-1)/n)^m, \quad \ell = 2,\ldots,n.
\end{align}
On $B_n(\ell)$ it holds that
\begin{align}
\begin{split}
    \left( 2 \sin \left|\frac{\theta_j \pm \theta_k}{2}\right| + \frac{1}{n} \right)^{-2\alpha^2} \geq& \left( \frac{2\ell+1}{n} \right)^{-2\alpha^2}, \quad 1 \leq j < k \leq m,
\end{split}
\end{align}
and when $-\alpha^2 \pm \alpha \leq 0$, then additionally
\begin{align}
\begin{split}
    \left( 2\sin \frac{\theta_j}{2} + \frac{1}{n} \right)^{-\alpha^2 \pm \alpha} \geq &\left( 2 \sin \theta_j + \frac{1}{n} \right)^{-\alpha^2 \pm \alpha} \geq \left( \frac{2\ell+1}{n} \right)^{-\alpha^2 \pm \alpha}, \quad 1 \leq j \leq m.
\end{split}
\end{align}
Further we see that for all $\alpha > 0$
\begin{align}
\begin{split}
    &\int_{B_n(\ell)} \prod_{j = 1}^m (2 \sin \theta_j)^{2\alpha} \text{d}\theta_j \geq \int_{B_n(\ell)} \prod_{j = 1}^m \left( 2 \sin \frac{\theta_j}{2} \right)^{2\alpha} \text{d}\theta_j \geq c \int_{B_n(\ell)} \prod_{j = 1}^m \theta_j^{2\alpha} \text{d}\theta_j  \\
    =& c'n^{-2m\alpha - m} \left( \ell^{2m\alpha + m} - (\ell-1)^{2m\alpha + m} \right) \geq c'' n^{-2m\alpha - m} \ell^{2m\alpha + m - 1},
\end{split}
\end{align}
and that for $-\alpha^2 + \alpha > 0$
\begin{align} 
\begin{split}
    &\int_{B_n(\ell)} \prod_{j = 1}^m \left(2 \sin \theta_j + \frac{1}{n} \right)^{-\alpha^2 + \alpha} \text{d}\theta_j \geq \int_{B_n(\ell)} \prod_{j = 1}^m \left(2 \sin \frac{\theta_j}{2} + \frac{1}{n} \right)^{-\alpha^2 + \alpha} \text{d}\theta_j \\
    \geq & c \int_{B_n(\ell)} \prod_{j = 1}^m \theta_j^{-\alpha^2 + \alpha} \text{d}\theta_j \\
    =& c'n^{m\alpha^2 - m\alpha - m} \left( \ell^{-m\alpha^2 + m\alpha + m}  - (\ell-1)^{-m\alpha^2 + m\alpha + m}\right) \geq c''n^{m\alpha^2 - m\alpha - m} \ell^{-m\alpha^2 + m\alpha + m - 1}. 
\end{split}
\end{align}
Thus, since $B_n(\ell)$, $\ell = 2,\ldots,n$, are disjoint, and since 
\begin{equation}
\alpha = \frac{\sqrt{8m-3} - 1}{4m-2} \iff - 2m(m-1)\alpha^2 - m\alpha^2 \pm m\alpha + m = 0, 
\end{equation} 
there exist constants $c_1, c_2>0$, independent of $n$, such that
\begin{align}
\begin{split}
    I_{H(n)} (\alpha, (0,\pi)^m) \geq & c_1\sum_{\ell = 2}^{n} \left( \ell/n \right)^{-2m(m-1)\alpha^2 -m\alpha^2 \pm m\alpha + m} \ell^{-1} \\
    =& c_1 \sum_{\ell = 2}^n \ell^{-1} > c_1 \log n - c_2. 
\end{split}
\end{align}
This provides a sharp lower bound in the critical case, except when both $m=2$ and $H(n)$ is equal to one of the orthogonal ensembles.

The upper bounds in the critical and supercritical phases of Theorem \ref{thm:MoM Sp(2n)}, and Theorem \ref{thm:MoM O(n)} in the case $m \neq 2$, are more involved: they follow from (\ref{eqn:MoM I}) and the following lemma, which will be proven in Section \ref{section:critical and supercritical}.

\begin{lemma} \label{lemma:I}
Let $\alpha > 0$. As $n \rightarrow \infty$, we have the following estimates for $m \in \mathbb{N} \setminus \{2\}$ when $H(n) \in \left\{ SO(2n), \, SO^-(2n), \, SO(2n+1), \, SO^-(2n+1), \, Sp(2n) \right\}$, and for $m \in \mathbb{N}$ when $H(n)=Sp(2n)$,
    \begin{equation}
        I_{H(n)}(\alpha,(0,\pi)^m) = \begin{cases} \mathcal{O}(\log n) & \alpha = \frac{\sqrt{8m-3} \pm 1}{4m-2}, \\
        \mathcal{O}\left(n^{2m(m-1)\alpha^2 - m(1 -\alpha^2 \pm \alpha)} \right) & \alpha > \frac{\sqrt{8m-3} \pm 1}{4m-2},
        \end{cases}
    \end{equation}
    with $\pm = -$ if $H(n)=Sp(2n)$ and $\pm = +$ otherwise.
\end{lemma}

\begin{remark}
    In the subcritical phase, i.e. when $\alpha < \frac{\sqrt{8m-3} \pm 1}{4m-2}$, the integral $I_{H(n)}(\alpha,(0,\pi)^m)$ converges to the (finite) Selberg-type integral $I_\infty^\pm(\alpha,(0,\pi)^m)$, defined in (\ref{eqn:I infty}) below. 
\end{remark}

The case where both $H(n)\in \left\{ SO(2n), \, SO^-(2n), \, SO(2n+1), \, SO^-(2n+1)\right\}$ and $m = 2$ remains. Then there are two additional phases. By integrating over
\begin{align}
    \left\{ (\theta_1,\ldots,\theta_m) \in (0,\pi)^m: \left|\theta_1 - \frac{\pi}{2}\right| < \frac{\pi}{4}, \, \, \max_{1 \leq j < k \leq m} |\theta_j - \theta_k| < \frac{1}{n} \right\}, 
\end{align}
we obtain the lower bound
\begin{align} \label{eqn:m = 2 bulk 1}
 I_{H(n)}(\alpha,(0,\pi)^m)&\geq    cn^{m(m-1)\alpha^2 + 1 - m} , 
\end{align}
which diverges if and only if 
\begin{align}
    \alpha^2 > 1/m \text{ and } m > 1 \iff & \alpha > \frac{1}{\sqrt{m}} \text{ and } m > 1.
\end{align}
Also, exactly as in Section 2.1.3 in \cite{Fahs2021}, one can show that for $\alpha = \frac{1}{\sqrt{m}}$ there exist constants $c_3,c_4$ such that
\begin{align} \label{eqn:m = 2 bulk 2}
   I_{H(n)}(\alpha,(0,\pi)^m)\geq c_3\log n - c_4.
\end{align}
For $m \geq 2$ in the case $\pm = -$, and $m \geq 3$ in the case $\pm = +$, it holds that 
\begin{align}
    \frac{1}{\sqrt{m}} >& \, \, \frac{\sqrt{8m-3} \pm 1}{4m-2}, \\
    m(m-1)\alpha^2 + 1 - m <& \, \, 2m(m-1)\alpha^2 - m(1 -\alpha^2 \pm \alpha), \quad \forall \, \, \alpha \geq \frac{1}{\sqrt{m}}, \nonumber
\end{align}
which implies that in those cases the lower bounds (\ref{eqn:m = 2 bulk 1}) and (\ref{eqn:m = 2 bulk 2}) are less sharp than the previously obtained ones in (\ref{eqn:supercritical lower bound}) and can thus be ignored. \\

However, when $m = 2$ in the orthogonal cases, it holds that $\frac{1}{\sqrt{m}} < \, \, \frac{\sqrt{8m-3} \pm 1}{4m-2}$, and the lower bounds (\ref{eqn:m = 2 bulk 1}) and (\ref{eqn:m = 2 bulk 2}) are optimal for $\alpha = \frac{1}{\sqrt{2}}$ and $\frac{1}{\sqrt{2}} < \alpha < \frac{\sqrt{5}+1}{4}$, respectively. For $\alpha = \frac{\sqrt{5}+1}{4}$, when both the lower bounds (\ref{eqn:supercritical lower bound}) and (\ref{eqn:m = 2 bulk 1}) diverge with the same power, it turns out that an extra $\log n$ term appears. The following lemma states the different phases of the asymptotics of $I_{H(n)}(\alpha,(0,\pi)^2)$ for $H(n) \in \{ SO(2n), \, SO(2n+1), \, SO^-(2n+1), \, SO^-(2n+1) \}$. It will be proven in Section \ref{section:m = 2}, and together with (\ref{eqn:MoM I}) implies Theorem \ref{thm:MoM O(n)} for $m = 2$ and the phases where $\alpha \geq \frac{1}{\sqrt{2}}$.
\begin{lemma} \label{lemma:I m = 2}
Let $\alpha \geq 1/\sqrt{2}$. For $m=2$ and $H(n)\in\{ SO(2n), \, SO^-(2n), \, SO^+(2n+1), \, SO^-(2n+1) \}$, as $n \rightarrow \infty$, it holds that
\begin{equation}
    I_{H(n)}(\alpha, (0,\pi)^m)= \begin{cases} e^{\mathcal{O}(1)} \log n & \alpha = \frac{1}{\sqrt{2}}, \\
    e^{\mathcal{O}(1)} n^{2\alpha^2 - 1}  & \alpha \in \left( \frac{1}{\sqrt{2}}, \frac{\sqrt{5} + 1}{4} \right), \\
    e^{\mathcal{O}(1)} n^{2\alpha^2 - 1} \log n & \alpha = \frac{\sqrt{5} + 1}{4}, \\
    e^{\mathcal{O}(1)} n^{6\alpha^2 - 2\alpha - 2}  & \alpha > \frac{\sqrt{5} + 1}{4}.
    \end{cases}
\end{equation}
\end{lemma}

\section{Proof of the subcritical phases} \label{section:subcritical}

When setting $1/n$ to zero in (\ref{eqn:I}), we obtain the integrals 
\begin{align} \label{eqn:I infty}
\begin{split}
I_\infty^\pm(\alpha,R) :=& \int_R \prod_{1 \leq j < k \leq m} \left|2\cos \theta_j - 2 \cos \theta_k \right|^{-2\alpha^2} \prod_{j = 1}^m \left|2\sin \theta_j \right|^{-\alpha^2 \pm \alpha} \frac{\text{d}\theta_1}{\pi} \cdots \frac{\text{d}\theta_m}{\pi},
\end{split}
\end{align}
with $\pm = +$ in all the orthogonal cases, and $\pm = -$ in the symplectic case. The integrals $I_\infty^\pm(\alpha,(0,\pi)^m)$ are finite in the case $m = 1$ if and only if $\alpha < \frac{\sqrt{5} \pm 1}{2}$, and in the case $m \geq 2$ they are finite if and only if $\alpha < \min \left\{ \frac{1}{\sqrt{m}}, \frac{\sqrt{8m-3} \pm 1}{4m-2} \right\}$. This follows by changing variables to $x_j = \frac{1}{2} + \frac{1}{2} \cos \theta_j$ to obtain a Selberg integral, and then using Theorem \ref{thm:Selberg} below: 
\begin{align}
& I_\infty^\pm(\alpha,(0,\pi)^m) \nonumber \\
=& \frac{4^{-\alpha^2m^2 \pm \alpha m}}{\pi^m} \int_{0}^{1} \cdots \int_{0}^{1} \prod_{1 \leq j < k \leq m} |x_j - x_k|^{-2\alpha^2} \prod_{j = 1}^m (x_j(1 - x_j))^{ \frac{1 -\alpha^2 \pm \alpha}{2} - 1} \text{d}x_1 \cdots \text{d}x_m \nonumber \\
=& \frac{4^{-\alpha^2m^2 \pm \alpha m}}{\pi^m} \prod_{j = 0}^{m-1} \frac{\Gamma(1 - \alpha^2 - j\alpha^2) \Gamma\left( \frac{1 - \alpha^2 \pm \alpha}{2} - j\alpha^2 \right)^2}{\Gamma(1 - \alpha^2) \Gamma \left( 1 \pm \alpha - \alpha^2(m + j) \right)}.
\end{align}

\begin{theorem}[Selberg, 1944 \cite{Selberg1944}] \label{thm:Selberg}
We have the identity
\begin{align}
\begin{split}
    &\int_{0}^1 \cdots \int_{0}^1 \prod_{1 \leq j < k \leq m} |x_j - x_k|^{2c} \prod_{j = 1}^m \left(1 - x_j \right)^{a - 1} x_j^{b - 1} \mathrm{d}x_1 \cdots \mathrm{d}x_m \\=& \prod_{j = 0}^{m-1} \frac{\Gamma(1 + c + jc) \Gamma\left( a + jc \right) \Gamma\left( b + jc \right)}{\Gamma(1 + c) \Gamma \left( a + b + c(m + j -1)\right)},
\end{split}
\end{align}
where either side is finite if and only if $\Re a, \Re b > 0$, $\Re c > - \min \{1/m, \Re a/(m-1), \Re b/(m-1) \}$.
\end{theorem}    

Thus to prove the subcritical phase in Theorems \ref{thm:MoM Sp(2n)} and \ref{thm:MoM O(n)}, in view of \eqref{eqn:C pm}, we need to show that for any $\alpha$ in the subcritical phase, and for any given $\delta > 0$, there is an integer $N$ such that for all $n > N$ 
\begin{align}
\begin{split}
\left| {\rm MoM}_{H(n)}(m,\alpha) - (2n)^{m\alpha^2} \frac{G(1+\alpha)^{2m}}{G(1+2\alpha)^m} I_\infty^\pm(\alpha,(0,\pi)^m) \right| <& \delta n^{m\alpha^2}.
\end{split}
\end{align}
To prove this we need the following lemma, which will be proven in Section \ref{section:critical and supercritical}:
\begin{lemma} \label{lemma:I infty bound}
Let $m\in\mathbb N$, $\alpha<\min \left\{ \frac{1}{\sqrt{m}}, \frac{\sqrt{8m-3} \pm 1}{4m-2} \right\}$, and \[H(n) \in \left\{ SO(2n), \, SO^-(2n), \, SO(2n+1), \, SO^-(2n+1), \, Sp(2n) \right\}.\]
There exists an $N \in \mathbb{N}$ and a constant $C > 0$, such that for all $n > N$, and any subset $R \subset (0,\pi)^m$ which is symmetric under permutation of the variables and symmetric around $\pi/2$ in each variable, it holds that 
\begin{align}
    I_{H(n)}(\alpha,R) \leq C I_\infty^{\pm}(\alpha,R),
\end{align}
where $\pm=-$ if $H(n)=Sp(2n)$ and $\pm=+$ otherwise.
\end{lemma}

For $\eta > 0$ we divide $(0,\pi)^m$ into two regions $R_1(\eta)$ and $R_2(\eta)$, where $R_1(\eta)$ is the region where $ \min \{ | 2\sin \frac{\theta_j - \theta_k}{2} |, | 2\sin \frac{\theta_j + \theta_k}{2} |, |2\sin \theta_j| \} > \eta $ for all $j \neq k$, and $R_2(\eta) = (0,\pi)^m \setminus R_1(\eta)$. Then by (\ref{eqn:Baik-Rains}), (\ref{eqn:uniform asymptotics}), (\ref{eqn:I}) and Lemma \ref{lemma:I infty bound} it follows that 
\begin{align}
\begin{split}
\int_{R_2(\eta)} \mathbb{E}_{ H(n)} \left( \prod_{j = 1}^m |p_{H(n)}(\theta_j)|^{2\alpha} \right) \frac{\text{d}\theta_1}{\pi} \cdots \frac{\text{d}\theta_m}{\pi} =& e^{\mathcal{O}(1)} n^{m\alpha^2} I_{H(n)}(\alpha,R_2(\eta)) \\
=& \mathcal O \left( n^{m\alpha^2} I^\pm_\infty(\alpha, R_2(\eta)) \right), 
\end{split}
\end{align}
as $n \rightarrow \infty$, uniformly for $0 < \eta < \pi$, with $\pm = +$ in the orthogonal cases, and $\pm = -$ in the symplectic case. Since $I_\infty^\pm(\alpha,R_2(\eta)) \rightarrow 0$ as $\eta \rightarrow 0$, it follows that for any $\delta > 0$ we can fix an $\eta_0 > 0$ and an $N_0 \in \mathbb{N}$ such that 
\begin{align} \label{eqn:int R2}
\begin{split}
\int_{R_2(\eta)} \mathbb{E}_{ H(n)} \left( \prod_{j = 1}^m |p_{H(n)}(\theta_j)|^{2\alpha} \right) \frac{\text{d}\theta_1}{\pi} \cdots \frac{\text{d}\theta_m}{\pi} < \delta n^{m\alpha^2}/2,
\end{split}
\end{align}
for all $n \geq N_0$ and $\eta < \eta_0$. \\

We now evaluate the integral of $\mathbb{E}_{ H(n)} \left( \prod_{j = 1}^m |p_{H(n)}(\theta_j)|^{2\alpha} \right)$ over $R_1(\eta)$. When all singularities $e^{i\theta_j}$, $j = 1,...,m$, of $f_m^{(\alpha)}$ are bounded away from each other and from $\pm 1$, then Theorem 1.25 in \cite{Deift2011} gives the asymptotics, including the leading order coefficient, of $D_n^{T+H,\kappa} \left( f_m^{(\alpha)} \right)$, $\kappa = 1,2,3,4$. As on $R_1(\eta)$ all singularities are bounded away from each other and from $\pm 1$, substituting those asymptotics into (\ref{eqn:Baik-Rains}) implies that
\begin{align} \label{eqn:pointwise asymptotics}
\begin{split}
&\mathbb{E}_{ H(n)} \left( \prod_{j = 1}^m |p_{H(n)}(\theta_j)|^{2\alpha} \right) = (1+o(1)) (2n)^{m\alpha^2} \frac{G(1+\alpha)^{2m}}{G(1+2\alpha)^m} \\
&\times \prod_{1 \leq j < k \leq m} \left| 2 \cos \theta_j  - 2 \cos \theta_k \right|^{-2\alpha^2} \prod_{j = 1}^m \left(2 \sin \theta_j \right)^{- \alpha^2 \pm \alpha},
\end{split}
\end{align}
as $n \rightarrow \infty$, uniformly for $(\theta_1,...,\theta_m) \in R_1(\eta)$. Combining (\ref{eqn:I infty}) and (\ref{eqn:pointwise asymptotics}) we see that
\begin{align} 
\begin{split}
&\int_{R_1(\eta)} \mathbb{E}_{ H(n)} \left( \prod_{j = 1}^m |p_{H(n)}(\theta_j)|^{2\alpha} \right) \frac{\text{d}\theta_1}{\pi} \cdots \frac{\text{d}\theta_m}{\pi} \\
=&(1+o(1)) (2n)^{m\alpha^2} \frac{G(1+\alpha)^{2m}}{G(1+2\alpha)^m} I_\infty^\pm(\alpha,R_1(\eta)) \\
=& (1+o(1))(2n)^{m\alpha^2} \frac{G(1+\alpha)^{2m}}{G(1+2\alpha)^m} \left( I_\infty^\pm(\alpha,(0,\pi)^m) - I_\infty^\pm(\alpha,R_2(\eta)) \right),
\end{split}
\end{align}
where the $o(1)$ term tends to zero for any fixed $\eta > 0$, as $n \rightarrow \infty$. This implies that 
\begin{align}
&\left| {\rm MoM}_{H(n)}(m,\alpha) - (2n)^{m\alpha^2} \frac{G(1+\alpha)^{2m}}{G(1+2\alpha)^m} I_\infty^\pm(\alpha,(0,\pi)^m) \right| \nonumber \\
=&\left| \int_{(0,\pi)^m} \mathbb{E}_{ H(n)} \left( \prod_{j = 1}^m |p_{H(n)}(\theta_j)|^{2\alpha} \right) \frac{\text{d}\theta_1}{\pi} \cdots \frac{\text{d}\theta_m}{\pi} - (2n)^{m\alpha^2} \frac{G(1+\alpha)^{2m}}{G(1+2\alpha)^m} I_\infty^\pm(\alpha,(0,\pi)^m) \right| \nonumber\\
\leq &  \left| \int_{R_1(\eta)} \mathbb{E}_{ H(n)} \left( \prod_{j = 1}^m |p_{H(n)}(\theta_j)|^{2\alpha} \right) \frac{\text{d}\theta_1}{\pi} \cdots \frac{\text{d}\theta_m}{\pi} - (2n)^{m\alpha^2} \frac{G(1+\alpha)^{2m}}{G(1+2\alpha)^m} I_\infty^\pm(\alpha,(0,\pi)^m) \right| \nonumber\\
&+\int_{R_2(\eta)} \mathbb{E}_{ H(n)} \left( \prod_{j = 1}^m |p_{H(n)}(\theta_j)|^{2\alpha} \right) \frac{\text{d}\theta_1}{\pi} \cdots \frac{\text{d}\theta_m}{\pi} \\
\leq & (2n^{m\alpha^2}) \frac{G(1+\alpha)^{2m}}{G(1+2\alpha)^m} \left( o(1) I_\infty^\pm(\alpha,(0,\pi)^m) + (1 + o(1) I_\infty^\pm(\alpha,R_2(\eta)) \right) \nonumber \\
&+\int_{R_2(\eta)} \mathbb{E}_{ H(n)} \left( \prod_{j = 1}^m |p_{H(n)}(\theta_j)|^{2\alpha} \right) \frac{\text{d}\theta_1}{\pi} \cdots \frac{\text{d}\theta_m}{\pi}.\nonumber
\end{align}
Combining this with (\ref{eqn:int R2}), the fact that $I_\infty^\pm(\alpha,R_2(\eta)) \rightarrow 0$ when $\eta \rightarrow 0$, and the fact that $o(1) \rightarrow 0$ as $n \rightarrow \infty$ for any fixed $\eta > 0$, we see that for any given $\delta > 0$ we can fix an $\eta < \eta_0$ and an $N \geq N_0$, such that 
\begin{align}
    &\left| {\rm MoM}_{H(n)}(m,\alpha) - (2n)^{m\alpha^2} \frac{G(1+\alpha)^{2m}}{G(1+2\alpha)^m} I_\infty^\pm(\alpha,(0,\pi)^m) \right| < \delta n^{m\alpha^2},
\end{align}
for all $n \geq N$. This finishes the proof of the subcritical phases in Theorems \ref{thm:MoM Sp(2n)} and \ref{thm:MoM O(n)}.

\section{Proof of Lemma \ref{lemma:I} and Lemma \ref{lemma:I infty bound}} \label{section:critical and supercritical}

We first see, with $\pm = +$ in all the orthogonal cases, and $\pm = -$ in the symplectic case, that
\begin{align} \label{eqn:I pm}
& I_{H(n)}(\alpha, R) \leq \int_{R} \prod_{1 \leq j < k \leq m} \left( \left|2\sin \frac{\theta_j - \theta_k}{2} \right| + \frac{1}{n} \right)^{-2\alpha^2} \left( \left| 2\sin \frac{\theta_j + \theta_k}{2} \right| + \frac{1}{n} \right)^{-2\alpha^2} \nonumber \\
&\times \prod_{j = 1}^m \left( \left| 2\sin \theta_j \right| + \frac{1}{n} \right)^{-\alpha^2 \pm \alpha} \frac{\text{d}\theta_1}{\pi} \cdots \frac{\text{d}\theta_m}{\pi} \nonumber \\
=& \int_{R} \prod_{1 \leq j < k \leq m} \left( 2|\cos \theta_j - \cos \theta_k| + \frac{1}{n} \left|2\sin \frac{\theta_j - \theta_k}{2} \right| + \frac{1}{n} \left| 2\sin \frac{\theta_j + \theta_k}{2} \right| + \frac{1}{n^2} \right)^{-2\alpha^2} \nonumber \\
&\times \prod_{j = 1}^m \left( \left| 2\sin \theta_j \right| + \frac{1}{n} \right)^{-\alpha^2 \pm \alpha} \frac{\text{d}\theta_1}{\pi} \cdots \frac{\text{d}\theta_m}{\pi} \\
\leq & C \int_{R} \prod_{1 \leq j < k \leq m} \left( |\cos \theta_j - \cos \theta_k| + \frac{1}{n^2} \right)^{-2\alpha^2} \prod_{j = 1}^m \left( |\sin \theta_j| + \frac{1}{n} \right)^{- \alpha^2 \pm \alpha} \text{d}\theta_1 \cdots \text{d}\theta_m, \nonumber
\end{align}
for a constant $C$ which is independent of $n$. Making the variable transformation $\cos \theta_j = t_j$ it follows that
\begin{align} \label{eqn:I hat}
    I_{H(n)}(\alpha, (0,\pi)^m) \leq C I_n^\pm(m,\alpha),
\end{align}
where
\begin{equation}
I_n^\pm(m,\alpha) 
    := \int_{(-1,1)^m} \prod_{1 \leq j < k \leq m} \left( |t_j - t_k| + \frac{1}{n^2} \right)^{-2\alpha^2} \prod_{j = 1}^m \left( \sqrt{1 - t_j^2} + \frac{1}{n} \right)^{- \alpha^2 \pm \alpha} \frac{\text{d}t_j}{\sqrt{1 - t_j^2}} .
\end{equation} 

We bound $I_n^\pm(m,\alpha)$ by the following simpler integral:
\begin{lemma} \label{lemma:I to J}
Let $\alpha>0$ and $m\in\mathbb N$. There exists $C>0$ such that for all $n \in \mathbb{N}$ it holds that
\begin{align}
    I_n^\pm(m,\alpha) \leq C J_n^\pm(m,\alpha),
\end{align}
where 
\begin{align}
    J_n^\pm(m,\alpha) = \int_{[0,1)^{m}} \prod_{1 \leq j < k \leq m} \left( |t_j - t_k| + \frac{1}{n^2} \right)^{-2\alpha^2} \prod_{j = 1}^m \left( \sqrt{t_j} + \frac{1}{n} \right)^{- \alpha^2 \pm \alpha} \frac{\text{d}t_j}{\sqrt{t_j}}.
\end{align}
\end{lemma}

\noindent \textbf{Proof:} Due to symmetry of the integrand in the $t_j$'s we see that
\begin{align}
\begin{split}
    I_n^\pm(m, \alpha) =& \sum_{\ell = 0}^m \binom{m}{\ell} I_n^\pm(m,\alpha,\ell) ,
\end{split}
\end{align}
where for $\ell \in \{0,\ldots,m\}$ 
\begin{align}
&I_n^\pm(m,\alpha,\ell) \\
:=& \int_{(-1,0]^\ell \times [0,1)^{m-\ell}} \prod_{1 \leq j < k \leq m} \left( |t_j - t_k| + \frac{1}{n^2} \right)^{-2\alpha^2} \prod_{j = 1}^m \left( \sqrt{1 - t_j^2} + \frac{1}{n} \right)^{- \alpha^2 \pm \alpha} \frac{\text{d}t_j}{\sqrt{1 - t_j^2}}. \nonumber
\end{align}
By setting $s_j = - t_j$ for $1 \leq j \leq \ell$ and $s_j = t_j$ for $\ell + 1 \leq j \leq m$ we see that  
\begin{align}
\begin{split}
    &I_n^\pm(m,\alpha,\ell) \\
    =& \int_{[0,1)^{m}} \prod_{1 \leq j < k \leq \ell} \left( |s_j - s_k| + \frac{1}{n^2} \right)^{-2\alpha^2} \prod_{1 \leq j \leq \ell < k \leq m} \left( |s_j + s_k| + \frac{1}{n^2} \right)^{-2\alpha^2} \\
    &\times \prod_{\ell + 1 \leq j < k \leq m} \left( |s_j - s_k| + \frac{1}{n^2} \right)^{-2\alpha^2} \prod_{j = 1}^m \left( \sqrt{1 - s_j^2} + \frac{1}{n} \right)^{- \alpha^2 \pm \alpha} \frac{\text{d}s_j}{\sqrt{1 - s_j^2}}  \\
    \leq & \int_{[0,1)^{m}} \prod_{1 \leq j < k \leq m} \left( |s_j - s_k| + \frac{1}{n^2} \right)^{-2\alpha^2} \prod_{j = 1}^m \left( \sqrt{1 - s_j^2} + \frac{1}{n} \right)^{- \alpha^2 \pm \alpha} \frac{\text{d}s_j}{\sqrt{1 - s_j^2}}  \\
    =& I_n^\pm(m,\alpha, 0).
\end{split}
\end{align}
Thus 
\begin{align} 
    I_n^\pm(m,\alpha) \leq 2^m I_n^\pm(m,\alpha,0).
\end{align}
Since $ 1 \leq \sqrt{1 + t_j} \leq \sqrt{2}$ for $t_j \in [0,1]$, and due to symmetry, it holds that
\begin{align}
\begin{split} 
    &I_n^\pm(m,\alpha,0) \\
    &\leq  C\int_{[0,1)^{m}} \prod_{1 \leq j < k \leq m} \left( |t_j - t_k| + \frac{1}{n^2} \right)^{-2\alpha^2} \prod_{j = 1}^m \left( \sqrt{1 - t_j} + \frac{1}{n} \right)^{- \alpha^2 \pm \alpha} \frac{\text{d}t_j}{\sqrt{1 - t_j}}, \\
    &= C\int_{[0,1)^{m}} \prod_{1 \leq j < k \leq m} \left( |t_j - t_k| + \frac{1}{n^2} \right)^{-2\alpha^2} \prod_{j = 1}^m \left( \sqrt{t_j} + \frac{1}{n} \right)^{- \alpha^2 \pm \alpha} \frac{\text{d}t_j}{\sqrt{t_j}}.
\end{split}
\end{align}
This finishes the proof. \qed \\

We now combine the factors $\left( \sqrt{t_j} + \frac{1}{n} \right)^{- \alpha^2 \pm \alpha}$ and $t_j^{-1/2}$:
\begin{lemma} \label{lemma:J to J l}
Let $\alpha>0$ and $m\in\mathbb N$. There exists $C>0$ such that for all $n \in \mathbb{N}$ it holds that
\begin{align}\label{eq:Jnpm}
     J_n^\pm(m,\alpha) \leq C \sum_{\ell = 0}^m n^{2(m-\ell)(m-\ell-1)\alpha^2 - (m-\ell)(1-\alpha^2 \pm \alpha)} J_n^\pm(m,\alpha,\ell),
\end{align}
where for $\ell = 1,\ldots,m$
\begin{align}
\begin{split}
    J_n^\pm(m,\alpha,\ell) :=& \int_{\left[ \frac{1}{n^2},1 \right]^\ell} \prod_{1 \leq j < k \leq \ell} \left( |t_j - t_k| + \frac{1}{n^2} \right)^{-2\alpha^2} \prod_{j = 1}^\ell t_j^{\frac{-1-(4(m-\ell)+1) \alpha^2 \pm \alpha}{2}} \text{d}t_{1} \cdots \text{d}t_\ell,
\end{split}
\end{align}
and $J_n^\pm(m,\alpha,0) := 1$.
\end{lemma}

\noindent \textbf{Proof:} We observe that
\begin{align}\label{eq:Jn1}
\begin{split}
    J_n^\pm(m,\alpha) = \sum_{\ell = 0}^m \binom{m}{l} \int_{\left[ \frac{1}{n^2},1 \right]^\ell \times \left[0,\frac{1}{n^2} \right]^{m-\ell}} \prod_{1 \leq j < k \leq m} \left( |t_j - t_k| + \frac{1}{n^2} \right)^{-2\alpha^2} \\
    \hspace{3cm} \times \prod_{j = 1}^m \left(\sqrt{t_j} + \frac{1}{n} \right)^{-\alpha^2 \pm \alpha} t_j^{-1/2} \text{d}t_j.
\end{split}
\end{align}
The integral on the right can be rewritten as 
\begin{align}\label{integral}
\begin{split}
    \int_{\left[\frac{1}{n^2}, 1\right]^\ell \times \left[0, \frac{1}{n^2}\right]^{m-\ell}} \prod_{1 \leq j < k \leq \ell} \left( |t_j - t_k| + \frac{1}{n^2} \right)^{-2\alpha^2} \prod_{1 \leq j \leq \ell < k \leq m} \left( |t_j - t_k| + \frac{1}{n^2} \right)^{-2\alpha^2} \\
    \times \prod_{\ell + 1 \leq j < k \leq m} \left( |t_j - t_k| + \frac{1}{n^2} \right)^{-2\alpha^2} \prod_{j = 1}^m \left(\sqrt{t_j}+\frac{1}{n}\right)^{-\alpha^2 \pm \alpha} t_j^{-1/2} \text{d}t_j. 
\end{split}
\end{align}
Now using the estimates
\begin{align}
\begin{split}
    \left( |t_j - t_k| + \frac{1}{n^2} \right)^{-2\alpha^2} \leq n^{4\alpha^2} &\text{ for } (t_j, t_k) \in \left[0,\frac{1}{n^2}\right] \times \left[0,\frac{1}{n^2}\right], \\
    \left(|t_j - t_k| + \frac{1}{n^2}\right)^{-2\alpha^2} \leq t_j^{-2\alpha^2} &\text{ for  } (t_j, t_k) \in \left[\frac{1}{n^2}, 1\right] \times \left[0,\frac{1}{n^2}\right], \\
    \left(\sqrt{t_j} + \frac{1}{n} \right)^{-\alpha^2 \pm \alpha} \leq n^{\alpha^2 \mp \alpha} &\text{ for } t_j \in \left[0,\frac{1}{n^2}\right] \text{ and } -\alpha^2 \pm \alpha \leq 0, \\ 
    \left(\sqrt{t_j} + \frac{1}{n} \right)^{-\alpha^2 \pm \alpha} \leq \left( \frac{n}{2} \right)^{\alpha^2 \mp \alpha} &\text{ for } t_j \in \left[0,\frac{1}{n^2}\right] \text{ and } -\alpha^2 \pm \alpha > 0, \\ 
    \left(\sqrt{t_j} + \frac{1}{n} \right)^{-\alpha^2 \pm \alpha} \leq t_j^{\frac{-\alpha^2 \pm \alpha}{2}} &\text{ for } t_j \in \left[\frac{1}{n^2},1\right] \text{ and } -\alpha^2 \pm \alpha \leq 0, \\ 
    \left( \sqrt{t_j} + \frac{1}{n} \right)^{-\alpha^2 \pm \alpha} \leq \left( 4 t_j \right)^{\frac{-\alpha^2 \pm \alpha}{2}} &\text{ for } t_j \in \left[\frac{1}{n^2},1\right] \text{ and } -\alpha^2 \pm \alpha > 0,
\end{split}
\end{align}
and the identity $\int_0^{\frac{1}{n^2}} t^{-1/2} \text{d}t = \frac{2}{n}$, we see that \eqref{integral} is bounded by
\begin{align}
\begin{split}
    C n^{2(m-\ell)(m-\ell-1)\alpha^2 - (m-\ell)(1 - \alpha^2 \pm \alpha)} \int_{\left[\frac{1}{n^2},1\right]^\ell} \prod_{1 \leq j < k \leq \ell} \left( |t_j - t_k| + \frac{1}{n^2} \right)^{-2\alpha^2} \\
    \times \prod_{j = 1}^\ell t_j^{\frac{-1-\alpha^2 \pm \alpha}{2}-2(m-\ell)\alpha^2} \text{d}t_j,
\end{split}
\end{align} 
for a suitably chosen $C>0$. Substituting this in \eqref{eq:Jn1}, we obtain the result. \qed \\

Now we are able to prove Lemma \ref{lemma:I infty bound}, which is needed to complete the proof in Section \ref{section:subcritical}, of the results in the subcritical phase:\\

\noindent \textbf{Proof of Lemma \ref{lemma:I infty bound}:} We see that $J_n^{\pm}(m,\alpha,\ell) \leq J_\infty^{\pm}(m,\alpha,\ell)$, where $J_\infty^{\pm}(m,\alpha,\ell)$ denotes the integrals one obtains when setting $1/n$ to zero in the integration ranges and integrands of $J_n^{\pm}(m,\alpha,\ell)$, $\ell = 0,\ldots,m$:
\begin{align}
    J_\infty^\pm(m,\alpha,\ell) := \int_{\left[0,1 \right]^\ell} \prod_{1 \leq j < k \leq \ell}  |t_j - t_k|^{-2\alpha^2} \prod_{j = 1}^\ell t_j^{\frac{-1-(4(m-\ell)+1) \alpha^2 \pm \alpha}{2}} \text{d}t_j.
\end{align}
$J_\infty^{\pm}(m,\alpha,\ell)$ is a Selberg integral and is finite by Theorem \ref{thm:Selberg} for all $\alpha$ in the subcritical phase $\alpha < \min \left\{ \frac{1}{\sqrt{m}}, \frac{\sqrt{8m-3} \pm 1}{4m-2} \right\}$ and for all $\ell = 0,\ldots,m$. Moreover, in the subcritical phase, the summands in \eqref{eq:Jnpm} contain $n$ with a negative power for $\ell = 0,\ldots,m-1$ and with power zero for $\ell = m$. Thus, by (\ref{eqn:I hat}), Lemma \ref{lemma:I to J}, and Lemma \ref{lemma:J to J l}, we see that there exists a constant $C$ such that for all $n \geq N$
\begin{align}
    I_{H(n)}(\alpha,(0,\pi)^m) \leq CJ^\pm_{\infty}(m,\alpha,m) = C2^m I^\pm_\infty(\alpha,(0,\pi)^m).
\end{align}
We can repeat \textit{mutatis mutandis} those estimates and arguments for subsets $R \subset (0,\pi)^m$ that are symmetric under permutation of the variables and symmetric around $\pi/2$ in each variable.
Transforming $R$ appropriately, i.e. splitting up the integration range or changing variables, we then obtain that there exists a constant $C$ such that for all $n \in \mathbb{N}$
\begin{align}
    I_{H(n)}(\alpha,R) \leq C2^m I^\pm_\infty(\alpha,R).
\end{align}
\qed  \\

By changing variables to $x_j = t_j^{-1}n^{-2}$ in the integrals $J_n^\pm(m,\alpha,\ell)$ from Lemma \ref{lemma:J to J l} we see that 
\begin{align} 
    J_n^\pm(m,\alpha) \leq Cn^{2m(m-1)\alpha^2 - m(1 - \alpha^2 \pm \alpha)} \sum_{\ell = 0}^m K_n^\pm(m,\alpha,\ell), 
\end{align}
where for $\ell = 1,\ldots,m$
\begin{align} \label{eqn:K}
\begin{split}
&K_n^\pm(m,\alpha,\ell) 
:= \int_{1/n^2}^1 \cdots \int_{1/n^2}^1 \prod_{1 \leq j < k \leq \ell} \left( \left| x_j - x_k \right| + x_jx_k \right)^{-2\alpha^2} \prod_{j = 1}^\ell x_j^{((4m - 3)\alpha^2 \mp \alpha - 3)/2} \text{d}x_j,
\end{split}
\end{align}
and $K_n^\pm(m,\alpha,0) := 1$. We see, since $(|x_j - x_k| + x_jx_k)^{-2\alpha^2} \geq 1$ for $x_j,x_k \in [0,1]$, that for $\ell = 1,...,m-1$
\begin{align}
    K_n(m,\alpha,m) \geq K_n(m,\alpha,\ell) \int_{1/n^2}^1 \cdots \int_{1/n^2}^1 \prod_{j = \ell + 1}^m x_j^{((4m - 3)\alpha^2 \mp \alpha - 3)/2} \text{d}x_{j} \geq C K_n(m,\alpha,\ell). 
\end{align}
Thus we see that 
\begin{align} \label{eqn:J to K}
    J_n^\pm(m,\alpha) \leq Cn^{2m(m-1)\alpha^2 - m(1 - \alpha^2 \pm \alpha)} K_n^\pm(m,\alpha,m).
\end{align}
The following lemma, combined with (\ref{eqn:I pm}), (\ref{eqn:I hat}), Lemma \ref{lemma:I to J}, Lemma \ref{lemma:J to J l}, and (\ref{eqn:J to K}), will complete the proof of Lemma \ref{lemma:I}.
\begin{lemma} \label{lemma:K}
Let $m \in \mathbb N \setminus\{2\}$ and $H(n) \in \left\{ SO(2n), \, SO^-(2n), \, SO(2n+1), \, SO^-(2n+1), \, Sp(2n) \right\}$, or let $m=2$ and $H(n)=Sp(2n)$.
As $n \rightarrow \infty$, with $\pm = -$ if $H(n) = Sp(2n)$ and $\pm=+$ otherwise:
\begin{align}
\begin{split}
    K_n^\pm (m,\alpha,m) = \begin{cases} \mathcal{O}(\log n) & \alpha = \frac{\sqrt{8m-3} \pm 1}{4m-2},\\
    \mathcal{O}(1) & \alpha > \frac{\sqrt{8m-3} \pm 1}{4m-2}.
    \end{cases}
\end{split}
\end{align}
\end{lemma}

\noindent \textbf{Proof of Lemma \ref{lemma:K}:} For $m = 1$ the proof is immediate, thus we let $m \geq 2$ for $\pm = +$ and $m \geq 3$ for $\pm = -$. For $\alpha = \frac{\sqrt{8m-3} \pm 1}{4m - 2}$ we see that
\begin{align}
    &K_n^\pm(m,\alpha,m) \nonumber \\
    \leq & m \int_{1/n^2}^1 \int_{1/n^2}^{x_m} \cdots \int_{1/n^2}^{x_m} \prod_{1 \leq j < k \leq m} \left| x_j - x_k \right|^{-2\alpha^2} \prod_{j = 1}^m x_j^{((4m - 3)\alpha^2 \mp \alpha - 3)/2} \text{d}x_1 \cdots \text{d}x_m \nonumber \\
    =& m\int_{1/n^2}^1 \int_{1/(n^2x_m)}^{1} \cdots \int_{1/(n^2x_m)}^{1} x_m^{m\frac{(4m-3)\alpha^2 \mp \alpha - 3}{2}} \prod_{1 \leq j < k \leq m-1} \left| x_mt_j - x_mt_k \right|^{-2\alpha^2} \nonumber \\
    &\hspace{3cm} \times \prod_{j = 1}^{m-1} |x_m - x_mt_j|^{-2\alpha^2} t_j^{((4m - 3)\alpha^2 \mp \alpha - 3)/2} \text{d}t_1 \cdots \text{d}t_{m-1} x_m^{m-1} \text{d}x_m \nonumber \\
    \leq & m\int_{1/n^2}^1 x_m^{m(m-1)\alpha^2 + m\frac{\alpha^2 \mp \alpha - 1}{2} - 1} \text{d}x_m \\
    &\times \int_{0}^1 \cdots \int_{0}^1 \prod_{1 \leq j < k \leq m-1} \left| t_j - t_k \right|^{-2\alpha^2} \prod_{j = 1}^{m-1} |1 - t_j|^{-2\alpha^2} t_j^{((4m - 3)\alpha^2 \mp \alpha - 3)/2} \text{d}t_1 \cdots \text{d}t_{m-1}, \nonumber
\end{align}
where we set $x_j = t_jx_m$ for $j = 2,...,m$. At the critical value the exponent in the first integral equals $-1$, thus the first integral exactly equals $2\log n$. The second integral is a Selberg integral which is finite if and only if $\alpha < 1/\sqrt{m-1}$, $-2\alpha^2 + 1 > 0$, $((4m - 3)\alpha^2 \mp \alpha - 1)/2 > 0$, $(m-2)\alpha^2 < -2\alpha^2 + 1$ and $(m-2)\alpha^2 < ((4m - 3)\alpha^2 \mp \alpha - 1)/2 > 0$. It is easy to check that all those conditions are fulfilled for $\alpha = \frac{\sqrt{8m-3} \pm 1}{4m - 2}$, which proves the lemma for the critical value. \\

For $\frac{\sqrt{8m-3} \pm 1}{4m-2} < \alpha < 1/\sqrt{m}$ we see that 
\begin{align}
    K_n^\pm(m,\alpha,m) \leq \int_{0}^1 \cdots \int_{0}^1 \prod_{1 \leq j < k \leq m} \left| x_j - x_k \right|^{-2\alpha^2} \prod_{j = 1}^m x_j^{((4m - 3)\alpha^2 \mp \alpha - 3)/2} \text{d}x_j.
\end{align}
The right-hand side is a Selberg integral, which by Theorem \ref{thm:Selberg} is finite exactly when $\frac{\sqrt{8m-3} \pm 1}{4m-2} < \alpha < 1/\sqrt{m}$. \\

When $\alpha \geq 1/\sqrt{m}$, then we see that for any $\epsilon > 0$ it holds that
\begin{align}
    \prod_{1 \leq j < k \leq m} \left( \left| x_j - x_k \right| + x_jx_k \right)^{-2\alpha^2} \leq \prod_{1 \leq j < k \leq m} |x_j - x_k|^{-\frac{2}{m} + \epsilon} \prod_{j = 1}^m x_j^{(m-1)(-2\alpha^2 + \frac{2}{m} - \epsilon)},
\end{align}
and thus 
\begin{align}
    K_n^\pm(m,\alpha,m) \leq \int_0^1 \cdots \int_0^1 \prod_{1 \leq j < k \leq m} |x_j - x_k|^{-\frac{2}{m} + \epsilon} \prod_{j = 1}^m x_j^{2\frac{m-1}{m} - (m-1)\epsilon + \frac{\alpha^2 \mp \alpha - 1}{2} - 1} \text{d}x_j. 
\end{align}
The right-hand side is again a Selberg integral which by Theorem \ref{thm:Selberg} is finite if and only if $\epsilon > 0$, $2\frac{m-1}{m} - (m-1)\epsilon + \frac{\alpha^2 \mp \alpha - 1}{2} > 0$, and 
\begin{align}
\begin{split}
    \frac{1}{m} - \frac{\epsilon}{2} <& \frac{2}{m} - \epsilon + \frac{\alpha^2 \mp \alpha - 1}{2(m-1)} \\
    \iff \frac{\epsilon}{2} <& \frac{1}{m} + \frac{\alpha^2 \mp \alpha - 1}{2(m-1)}.
\end{split}
\end{align}
For $\epsilon < \frac{2}{m}$ the third condition implies the second. For $\alpha > 0$ it holds that $\alpha^2 - \alpha \geq -1/4$ and $\alpha^2 + \alpha > 0$, thus we see that (except for $\pm = +$ and $m = 2$) the third condition holds for $\epsilon$ small enough:
\begin{align}
\begin{split}
    \frac{1}{m} + \frac{\alpha^2 - \alpha - 1}{2(m-1)} \geq & \frac{1}{m} - \frac{5}{8(m-1)} = \frac{3m - 8}{8m(m-1)} > 0, \\
    \frac{1}{m} + \frac{\alpha^2 + \alpha - 1}{2(m-1)} > & \frac{1}{m} - \frac{1}{2(m-1)} = \frac{m - 2}{2m(m-1)} \geq 0. 
\end{split}
\end{align}
This finishes the proof of the supercritical phase. \qed 

\section{Proof of Lemma \ref{lemma:I m = 2}} \label{section:m = 2}
Let $H(n) \in \{ SO(2n), \, SO^-(2n), \, SO(2n+1), \, SO^-(2n+1) \}$. We split up the integration range of $I_{H(n)}(\alpha,(0,\pi)^2)$ in \eqref{eqn:I} into $(0,\pi/2)^2$, $(0,\pi/2) \times (\pi/2, \pi)$, $(\pi/2, \pi) \times (0, \pi/2)$, and $(\pi/2, \pi)^2$. We see that (where $SO^+(n) := SO(n)$)
\begin{align} \label{eqn:ranges}
\begin{split}
    I_{SO^{\pm}(2n)}(\alpha,(0,\pi/2)^2) =& I_{SO^{\pm}(2n)}(\alpha,(\pi/2,\pi)^2),\\
    I_{SO^{\pm}(2n)}(\alpha,(0,\pi/2) \times (\pi/2, \pi)) =& I_{SO^{\pm}(2n)}(\alpha, (\pi/2,\pi) \times (0, \pi/2)),\\
    I_{SO^{\pm}(2n+1)}(\alpha,(0,\pi/2)^2) =& I_{SO^{\mp}(2n+1)}(\alpha,(\pi/2,\pi)^2),\\
    I_{SO^{\pm}(2n+1)}(\alpha,(0,\pi/2) \times (\pi/2, \pi)) =& I_{SO^{\mp}(2n+1)}(\alpha, (\pi/2,\pi) \times (0, \pi/2)).\\
\end{split}
\end{align}
Thus it suffices to prove Lemma \ref{lemma:I m = 2} for each of the eight integrals on the left-hand sides of (\ref{eqn:ranges}). We will only prove it for $I_{SO(2n)}(\alpha,(0,\pi/2)^2)$ and $I_{SO(2n)}(\alpha,(\pi/2, \pi)^2)$, as for the other six integrals in (\ref{eqn:ranges}) the proof is essentially the same. \\

We use that $2\theta/\pi \leq \sin \theta \leq \theta$ for $0 < \theta < \pi /2$ to obtain 
\begin{align} \label{eqn:theta bound}
    &I_{SO(2n)}(\alpha,(0,\pi/2)^2) \nonumber \\
    \leq & C\int_{0}^{\pi/2} \int_{\theta_2}^{\pi/2} \left( \theta_1^2 - \theta_2^2 + 2 \theta_1/n + \frac{1}{n^2} \right)^{-2\alpha^2} \left(\theta_1 + \frac{1}{n} \right)^{-\alpha^2 + \alpha} \left(\theta_2 + \frac{1}{n} \right)^{-\alpha^2 + \alpha} \text{d}\theta_1 \text{d}\theta_2 \nonumber \\
    \leq & C' I_{SO(2n)}(\alpha,(0,\pi/2)^2). 
\end{align} 
We see that for all $\alpha > 0$
\begin{align}
    \left( \theta_1^2 - \theta_2^2 + 2 \theta_1/n + \frac{1}{n^2} \right)^{-2\alpha^2} \leq n^{4\alpha^2}, 
\end{align} 
and when additionally $-\alpha^2 + \alpha + 1 < 0$ then we can bound the integral in the middle of (\ref{eqn:theta bound}) by $\mathcal{O}(n^{6\alpha^2 - 2\alpha - 2})$. Since we get a lower bound of the same power by (\ref{eqn:supercritical lower bound}) those upper and lower bounds are optimal. \\

For $-\alpha^2 + \alpha + 1 \geq 0$ we set $\theta_j = s_j/n$ and find that the integral in the middle of (\ref{eqn:theta bound}) is equal to
\begin{align}     
    Cn^{6\alpha^2 - 2 \alpha - 2} \int_0^{n\pi/2} \int_{s_2}^{n\pi/2} \left( s_1^2 - s_2^2 + 2s_1 + 1\right)^{-2\alpha^2} (s_1 + 1)^{-\alpha^2 + \alpha} (s_2 + 1)^{-\alpha^2 + \alpha} \text{d}s_1 \text{d}s_2. 
\end{align}
After another change of variables $s_1=s$, $s_2=st$, we see that this is equal to
\begin{align}
    &Cn^{6\alpha^2 - 2 \alpha - 2} \int_0^{n\pi/2} \int_{0}^{1} \left( s^2(1 - t^2) + 2st + 1\right)^{-2\alpha^2} (s + 1)^{-\alpha^2 + \alpha} (st + 1)^{-\alpha^2 + \alpha} s \text{d}t \text{d}s \nonumber \\
    = & Cn^{6\alpha^2 - 2 \alpha - 2} \int_0^2 \int_{0}^{1} \left( s^2(1 - t^2) + 2st + 1\right)^{-2\alpha^2} (s + 1)^{-\alpha^2 + \alpha} (st + 1)^{-\alpha^2 + \alpha} s \text{d}t \text{d}s \\
    & + Cn^{6\alpha^2 - 2 \alpha - 2} \int_2^{n\pi/2} \int_{0}^{1/2} \left( s^2(1 - t^2) + 2st + 1\right)^{-2\alpha^2} (s + 1)^{-\alpha^2 + \alpha} (st + 1)^{-\alpha^2 + \alpha} s \text{d}t \text{d}s \nonumber \\
    & + Cn^{6\alpha^2 - 2 \alpha - 2} \int_2^{n\pi/2} \int_{1/2}^{1} \left( s^2(1 - t^2) + 2st + 1\right)^{-2\alpha^2} (s + 1)^{-\alpha^2 + \alpha} (st + 1)^{-\alpha^2 + \alpha} s \text{d}t \text{d}s. \nonumber
\end{align}
We then make the following observations. 
\begin{itemize}
\item The first integral on the right-hand side is finite and non-zero.
\item The second integral is bounded above and below by
    \begin{align}
        &\int_2^{n\pi/2} s^{-5\alpha^2 + \alpha + 1} \int_0^{1/2} (st + 1)^{-\alpha^2 + \alpha} \text{d}t \text{d}s,
    \end{align}
    multiplied by suitable constants. It holds that
    \begin{align}
        \int_0^{1/2} (st + 1)^{-\alpha^2 + \alpha} \text{d}t \leq \begin{cases} C' s^{-\alpha^2 + \alpha}, & -\alpha^2 + \alpha + 1 < 0, \\ C' s^{-\alpha^2 + \alpha} \log s, & -\alpha^2 + \alpha + 1 = 0, \end{cases}
    \end{align}
    and it is easy to check that these bounds are optimal for $s \geq 2$, in the sense that the left-hand side is bounded below by $c(1_{-\alpha^2+\alpha + 1 < 0} s^{-\alpha^2 + \alpha} + 1_{-\alpha^2+\alpha + 1 = 0} s^{-\alpha^2 + \alpha}\log n)$ for some $c>0$. Thus for large $n$, the second integral is bounded by 
    \begin{align}
    \begin{split}
        &\int_2^{n\pi/2} s^{-6\alpha^2 + 2\alpha + 1} (1_{-\alpha^2 + \alpha + 1 < 0} + 1_{-\alpha^2 + \alpha + 1 = 0} \log s) \text{d}s\\ 
        =& \mathcal{O}(1) + \mathcal{O}\left( n^{-6\alpha^2 + 2\alpha + 2} \right) + 1_{- 6 \alpha^2 + 2\alpha + 2 = 0} \mathcal{O}(\log n)
    \end{split}
    \end{align}
    as $n\to\infty$,
    and it is straightforward to see that this bound is optimal.
\item  Since $s \geq 2$ and $t \in [1/2,1]$ in its integration range, the third integral is bounded above and below by 
    \begin{align}
    \begin{split}
        & \int_2^{n\pi/2} s^{-2\alpha^2 + 2\alpha + 1} \int_{1/2}^{1} \left( s^2(1 - t) + st \right)^{-2\alpha^2} \text{d}t \text{d}s,
    \end{split}
    \end{align}
    multiplied by suitable constants. For $\alpha > 1/\sqrt{2}$ we see that
    \begin{align}
    \begin{split}
        \int_{1/2}^{1} \left( s^2(1 - t) + st\right)^{-2\alpha^2} \text{d}t =& \left[ \left( s^2(1 - t) + st\right)^{-2\alpha^2 + 1} \frac{(s - s^2)^{-1}}{1 - 2\alpha^2} \right]_{1/2}^1 \\
        =& \frac{(s^2 - s)^{-1}}{2\alpha^2 - 1} \left( s^{-2\alpha^2 + 1} - (s^2/2 + s/2)^{-2\alpha^2 + 1} \right) \\
        \leq & C' s^{-2\alpha^2 - 1},
    \end{split}
    \end{align}
    while for $\alpha = \frac{1}{\sqrt{2}}$ we see that
    \begin{align}
    \begin{split}
        \int_{1/2}^{1} \left( s^2(1 - t) + st\right)^{-2\alpha^2} \text{d}t  =& (s^2 - s)^{-1} \log \frac{s + 1}{2} \\
        \leq & C'' s^{-2} \log s,
    \end{split}
    \end{align}
    and we easily obtain lower bounds of the same order. Thus for large $n$, we get the following optimal bound for the third integral
    \begin{align}
        &\int_2^{n\pi/2} s^{-4\alpha^2 + 2\alpha} (1 + 1_{2\alpha^2 = 1}\log s) \text{d} s \\
        =& \mathcal{O}(1) + \mathcal{O} \left( n^{1 - 4\alpha^2 + 2\alpha} \right) + 1_{2\alpha^2 - 1 = 0} \mathcal{O} \left( n^{1 - 4\alpha^2 + 2\alpha} \log n \right)+ 1_{1 - 4 \alpha^2 + 2\alpha = 0} \mathcal{O}(\log n). \nonumber
    \end{align} 
\end{itemize}
Further we see that for large $n$
\begin{align}
    &I_{SO(2n)}(\alpha,(0,\pi/2) \times (\pi/2, \pi)) \nonumber \\
    \leq & C\int_{0}^{\pi/2} \int_{\pi/2}^{\pi} \left( \theta_1^2 - \theta_2^2 + 2 \theta_1/n + \frac{1}{n^2} \right)^{-2\alpha^2} \left(\pi - \theta_1 + \frac{1}{n} \right)^{-\alpha^2 + \alpha} \left(\theta_2 + \frac{1}{n} \right)^{-\alpha^2 + \alpha} \text{d}\theta_1 \text{d}\theta_2 \nonumber \\
    \leq & C' + C' \int_{\pi/4}^{\pi/2} \int_{\pi/2}^{3\pi/4} \left( \theta_1^2 - \theta_2^2 + 2 \theta_1/n + \frac{1}{n^2} \right)^{-2\alpha^2} \text{d}\theta_1 \text{d}\theta_2 \\
    \leq & C'' + C'' \int_{\pi/4}^{\pi/2} \int_{\pi/2}^{3\pi/4} \left( \theta_1 - \theta_2 + 1/n \right)^{-2\alpha^2} \text{d}\theta_1 \text{d}\theta_2 \nonumber \\
    =& \mathcal{O}(1) + \mathcal{O}(n^{2\alpha^2 - 1}) + 1_{1 - 2 \alpha^2 = 0} \mathcal{O}(\log n), \nonumber
\end{align}
and one can easily see that this bound is optimal as well.\\

Putting all the obtained bounds together we see that
\begin{align}
    &I_{SO(2n)}(\alpha,(0,\pi)^2) \nonumber \\
    =& \mathcal{O}(1) + 1_{1 - 2 \alpha^2 = 0} \mathcal{O}(\log n) +  \mathcal{O}(n^{2\alpha^2 - 1}) + 1_{1 - 4 \alpha^2 + 2\alpha = 0} \mathcal{O}(n^{2\alpha^2 - 1} \log n) + \mathcal{O}(n^{6\alpha^2 - 2\alpha - 2}) \nonumber \\
    =& \begin{cases} \mathcal{O}(1) & \alpha < \frac{1}{\sqrt{2}}, \\
        \mathcal{O}(\log n) & \alpha = \frac{1}{\sqrt{2}}, \\
        \mathcal{O}(n^{2\alpha^2 - 1})  & \alpha \in \left( \frac{1}{\sqrt{2}}, \frac{\sqrt{5} + 1}{4} \right), \\
        \mathcal{O}(n^{2\alpha^2 - 1}) \log n & \alpha = \frac{\sqrt{5} + 1}{4}, \\
        \mathcal{O}(n^{6\alpha^2 - 2\alpha - 2})  & \alpha > \frac{\sqrt{5} + 1}{4},
        \end{cases}
\end{align}
and since those bounds are optimal in the sense that we get lower bounds of the same order, this finishes the proof.

%\section*{Acknowledgements}
%JF and JK were supported by ERC Advanced Grant 740900 (LogCorRM). 
%TC was supported by the Fonds de la Recherche Scientifique-FNRS under EOS project O013018F.
%We are grateful to Brian Conrey for very helpful discussions at the American Institute of Mathematics. 
\chapter{The Characteristic Polynomial of Unitary Brownian Motion and the Gaussian Free Field} \label{chapter:UBM}

In this chapter we prove Theorem \ref{thm:FS}, which is the natural dynamical analogue to Theorem \ref{thm:HughesKeatingOConnell1} by Hughes-Keating-O'Connell, i.e. we prove that the convergence of the real and imaginary parts of the logarithm of the characteristic polynomial of unitary Brownian motion toward Gaussian free fields on the cylinder, as the matrix dimension goes to infinity, holds in certain suitable Sobolev spaces\footnote{Several Sobolev spaces are involved because we can improve the regularity with respect to one of the variables at the cost of sacrificing some regularity with respect to the other variable.}, whose regularity we prove to be optimal. Our result is also related to the work of Spohn \cite{Spohn1998}, from which the identification of the above limit as the Gaussian free field first followed, albeit in a different function space. \\

In the course of this research we also proved a Wick-type identity, which we include in the last section of this chapter, as it might be of independent interest. The identity allows us to express the second moment of the trace of arbitrary products of a GUE matrix $H$ and an independent Haar-distributed unitary matrix $U$ in terms of moments of traces of $U$ only. When the dimension $n$ is large enough, the Diaconis-Shahshahani (and Diaconis-Evans) theorem on moments of traces of unitary matrices \cite{Diaconis1994}, Theorem \ref{thm:DiaconisShahshahani} in this thesis, allows us to then compute this new expression explicitly as a polynomial in $n$. This chapter is based on joint work with Isao Sauzedde \cite{FS22}.

\section{Context and Statement of Results}
We let $(U_t)_{t \in \mathbb{R}}$ be a unitary Brownian motion (see Section \ref{section:CCG} for a definition) at its equilibrium measure, which is Haar measure, and recall the definition of its characteristic polynomial:
\begin{align}
p_{U(n)}(t,\theta): = \text{det}\left(I-e^{-i\theta}U_t\right) = \prod_{k=1}^n (1-e^{i(\theta_k(t)-\theta)}), \quad (t,\theta) \in \mathbb{R} \times [0,2\pi),
\end{align}
where $\theta_1,...,\theta_n(t)$ denote the eigenangles of $U_t$. We define its logarithm by
\begin{align}
\begin{split}
\log p_{U(n)}(t,\theta):=& \sum_{k = 1}^n \log (1-e^{i(\theta_k(t) -\theta)}),
\end{split}
\end{align}
with the branches on the RHS being the principal branches, such that
\begin{align}
\Im \log (1-e^{i(\theta_k(t) - \theta)}) \in \left( - \frac{\pi}{2}, \frac{\pi}{2} \right],
\end{align}
with $\Im \log 0 := \pi/2$. Since
\begin{align}
\log (1 - z) = - \sum_{k = 1}^\infty \frac{z^k}{k}
\end{align}
for $|z| \leq 1$, where for $z = 1$ both sides equal $-\infty$, and by the identity $\log \det = \Tr \log$, we see that the Fourier expansion of $\log p_{U(n)}$ w.r.t. the spatial variable $\theta$ is given as follows: 
\begin{align}
\begin{split}
\log p_{U(n)}(t,\theta) =& - \sum_{k = 1}^{\infty} \frac{\Tr (U_t^k)}{k} e^{-ik\theta}.
\end{split}
\end{align}
By Theorem \ref{thm:DiaconisShahshahani} it follows that for any $l \in \mathbb{N}$ and fixed $t \in \mathbb{R}$, the random variables $k^{-1}\Tr (U_t^{k})$, $k = 1,...,l$, converge in distribution to independent complex Gaussians $A_k$, $k = 1,...,l$, whose real and imaginary parts are independent centered real Gaussians with variance $1/(2k)$. By also showing tightness in the negative Sobolev spaces $H^{-\epsilon}_0(S^1)$ (see Section \ref{section:Sobolev} for a definition), Hughes, Keating and O'Connell proved that for any fixed $t \in \mathbb{R}$, $\log p_{U(n)}(t, \cdot)$ converges to a generalized Gaussian field as $n \rightarrow \infty$. Their result was already stated as Theorem \ref{thm:HughesKeatingOConnell1} in Chapter \ref{chapter:Intro}, but for the convenience of the reader we restate it here:

\begin{theorem} [Hughes, Keating, O'Connell \cite{Hughes2005}] \label{thm:HughesKeatingOConnell}
    For any $\epsilon > 0$, the sequence of pairs of fields $\left( \Re \log p_{U(n)}(\cdot), \Im \log p_{U(n)}(\cdot) \right)_{n \in \mathbb{N}}$ converges in distribution in $H_0^{-\epsilon}(S^1) \times H_0^{-\epsilon}(S^1)$ to the pair of generalized Gaussian fields $(\Re Z, \Im Z)$, where  
    \begin{align}
        Z(\theta) =& \sum_{k = 1}^\infty A_k e^{-ik\theta},
    \end{align}
    with $A_k$ being complex Gaussians whose real and imaginary parts are independent centered Gaussians with variance $1/(2k)$.
\end{theorem}

In the dynamic case, i.e. when considering $\log p_{U(n)}$ also as a function of $t$, the natural candidate for the limit of $\log_{U(n)}(\cdot, \cdot)$ is the infinite-dimensional Orstein-Uhlenbeck process
\begin{align} \label{eqn:Z}
    Z(t,\theta) := \sum_{k = 1}^\infty A_k(t) e^{-ik\theta},
\end{align}
where $A_k(\cdot)$, $k \in \mathbb{N}$, are independent complex Ornstein-Uhlenbeck processes at their stationary distribution, i.e. (up to a linear time change) solutions to the SDEs
\begin{align}
    \label{eq:SDEA}
    \text{d}A_k(t) = - kA_k(t) \text{d}t + \text{d}\left(W_k(t) + i\tilde{W}_k(t) \right),
\end{align}
with the real and imaginary parts of $A_k(t)$ at any fixed time $t$ being independent real Gaussians with variance $1/(2k)$, and $W_k$, $\tilde{W}_k$, $k \in \mathbb{N}$, denoting real standard Brownian motions. The identification of $Z$ as the limit of $\log p_{U(n)}$, and in particular the weak convergence of the dynamic Fourier coefficients $t \mapsto k^{-1} \Tr (U_t^k)$, $k \in \mathbb{N}$, to the independent Ornstein-Uhlenbeck processes $A_k$, $k \in \mathbb{N}$, follows from the work of Spohn \cite{Spohn1998}, as we now explain.\\

For fixed $\epsilon > 0$, and real-valued functions $f \in H_0^{3/2 + \epsilon}(S^1, \mathbb{R})$, Spohn considered the linear statistics 
\begin{align} \label{eq:Spohn xi_n}
    \xi_n(t,f): = \sum_{j=1}^n f(e^{i\theta_j(t)}) = \sum_{j = 1}^n \sum_{k \in \mathbb{Z} \setminus \{0\}} f_k e^{ik\theta_j(t)} = \sum_{k \in \mathbb{Z} \setminus \{0\}} f_k \Tr (U_t^k),
\end{align}
of the eigenvalues $e^{i\theta_1(t)},...,e^{i\theta_n(t)}$ of unitary Brownian motion (in fact he more generally considered interacting particles on the unit circle with different repulsion strengths). Since $H_0^{-3/2 - \epsilon}(S^1, \mathbb{R})$ is the dual space of $H_0^{3/2 + \epsilon}(S^1, \mathbb{R})$, one can consider $\xi_n$ as a random variable in $\mathcal{C}(\mathbb{R}, H^{-3/2-\epsilon}(S^1, \mathbb{R} )):= \{f:\mathbb{R} \rightarrow H^{-3/2-\epsilon}(S^1, \mathbb{R}): \, f \text{ is continuous}\}$, and (with an abuse of notation) represent it as a time-dependent formal Fourier series, and see that
\begin{align}
\begin{split}
    \xi_n(t,\theta) = \sum_{k \in \mathbb{Z} \setminus \{0\}} \Tr (U_t^{-k}) e^{ik\theta} =& - i \partial_\theta \left( \sum_{k \in \mathbb{Z}\setminus \{0\}} \frac{\Tr (U_t^{-k})}{k} e^{ik\theta} \right) = - 2  \partial_{\theta} \Im  \log p_{U(n)}(t,\theta) \\ 
    =& \sqrt{-\partial_\theta^2} \left( \sum_{k \in \mathbb{Z}\setminus \{0\}} \frac{\Tr (U_t^{-k})}{|k|} e^{ik\theta} \right) = 2 \sqrt{-\partial_\theta^2} \Re \log p_{U(n)}(t,\theta),
\end{split}
\end{align}
where $\sqrt{-\partial_\theta^2}$ denotes the operator that multiplies the $k$-th Fourier coefficient by $|k|$. Spohn's result, reformulated to our setting, is then as follows: 

\begin{theorem}[Spohn \cite{Spohn1998}] \label{thm:Spohn}
    Let $\epsilon>0$ and endow $\mathcal{C}(\mathbb{R}, H^{-3/2-\epsilon}(S^1, \mathbb{R}))$ with the topology of locally uniform convergence. Then, as $n \rightarrow \infty$, the random variables $\xi_n$ converge in distribution to the infinite-dimensional Ornstein-Uhlenbeck process (where $A_{-k} = \overline{A_k}$ for $k \in \mathbb{N}$ and $\sgn(k) = 1_{k > 0} - 1_{k < 0}$) %\todo{does a similar statement follow for the real part, using that proof technique, even though it cannot be written as $\sum f(z_j)$?}
\begin{align} \label{eq:xi}
\begin{split}
    \xi(t,\theta) := \sum_{k \in \mathbb{Z} \setminus \{0\}} |k|A_{-k}(t) e^{ik\theta} =& - i \partial_\theta \left( \sum_{k \in \mathbb{Z}\setminus \{0\}}^\infty \sgn (k) A_{-k}(t) e^{ik\theta} \right) = - 2 \partial_\theta \Im Z(t,\theta) \\
    =& \sqrt{-\partial_\theta^2} \left( \sum_{k \in \mathbb{Z}\setminus \{0\}} A_{-k}(t) e^{ik\theta} \right) = 2 \sqrt{-\partial_\theta^2} \Re Z(t,\theta).
\end{split}
\end{align}

\end{theorem}
\begin{remark} \label{remark:SPDE}
    Spohn defined $\xi$ as a Gaussian process indexed by $(t,f) \in \mathbb{R} \times H^{3/2 + \epsilon}(S^1,\mathbb{R})$ (as in (\ref{eq:Spohn xi_n})), i.e. $\xi(t,f) = \langle \xi(t,\cdot), f \rangle$, which satisfies a certain stochastic partial differential equation. To see that the two definitions are equivalent, one can either use Gaussianity and do a covariance calculation, see Section 2.2 in \cite{Bourgade2022}, or see that the SDEs for $\xi(t,e^{ik \cdot}/k) := \langle \xi(t,\cdot), e^{ik \cdot }/k \rangle = \langle \xi(t, \cdot) , \cos (k \cdot)/k) \rangle + i \langle \xi (t, \cdot), \sin (k \cdot )/k \rangle$, $k \in \mathbb{Z} \setminus \{0\}$, that follow from Spohn's representation, are exactly the SDEs (\ref{eq:SDEA}) that define $A_k$. \\
    In particular, Spohn's result implies that for any $l \in \mathbb{N}$
    \begin{align}
        \left( \xi_n(\cdot,e^{ik\cdot}/k) \right)_{k = 1}^l = \left( |k|^{-1} \Tr(U^{-k}_{(\cdot)}) \right)_{k = 1}^l \implies \left( \xi(\cdot,e^{ik\cdot}/k) \right)_{k = 1}^l = \left( A_{k}(\cdot) \right)_{k = 1}^l,
    \end{align}
    weakly as random variables in $\mathcal{C}(\mathbb{R},\mathbb{C}^l)$, i.e. it implies convergence of the dynamic Fourier coefficients of $\log p_{U(n)}$ to those of $Z$.
\end{remark}

We now recall Theorem \ref{thm:FS}, which on the one hand can be seen as the natural dynamic analogue of Theorem \ref{thm:HughesKeatingOConnell}, and on the other hand as an analogue to Spohn's result (see below Remark \ref{remark:covariance}): 

\begin{theorem} \label{thm:UBM}
For any $s \in [0,\frac{1}{2})$, $\epsilon > s$, and $T>0$, the sequence of random fields $ \left( \log p_{U(n)} (\cdot,\cdot) \right)_{n \in \mathbb{N}}$ converges in distribution in the tensor product of Hilbert spaces $H^{s}([0,T]) \otimes  H^{-\epsilon}_0(S^1)$ (see Section \ref{section:Sobolev} for a definition) to the generalized Gaussian field $Z(\cdot,\cdot)$ in (\ref{eqn:Z}). \\
Furthermore, those regularity parameters $s$ and $-\epsilon$ are optimal, in the sense that for $s = 1/2$ or $s \geq \epsilon \geq 0$ the sequence $ \left( \log p_{U(n)} (\cdot,\cdot) \right)_{n \in \mathbb{N}}$ almost surely does not converge and $Z(\cdot,\cdot)$ is almost surely not an element of  $H^{s}([0,T]) \otimes  H^{-\epsilon}_0(S^1)$. 
\end{theorem}

\begin{remark} \label{remark:covariance}
A straightforward calculation shows that the covariance functions of $\Re Z$ and $\Im Z$ are given by
\begin{align}
\begin{split}
    & \mathbb{E}(\Re Z(t,\theta), \Re Z (t',\theta')) = \mathbb{E}(\Im Z(t,\theta), \Im Z (t',\theta'))  \\ 
    =& \sum_{k = 1}^\infty \frac{e^{-k|t - t'|}}{2k} \cos (k(\theta - \theta')) =
    \frac{1}{2} \log \frac{\max \{e^{-t}, e^{-t'}\}}{|e^{-t}e^{i\theta} - e^{-t'}e^{i\theta'}|}.
\end{split}
\end{align}
This implies that $\Re Z$ and $\Im Z$ are Gaussian free fields on the infinite cylinder $\mathbb{R} \times \mathbb{R} / 2\pi \mathbb{Z}$, since 
\begin{align}
    \frac{1}{2} \log \frac{\max \{e^{-t}, e^{-t'}\}}{|e^{-t}e^{i\theta} - e^{-t'}e^{i\theta'}|} = \pi (-\Delta)^{-1}(t,\theta,t',\theta'),
\end{align}
where $(-\Delta)^{-1}$ denotes the Green's function associated to the Laplacian $\Delta = \partial_t^2 + \partial_\theta^2$, see \cite[Section 2.2]{Bourgade2022}.
\end{remark}

%\begin{remark} \label{remark:fdd}
%    As Spohn's result already implies convergence of the finite dimensional distributions of the stochastic process $(t,k) \mapsto k^{-1} \Tr (U_t^{-k})$ to those of the stochastic process $(t,k) \mapsto A_k(t)$, proving Theorem \ref{thm:UBM} boils down to showing tightness of $\log p_{U(n)}$ in $H^{s}([0,T]) \otimes  H^{-\epsilon}_0(S^1)$.
%\end{remark}

To compare our Theorem \ref{thm:UBM} to Spohn's Theorem \ref{thm:Spohn}, observe that since by the latter it holds that $\sqrt{-\partial_\theta^2} \Re Z(t,\theta), \, \partial_\theta \Im Z \in \mathcal{C}([0,T], H_0^{-3/2 - \epsilon}(S^1, \mathbb{R}))$, for $\epsilon > 0$, it follows that $\Re Z, \, \Im Z \in \mathcal{C}([0,T], H_0^{- \epsilon}(S^1, \mathbb{R}))$ for $\epsilon > 1/2$. Further note that the inclusions %\todo{interpolation}
\begin{align}
\begin{split}
    H^{s}([0,T]) \otimes H_0^{- \epsilon}(S^1, \mathbb{R} ) \subset & \mathcal{C}([0,T], H_0^{- \epsilon}(S^1, \mathbb{R})) , \\
    H_0^{\epsilon}(S^1, \mathbb{R}) \subset & \mathcal{C}(S^1,\mathbb{R}),
\end{split}
\end{align}
are valid if and only if $s > 1/2$, and $\epsilon > 1/2$, respectively. Thus Spohn's Theorem \ref{thm:Spohn} implies that $Z$ can be considered as a continuous function in $t$, but then it is of regularity $-\epsilon < -1/2$ in $\theta$, and can thus only be integrated against functions of regularity $\epsilon > 1/2$, which are in particular continuous. Our Theorem \ref{thm:UBM} treats the case where $Z$ is considered to be of regularity $s < 1/2$ in $t$, in which it is not continuous in $t$ anymore, and shows that then $Z$ is of regularity $-\epsilon > -s > -1/2$ in $\theta$, and can thus be integrated against functions of regularity $\epsilon < s < 1/2$, which includes discontinuous functions.\\

Recently Bourgade and Falconet proved the first multi-time generalization of Fisher-Hartwig asymptotics \cite{Bourgade2022}[Theorem 1.2] which allowed them to show that for $\gamma \in (0,2\sqrt{2})$ the random measures on $\mathbb{R} \times \mathbb{R} / 2\pi \mathbb{Z}$
\begin{align}
    \frac{|p_n(t,\theta)|^\gamma}{\mathbb{E}\left( |p_n(t,\theta)|^\gamma \right)} \text{d}t \text{d}\theta,
\end{align}
converge weakly to a Gaussian multiplicative chaos measure, which can be formally written as
\begin{align}
    \frac{e^{\gamma \Re Z(t,\theta)}}{\mathbb{E}\left( e^{\gamma \Re Z(t,\theta)} \right)} \text{d}t \text{d}\theta.
\end{align}
Their dynamic Fisher-Hartwig asymptotics are related to this chapter since they in particular imply that the exponential moments of the linear statistics
\begin{align}
    \xi_n(t,f) = \langle \xi_n(t, \cdot), f \rangle,
\end{align}
converge to the exponential moments of $\xi(t,f) := \langle \xi(t, \cdot), f \rangle$, jointly for any finite number of $f$ in a certain class of functions, and times $t \in \mathbb{R}$, thus providing another proof for the convergence of the dynamic Fourier coefficients of $\log p_{U(n)}$ to those of $Z$. The specific result from Bourgade and Falconet that we use in our proof of Theorem \ref{thm:UBM} is the following \cite[Corollary 3.2]{Bourgade2022}:

\begin{corollary}[Bourgade, Falconet] \label{thm:multi time}
    Let $(e^{i\theta_1(t)},...,e^{i\theta_n(t)})_{t \geq 0}$ denote the eigenvalue process of unitary Brownian motion, started at Haar measure. For real functions $f,g \in H_0^{1/2}(S^1, \mathbb{R})$, $n \in \mathbb{N}$ and $t \geq 0$, it holds that 
    \begin{align}
        &\mathbb{E} \Big[ \xi_n(0,f) \xi_n(t,g) \Big]
        = \hspace{-0.2cm}\sum_{|k| \leq n - 1} \hspace{-0.2cm} f_k g_{-k} \sgn (k) e^{-|k|t} \frac{\sinh (\frac{k^2t}{n})}{\sinh (\frac{kt}{n})} + \sum_{|k| \geq n} f_k g_{-k} e^{-\frac{k^2t}{n}} \frac{\sinh (kt)}{\sinh (\frac{kt}{n})}.
    \end{align}
\end{corollary}

\section{Sobolov spaces and their Tensor Product} \label{section:Sobolev}

Consider the space of square integrable $\mathbb{C}$-valued functions on the unit circle, with vanishing mean:
\begin{equation}
    L^2_0(S^1)= \left\{ f(\theta) = \sum_{k \in \mathbb{Z}} f_k e^{ik\theta} : \sum_{k \in \mathbb{Z}} |f_k|^2 < \infty, f_0=0 \right\}.
\end{equation}
For $s \geq 0$, we define $H^s_0(S^1)$ as the restriction of $L^2_0(S^1)$ w.r.t. the functions for which the inner product
\begin{equation}
    \langle f, g \rangle_s = \sum_{k \in \mathbb{Z}} |k|^{2s} f_k\overline{g_k}
\end{equation}
is finite. For $s <0$, we define $H^s_0(S^1)$ as the completion of $L^2_0(S^1)$ w.r.t. the scalar product $\langle \cdot , \cdot \rangle_s$. Thus we see that for all $s \in \mathbb{R}$ it holds that
\begin{align}
    H_0^s(S^1) = \left\{ f(\theta) = \sum_{k \in \mathbb{Z}} f_k e^{ik\theta} : \sum_{k \in \mathbb{Z}} k^{2s} |f_k|^2 < \infty, f_0=0 \right\},
\end{align}
i.e. $H^s(S^1)$ can be seen as a space of (formal) Fourier series. Note that $\left(H_0^s(S^1), \langle \cdot , \cdot \rangle_s \right)$ is a Hilbert space for all $s \in \mathbb{R}$. For $s \geq 0$ it is a subspace of $H^0_0(S^1) = L_0^2(S^1)$, i.e. the space of square-integrable functions with zero mean, while for $s < 0$, $H_0^s(S^1)$ can be interpreted as the dual space of $H_0^{-s}(S^1)$, i.e. as a space of generalized functions defined up to an additive constant. We denote by $H_0^s(S^1,\mathbb{R})$ the subset of $H_0^s(S^1)$ consisting of real-valued (generalized) functions, i.e. those for which $f_k = \overline{f_k}$.\\

Fix $T>0$. For $s = 0$ we set $H^s([0,T]) := L^2([0,T])$, and for $s \in (0,1)$, we define the fractional Sobolev space $H^s([0,T])$ as the subspace of $L^2([0,T])$, where the Slobodeckij inner product 
\begin{equation}
(f, g)_s := \int_0^T f(t) \overline{g(t)} \text{d}t + \int_0^T\int_0^T \frac{(f(t)-f(u))(\overline{g(t)-g(u)} )}{|t-u|^{1+2s}} \text{d}t \text{d}u
\end{equation}
is finite. Note that $(H^s([0,T]),(\cdot, \cdot)_s)$ is a Hilbert space for all $s \geq 0$.

\begin{remark}
For the fact that the fractional Sobolev spaces defined through Fourier series or through the Slobodeckij norm agree, the reader can consult e.g. \cite{DiNezza2012}.
\end{remark}

For $s \geq 0$ and $\epsilon > 0$ we let $H^s([0,T]) \otimes H^{-\epsilon}_0(S^1)$ denote the tensor product of Hilbert spaces $H^s([0,T])$ and $H^{-\epsilon}_0(S^1)$. Since the inner product on that space is determined by 
\begin{align}
\begin{split}
&\langle f\otimes g, h\otimes k\rangle_{s,-\epsilon} \\
&=
(f,h)_s \langle g, k\rangle_{-\epsilon}\\
&= \int_0^T f(t)\overline{h(t)} \text{d}t \langle  g, k \rangle_{-\epsilon}
+ 1_{s \neq 0} \int_0^T\int_0^T \frac{ (f(t) -f(u) )  \overline{(h(t) -h(u) )} }{|t-u|^{1+2s}} \text{d}t \text{d}u \langle  g, k \rangle_{-\epsilon}, \\
&= \int_0^T \langle f(t) g, h(t)k \rangle_{-\epsilon} \text{d}t
+ 1_{s \neq 0} \int_0^T\int_0^T \frac{\langle (f(t) -f(u) )g , (h(t) -h(u) )k \rangle_{-\epsilon}}{|t-u|^{1+2s}} \text{d}t \text{d}u,
\end{split}
\end{align}
we obtain that
\begin{align}
\begin{split}
    \langle F,G\rangle_{s,-\epsilon} =& \int_0^T \langle F(t,\cdot),G(t,\cdot)\rangle_{-\epsilon} \text{d}t \\
    &+  1_{s \neq 0} \int_0^T\int_0^T \frac{\langle F(t,\cdot)-F(u,\cdot), G(t,\cdot)-G(u,\cdot)\rangle_{-\epsilon}}{|t-u|^{1+2s}} \text{d}t \text{d}u,
\end{split}
\end{align}

first when $F$ and $G$ are linear combinations of pure tensor products, and then for all $F,G\in H^s([0,T])\otimes H^{-\epsilon}_0(S^1)$ by density and continuity. 

\section{Proof of the main result Theorem \ref{thm:UBM}}

The strategy for proving convergence is as in the stationary case in \cite{Hughes2005}: we treat $\left( \log p_{U(n)} \right)_{n \in \mathbb{N}}$ as a sequence of random variables in $H^{s}([0,T])\otimes  H^{-\epsilon}_0(S^1) $, and first establish that if any of its subsequences has a limit then that limit has to be $Z$. We do this by showing that the finite-dimensional distributions of $\log p_{U(n)}$ converge to those of $Z$, i.e. that 
\begin{align} \label{eqn:fdd convergence}
    \left( k_j^{-1} \text{Tr}(U_{t_j}^{k_j}) \right)_{j = 1}^l \implies \left( A_{k_j}(t_j) \right)_{j = 1}^l,
\end{align} 
for any $l \in \mathbb{N}$, and $(t_j,k_j) \in [0,T] \times \mathbb{N}$, $j = 1,...,l$. This follows immediately from Spohn's result, here Theorem \ref{thm:Spohn}, as it implies in particular that (see Remark \ref{remark:SPDE})
\begin{align} \label{eqn:fdd convergence2}
    \left( k_j^{-1} \text{Tr}(U_{(\cdot)}^{k_j}) \right)_{j = 1}^l \implies \left( A_{k_j}(\cdot) \right)_{j = 1}^l,    
\end{align}
in $C(\mathbb{R},\mathbb{C}^l)$, for any $l \in \mathbb{N}$, and $k_j \in \mathbb{N}$, $j = 1,...,l$. Applying the continuous mapping theorem to (\ref{eqn:fdd convergence2}), with the appropriate function from $C(\mathbb{R},\mathbb{C}^l)$ to $\mathbb{C}^l$, then implies (\ref{eqn:fdd convergence}). \\

We then show that the set $\left( \log p_{U(n)} \right)_{n \in \mathbb{N}}$ is tight in $H^s([0,T])\otimes  H^{-\epsilon}_0(S^1)$ for $s \in [0,1/2)$ and $\epsilon > s$. Since $H^{s}([0,T])\otimes  H^{-\epsilon}_0(S^1) $ is complete and separable, Prokhorov's theorem implies that the closure of $\left(\log p_{U(n)}\right)_{n \in \mathbb{N}}$ is sequentially compact w.r.t. the topology of weak convergence. In particular this means that every subsequence of $\left( \log p_{U(n)} \right)_{n \in \mathbb{N}}$ has a weak limit in $H^{s}([0,T])\otimes  H^{-\epsilon}_0(S^1) $. Since any such limit has to be $Z$ it follows that the whole sequence $\left(\log p_{U(n)} \right)_{n \in \mathbb{N}}$, must converge weakly to $Z$. This implies in particular that $Z$ is almost surely an element of $H^{s}([0,T])\otimes  H^{-\epsilon}_0(S^1)$ for $s \in [0,1/2)$ and $\epsilon > s$. \\
We will conclude by proving optimality of the regularity parameters $s$ and $-\epsilon$, by showing in Lemma \ref{lemma:optimal} that when $s = 1/2$ or $s \geq \epsilon \geq 0$ the field $Z$ is almost surely not an element of $H^{s}([0,T]) \otimes  H^{-\epsilon}_0(S^1)$. Thus for $s = 1/2$ or $s = \epsilon > 0$ the sequence $\left( \log p_{U(n)} \right)_{n \in \mathbb{N}}$ almost surely cannot converge (and cannot be tight) in $H^{s}([0,T])\otimes  H^{-\epsilon}_0(S^1)$ as its limit would have to be $Z$.\\

We now show tightness of $\left( \log p_{U(n)} \right)_{n \in \mathbb{N}}$ in $H^s([0,T]) \otimes H^{-\epsilon}_0(S^1)$, i.e. that for every $\delta > 0$ we construct a compact $K_\delta \subset H^s([0,T]) \otimes H^{-\epsilon}_0(S^1)$ for which
\begin{equation}
\sup_{n \in \mathbb{N}} \mathbb{P} \left( \log p_{U(n)} \notin K_\delta \right) < \delta.
\end{equation}
We let $0 \leq s' < \epsilon'$ be such that $0 \leq s < s' < \epsilon' < \epsilon$, and choose
\begin{align}
    K_\delta = \left\{ F \in H^{s}([0,T])\otimes H_0^{-\epsilon}(S^1): ||F||^2_{s',-\epsilon'} \leq C_\delta \right\},
\end{align}
for a $C_\delta$ depending on $\delta$. By Lemma \ref{lemma:compact inclusion} below we see that $K_\delta$ is compact in $H^{s}([0,T])\otimes H_0^{-\epsilon}(S^1)$, and by Lemma \ref{lemma:boundedness} below we see that $\sup_{n \in \mathbb{N}} \mathbb{E} \left( ||\log p_{U(n)}||^2_{s',-\epsilon'} \right) < \infty$. Thus, by choosing $C_\delta$ large enough and then using Markov's inequality, we see that
\begin{align}
\begin{split}
    \sup_{n \in \mathbb{N}} \mathbb{P} \left( \left( \log p_{U(n)} \right) \notin K_\delta \right) =& \sup_{n \in \mathbb{N}} \mathbb{P} \left( || \log p_{U(n)} ||_{s,-\epsilon}^2 > C_\delta \right) \\
    \leq& \frac{\sup_{n \in \mathbb{N}} \mathbb{E} \left( ||\log p_{U(n)}||^2_{s',-\epsilon'} \right)}{C_\delta} \\
    <& \delta,
\end{split}
\end{align}
which shows tightness of $\log p_{U(n)}$.
\begin{lemma}\label{lemma:compact inclusion}
    Let $0 \leq s < s' < \epsilon' < \epsilon$. Then, the inclusion of $H^{s'}([0,T])\otimes H_0^{-\epsilon'}(S^1)$ into  $H^{s}([0,T])\otimes H_0^{-\epsilon}(S^1)$ is compact.
\end{lemma}
\begin{proof}
    From the Kondrachov embedding theorem, the inclusion $\iota_1$ of $H^{s'}([0,T])$ into $H^{s}([0,T])$ is compact, as well as the inclusion $\iota_2$ from $H_0^{\epsilon}(S^1)$ into $H_0^{\epsilon'}(S^1)$. Then the dual operator $\iota_2^*: H_0^{-\epsilon'}(S^1) \rightarrow H_0^{-\epsilon}(S^1)$ is also compact. On Hilbert spaces, the tensor product of two compact operators is also compact (see e.g. \cite{Kubrusly2011} \footnote{In \cite{Kubrusly2011}, the result is stated for endomorphisms, but this extra assumption is not used in the proof.}), so that $\iota_1\otimes \iota_2^*$ is compact indeed.
\end{proof}

\begin{lemma}\label{lemma:boundedness}
    For all $s \in [0,1/2)$ and all $\epsilon > s$, it holds that $\sup_{n \in \mathbb{N}} \mathbb{E} \left( ||\log p_{U(n)}||^2_{s,-\epsilon} \right) < \infty$.
\end{lemma}
\noindent \textbf{Proof:} We see that
\begin{align}
    \mathbb{E} \left( ||\log p_{U(n)}||^2_{s,-\epsilon} \right) =& \mathbb{E} \left( \int_0^T ||\log p_{U(n)}(\cdot , t)||_{-\epsilon}^2 \text{d}t \right) \\
    &+ 1_{s \neq 0} \mathbb{E} \left( \int_0^T \int_0^T \frac{||  \log p_{U(n)}(\cdot, t) -  \log p_{U(n)}(\cdot, u)||^2_{-\epsilon}}{|t-u|^{2s + 1}} \text{d}t \text{d}u \right). \nonumber
\end{align}
For the first summand it holds that (with $k \wedge n$ denoting $\min \{k,n\}$)
\begin{align}
\begin{split}
\mathbb{E} \left( \int_0^T ||\log p_{U(n)}(\cdot , t)||^2_{-\epsilon} \text{d}t \right) =& \int_0^T \mathbb{E} \left( \sum_{k=1}^\infty k^{-2\epsilon} \frac{ \left| \Tr (U_t^k) \right|^2}{k^2} \right) \text{d}t \\
=& T \sum_{k = 1}^\infty k^{-2 - 2\epsilon} \mathbb{E} \left( \left| \Tr (U_0^k) \right|^2 \right) \\
=& T \sum_{k = 1}^\infty k^{-2 - 2\epsilon} (k \wedge n) \\
< & T \sum_{k = 1}^\infty k^{-1-2\epsilon} < \infty.
\end{split}
\end{align}
For the second summand it holds that:
\begin{align}
\label{eq:split}
    &\hspace{-2cm}\mathbb{E} \Big( \int_0^T \int_0^T \frac{||  \log p_{U(n)}(\cdot, t) -  \log p_{U(n)}(\cdot, u)||^2_{-\epsilon}}{|t-u|^{2s + 1}} \text{d}t \text{d}u \Big) \nonumber \\
    =& \sum_{k = 1}^\infty k^{-2 - 2\epsilon} \int_0^T \int_0^T \frac{ \mathbb{E} \left( |\Tr (U_t^k - U_u^k|^2 \right) }{|t - u|^{2s+1}} \text{d}t \text{d}u \nonumber \\
    \leq & CT\sum_{k = 1}^\infty k^{-2 - 2\epsilon} \int_0^T \frac{ \mathbb{E} \left( |\Tr (U_t^k - U_0^k)|^2 \right) }{t^{2s+1}} \text{d}t \nonumber \\
    \leq & CT\sum_{k = 1}^\infty k^{-2 - 2\epsilon} \int_0^{k^{-1}} \frac{ \mathbb{E} \left( |\Tr (U_t^k - U_0^k)|^2 \right) }{t^{2s+1}} \text{d}t\\
    & + CT\sum_{k = 1}^\infty k^{-2 - 2\epsilon} \int_{k^{-1}}^{\infty} \frac{ 4 \mathbb{E} \left( |\Tr (U_0^k )|^2 \right) }{t^{2s+1}} \text{d}t. \nonumber
\end{align}
For the second summand in (\ref{eq:split}) we get
\[
\int_{k^{-1}}^{\infty} \frac{ 4 \mathbb{E} \left( |\Tr (U_0^k )|^2 \right) }{t^{2s+1}} \text{d}t
= 8 s (n\wedge k) k^{2s},
\]
which is sufficient since $\sum_{k = 1}^\infty k^{-2 - 2\epsilon+1+2s}$ is finite as soon as
$s < \epsilon $. \\

For the first sum in \eqref{eq:split} we use Corollary \ref{thm:multi time} with $f(z) = g(z) = z^k$, which implies that for all $k \geq 1$
\begin{align}
\begin{split}
    \mathbb{E}\left( \Tr (U_t^k) \overline{\Tr (U_0^k)} \right) =& 1_{k < n} e^{-kt} \frac{ \sinh (\frac{k^2t}{n})}{\sinh (\frac{kt}n)} + 1_{k \geq n} e^{-\frac{k^2t}{n}} \frac{\sinh(kt)}{\sinh(\frac{kt}{n})} \\
    =& e^{-\frac{k(k \vee n) t}{n}} \frac{\sinh \left( \frac{k(k\wedge n)t}{n} \right)}{\sinh \left( \frac{kt}{n} \right)},
\end{split}
\end{align}
with $k \vee n := \max \{k, n\}$. Using this, and the fact that $\sinh x \geq x$ and $1/\sinh x \geq 1/x - x/6$ for all $x > 0$, we see that for $t < k^{-1}$:
\begin{align}
\begin{split}
    &\mathbb{E} \left( |\Tr (U_t^k - U_0^k)|^2 \right) \\
    =& \mathbb{E} \left( |\Tr (U_0^k)|^2 \right)+ \mathbb{E} \left( |\Tr (U_t^k)|^2 \right) - 2\mathbb{E} \left( \Tr (U_t^k) \overline{\Tr (U_0^k)} \right) \\
    =& 2(k \wedge n) - 2e^{-\frac{k(k \vee n)t}{n}} \frac{\sinh \left( \frac{k(k\wedge n)t}{n} \right)}{\sinh \left( \frac{kt}{n} \right)} \\
    \leq & 2(k \wedge n) - 2e^{-\frac{k(k \vee n)t}{n}} \frac{k(k\wedge n)t}{n} \left( \left(\frac{kt}{n} \right)^{-1} -  \frac{kt}{6n} \right) \\
    = & 2(k \wedge n) - 2e^{-\frac{k(k \vee n)t}{n}} \left( k \wedge n - \frac{k^2t^2 (k \wedge n)}{6n^2} \right) \\
    =&  2e^{-\frac{k(k \vee n)t}{n}} \frac{k^2 t^2 (k \wedge n)}{6n^2} + 2(k \wedge n )(1 - e^{-\frac{k(k \vee n)t}{n}}) \\
    \leq & 2k^3t^2 + 2(k \wedge n) \frac{k(k \vee n)t}{n} \\
    \leq & 4k^2t.
\end{split}
\end{align}
Thus we see that when $s < 1/2$, the first sum in \eqref{eq:split} is bounded by $
    CT \sum_{k = 1}^\infty k^{-1 - 2\epsilon + 2s}$,
which is finite for $s < \epsilon$. This finishes the proof. \qed 

\begin{lemma} \label{lemma:optimal}
    For $s = 1/2$ or $s \geq \epsilon \geq 0$ the field $Z$ is almost surely not an element of $H^{s}([0,T])\otimes  H^{-\epsilon}_0(S^1)$. 
\end{lemma}
\begin{proof}
    We see that
    \begin{align} \label{eq:expected norm finite}
    \begin{split}
        || Z ||_{s,-\epsilon}^2  =& \sum_{k = 1}^\infty k^{-2\epsilon} \int_0^T |A_k(t)|^2 \text{d}t + 1_{s \neq 0} \sum_{k = 1}^\infty k^{-2\epsilon} \int_0^T \int_0^T \frac{|A_k(t) - A_k(u)|^2}{|t - u|^{1 + 2s}} \text{d}t \text{d}u.
    \end{split}
    \end{align}
    For $s = \epsilon = 0$ we see that
    \begin{align}
        ||Z||^2_{0,0} = \int_0^T ||Z(t,\cdot)||_0^2 \text{d}t.
    \end{align}
    If this were almost surely finite, then $\mathbb{P}\left( \text{for almost all } t: \, ||Z(t,\cdot)||_0 < \infty  \right) = 1$, and  
    \begin{align}
        T = \mathbb{E} \left( \int_0^T 1_{||Z(t,\cdot)||_0 < \infty} \text{d}t \right) = \int_{0}^T \mathbb{P}\left( ||Z(t,\cdot)||_0 < \infty \right) \text{d}t,
    \end{align}
    which is a contradiction since for all $t$ we know that $\mathbb{P}(||Z(t,\cdot)||_0 < \infty) = 0$. Thus we see that almost surely $||Z||_{0,0} = \infty$. \\    
    
    For $s = \epsilon > 0$ we use the fact that stationary Ornstein-Uhlenbeck processes $A_k(t)$, $k \in \mathbb{N}$, can be represented as scaled, time-changed Brownian motions:
    \begin{align}
        A_k(t) = \frac{e^{-kt}}{\sqrt{2k}} \left( B_k(e^{2kt}) + i\tilde{B}_k(e^{2kt}) \right),
    \end{align}
    where $B_k$, $\tilde{B}_k$, $k \in \mathbb{N}$, are real standard Brownian motions. Thus it follows that the second summand in $||Z||_{s,-\epsilon}^2$ is bounded below by
    \begin{align}
        & \frac{1}{2} \sum_{k = 1}^\infty k^{-1-2\epsilon} \int_0^T \int_0^T \frac{|e^{-kt} B_k(e^{2kt}) - e^{-ku} B_k(e^{2ku})|^2}{|t - u|^{1 + 2s}} \text{d}t \text{d}u \nonumber \\ 
        =& \frac{1}{2} \sum_{k = 1}^\infty k^{-1-2\epsilon + 2s} \int_0^{kT} \int_0^{kT} \frac{|e^{-t} B_k(e^{2t}) - e^{-u} B_k(e^{2u})|^2}{|t - u|^{1 + 2s}} \text{d}t \text{d}u \\
        \geq & \frac{1}{2} \sum_{k = 1}^\infty k^{-1} \int_0^{T} \int_0^{T} \frac{|e^{-t} B_k(e^{2t}) - e^{-u} B_k(e^{2u})|^2}{|t - u|^{1 + 2s}} \text{d}t \text{d}u. \nonumber 
    \end{align}
We see that $Y_k :=\int_0^{T} \int_0^{T} \frac{|e^{-t} B_k(e^{2t}) - e^{-u} B_k(e^{2u})|^2}{|t - u|^{1 + 2s}} \text{d}t \text{d}u$, $k \in \mathbb{N}$, are an i.i.d. family of almost surely positive random variables. Let $\delta>0$ be such that $\mathbb{P}(Y_k\geq \delta)\geq \frac{1}{2}$, and set $Z_k= \delta 1_{Y_k\geq \delta}$, so that $Y_k\geq Z_k$ for all $k$. Notice that $(\delta-Z_k)$ is stochastically dominated by $Z_k$, so $\mathbb{P}( \sum_{k = 1}^\infty k^{-1} Z_k <\infty)\leq  \mathbb{P}( \sum_{k = 1}^\infty k^{-1} (\delta-Z_k) <\infty)  $. Since $\sum_{k = 1}^\infty k^{-1} =+ \infty$, these two events cannot happen at the same time, thus $\mathbb{P}( \sum_{k = 1}^\infty k^{-1} Z_k <\infty)\leq \frac{1}{2}$. By Kolmogorov's zero-one law,  $\mathbb{P}( \sum_{k = 1}^\infty k^{-1} Z_k <\infty)=0$. Since $Y_k\geq Z_k$, we deduce that, for $\epsilon\leq s$, $\| Z\|^2_{s,-\epsilon} \geq \sum_{k = 1}^\infty k^{-1} Y_k= + \infty$ almost surely. \\    

For $s\geq\frac{1}{2}$, we consider the event $E$ in which $Z$ is an element of $H^s([0,T]) \otimes H_0^{-\epsilon}(S^1)$. Then, since $\theta\mapsto e^{-i\theta}$ is in $H^\epsilon_0(S^1)$, by duality in $L^2(S^1)$, we deduce that the function $t\mapsto \int_0^{2\pi} Z(s,\theta) e^{-i \theta} \text{d}\theta$ belongs to $H^{\frac{1}{2}}([0,T])$. This function is nothing but $A_1$, which is an Ornstein-Uhlenbeck process. By Girsanov's theorem, its distribution on a bounded set $[0,T]$ is absolutely continuous with respect to the one of Brownian motion restricted to $[0,T]$. Thus, the probability of $E$ is equal to one if and only if Brownian motion almost surely belongs to $H^{\frac{1}{2}}([0,T])$, which is known not to be the case.
\end{proof}

\section{A Wick-type identity} \label{section:Wick}

In this section we prove Proposition \ref{prop:orbits} below, which is about expectations of the form
\begin{align}
&\mathbb{E} \Big( \text{Tr}\left( H U^{\sigma_1} H U^{\sigma_2} \dots H U^{\sigma_j} \right) \overline{\text{Tr}\left( H U^{\sigma_1} H U^{\sigma_2} \dots H U^{\sigma_j} \right)} \Big),
\end{align}
where $\sigma_1,...,\sigma_j \in \mathbb{Z}$, $U \in U(n)$ is Haar-distributed and independent from $H$, which is a $GUE(n)$ matrix, i.e. $H_{ii} \sim \mathcal{N}(0,1)$ for $i = 1,...,n$, and $\Re H_{ij} = \Re H_{ji} \sim \mathcal{N}(0,1/2)$, $\Im H_{ij} = - \Im H_{ji} \sim \mathcal{N}(0,1/2)$ for $1 \leq i < j \leq n$, with entries being independent up to the Hermitian symmetry. 

Such expressions appear rather naturally when we consider powers of a unitary Brownian motion $U$, since for small $s$ it holds that $U_{t+s} \simeq (1+\sqrt{2s}H)U_t$, see (\ref{eqn:UBM SDE}). \\

Let $C_{2j} = \{ \pi\in S_{2j}: \pi^2=\text{Id}, \forall \ell \in \{1,\dots, 2j\}, \pi(\ell)\neq \ell \}$, i.e. $C_{2j}$ is the set of pairings on $\{1,...,2j\}$. Using independence of $U$ and $H$, and Isserlis' theorem (also called Wick's probability theorem), we see that
\begin{align} \label{eqn:C2j}
    &\mathbb{E} \Big( \text{Tr}\left( H U^{\sigma_1} H U^{\sigma_2} \dots H U^{\sigma_j} \right) \overline{\text{Tr}\left( H U^{\sigma_1} H U^{\sigma_2} \dots H U^{\sigma_j} \right)} \Big) \nonumber \\
    =& \sum_{i_1,...,i_{4j}} \mathbb{E} \left[ H_{i_1i_2} H_{i_3i_4} \cdots H_{i_{2j-1} i_{2j}} H_{i_{2j+1} i_{2j+2}} H_{i_{2j+3} i_{2j+4}} \cdots H_{i_{4j-1} i_{4j}}\right] \\
    & \times \mathbb{E} \left[ (U^{\sigma_1})_{i_2i_3} (U^{\sigma_2})_{i_4i_5} \cdots (U^{\sigma_j})_{i_{2j}i_1} \overline{(U^{\sigma_1})_{i_{2j+1}i_{2j+4}} (U^{\sigma_2})_{i_{2j+3}i_{2j+6}} \cdots (U^{\sigma_j})_{i_{4j-1}i_{2j+2} }}\right] \nonumber \\
    =& \sum_{i_1,...,i_{4j}} \sum_{\pi\in C_{2j}} 1_{\forall \ell \in \{1,\dots, 2j\}, (i_{2\ell-1},i_{2\ell})=(i_{2\pi(\ell)},i_{2\pi(\ell)-1} ) } \nonumber \\
    & \times \mathbb{E} \left[ (U^{\sigma_1})_{i_2i_3} (U^{\sigma_2})_{i_4i_5} \cdots (U^{\sigma_j})_{i_{2j}i_1} \overline{(U^{\sigma_1})_{i_{2j+1}i_{2j+4}} (U^{\sigma_2})_{i_{2j+3}i_{2j+6}} \cdots (U^{\sigma_j})_{i_{4j-1}i_{2j+2} }} \right]. \nonumber
\end{align}

The condition $(i_{2\ell-1},i_{2\ell})=(i_{2\pi(\ell)},i_{2\pi(\ell)-1}) \, \, \forall \ell \in \{1,\dots, 2j\} \, \, \forall i_1,...,i_{4j} \in \{1,...,n\}$ allows us to define a map $\pi \mapsto \tilde{\pi}$ from $C_{2j}$ to $C_{4j}$ by the formula
\begin{align} \label{eqn:pi tilde}
    \tilde{\pi}(2\ell-1) =& 2\pi(\ell), \quad \tilde{\pi}(2\ell) = 2 \pi(\ell) - 1, \quad \ell = 1,...,2j.
\end{align}
Further we define the pairing $\rho \in C_{4j}$ as
\begin{align} \label{eqn:rho}
\begin{split}
    \rho :=& (23)(45) \cdots (2j, 1) (2j+1, 2j+4) (2j+3, 2j+6) \\
    & \hspace{2cm} \cdots (2j+2\ell-1, 2j+2\ell+2) \cdots (4j-1, 2j+2).
\end{split}
\end{align}
See Example \ref{example:j = 2} for a list of the pairings $\pi$, $\tilde{\pi}$ and $\rho$, for $j=2$, and Figure 5.1 for their depiction. \\

Note that $\rho$ and all pairings $\tilde{\pi}$ pair even numbers with odd numbers, thus $\tilde{\pi} \rho$ maps even numbers to even numbers and odd numbers to odd numbers.  Using the pairing $\tilde{\pi}$, the even numbers $i_2,i_4,...i_{4j}$ determine all the odd ones. Thus we see that
\begin{align} \label{eqn:C2j to sigma hat}
    &\sum_{i_1,...,i_{4j}} \sum_{\pi\in C_{2j}} 1_{\forall \ell \in \{1,\dots, 2j\}, (i_{2\ell-1},i_{2\ell})=(i_{2\pi(\ell)},i_{2\pi(\ell)-1} ) } \nonumber \\
    &\times \mathbb{E} \left[ (U^{\sigma_1})_{i_2i_3} (U^{\sigma_2})_{i_4i_5} \cdots (U^{\sigma_j})_{i_{2j}i_1} \overline{(U^{\sigma_1})_{i_{2j+1}i_{2j+4}} (U^{\sigma_2})_{i_{2j+3}i_{2j+6}} \cdots (U^{\sigma_j})_{i_{4j-1}i_{2j+2} }} \right] \nonumber \\
    =&\sum_{\pi\in C_{2j}} \sum_{i_1,...,i_{4j}} {1}_{\forall \ell \in \{1,\dots, 4j\}, i_{l} = i_{\tilde{\pi}(\ell)}} \\
    &\times \mathbb{E} \left[ (U^{\sigma_1})_{i_2i_{\rho(2)}} \cdots (U^{\sigma_j})_{i_{2j}i_{\rho(2j)}} \overline{(U^{\sigma_1})_{i_{\rho(2j+4)} i_{2j+4}} \cdots (U^{\sigma_j})_{i_{\rho(2j+2)}i_{2j + 2} }} \right] \nonumber \\
    =&\sum_{\pi\in C_{2j}} \sum_{i_2,i_4,...,i_{4j}} \mathbb{E} \left[ (U^{\sigma_1})_{i_2i_{\tilde{\pi} \rho(2)}} \cdots (U^{\sigma_j})_{i_{2j}i_{\tilde{\pi} \rho(2j)}} \overline{(U^{\sigma_1})_{i_{ \tilde{\pi} \rho(2j+4)} i_{2j+4}} \cdots (U^{\sigma_j})_{i_{\tilde{\pi} \rho(2j+2)}i_{2j + 2} }} \right] \nonumber \\
    =& \sum_{\pi \in C_{2j}} \mathbb{E} \left( \sum_{i_2,i_4,...i_{4j}} \prod_{\ell = 1}^{2j} (U^{\hat{\sigma}_{l}})_{i_{2\ell} i_{\tilde{\pi}\rho(2\ell)}} \right), \nonumber
\end{align}
where
\begin{align} \label{eqn:sigma hat2}
    \hat{\sigma}_\ell = \begin{cases} \sigma_\ell, & \ell = 1,2,...,j, \\ -\sigma_{\ell-j-1}, & \ell= j + 2,...,2j, \\ -\sigma_{j}, & \ell= j+1. \end{cases}
\end{align}

By repeatedly applying $\tilde{\pi}\rho$ to $\{2,4,...,4j\}$, we get a partition of $\{2,4,...,4j\}$ into orbits. The set of these orbits we denote by $\mathcal{O}_{\tilde{\pi}\rho}$. We see that

\begin{align} \label{eqn:sigma hat to orbits}
\begin{split}
    & \sum_{\pi \in C_{2j}} \mathbb{E} \left( \sum_{i_2,i_4,...i_{4j}} \prod_{\ell = 1}^{2j} (U^{\hat{\sigma}_{2\ell}})_{i_{2\ell} i_{\tilde{\pi}\rho(2\ell)}} \right) \\
    =& \sum_{\pi \in C_{2j}} \mathbb{E} \left( \sum_{i_2,i_4,...i_{4j}} \prod_{o \in \mathcal{O}_{\tilde{\pi}\rho}} \prod_{w \in o} (U^{\hat{\sigma}_{w}})_{i_{w} i_{\tilde{\pi}\rho(w)}} \right) \\
    =& \sum_{\pi \in C_{2j}} \mathbb{E} \left( \prod_{o \in \mathcal{O}_{\tilde{\pi}\rho}} \text{Tr}\left( \prod_{w \in o} U^{\hat{\sigma}_w} \right) \right) \\
    =& \sum_{\pi \in C_{2j}} \mathbb{E} \left( \prod_{o \in \mathcal{O}_{\tilde{\pi}\rho}} \text{Tr}\left( U^{\sum_{w \in o} \hat{\sigma}_w} \right) \right).
\end{split}
\end{align}

Putting together (\ref{eqn:C2j}), (\ref{eqn:C2j to sigma hat}), (\ref{eqn:sigma hat2}) and (\ref{eqn:sigma hat to orbits}), we have proven the following proposition:

\begin{proposition} \label{prop:orbits}
Let $H$ be an $n \times n$ matrix from the $GUE(n)$, and let $U \in U(n)$ be independent and Haar-distributed. Let $C_{2j} = \{ \pi\in S_{2j}: \pi^2=\text{Id}, \forall \ell \in \{1,\dots, 2j\}, \pi(\ell)\neq \ell \}$, and for $\pi \in C_{2j}$ define the pairing $\tilde{\pi} \in C_{4j}$ by (\ref{eqn:pi tilde}). Furthermore let $\rho$ be as in (\ref{eqn:rho}), $\hat{\sigma}$ as in (\ref{eqn:sigma hat2}), and let $\mathcal{O}_{\tilde{\pi} \rho}$ denote the set of orbits obtained by repeatedly applying $\tilde{\pi}\rho$ to $\{2,4,...,4j\}$. Then for $j \in \mathbb{N}$ and $\sigma_1, ..., \sigma_j \in \mathbb{N}$,  it holds that
\begin{align}
\begin{split}
    &\mathbb{E} \Big( \text{Tr}\left( H U^{\sigma_1} H U^{\sigma_2} \dots H U^{\sigma_j} \right) \overline{\text{Tr}\left( H U^{\sigma_1} H U^{\sigma_2} \dots H U^{\sigma_j} \right)} \Big) \\
    =& \sum_{\pi \in C_{2j}} \mathbb{E} \left( \prod_{o \in \mathcal{O}_{\tilde{\pi}\rho}} \text{Tr}\left( U^{\sum_{w \in o} \hat{\sigma}_w} \right) \right).
\end{split}
\end{align}
\end{proposition}

\begin{example}\label{example:j = 2} For $j = 2$ we see that $\rho = (14)(23)(58)(67)$, and (see Figure 5.1)
\begin{align}
\begin{split}
    &\pi = (12)(34), \quad \tilde{\pi} = (14)(23)(58)(67), \quad \tilde{\pi} \rho = (2)(4)(6)(8),  \\
    &\pi = (13)(24), \quad \tilde{\pi} = (16)(25)(38)(47), \quad \tilde{\pi} \rho = (28)(46),  \\
    &\pi = (14)(23), \quad \tilde{\pi} = (18)(27)(36)(45), \quad \tilde{\pi} \rho = (26)(48),
\end{split}
\end{align}
and that $\hat{\sigma}_2 = \sigma_1$, $\hat{\sigma}_4 = \sigma_2$, $\hat{\sigma}_6 = -\sigma_2$ and $\hat{\sigma}_8 = - \sigma_1$. Thus from Lemma \ref{prop:orbits} it follows that
\begin{align}
\begin{split}
    &\mathbb{E} \Big( \text{Tr}\left( H U^{\sigma_1} H U^{\sigma_2} \right) \overline{\text{Tr}\left( H U^{\sigma_1} H U^{\sigma_2} \right)} \Big) \\
    =& \mathbb{E} \left( \text{Tr}U^{\sigma_1} \text{Tr}U^{\sigma_2} \text{Tr}U^{-\sigma_2} \text{Tr}U^{-\sigma_1} \right) \\
    &+\mathbb{E} \left( \text{Tr}U^{\sigma_1-\sigma_1} \text{Tr}U^{\sigma_2-\sigma_2} \right) \\
    &+\mathbb{E} \left( \text{Tr}U^{\sigma_1-\sigma_2} \text{Tr}U^{\sigma_2-\sigma_1} \right) \\
    =& \begin{cases} 2\sigma_1^2 + n^2 + n^2 & \sigma_1 = \sigma_2 \\ \sigma_1\sigma_2 + n^2 + |\sigma_1 - \sigma_2| & \sigma_1 \neq \sigma_ 2 \end{cases},
\end{split}
\end{align}
where the last equality holds for large enough $n$ by Theorem \ref{thm:DiaconisShahshahani}.
\end{example}

\begin{figure}[H] \label{figure:j = 2}
\center
\begin{tikzpicture}[scale = 0.8]
\node at ({2*cos(0)-6},{2*sin(0)}) [circle,fill,inner sep=1.5pt]{};
\fill ({2.5*cos(0)-6},{2.5*sin(0)}) node[] {$3$};
\node at ({2*cos(45)-6},{2*sin(45)}) [circle,fill,inner sep=1.5pt]{};
\fill ({2.5*cos(45)-6},{2.5*sin(45)}) node[] {$2$};
\node at ({2*cos(90)-6},{2*sin(90)}) [circle,fill,inner sep=1.5pt]{};
\fill ({2.5*cos(90)-6},{2.5*sin(90)}) node[] {$1$};
\node at ({2*cos(135)-6},{2*sin(135)}) [circle,fill,inner sep=1.5pt]{};
\fill ({2.5*cos(135)-6},{2.5*sin(135)}) node[] {$8$};
\node at ({2*cos(180)-6},{2*sin(180)}) [circle,fill,inner sep=1.5pt]{};
\fill ({2.5*cos(180)-6},{2.5*sin(180)}) node[] {$7$};
\node at ({2*cos(225)-6},{2*sin(225)}) [circle,fill,inner sep=1.5pt]{};
\fill ({2.5*cos(225)-6},{2.5*sin(225)}) node[] {$6$};
\node at ({2*cos(270)-6},{2*sin(270)}) [circle,fill,inner sep=1.5pt]{};
\fill ({2.5*cos(270)-6},{2.5*sin(270)}) node[] {$5$};
\node at ({2*cos(315)-6},{2*sin(315)}) [circle,fill,inner sep=1.5pt]{};
\fill ({2.5*cos(315)-6},{2.5*sin(315)}) node[] {$4$};

\draw ({2*cos(90)-6.05},{2*sin(90)}) -- ({2*cos(315)-6.05},{2*sin(315)});
\draw ({2*cos(0)-6.05},{2*sin(0)}) -- ({2*cos(45)-6.05},{2*sin(45)});
\draw ({2*cos(135)-6.05},{2*sin(135)}) -- ({2*cos(270)-6.05},{2*sin(270)});
\draw ({2*cos(180)-6.05},{2*sin(180)}) -- ({2*cos(225)-6.05},{2*sin(225)});

\draw[red] ({2*cos(90)-5.95},{2*sin(90)}) -- ({2*cos(315)-5.95},{2*sin(315)});
\draw[red] ({2*cos(0)-5.95},{2*sin(0)}) -- ({2*cos(45)-5.95},{2*sin(45)});
\draw[red] ({2*cos(135)-5.95},{2*sin(135)}) -- ({2*cos(270)-5.95},{2*sin(270)});
\draw[red] ({2*cos(180)-5.95},{2*sin(180)}) -- ({2*cos(225)-5.95},{2*sin(225)});

\node at ({2*cos(0)},{2*sin(0)}) [circle,fill,inner sep=1.5pt]{};
\fill ({2.5*cos(0)},{2.5*sin(0)}) node[] {$3$};
\node at ({2*cos(45)},{2*sin(45)}) [circle,fill,inner sep=1.5pt]{};
\fill ({2.5*cos(45)},{2.5*sin(45)}) node[] {$2$};
\node at ({2*cos(90)},{2*sin(90)}) [circle,fill,inner sep=1.5pt]{};
\fill ({2.5*cos(90)},{2.5*sin(90)}) node[] {$1$};
\node at ({2*cos(135)},{2*sin(135)}) [circle,fill,inner sep=1.5pt]{};
\fill ({2.5*cos(135)},{2.5*sin(135)}) node[] {$8$};
\node at ({2*cos(180)},{2*sin(180)}) [circle,fill,inner sep=1.5pt]{};
\fill ({2.5*cos(180)},{2.5*sin(180)}) node[] {$7$};
\node at ({2*cos(225)},{2*sin(225)}) [circle,fill,inner sep=1.5pt]{};
\fill ({2.5*cos(225)},{2.5*sin(225)}) node[] {$6$};
\node at ({2*cos(270)},{2*sin(270)}) [circle,fill,inner sep=1.5pt]{};
\fill ({2.5*cos(270)},{2.5*sin(270)}) node[] {$5$};
\node at ({2*cos(315)},{2*sin(315)}) [circle,fill,inner sep=1.5pt]{};
\fill ({2.5*cos(315)},{2.5*sin(315)}) node[] {$4$};

\draw ({2*cos(90)},{2*sin(90)}) -- ({2*cos(315)},{2*sin(315)});
\draw ({2*cos(0)},{2*sin(0)}) -- ({2*cos(45)},{2*sin(45)});
\draw ({2*cos(135)},{2*sin(135)}) -- ({2*cos(270)},{2*sin(270)});
\draw ({2*cos(180)},{2*sin(180)}) -- ({2*cos(225)},{2*sin(225)});

\draw[red] ({2*cos(90)},{2*sin(90)}) -- ({2*cos(225)},{2*sin(225)});
\draw[red] ({2*cos(45)},{2*sin(45)}) -- ({2*cos(270)},{2*sin(270)});
\draw[red] ({2*cos(0)},{2*sin(0)}) -- ({2*cos(135)},{2*sin(135)});
\draw[red] ({2*cos(315)},{2*sin(315)}) -- ({2*cos(180)},{2*sin(180)});

\node at ({2*cos(0)+6},{2*sin(0)}) [circle,fill,inner sep=1.5pt]{};
\fill ({2.5*cos(0)+6},{2.5*sin(0)}) node[] {$3$};
\node at ({2*cos(45)+6},{2*sin(45)}) [circle,fill,inner sep=1.5pt]{};
\fill ({2.5*cos(45)+6},{2.5*sin(45)}) node[] {$2$};
\node at ({2*cos(90)+6},{2*sin(90)}) [circle,fill,inner sep=1.5pt]{};
\fill ({2.5*cos(90)+6},{2.5*sin(90)}) node[] {$1$};
\node at ({2*cos(135)+6},{2*sin(135)}) [circle,fill,inner sep=1.5pt]{};
\fill ({2.5*cos(135)+6},{2.5*sin(135)}) node[] {$8$};
\node at ({2*cos(180)+6},{2*sin(180)}) [circle,fill,inner sep=1.5pt]{};
\fill ({2.5*cos(180)+6},{2.5*sin(180)}) node[] {$7$};
\node at ({2*cos(225)+6},{2*sin(225)}) [circle,fill,inner sep=1.5pt]{};
\fill ({2.5*cos(225)+6},{2.5*sin(225)}) node[] {$6$};
\node at ({2*cos(270)+6},{2*sin(270)}) [circle,fill,inner sep=1.5pt]{};
\fill ({2.5*cos(270)+6},{2.5*sin(270)}) node[] {$5$};
\node at ({2*cos(315)+6},{2*sin(315)}) [circle,fill,inner sep=1.5pt]{};
\fill ({2.5*cos(315)+6},{2.5*sin(315)}) node[] {$4$};

\draw ({2*cos(90)+6},{2*sin(90)}) -- ({2*cos(315)+6},{2*sin(315)});
\draw ({2*cos(0)+6},{2*sin(0)}) -- ({2*cos(45)+6},{2*sin(45)});
\draw ({2*cos(135)+6},{2*sin(135)}) -- ({2*cos(270)+6},{2*sin(270)});
\draw ({2*cos(180)+6},{2*sin(180)}) -- ({2*cos(225)+6},{2*sin(225)});

\draw[red] ({2*cos(90)+6},{2*sin(90)}) -- ({2*cos(135)+6},{2*sin(135)});
\draw[red] ({2*cos(45)+6},{2*sin(45)}) -- ({2*cos(180)+6},{2*sin(180)});
\draw[red] ({2*cos(0)+6},{2*sin(0)}) -- ({2*cos(225)+6},{2*sin(225)});
\draw[red] ({2*cos(315)+6},{2*sin(315)}) -- ({2*cos(270)+6},{2*sin(270)});

\end{tikzpicture}
\caption{The pairing $\rho$ is in black, the three pairings $\tilde{\pi}$ in $\tilde{C}_{8}$ are in red.}
\end{figure}
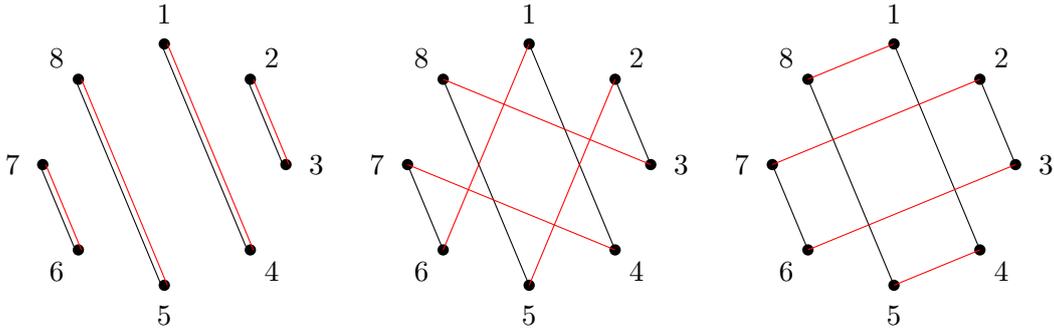

%\section*{Acknowledgments}
%Our work was supported by ERC Advanced Grant 740900 (LogCorRM). Additionally, Isao Sauzedde was funded by the EPSRC grant EP/W006227/1 during the later stage of the writing process. We are most grateful to Jon Keating and Hugo Falconet for very helpful discussions and suggestions, and to Thierry Lévy for his valuable comments.

%now enable appendix numbering format and include any appendices
%\appendix
%\include{appendix1}
%\include{appendix2}

%next line adds the Bibliography to the contents page
\addcontentsline{toc}{chapter}{Bibliography}
%uncomment next line to change bibliography name to references
%\renewcommand{\bibname}{References}
\bibliography{main.bbl}       
\bibliographystyle{alpha}  %use the plain bibliography style

\end{document}